\documentclass[a4paper,12pt]{amsart}
\usepackage[hmarginratio={1:1},vmarginratio={1:1},heightrounded,textwidth=455pt]{geometry}
\usepackage{amsfonts}
\usepackage{amsmath}
\usepackage{amssymb, colonequals}
\usepackage{hyperref}
\usepackage{tikz}
\usetikzlibrary{arrows}
\usetikzlibrary{patterns}
\usetikzlibrary{shapes,snakes}
\usetikzlibrary{arrows}
\usetikzlibrary{patterns}
\usetikzlibrary{shapes,snakes}
\usetikzlibrary{calc}
\usepackage[textsize=small]{todonotes}
\usepackage{url}
\usepackage{enumitem}
\usepackage[vlined,norelsize,ruled,linesnumbered]{algorithm2e}
\usepackage{multirow}

\newcommand{\di}{\hspace{0.1cm} \textbf{\textit{di}} \hspace{0.1cm} }
\newcommand{\cs}{\hspace{0.1cm} \textbf{\textit{cs}} \hspace{0.1cm} }
\newcommand{\cd}{\hspace{0.1cm} \textbf{\textit{cd}} \hspace{0.1cm} }
\newcommand{\cc}{\hspace{0.1cm} \textbf{\textit{cc}} \hspace{0.1cm} }
\newcommand{\ov}{\hspace{0.1cm} \textbf{\textit{ov}} \hspace{0.1cm} }
\newcommand{\bdi}{\textbf{\textit{di}}}
\newcommand{\bcs}{\textbf{\textit{cs}}}
\newcommand{\bcd}{\textbf{\textit{cd}}}
\newcommand{\bcc}{\textbf{\textit{cc}}}
\newcommand{\bov}{\textbf{\textit{ov}}}

\newcommand{\inside}{inside}

\newcommand{\leftside}{left}
\newcommand{\rightside}{right}
\let\leq\leqslant
\let\geq\geqslant
\let\setminus\smallsetminus

\let\rho\varrho

\makeatletter
\newcommand\ie{i.e\@ifnextchar.{}{.\@}}
\newcommand\etc{etc\@ifnextchar.{}{.\@}}
\newcommand\etal{et~al\@ifnextchar.{}{.\@}}
\makeatother

\newcounter{dummy} 
\numberwithin{dummy}{section}
\newtheorem{theorem}[dummy]{Theorem}
\newtheorem{claim}[dummy]{Claim}
\newtheorem{lemma}[dummy]{Lemma}

\newtheorem{definition}[dummy]{Definition}

\newcounter{hackcount}

\title[Testing isomorphism of circular-arc graphs]{Testing isomorphism of circular-arc graphs -- Hsu's approach revisited}

\thanks{
Tomasz Krawczyk is partially supported by the Polish National Science Center grant UMO-2015/17/B/ST6/01873.
}

\author[T.~Krawczyk]{Tomasz Krawczyk}
\address[T.~Krawczyk]{Theoretical Computer Science Department, Faculty of Mathematics and Computer Science, Jagiellonian University, Krak\'ow, Poland}
\email{krawczyk@tcs.uj.edu.pl}
\date{}

\begin{document}

\thispagestyle{empty}

\maketitle

\begin{abstract}
Circular-arc graphs are intersection graphs of arcs on the circle.
The aim of our work is to present a polynomial time algorithm testing whether two 
circular-arc graphs are isomorphic.
To accomplish our task we construct decomposition trees, which are
the structures representing all normalized intersection models of circular-arc graphs.
Normalized models reflect the neighbourhood relation in circular-arc graphs
and can be seen as their canonical representations;
in particular, every intersection model can be easily transformed into a normalized one.
Decomposition trees generalize PQ-trees, which are the structures that represent all intersection models of interval graphs.

Our work adapts and appropriately extends the previous work on the similar topic done by Hsu [\emph{SIAM J. Comput. 24(3), 411--439, (1995)}].
In his work, Hsu developed decomposition trees representing all normalized models
of circular-arc graphs.
However due to the counterexample given in [\emph{Discrete Math. Theor. Comput. Sci., 15(1), 157--182, 2013}], his decomposition trees can not be used by algorithms testing isomorphism of circular-arc graphs.

\end{abstract}

\section{Introduction}
\label{sec:introduction}
Circular-arc graphs are intersection graphs of arcs on the circle. 
Circular-arc graphs generalize interval graphs, 
which are the intersection graphs of intervals on a real line.
Usually, the problems for circular-arc graphs tend to be harder than for their interval counterparts.
A good example illustrating our remark is the problem of compiling 
the lists of minimal forbidden induced subgraphs for these classes of graphs.
For interval graphs such a list was completed by Lekkerkerker and Boland already in the 1960s \cite{LekBol62}
but for circular-arc graphs, despite a flurry of research \cite{Bon09,Fed99,Fra15,Klee69,Chih06,Tro76,Tuc70}, it is still unknown.
We refer the readers to the survey papers \cite{surDur14, surLinSchw09}, where the state of 
research on the structural properties of circular-arc graphs is outlined.

The first linear time algorithm for the recognition of interval graphs was given by Booth and Lueker \cite{BoothLueker76} in the 1970s.
A few years later, the first polynomial time algorithm for the recognition of circular-arc graphs was constructed by Tucker \cite{Tuck80}.
The complexity of this algorithm has been subsequently improved in \cite{EschSpin93, Hsu95}.
Currently, there are known at least two linear-time algorithms recognizing circular-arc graphs \cite{KapNus11, McCon03}.

In the 1970s Booth and Lueker \cite{BoothLueker76} introduced \emph{PQ-trees}, structures
that appear to be useful to represent all intersection models of interval graphs.
A few years later Lueker and Booth used PQ-trees in the construction of a linear time algorithm testing isomorphism of interval graphs~\cite{LuekerBooth79}.
On the other hand, the isomorphism problem for circular-arc graphs has been open for almost 40 years.
There are known linear algorithms solving the isomorphism problem on proper circular-arc graphs \cite{counterex13, Lin08} and co-bipartite circular-arc graphs \cite{Eschen97}.
The isomorphism problem can be solved in linear time \cite{counterex13} and logarithmic space \cite{Kob16} in the class of Helly circular-arc graphs.
The partial results for the general case have been given in \cite{Chan16}.
Only recently, the first polynomial time-algorithm for the isomorphism problem for circular-arc graphs was announced by Nedela, Ponomarenko, and Zeman \cite{NPZ19}.
We mention that their algorithm uses quite different techniques from those presented in this paper.

In 1990's Hsu claimed a theorem describing the structure of all normalized intersection models of circular-arc graphs and a polynomial time algorithm for the isomorphism problem \cite{Hsu95}.
In his work Hsu developed decomposition trees, which are structures that represent all normalized models of a circular-arc graph.
Based on his decomposition trees, Hsu proposed a polynomial time algorithm testing isomorphism of circular-arc graphs.
However, Hsu's algorithm was proven to be incorrect and a few years ago a counterexample 
for its correctness was constructed by Curtis, Lin, McConnell, Nussbaum, Soulignac, Spinrad, and Szwarcfiter~\cite{counterex13}.
In particular, decomposition trees proposed by Hsu can not be used to test whether two circular-arc graphs are isomorphic.

\subsection{Our work}
We adapt and extend Hsu's ideas appropriately and we construct refined decomposition trees 
representing all normalized models of a circular-arc graph.
To attain our goal we exploit the ideas invented by Spinrad \cite{Spin88}, 
which enabled him to reduce the recognition problem of co-bipartite circular-arc graphs
to testing whether some carefully designed posets have dimension at most two.
We extend Spinrad's ideas to the whole class of circular-arc graphs (to characterize normalized models of some parts of circular-arc graphs) and we plug them appropriately to Hsu's framework.
Eventually, we develop a decomposition tree representing all normalized models of a circular-arc graph.
Decomposition trees presented here generalize PQ-trees, which are the structures representing all intersection models of
interval graphs.
Given such decomposition trees, we propose a polynomial time algorithm for the isomorphism problem on circular-arc graphs.

Our paper is organized as follows:
\begin{itemize}
 \item In Section \ref{sec:normalized_models_and_Hsu_approach} we compare our approach and Hsu's approach to the problem of characterization of all normalized models of circular-arc graphs.
 We also quote a counterexample to the correctness of Hsu's isomorphism algorithm constructed in~\cite{counterex13}.
 \item In Section \ref{sec:preliminaries} we introduce notation used throughout the paper.
 \item In Section \ref{sec:tools} we describe all tools required to prove our results, including
 split decomposition of circle graphs, modular decomposition, and transitive orientations of graphs.
 \item In Section \ref{sec:conformal_models} we describe a decomposition tree that keeps a track of all normalized models of a circular-arc graph.
 \item In Section \ref{sec:isomorphism_problem} we present a polynomial time algorithm for the isomorphism problem on circular-arc graphs.
\end{itemize}

\section{Two approaches to the problem of characterization of all normalized models of circular-arc graphs}
\label{sec:normalized_models_and_Hsu_approach}
A \emph{circular-arc model} $\psi$ of a graph 
$G = (V,E)$ is a collection of arcs $\{\psi(v): v \in V\}$ of a given circle $C$
such that for every $u,v \in V$ we have $uv \in E$ iff $\psi(u) \cap \psi(v) \neq \emptyset$.
A graph $G$ is a \emph{circular-arc graph} if $G$ admits a circular-arc model.
In this paper we only consider circular-arc models $\psi$
in which the arcs from $\{\psi(v): v \in V\}$ have different endpoints: 
one can easily verify that any arc model of $G$ can be turned into a model that satisfies this property.

A \emph{chord model} $\phi$ of a graph 
$G = (V,E)$ is a collection of chords $\{\phi(v): v \in V\}$ of a given circle $C$
such that for every $u,v \in V$ we have $uv \in E$ iff $\phi(u) \cap \phi(v) \neq \emptyset$.
A graph $G$ is a \emph{circle graph} if $G$ admits a chord model.


Let $G$ be a circular-arc graph with no twins and no universal vertices.
Suppose $\psi$ is an arc model of $G$.
Since $G$ has no universal vertices and since the endpoints of the arcs from $\{\psi(v): v \in V\}$ are pairwise different, 
we can distinguish five possibilities describing the mutual positions of every two arcs from $\{\psi(v): v \in V\}$.
Let $(v,u)$ be a pair of distinct vertices in $G$.
We say that:
\begin{itemize}
 \item $\psi(v)$ and $\psi(u)$ are \emph{disjoint} if $\psi(v) \cap \psi(u) = \emptyset$,
 \item $\psi(v)$ \emph{contains} $\psi(u)$ if $\psi(v) \supsetneq \psi(u)$,
 \item $\psi(v)$ \emph{is contained} in $\psi(u)$ if $\psi(v) \subsetneq \psi(u)$,
 \item $\psi(v)$ and $\psi(u)$ \emph{cover the circle} if $\psi(v) \cup \psi(u) = C$,
 \item $\psi(v)$ and $\psi(u)$ \emph{overlap}, otherwise.
\end{itemize}
See Figure \ref{fig:mutual_arc_position} for an illustration.

\begin{figure}[htp!]
\begin{tikzpicture}[scale=0.5]
\coordinate (center) at (0,0) {};
\coordinate (v) at ($(center)+(90:2cm)$) {};
\coordinate (u) at ($(center)+(270:2cm)$) {};

\coordinate (lv) at ($(center)+(90:2.6cm)$) {};
\coordinate (lu) at ($(center)+(270:2.6cm)$) {};

\tikzstyle{every node}=[inner sep=1pt]
\begin{scriptsize}
\node at (lv) {$\psi(v)$};
\node at (lu) {$\psi(u)$};
\end{scriptsize}

\draw[thick] ([shift=(30:2.1cm)]0,0) arc (30:150:2.1cm);
\draw[thick] ([shift=(210:2.1cm)]0,0) arc (210:330:2.1cm);

\draw[thick, white] (-3.0,-3)--(-3.0,-2.8);
\draw[thick, white] (3.0,3)--(3.0,2.8);

\end{tikzpicture}
\begin{tikzpicture}[scale=0.5]
\coordinate (center) at (0,0) {};
\coordinate (v) at ($(center)+(90:2cm)$) {};
\coordinate (u) at ($(center)+(270:2cm)$) {};

\coordinate (lv) at ($(center)+(90:2.6cm)$) {};
\coordinate (lu) at ($(center)+(90:1.25cm)$) {};

\tikzstyle{every node}=[inner sep=1pt]
\begin{scriptsize}
\node at (lv) {$\psi(v)$};
\node at (lu) {$\psi(u)$};
\end{scriptsize}

\draw[thick] ([shift=(0:2.1cm)]0,0) arc (0:180:2.1cm);
\draw[thick] ([shift=(45:1.8cm)]0,0) arc (45:135:1.8cm);

\draw[thick, white] (-3.0,-3)--(-3.0,-2.8);
\draw[thick, white] (3.0,3)--(3.0,2.8);

\end{tikzpicture}
\begin{tikzpicture}[scale=0.5]
\coordinate (center) at (0,0) {};
\coordinate (v) at ($(center)+(90:2cm)$) {};
\coordinate (u) at ($(center)+(270:2cm)$) {};

\coordinate (lu) at ($(center)+(90:2.6cm)$) {};
\coordinate (lv) at ($(center)+(90:1.25cm)$) {};

\tikzstyle{every node}=[inner sep=1pt]
\begin{scriptsize}
\node at (lv) {$\psi(v)$};
\node at (lu) {$\psi(u)$};
\end{scriptsize}

\draw[thick] ([shift=(0:2.1cm)]0,0) arc (0:180:2.1cm);
\draw[thick] ([shift=(45:1.8cm)]0,0) arc (45:135:1.8cm);

\draw[thick, white] (-3.0,-3)--(-3.0,-2.8);
\draw[thick, white] (3.0,3)--(3.0,2.8);

\end{tikzpicture}
\begin{tikzpicture}[scale=0.5]
\coordinate (center) at (0,0) {};
\coordinate (v) at ($(center)+(90:2cm)$) {};
\coordinate (u) at ($(center)+(270:2cm)$) {};

\coordinate (lv) at ($(center)+(90:2.6cm)$) {};
\coordinate (lu) at ($(center)+(270:2.4cm)$) {};

\tikzstyle{every node}=[inner sep=1pt]
\begin{scriptsize}
\node at (lv) {$\psi(v)$};
\node at (lu) {$\psi(u)$};
\end{scriptsize}

\draw[thick] ([shift=(-20:2.1cm)]0,0) arc (-20:200:2.1cm);
\draw[thick] ([shift=(160:1.9cm)]0,0) arc (160:380:1.9cm);

\draw[thick, white] (-3.0,-3)--(-3.0,-2.8);
\draw[thick, white] (3.0,3)--(3.0,2.8);

\end{tikzpicture}
\begin{tikzpicture}[scale=0.5]
\coordinate (center) at (0,0) {};
\coordinate (v) at ($(center)+(90:2cm)$) {};
\coordinate (u) at ($(center)+(270:2cm)$) {};

\coordinate (lv) at ($(center)+(90:2.6cm)$) {};
\coordinate (lu) at ($(center)+(270:2.4cm)$) {};

\tikzstyle{every node}=[inner sep=1pt]
\begin{scriptsize}
\node at (lv) {$\psi(v)$};
\node at (lu) {$\psi(u)$};
\end{scriptsize}
\draw[thick] ([shift=(70:2.1cm)]0,0) arc (70:200:2.1cm);
\draw[thick] ([shift=(160:1.9cm)]0,0) arc (160:290:1.9cm);

\draw[thick, white] (-3.0,-3)--(-3.0,-2.8);
\draw[thick, white] (3.0,3)--(3.0,2.8);
\end{tikzpicture}

\caption{\label{fig:mutual_arc_position} From left to right:
$\psi(v)$ and $\psi(u)$ are disjoint, $\psi(v)$ contains $\psi(u)$, $\psi(v)$ is contained in $\psi(u)$,
$\psi(v)$ and $\psi(u)$ cover the circle, and $\psi(v)$ and $\psi(u)$ overlap.
}

\end{figure}
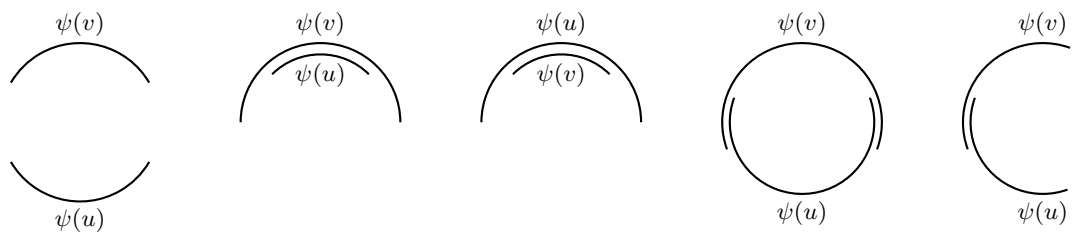

Following the ideas from \cite{Hsu95, Tuck80}, the \emph{intersection matrix} of $G$ is an $|V| \times |V|$ matrix $M_{G}$, whose rows and columns correspond to the vertices of $G$. 
Assuming that $N[v] = \{u \in V: uv \in E\} \cup \{v\}$
denotes the \emph{closed neighborhood} of the vertex $v$ in $G$,
the entries of $M_G[v,u]$ are defined such that:
$$M_G[v,u] = 
\left\{
\begin{array}{ccl}
\bdi &
\text{if} &\ vu \notin E, \\
\smallskip
\bcs
& \text{if} &N[u] \subsetneq N[v],\\
\smallskip
\bcd
& \text{ if} &N[v] \subsetneq N[u],\\
\smallskip
\bcc
&  \text{if} & 
\smallskip
\begin{array}{l}
vu \in E,\ N[v] \cup N[u] = V, \\
\text{foreach } w \in N[v]\setminus N[u] \text{ we have } N[w] \subsetneq N[v], \text{ and } \\
\text{foreach } w \in N[u]\setminus N[v] \text{ we have } N[w] \subsetneq N[u], \\
\end{array} \\
\smallskip
\bov
&& \text{ otherwise.}
\end{array}
\right.
$$
Note that the matrix $M_G$ is symmetric except that for every $u,v \in V$ we have $M_G[v,u] = \bcs \text{ iff }  M_G[u,v] = \bcd$.
In what follows we abbreviate and we write $v \bullet u$ if $M_G[v,u] = \bullet$, for $\bullet \in \{ \bdi, \bcs, \bcd, \bcc, \bov \}$. 

The intersection matrix $M_G$ encodes the relative relation between the closed neighborhoods of the vertices in the graph $G$.
The matrix $M_G$ tries to capture some relations between the entries of $M_G$ and the relative positions of the arcs in a circular-arc model of $G$.
In particular, one can easily verify that for every circular-arc model $\psi$ of $G$ and every pair of distinct vertices $(v, u)$ in $G$:
\begin{itemize}
 \item $\psi(v)$ and $\psi(u)$ are disjoint iff $v \di u$.
 \item If $\psi(v)$ contains $\psi(u)$, then $N[u] \subsetneq N[v]$, and hence $v \cs u$.
 \item If $\psi(v)$ is contained in $\psi(u)$, then $N[v] \subsetneq N[u]$, and hence $v \cd u$.
 \item If $\psi(v)$ and $\psi(u)$ cover the circle, then $N[v] \cup N[u] = V$, 
 $N[w] \subsetneq N[v]$ for every $w \in N[v] \setminus N[u]$, and $N[w] \subsetneq N[u]$
 for every $w \in N[u] \setminus N[v]$, and hence $v \cc u$. 
\end{itemize}
In so-called \emph{normalized models}, introduced by Hsu in \cite{Hsu95}, the relationship between the entries of $M_G$ and the relative positions of the arcs is even more rigid.
\begin{definition}
\label{def:normalized_model}
A circular-arc model $\psi$ of $G$ is \emph{normalized} if for every pair $(v,u)$ of distinct vertices of $G$ the following conditions are satisfied:
\begin{enumerate}
 \item \label{item:norm_di} $v \di u \iff \text{$\psi(v)$ and $\psi(u)$ are disjoint,}$
 \item \label{item:norm_cs} $v \cs u \iff \text{$\psi(v)$ contains $\psi(u)$,}$
 \item \label{item:norm_cd} $v \cd u \iff \text{$\psi(v)$ is contained in $\psi(u)$,}$
 \item \label{item:norm_cc} $v \cc u \iff \text{$\psi(v)$ and $\psi(u)$ cover the circle,}$
 \item \label{item:norm_ov} $v \ov u \iff \text{$\psi(v)$ and $\psi(u)$ overlap.}$
\end{enumerate}
\end{definition}
Every circular-arc model of $G$ fulfills \eqref{item:norm_di}, 
but it might not satisfy \eqref{item:norm_cs}, \eqref{item:norm_cd}, \eqref{item:norm_cc}, or \eqref{item:norm_ov}.
However, every circular-arc model $\psi$ of $G$ can be turned into a normalized model
by carrying out a normalization procedure on $\psi$.
The normalization procedure performs the following transformation on $\psi$ whenever there are adjacent vertices $(v,u)$ in $G$ violating \eqref{item:norm_cs}, \eqref{item:norm_cd} or \eqref{item:norm_cc}:
\begin{itemize}
 \item if $v \cs u$ but $\psi(v)$ does not contain $\psi(u)$, 
 it picks the endpoint of $\psi(v)$ contained in $\psi(u)$ and pulls it outside $\psi(u)$ so as
 $\psi(v)$ contains $\psi(u)$,
 \item if $v \cd u$ but $\psi(v)$ is not contained in $\psi(u)$, 
 it picks the endpoint of $\psi(u)$ contained in $\psi(v)$ and pulls it outside $\psi(v)$ so as $\psi(v)$ is contained in $\psi(u)$, 
 \item if $v \cc u$ but $\psi(v)$ and $\psi(u)$ do not cover the circle, 
 it picks the endpoint of $\psi(v)$ from outside $\psi(u)$ and the endpoint of $\psi(u)$ from outside $\psi(v)$ and pulls these endpoints towards each other until they pass somewhere on the circle $C$.
\end{itemize}
The above transformations keep $\psi$ a model of $G$ and, 
if performed in an appropriate order, eventually lead to a normalized model of $G$ -- see \cite{ Hsu95, Tuck80} for more details.
\begin{theorem}[\cite{Hsu95, Tuck80}]
\label{thm:circular_arc_graphs_normalized_models}
Suppose $G$ is a graph with no twins and no universal vertices.
Then, $G$ is a circular-arc graph if and only if $G$ has a normalized model.
\end{theorem}
Our goal is to describe the structure representing all normalized models of a circular-arc graph $G$.
To achieve our goal, we follow the approach taken by Hsu~\cite{Hsu95}.
We consider the overlap graph $G_{ov}$ associated with $G$ which joins with an edge every two vertices $u,v$
such that $M_G[u,v]=\bov$.
Then, we are searching for some particular chord models of $G_{ov}$, called \emph{conformal}, 
which are in one-to-one correspondence with normalized models of $G$.
Then, we describe the structure of all conformal models of $G_{ov}$,
thus obtaining a description of all normalized models of $G$.
Similarly to Hsu's work, to achieve our goals we exploit a \emph{split decomposition} of $G_{ov}$, 
a structure describing all chord models of $G_{ov}$, and a \emph{modular decomposition} of $G_{ov}$, 
a structure that appears to be appropriate to characterize all conformal models of $G_{ov}$.
Below we detail our approach.

\begin{definition}[\cite{Hsu95}]
Let $G = (V,E)$ be a circular-arc graph with no twins and no universal vertices.
The \emph{overlap graph} $G_{ov} = (V,{\sim})$ of $G$ joins with an edge $\sim$ every two vertices $u, v \in V$ such that $M_G[u,v] = M_G[v,u] = \bov$.
\end{definition}

There is a natural \emph{straightening procedure} that transforms normalized models $\psi$ of $G$ 
into \emph{oriented chord models} $\phi$ of $G_{ov}$:
it converts every arc $\psi(v)$ into an oriented chord $\phi(v)$ such that $\phi(v)$ has the same endpoint as
$\psi(v)$ and $\phi(v)$ is oriented so as it has the arc $\psi(v)$ on its left side if we traverse 
$\phi(v)$ from its tail to its head -- see Figure \ref{fig:straightening_bending} for an illustration.

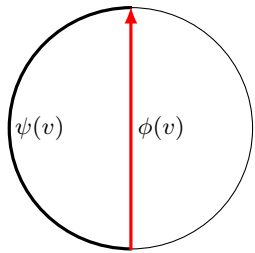
\begin{figure}[htp!]
\begin{tikzpicture}[xscale=0.80,yscale=0.80,>=latex,shorten >=-0.4pt,shorten <=-0.4pt]

\coordinate (lphiv) at (0.5,0) {};
\coordinate (lpsiv) at (-1.5,0) {};

\tikzstyle{every node}=[inner sep=1pt]
\begin{scriptsize}
\node at (lphiv) {$\phi(v)$};
\node at (lpsiv) {$\psi(v)$};
\end{scriptsize}

\draw (0,0) circle (2cm);
\draw[very thick] ([shift=(90:2cm)]0,0) arc (90:270:2cm);
\draw[red,<-,very thick] ([shift=(90:2cm)]0,0)--([shift=(270:2cm)]0,0);
\end{tikzpicture} 
\caption{\label{fig:straightening_bending} The straightening procedure transforms the arc $\psi(v)$ into an oriented chord $\phi(v)$. The bending procedure transforms the oriented chord $\phi(v)$ into an arc $\psi(v)$.}
\end{figure}

Clearly, for every $v,u \in V(G)$, the oriented chords $\phi(v)$ and $\phi(u)$ intersect 
if and only if the arcs $\psi(v)$ and $\psi(u)$ overlap.
Hence, for every $v,u \in V$ we have $v \sim u$ iff the chords $\phi(v)$ and $\phi(u)$ intersect. 
This means that $\phi$ is an \emph{oriented chord model} of $G_{ov}$ and $G_{ov}$ is a circle graph.
\begin{lemma}[\cite{Hsu95}]
\label{lem:G_ov_is_circle}
Suppose $G$ is a circular-arc graph with no twins and no universal vertices.
Then $G_{ov}$ is a circle graph.
\end{lemma}
However, the converse operation is not always possible -- we can not convert any oriented chord model of $G_{ov}$ into a normalized circular-arc model of $G$.
To describe models for which such operation is feasible,
we first note a simple property of the oriented chord models of $G_{ov}$ obtained by the straightening procedure.
We associate with every vertex $v \in V(G)$ two sets, $\leftside(v)$ and $\rightside(v)$:
$$
\begin{array}{rcl}
\leftside(v) &=& \{ u \in V(G): \quad
v \cs u \quad \text {or} \quad v \cc u\}, 
\\
\rightside(v) &=& \{ u \in V(G):
\quad v \di u \quad \text{or} \quad v \cd u \}.
\end{array}
$$
If $\phi$ is an oriented chord model of $G_{ov}$ obtained from the straightening of a normalized model $\psi$, the oriented chords $\phi(u)$ for $u \in \leftside(v)$ lie on the left side of $\phi(v)$ and the oriented chords $\phi(u)$ for $u \in \rightside(v)$ lie on the right side of $\phi(v)$, for every $v \in V(G)$. 
See Figure~\ref{fig:uv_pairs} for an illustration.
\begin{definition}
An oriented chord model $\phi$ of $G_{ov}$ is \emph{conformal} if for every $v,u \in V(G)$:
\begin{itemize}
\item $u \in \leftside(v)$ iff $\phi(u)$ lies on the left side of $\phi(v)$,
\item $u \in \rightside(v)$ iff $\phi(u)$ lies on the right side of $\phi(v)$.
\end{itemize}
\end{definition}
So, the straightening procedure transforms normalized models of $G$ into conformal models of $G_{ov}$.

\begin{figure}[htp!]
\begin{tikzpicture}[scale=0.6,>=latex,shorten >=-0.4pt,shorten <=-0.4pt]
\coordinate (label) at (0,-3) {};

\coordinate (lv) at (-0.3,0) {};
\coordinate (lu) at (1.35,0) {};

\tikzstyle{every node}=[inner sep=1pt]
\begin{scriptsize}
\node at (lv) {$v$};
\node at (lu) {$u$};
\end{scriptsize}

\draw[thick] ([shift=(-60:2.2cm)]0,0) arc (-60:60:2.2cm);
\draw[thick,<-] ([shift=(-60:2.2cm)]0,0) -- ([shift=(60:2.2cm)]0,0);

\draw[red, thick] ([shift=(90:2.0cm)]0,0) arc (90:270:2.0cm);
\draw[red, thick,<-] ([shift=(90:2.0cm)]0,0) -- ([shift=(270:2.0cm)]0,0);

\draw[white] (-2.5,-2.5)--(-2.5,-2.3);
\draw[white] (2.5,2.5)--(2.5,2.3);
\end{tikzpicture}
\hspace{0.2cm}
\begin{tikzpicture}[scale=0.6,>=latex,shorten >=-0.4pt,shorten <=-0.4pt]
\coordinate (label) at (0,-3) {};

\coordinate (lv) at (0.3,0) {};
\coordinate (lu) at (-1.25,0) {};

\tikzstyle{every node}=[inner sep=1pt]
\begin{scriptsize}
\node at (lv) {$v$};
\node at (lu) {$u$};
\end{scriptsize}

\draw[thick] ([shift=(120:2.0cm)]0,0) arc (120:240:2.0cm);
\draw[thick,->] ([shift=(240:2.0cm)]0,0) -- ([shift=(120:2.0cm)]0,0);

\draw[red, thick] ([shift=(90:2.2cm)]0,0) arc (90:270:2.2cm);
\draw[red, thick,<-] ([shift=(90:2.2cm)]0,0) -- ([shift=(270:2.2cm)]0,0);

\draw[white] (-2.5,-2.5)--(-2.5,-2.3);
\draw[white] (1.1,2.5)--(1.1,2.3);
\end{tikzpicture}
\hspace{0.2cm}
\begin{tikzpicture}[scale=0.6,>=latex,shorten >=-0.4pt,shorten <=-0.4pt]
\coordinate (label) at (0,-3) {};

\coordinate (lv) at (-0.3,0) {};
\coordinate (lu) at (1.35,0) {};

\tikzstyle{every node}=[inner sep=1pt]
\begin{scriptsize}
\node at (lv) {$v$};
\node at (lu) {$u$};
\end{scriptsize}

\draw[thick] ([shift=(60:2.2cm)]0,0) arc (60:300:2.2cm);
\draw[thick,->] ([shift=(300:2.2cm)]0,0) -- ([shift=(60:2.2cm)]0,0);

\draw[red, thick] ([shift=(90:2.0cm)]0,0) arc (90:270:2.0cm);
\draw[red, thick,<-] ([shift=(90:2.0cm)]0,0) -- ([shift=(270:2.0cm)]0,0);

\draw[white] (-2.5,-2.5)--(-2.5,-2.3);
\draw[white] (2.3,2.5)--(2.3,2.3);
\end{tikzpicture}
\hspace{0.2cm}
\begin{tikzpicture}[scale=0.6,>=latex,shorten >=-0.4pt,shorten <=-0.4pt]
\coordinate (label) at (0,-3) {};

\coordinate (lv) at (0.3,0) {};
\coordinate (lu) at (-1.0,0) {};

\tikzstyle{every node}=[inner sep=1pt]
\begin{scriptsize}
\node at (lv) {$v$};
\node at (lu) {$u$};
\end{scriptsize}

\draw[thick] ([shift=(250:2.0cm)]0,0) arc (250:470:2.0cm);
\draw[thick,<-] ([shift=(250:2.0cm)]0,0) -- ([shift=(470:2.0cm)]0,0);

\draw[red, thick] ([shift=(90:2.2cm)]0,0) arc (90:270:2.2cm);
\draw[red, thick,<-] ([shift=(90:2.2cm)]0,0) -- ([shift=(270:2.2cm)]0,0);

\draw[white] (-2.5,-2.5)--(-2.5,-2.3);
\draw[white] (2.5,2.5)--(2.5,2.3);
\end{tikzpicture}
\hspace{0.2cm}
\begin{tikzpicture}[scale=0.6,>=latex,shorten >=-0.4pt,shorten <=-0.4pt]
\coordinate (label) at (0,-3) {};

\coordinate (lv) at (-0.3,0) {};
\coordinate (lu) at (1.05,0.7) {};

\tikzstyle{every node}=[inner sep=1pt]
\begin{scriptsize}
\node at (lv) {$v$};
\node at (lu) {$u$};
\end{scriptsize}

\draw[thick] ([shift=(30:2.0cm)]0,0) arc (30:150:2.0cm);
\draw[thick,<-] ([shift=(30:2.0cm)]0,0) -- ([shift=(150:2.0cm)]0,0);

\draw[red, thick] ([shift=(90:2.2cm)]0,0) arc (90:270:2.2cm);
\draw[red, thick,<-] ([shift=(90:2.2cm)]0,0) -- ([shift=(270:2.2cm)]0,0);

\draw[white] (-2.5,-2.5)--(-2.5,-2.3);
\draw[white] (2.5,2.5)--(2.5,2.3);
\end{tikzpicture}
\caption{\label{fig:uv_pairs} The mutual positions of the arcs $\psi(v)$ and $\psi(u)$ and the mutual positions of the corresponding oriented chords $\phi(v)$ and $\phi(u)$ for the cases: $v \di u$, $v \cs u$, $v \cd u$, $v \cc u$, and $v \ov u$, respectively.}
\end{figure}
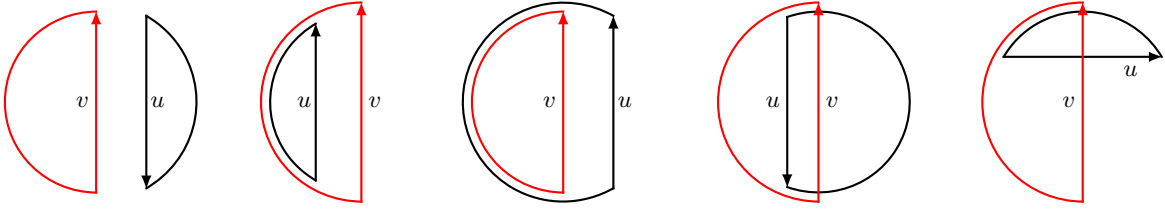

Now, suppose $\phi$ is a conformal model of $G_{ov}$.
A \emph{bending procedure} transforms every oriented chord $\phi(v)$ into an arc $\psi(v)$ 
with the same endpoints as $\phi(v)$ placed on the left side of $\phi(v)$, for $v \in V$.
So, the bending procedure is the inverse of the straightening procedure.
One can easily check that the bending procedures transforms $\phi$ into a normalized model $\psi$ of $G$.
Indeed, note that for every pair $(v,u)$ of distinct vertices in $G$ the statements:
\begin{itemize}
 \item $v \di u$, $v \in \rightside(u)$ and $u \in \rightside(v)$,
 $\psi(v)$ and $\psi(u)$ are disjoint, $\phi(v)$ has $\phi(u)$ on its right side and 
 $\phi(u)$ has $\phi(v)$ on its right side, are equivalent.
 \item $v \cs u$, $v \in \rightside(u)$ and $u \in \leftside(v)$, $\psi(v)$ contains $\psi(u)$, 
 $\phi(u)$ has $\phi(v)$ on its right side and $\phi(v)$ has $\phi(u)$ on its left side, are equivalent. 
 \item $v \cd u$, $v \in \leftside(u)$ and $u \in \rightside(v)$, $\psi(v)$ is contained in $\psi(u)$,
 $\phi(u)$ has $\phi(v)$ on its left side, $\phi(v)$ has $\phi(u)$ on its right side, are equivalent.
 \item $v \cc u$, $v \in \leftside(u)$ and $u \in \leftside(v)$, $\psi(v)$ and $\psi(u)$ cover the circle,
 $\phi(v)$ has $\phi(u)$ on its left side and $\phi(u)$ has $\phi(v)$ on its left side, are equivalent.
 \item $v \ov u$, $\psi(v)$ and $\psi(u)$ overlap, $\phi(v)$ and $\phi(u)$ intersect, are equivalent.
\end{itemize}

Since the straightening procedure and the bending procedure establish a one-to-one correspondence between
normalized models of $G$ and conformal model of $G_{ov}$, we have the following theorems.
\begin{theorem}[\cite{Hsu95}]
\label{thm:normalized_models_versus_conformal_models}
Let $G$ be a graph with no twins and no universal vertices.
Let $G_{ov}$ be an overlap graph associated with $G$.
Then, $G$ is a circular-arc graph if and only if $G_{ov}$ is a circle graph that admits a conformal model.
\end{theorem}
\begin{theorem}[\cite{Hsu95}]
\label{thm:normalized_conformal_correspondence}
Let $G$ be a circular-arc graph with no twins and no universal vertices.
There is a one-to-one correspondence between normalized models of $G$ 
and conformal models of $G_{ov}$.
\end{theorem}

\subsection{Hsu's approach}
The straightening procedure and the bending procedure were introduced by Hsu \cite{Hsu95}.
However, Hsu's straightening procedure does not orient the chords in $\phi$.
Thus, the bending procedure needs to be performed more carefully,
but still can be uniquely performed unless $G$ has universal vertices. 
In particular, Lemma \ref{lem:G_ov_is_circle} and Theorems \ref{thm:normalized_models_versus_conformal_models} and \ref{thm:normalized_conformal_correspondence} were proved by Hsu \cite{Hsu95}, 
but in a slightly different setting.
The main difference between our approaches lies in the definition of conformal models of $G_{ov}$.
In fact, Hsu assumes the following definition: a non-oriented \emph{chord model $\phi$ of $G_{ov}$ is conformal 
if for every vertex $v$ of $G$ the chords associated with vertices in $\leftside(v)$ are on one side of $\phi(v)$
and those associated with vertices in $\rightside(v)$ are on the other side of $\phi(v)$} (Section 5.2 in \cite{Hsu95}).
Such a definition has one drawback:
there might exist two non-isomorphic circular-arc graphs $G=(V,E)$ and $G'=(V,E')$, defined on the same set of vertices $V$, such that $G_{ov} = G'_{ov}$ and such that both $G_{ov}$ and $G'_{ov}$ have the same conformal model $\phi$ -- see Figures~\ref{fig:counterex_graphs_models}--\ref{fig:counterex_overlap_and_hsu_conformal}.
This observation was noted by Curtis, Lin, McConnell, Nussbaum, Soulignac, Spinrad, and Szwarcfiter \cite{counterex13} and resulted in the construction of a counterexample to the correctness of Hsu's isomorphism algorithm.
As is stated in \cite{counterex13}: ``The origin of the mistake in Hsu's algorithm is the statement:
\emph{To test the isomorphism between two circular-arc graphs $G$ and $G'$, it suffices to test whether there exists isomorphic conformal models for $G_{ov}$ and $G'_{ov}$}'', which is not true due to the example constructed in \cite{counterex13} and shown in Figures~\ref{fig:counterex_graphs_models}--\ref{fig:counterex_overlap_and_hsu_conformal}.
Note that, assuming our definition, the conformal models corresponding to normalized models of $G$ and $G'$
are not isomorphic -- see Figure \ref{fig:counterex_models}.

\begin{figure}[htp!]
\centering
\begin{tikzpicture}[scale=0.70]
\coordinate (center) at (0,0) {};
\coordinate (a1) at ($(center)+(90:2cm)$) {};
\coordinate (a2) at ($(center)+(141:2cm)$) {};
\coordinate (a3) at ($(center)+(192:2cm)$) {};
\coordinate (a4) at ($(center)+(250:2cm)$) {};
\coordinate (a5) at ($(center)+(296:2cm)$) {};
\coordinate (a6) at ($(center)+(348:2cm)$) {};
\coordinate (a7) at ($(center)+(39:2cm)$) {};

\coordinate (lcenter) at ($(center)+(0.3,0)$) {};
\coordinate (la1) at ($(center)+(90:2.3cm)$) {};
\coordinate (la2) at ($(center)+(141:2.3cm)$) {};
\coordinate (la3) at ($(center)+(192:2.3cm)$) {};
\coordinate (la4) at ($(center)+(250:2.3cm)$) {};
\coordinate (la5) at ($(center)+(296:2.3cm)$) {};
\coordinate (la6) at ($(center)+(348:2.3cm)$) {};
\coordinate (la7) at ($(center)+(39:2.3cm)$) {};

\tikzstyle{every node}=[circle,minimum size=5pt,inner sep=0pt,draw,fill]
\node at (center) {};
\node at (a1) {};
\node at (a2) {};
\node at (a3) {};
\node at (a4) {};
\node at (a5) {};
\node at (a6) {};
\node at (a7) {};

\tikzstyle{every node}=[inner sep=1pt]

\path (center) edge[thick] (a1);
\path (center) edge[thick] (a2);
\path (center) edge[thick] (a3);
\path (center) edge[thick] (a4);
\path (center) edge[thick] (a5);

\path (a1) edge[thick] (a2);
\path (a2) edge[thick] (a3);
\path (a3) edge[thick] (a4);
\path (a4) edge[thick] (a5);
\path (a5) edge[thick] (a6);
\path (a6) edge[thick] (a7);
\path (a7) edge[thick] (a1);


\end{tikzpicture}
\hspace{0.1cm}
\begin{tikzpicture}[scale=0.60]
\coordinate (center) at (0,0) {};
\coordinate (a) at ($(center)+(90:2cm)$) {};
\coordinate (b) at ($(center)+(141:2cm)$) {};
\coordinate (c) at ($(center)+(192:2cm)$) {};
\coordinate (d) at ($(center)+(250:2cm)$) {};
\coordinate (e) at ($(center)+(296:2cm)$) {};
\coordinate (f) at ($(center)+(348:2cm)$) {};
\coordinate (g) at ($(center)+(39:2cm)$) {};

\coordinate (la) at ($(center)+(90:2cm)$) {};
\coordinate (lb) at ($(center)+(141:2cm)$) {};
\coordinate (lc) at ($(center)+(192:2cm)$) {};
\coordinate (ld) at ($(center)+(250:2cm)$) {};
\coordinate (le) at ($(center)+(296:2cm)$) {};
\coordinate (lf) at ($(center)+(348:2cm)$) {};
\coordinate (lg) at ($(center)+(39:2cm)$) {};

\tikzstyle{every node}=[inner sep=1pt]

\draw[thick] ([shift=(60:2.3cm)]0,0) arc (60:120:2.3cm);
\draw[thick] ([shift=(111:2.4cm)]0,0) arc (111:171:2.4cm);
\draw[thick] ([shift=(162:2.3cm)]0,0) arc (162:225:2.3cm);
\draw[thick] ([shift=(217:2.43cm)]0,0) arc (217:277:2.43cm);
\draw[thick] ([shift=(267:2.35cm)]0,0) arc (266:325:2.35cm);
\draw[thick] ([shift=(318:2.3cm)]0,0) arc (318:378:2.3cm);
\draw[thick] ([shift=(9:2.4cm)]0,0) arc (9:69:2.4cm);

\draw[thick] ([shift=(90:2.1cm)]0,0) arc (90:296:2.1cm);


\end{tikzpicture}
\hspace{1cm}
\begin{tikzpicture}[scale=0.70]
\coordinate (center) at (0,0) {};
\coordinate (a1) at ($(center)+(90:2cm)$) {};
\coordinate (a2) at ($(center)+(141:2cm)$) {};
\coordinate (a3) at ($(center)+(192:2cm)$) {};
\coordinate (a4) at ($(center)+(250:2cm)$) {};
\coordinate (a5) at ($(center)+(296:2cm)$) {};
\coordinate (a6) at ($(center)+(348:2cm)$) {};
\coordinate (a7) at ($(center)+(39:2cm)$) {};

\coordinate (lcenter) at ($(center)+(-0.3,0)$) {};
\coordinate (la1) at ($(center)+(90:2.3cm)$) {};
\coordinate (la2) at ($(center)+(141:2.3cm)$) {};
\coordinate (la3) at ($(center)+(192:2.3cm)$) {};
\coordinate (la4) at ($(center)+(250:2.3cm)$) {};
\coordinate (la5) at ($(center)+(296:2.3cm)$) {};
\coordinate (la6) at ($(center)+(348:2.3cm)$) {};
\coordinate (la7) at ($(center)+(39:2.3cm)$) {};

\tikzstyle{every node}=[circle,minimum size=5pt,inner sep=0pt,draw,fill]
\node at (center) {};
\node at (a1) {};
\node at (a2) {};
\node at (a3) {};
\node at (a4) {};
\node at (a5) {};
\node at (a6) {};
\node at (a7) {};

\tikzstyle{every node}=[inner sep=1pt]

\path (center) edge[thick] (a5);
\path (center) edge[thick] (a6);
\path (center) edge[thick] (a7);
\path (center) edge[thick] (a1);

\path (a1) edge[thick] (a2);
\path (a2) edge[thick] (a3);
\path (a3) edge[thick] (a4);
\path (a4) edge[thick] (a5);
\path (a5) edge[thick] (a6);
\path (a6) edge[thick] (a7);
\path (a7) edge[thick] (a1);
\end{tikzpicture}
\hspace{0.1cm}
\begin{tikzpicture}[scale=0.60]
\coordinate (center) at (0,0) {};
\coordinate (a) at ($(center)+(90:2cm)$) {};
\coordinate (b) at ($(center)+(141:2cm)$) {};
\coordinate (c) at ($(center)+(192:2cm)$) {};
\coordinate (d) at ($(center)+(250:2cm)$) {};
\coordinate (e) at ($(center)+(296:2cm)$) {};
\coordinate (f) at ($(center)+(348:2cm)$) {};
\coordinate (g) at ($(center)+(39:2cm)$) {};

\coordinate (la) at ($(center)+(90:2cm)$) {};
\coordinate (lb) at ($(center)+(141:2cm)$) {};
\coordinate (lc) at ($(center)+(192:2cm)$) {};
\coordinate (ld) at ($(center)+(250:2cm)$) {};
\coordinate (le) at ($(center)+(296:2cm)$) {};
\coordinate (lf) at ($(center)+(348:2cm)$) {};
\coordinate (lg) at ($(center)+(39:2cm)$) {};

\tikzstyle{every node}=[inner sep=1pt]

\draw[thick] ([shift=(60:2.3cm)]0,0) arc (60:120:2.3cm);
\draw[thick] ([shift=(111:2.4cm)]0,0) arc (111:171:2.4cm);
\draw[thick] ([shift=(162:2.3cm)]0,0) arc (162:225:2.3cm);
\draw[thick] ([shift=(217:2.43cm)]0,0) arc (217:277:2.43cm);
\draw[thick] ([shift=(267:2.35cm)]0,0) arc (266:325:2.35cm);
\draw[thick] ([shift=(318:2.25cm)]0,0) arc (318:378:2.25cm);
\draw[thick] ([shift=(9:2.4cm)]0,0) arc (9:69:2.4cm);

\draw[thick] ([shift=(296:2.1cm)]0,0) arc (296:450:2.1cm);


\end{tikzpicture}

\caption{\label{fig:counterex_graphs_models} Circular-arc graphs $G$ and $G'$ and their normalized models.
Graphs $G$ and $G'$ are not isomorphic as they have different number of edges.}
\end{figure}
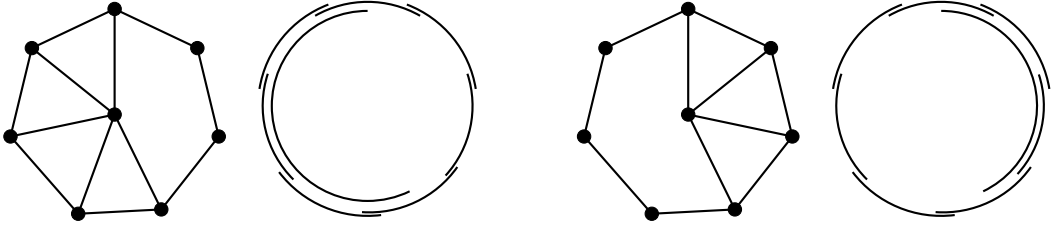
\begin{figure}[htp!]
\centering
\begin{tikzpicture}[scale=0.70]
\coordinate (center) at (0,0) {};
\coordinate (a1) at ($(center)+(90:2cm)$) {};
\coordinate (a2) at ($(center)+(141:2cm)$) {};
\coordinate (a3) at ($(center)+(192:2cm)$) {};
\coordinate (a4) at ($(center)+(250:2cm)$) {};
\coordinate (a5) at ($(center)+(296:2cm)$) {};
\coordinate (a6) at ($(center)+(348:2cm)$) {};
\coordinate (a7) at ($(center)+(39:2cm)$) {};

\coordinate (lcenter) at ($(center)+(0.3,0)$) {};
\coordinate (la1) at ($(center)+(90:2.3cm)$) {};
\coordinate (la2) at ($(center)+(141:2.3cm)$) {};
\coordinate (la3) at ($(center)+(192:2.3cm)$) {};
\coordinate (la4) at ($(center)+(250:2.3cm)$) {};
\coordinate (la5) at ($(center)+(296:2.3cm)$) {};
\coordinate (la6) at ($(center)+(348:2.3cm)$) {};
\coordinate (la7) at ($(center)+(39:2.3cm)$) {};

\tikzstyle{every node}=[circle,minimum size=5pt,inner sep=0pt,draw,fill]
\node at (center) {};
\node at (a1) {};
\node at (a2) {};
\node at (a3) {};
\node at (a4) {};
\node at (a5) {};
\node at (a6) {};
\node at (a7) {};

\tikzstyle{every node}=[inner sep=1pt]

\path (center) edge[thick] (a1);
\path (center) edge[thick] (a5);

\path (a1) edge[thick] (a2);
\path (a2) edge[thick] (a3);
\path (a3) edge[thick] (a4);
\path (a4) edge[thick] (a5);
\path (a5) edge[thick] (a6);
\path (a6) edge[thick] (a7);
\path (a7) edge[thick] (a1);


\end{tikzpicture}
\hspace{1cm}
\begin{tikzpicture}[scale=0.70]
\coordinate (la) at ($(center)+(125:2cm)$) {};
\coordinate (ra) at ($(center)+(55:2cm)$) {};
\coordinate (lb) at ($(center)+(176:2cm)$) {};
\coordinate (rb) at ($(center)+(106:2cm)$) {};
\coordinate (lc) at ($(center)+(229:2cm)$) {};
\coordinate (rc) at ($(center)+(157:2cm)$) {};
\coordinate (d) at ($(center)+(250:2cm)$) {};
\coordinate (rd) at ($(center)+(215:2cm)$) {};
\coordinate (ld) at ($(center)+(285:2cm)$) {};
\coordinate (le) at ($(center)+(261:2cm)$) {};
\coordinate (re) at ($(center)+(331:2cm)$) {};
\coordinate (lf) at ($(center)+(313:2cm)$) {};
\coordinate (rf) at ($(center)+(383:2cm)$) {};
\coordinate (lg) at ($(center)+(4:2cm)$) {};
\coordinate (rg) at ($(center)+(74:2cm)$) {};

\coordinate (lv) at ($(center)+(90:2cm)$) {};
\coordinate (rv) at ($(center)+(298:2cm)$) {};
\draw (la)--(ra);
\draw (lb)--(rb);
\draw (lc)--(rc);
\draw (ld)--(rd);
\draw (le)--(re);
\draw (lf)--(rf);
\draw (lg)--(rg);
\draw (rv)--(lv);

\draw (0,0) circle (2cm);

\end{tikzpicture} 

\caption{\label{fig:counterex_overlap_and_hsu_conformal} Graphs $G$ and $G'$ have the same overlap graph $G_{ov} = G'_{ov}$ (shown to the left). The overlap graphs $G_{ov}$ and $G'_{ov}$ have the same conformal model (shown to the right).}
\end{figure}
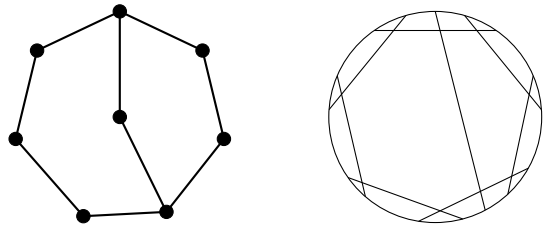

\begin{figure}[!htp]
\begin{tikzpicture}[scale=0.70]
\coordinate (center) at (0,0) {};
\coordinate (a) at ($(center)+(90:2cm)$) {};
\coordinate (b) at ($(center)+(141:2cm)$) {};
\coordinate (c) at ($(center)+(192:2cm)$) {};
\coordinate (d) at ($(center)+(250:2cm)$) {};
\coordinate (e) at ($(center)+(296:2cm)$) {};
\coordinate (f) at ($(center)+(348:2cm)$) {};
\coordinate (g) at ($(center)+(39:2cm)$) {};

\coordinate (la) at ($(center)+(90:2cm)$) {};
\coordinate (lb) at ($(center)+(141:2cm)$) {};
\coordinate (lc) at ($(center)+(192:2cm)$) {};
\coordinate (ld) at ($(center)+(250:2cm)$) {};
\coordinate (le) at ($(center)+(296:2cm)$) {};
\coordinate (lf) at ($(center)+(348:2cm)$) {};
\coordinate (lg) at ($(center)+(39:2cm)$) {};

\tikzstyle{every node}=[inner sep=1pt]

\draw[thick] ([shift=(60:2.3cm)]0,0) arc (60:120:2.3cm);
\draw[thick] ([shift=(111:2.4cm)]0,0) arc (111:171:2.4cm);
\draw[thick] ([shift=(162:2.3cm)]0,0) arc (162:225:2.3cm);
\draw[thick] ([shift=(217:2.43cm)]0,0) arc (217:277:2.43cm);
\draw[thick] ([shift=(267:2.35cm)]0,0) arc (266:325:2.35cm);
\draw[thick] ([shift=(318:2.3cm)]0,0) arc (318:378:2.3cm);
\draw[thick] ([shift=(9:2.4cm)]0,0) arc (9:69:2.4cm);

\draw[thick] ([shift=(90:2.1cm)]0,0) arc (90:296:2.1cm);


\end{tikzpicture}
\hspace{0.1cm}
\begin{tikzpicture}[xscale=0.80,yscale=0.80,>=latex,shorten >=-0.4pt,shorten <=-0.4pt]

\coordinate (la) at ($(center)+(125:2cm)$) {};
\coordinate (ra) at ($(center)+(55:2cm)$) {};
\coordinate (lb) at ($(center)+(176:2cm)$) {};
\coordinate (rb) at ($(center)+(106:2cm)$) {};
\coordinate (lc) at ($(center)+(229:2cm)$) {};
\coordinate (rc) at ($(center)+(157:2cm)$) {};
\coordinate (ld) at ($(center)+(285:2cm)$) {};
\coordinate (rd) at ($(center)+(210:2cm)$) {};
\coordinate (le) at ($(center)+(331:2cm)$) {};
\coordinate (re) at ($(center)+(261:2cm)$) {};
\coordinate (lf) at ($(center)+(383:2cm)$) {};
\coordinate (rf) at ($(center)+(313:2cm)$) {};
\coordinate (lg) at ($(center)+(74:2cm)$) {};
\coordinate (rg) at ($(center)+(4:2cm)$) {};

\coordinate (lv) at ($(center)+(90:2cm)$) {};
\coordinate (rv) at ($(center)+(298:2cm)$) {};
\draw (0,0) circle (2cm);
\draw[white] (0,0) circle (2.2cm);
\draw[->] (la)--(ra);
\draw[->] (lb)--(rb);
\draw[->] (lc)--(rc);
\draw[->] (ld)--(rd);
\draw[->] (le)--(re);
\draw[->] (lf)--(rf);
\draw[->] (lg)--(rg);
\draw[->] (rv)--(lv);
\draw[red, very thick] ([shift=(90:2cm)]0,0) arc (90:298:2cm);

\draw (rv)--(lv);
\end{tikzpicture} 
\hspace{0.5cm}
\begin{tikzpicture}[scale=0.70]
\coordinate (center) at (0,0) {};
\coordinate (a) at ($(center)+(90:2cm)$) {};
\coordinate (b) at ($(center)+(141:2cm)$) {};
\coordinate (c) at ($(center)+(192:2cm)$) {};
\coordinate (d) at ($(center)+(250:2cm)$) {};
\coordinate (e) at ($(center)+(296:2cm)$) {};
\coordinate (f) at ($(center)+(348:2cm)$) {};
\coordinate (g) at ($(center)+(39:2cm)$) {};

\coordinate (la) at ($(center)+(90:2cm)$) {};
\coordinate (lb) at ($(center)+(141:2cm)$) {};
\coordinate (lc) at ($(center)+(192:2cm)$) {};
\coordinate (ld) at ($(center)+(250:2cm)$) {};
\coordinate (le) at ($(center)+(296:2cm)$) {};
\coordinate (lf) at ($(center)+(348:2cm)$) {};
\coordinate (lg) at ($(center)+(39:2cm)$) {};

\tikzstyle{every node}=[inner sep=1pt]

\draw[thick] ([shift=(60:2.3cm)]0,0) arc (60:120:2.3cm);
\draw[thick] ([shift=(111:2.4cm)]0,0) arc (111:171:2.4cm);
\draw[thick] ([shift=(162:2.3cm)]0,0) arc (162:225:2.3cm);
\draw[thick] ([shift=(217:2.43cm)]0,0) arc (217:277:2.43cm);
\draw[thick] ([shift=(267:2.35cm)]0,0) arc (266:325:2.35cm);
\draw[thick] ([shift=(318:2.25cm)]0,0) arc (318:378:2.25cm);
\draw[thick] ([shift=(9:2.4cm)]0,0) arc (9:69:2.4cm);

\draw[thick] ([shift=(296:2.1cm)]0,0) arc (296:450:2.1cm);


\end{tikzpicture}
\hspace{0.1cm}
\begin{tikzpicture}[scale=0.80,>=latex,shorten >=-0.4pt,shorten <=-0.4pt]

\coordinate (la) at ($(center)+(125:2cm)$) {};
\coordinate (ra) at ($(center)+(55:2cm)$) {};
\coordinate (lb) at ($(center)+(176:2cm)$) {};
\coordinate (rb) at ($(center)+(106:2cm)$) {};
\coordinate (lc) at ($(center)+(229:2cm)$) {};
\coordinate (rc) at ($(center)+(157:2cm)$) {};
\coordinate (ld) at ($(center)+(285:2cm)$) {};
\coordinate (rd) at ($(center)+(210:2cm)$) {};
\coordinate (le) at ($(center)+(331:2cm)$) {};
\coordinate (re) at ($(center)+(261:2cm)$) {};
\coordinate (lf) at ($(center)+(383:2cm)$) {};
\coordinate (rf) at ($(center)+(313:2cm)$) {};
\coordinate (lg) at ($(center)+(74:2cm)$) {};
\coordinate (rg) at ($(center)+(4:2cm)$) {};

\coordinate (lv) at ($(center)+(90:2cm)$) {};
\coordinate (rv) at ($(center)+(298:2cm)$) {};

\draw (0,0) circle (2cm);
\draw[white] (0,0) circle (2.2cm);
\draw[red,very thick] ([shift=(298:2cm)]0,0) arc (298:450:2cm);

\draw[->] (la)--(ra);
\draw[->] (lb)--(rb);
\draw[->] (lc)--(rc);
\draw[->] (ld)--(rd);
\draw[->] (le)--(re);
\draw[->] (lf)--(rf);
\draw[->] (lg)--(rg);
\draw[->] (lv)--(rv);
\end{tikzpicture} 

\caption{\label{fig:counterex_models}}
\end{figure}

Despite the mistake, Hsu's paper \cite{Hsu95} contains many brilliant ideas that are used in our work.
This includes:
\begin{itemize}
 \item reducing the problem of characterization of the normalized models of $G$ to the problem of characterization of the conformal models of the circle graph $G_{ov}$,
 \item the use of the modular decomposition of $G_{ov}$ to devise the structure representing all conformal models of $G_{ov}$. This includes, in particular:
 \begin{itemize}
 \item the concept of consistent decompositions, introduced in Subsection \ref{subsec:conformal_models_of_improper_prime_modules},
 \item the concept of $T_{NM}$-tree, introduced in Subsection \ref{subsec:conformal_models_of_improper_parallel_modules}.
\end{itemize}
\end{itemize}
Although we use similar concepts to those introduced by Hsu,
the details hidden behind them are different. 
Since we deal with conformal models defined in a different way,
definitions and proofs contained in our work differ (sometimes a lot) from those proposed in \cite{Hsu95}.
The differences are also caused by the approach in searching for conformal models.
Hsu's divides all the triples $(u,v,w)$ consisting of pairwise non-adjacent vertices in $G_{ov}$ into two categories.
Such a triple is:
\begin{itemize}
 \item \emph{in parallel}, written $u|v|w$, if the vertex $v$ has the vertices $u$ and $w$ on its different sides (that is, either $u \in \leftside(v)$ and $w \in \rightside(v)$ or $w \in \leftside(v)$ and $u \in \rightside(v)$),
 \item \emph{in series}, written $u-v-w$, if any vertex from $\{u,v,w\}$ has the remaining two vertices on the same side,
\end{itemize}
(see Section 5.1 of \cite{Hsu95}).
Then, Hsu is searching for chord models $\phi$ of $G_{ov}$ that satisfy 
the conditions:
\begin{itemize}
 \item if $u|v|w$ then the chord $\phi(v)$ has $\phi(u)$ and $\phi(w)$ on its different sides,
 \item if $u-v-w$ then every chord from $\{\phi(u), \phi(v), \phi(w)\}$ has the remaining two chords on the same side.
\end{itemize}
Hsu showed that such chord models correspond to conformal models (Section 5.2 in \cite{Hsu95}).
Consequently, Hsu builds his decomposition trees based on the types of the triples $(u,v,w)$ --
in particular, he tries to describe what kind of transformations on conformal models keep the relations between
every three non-intersecting chords unchanged.
Our approach is different: instead of looking at the triples, every vertex $v$ is
responsible for itself to be represented by a chord that has the chords representing the
vertices from $\leftside(v)$ on its left side and the chords representing the vertices from $\rightside(v)$
on its right side. 
This explains why we need to use the orientations of the chords in conformal models $\phi$ of $G_{ov}$: 
just to distinguish the left side of $\phi(u)$ from the right side of
$\phi(u)$.
Moreover, to characterize transformations between conformal models we do not need to analyze triples:
we are searching for transformations that keep the relative relations between the pairs of non-intersecting oriented chords unchanged.

\subsection{Our approach} 
In our work we use the framework proposed by Hsu \cite{Hsu95} to describe all conformal models of $G_{ov}$.
In addition, we use the ideas invented by Spinrad \cite{Spin88} that allow to characterize all normalized models of co-bipartite circular-arc graphs in terms of two-dimensional realizers of appropriately chosen two-dimensional posets (see Subsection \ref{subsec:conformal_models_of_proper_prime_parallel_modules} for more details).
We extend Spinrad's ideas on the whole class of circular-arc graphs:
in particular, we use them to characterize conformal models of some parts of the overlap graph $G_{ov}$.
Eventually, we stick all these pieces together and we develop a decomposition tree that represents all conformal models of $G_{ov}$.

Suppose $G = (V,E)$ and $G' = (V',E')$ are circular-arc graphs with no twins and no universal vertices.
The isomorphism algorithm devised in this paper tests whether there exists a bijection $\alpha: V \to V'$ 
that satisfies for every $(v,u) \in V \times V$ the following conditions:
\begin{itemize}
 \item $u$ in $\leftside(v)$ iff $\alpha(u)$ in $\leftside(\alpha(v))$,
 \item $u$ in $\rightside(v)$ iff $\alpha(u)$ in $\rightside(\alpha(v))$.
\end{itemize}
Hence, for every $v,u \in V$ the bijection $\alpha$ satisfies the properties:
\begin{itemize} 
 \item $v \di u$ iff $\alpha(v) \di \alpha(u)$,
 \item $v \cs u$ iff $\alpha(v) \cs \alpha(u)$,
 \item $v \cd u$ iff $\alpha(v) \cd \alpha(u)$,
 \item $v \cc u$ iff $\alpha(v) \cc \alpha(u)$,
 \item $v \ov u$ iff $\alpha(v) \ov \alpha(u)$,
\end{itemize}
which show that the graphs $G$ and $G'$ are indeed isomorphic.
To test whether such bijection exists we exploit decomposition trees of $G$ and of $G'$.
We traverse these trees bottom-up and for every pair of the nodes from these trees we test whether there is a bijection $\alpha$ satisfying the above properties with the restriction to the vertices kept in these two nodes.

One can extend the above ideas to handle also the case when $G$ and $G'$ contain twins and universal vertices.

\section{Preliminaries}
\label{sec:preliminaries}
A sequence $\tau$ over an alphabet $\Sigma$ is a \emph{word}. 
A \emph{circular word} represents the set of words which are cyclical shifts of one another. 
In the sequel, we represent a circular word by a word
from its corresponding set of words.
We do not introduce any notation to distinguish between words and circular words;
if it is not clear from the context we state explicitly whether we are dealing with a word or a circular word. 
We do one exception: we use operator $\equiv$ to indicate that the equality holds between two circular words.

Suppose $G=(V,E)$ is a circular-arc graph with no twins and no universal vertices.
Suppose $\psi$ is a normalized model of $G$ and $\phi$ is a conformal model of $G_{ov}$ associated with $\psi$.
Conformal model $\phi$ is represented by means of a circular word $\tau$ over the set of letters $V^{*} = \{v^{0}, v^{1} : V\}$. 
The circular word $\tau$ is obtained from the model $\phi$ as follows.
We choose a point $P$ on the circle $C$ and then we traverse $C$ in the clockwise order:
if we pass the tail of the chord $\phi(v)$ we append the letter $v^{0}$ to $\tau$ and when we pass the head of the chord $\phi(v)$ 
we append the letter $v^{1}$ to $\tau$. 
When we encounter $P$ again, we make the word $\tau$ circular.
We write $\phi \equiv \tau$ to denote that $\tau$ is a \emph{word representation} of $\phi$.
We consider two conformal models $\phi_1$ and $\phi_2$ of $G_{ov}$ \emph{equivalent}, 
written $\phi_1 \equiv \phi_2$, if the word representations of $\phi_1$ and $\phi_2$ are equal.
Usually we use the same symbol to denote a conformal model of $G_{ov}$ and its word representation.
Figure \ref{fig:examp} shows a circular-arc graph $G = (V,E)$, where $V = \{v_1,\ldots,v_6\}$ and $E = \{ v_iv_{i+1}: i \in [5]\}\cup \{v_6v_1\}$, its normalized model $\psi$, and a conformal model $\phi$ of $G_{ov}$ associated with $\psi$.
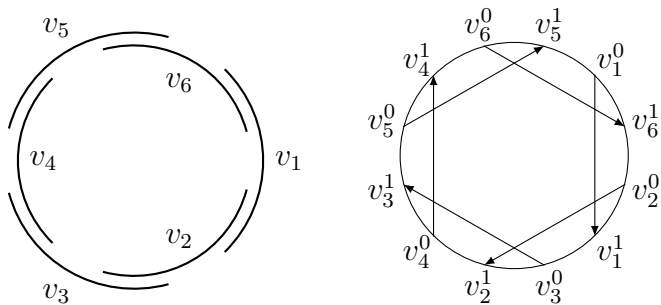
\begin{figure}[htp!]
\centering
\begin{tikzpicture}[xscale=0.85,yscale=0.85,>=latex,shorten >=-0.4pt,shorten <=-0.4pt]
\coordinate (center) at (0,0) {};
\coordinate (Lv1) at ($(center)+(0:2.4cm)$) {};
\coordinate (Lv6) at ($(center)+(60:1.4cm)$) {};
\coordinate (Lv5) at ($(center)+(120:2.4cm)$) {};
\coordinate (Lv4) at ($(center)+(180:1.4cm)$) {};
\coordinate (Lv3) at ($(center)+(240:2.4cm)$) {};
\coordinate (Lv2) at ($(center)+(300:1.4cm)$) {};

\coordinate (lv1) at ($(center)+(-45:2cm)$) {};
\coordinate (rv1) at ($(center)+(45:2cm)$) {};
\coordinate (lv6) at ($(center)+(15:2cm)$) {};
\coordinate (rv6) at ($(center)+(105:2cm)$) {};
\coordinate (lv5) at ($(center)+(75:2cm)$) {};
\coordinate (rv5) at ($(center)+(165:2cm)$) {};
\coordinate (lv4) at ($(center)+(135:2cm)$) {};
\coordinate (rv4) at ($(center)+(225:2cm)$) {};
\coordinate (lv3) at ($(center)+(195:2cm)$) {};
\coordinate (rv3) at ($(center)+(285:2cm)$) {};
\coordinate (lv2) at ($(center)+(255:2cm)$) {};
\coordinate (rv2) at ($(center)+(345:2cm)$) {};

\tikzstyle{every node}=[inner sep=1pt]
\node at (Lv1) {$v_1$};
\node at (Lv2) {$v_2$};
\node at (Lv3) {$v_3$};
\node at (Lv4) {$v_4$};
\node at (Lv5) {$v_5$};
\node at (Lv6) {$v_6$};

\draw[thick] ([shift=(-45:2cm)]0,0) arc (-45:45:2cm);
\draw[thick] ([shift=(15:1.8cm)]0,0) arc (15:105:1.8cm);
\draw[thick] ([shift=(75:2cm)]0,0) arc (75:165:2cm);
\draw[thick] ([shift=(135:1.8cm)]0,0) arc (135:225:1.8cm);
\draw[thick] ([shift=(195:2cm)]0,0) arc (195:285:2cm);
\draw[thick] ([shift=(255:1.8cm)]0,0) arc (255:345:1.8cm);

\end{tikzpicture}
\hspace{0.5cm}
\begin{tikzpicture}[xscale=0.75,yscale=0.75,>=latex,shorten >=-0.4pt,shorten <=-0.4pt]
\coordinate (center) at (0,0) {};
\coordinate (Llv1) at ($(center)+(-45:2.4cm)$) {};
\coordinate (Lrv1) at ($(center)+(45:2.4cm)$) {};
\coordinate (Llv6) at ($(center)+(15:2.4cm)$) {};
\coordinate (Lrv6) at ($(center)+(105:2.4cm)$) {};
\coordinate (Llv5) at ($(center)+(75:2.4cm)$) {};
\coordinate (Lrv5) at ($(center)+(165:2.4cm)$) {};
\coordinate (Llv4) at ($(center)+(135:2.4cm)$) {};
\coordinate (Lrv4) at ($(center)+(225:2.4cm)$) {};
\coordinate (Llv3) at ($(center)+(195:2.4cm)$) {};
\coordinate (Lrv3) at ($(center)+(285:2.4cm)$) {};
\coordinate (Llv2) at ($(center)+(255:2.4cm)$) {};
\coordinate (Lrv2) at ($(center)+(345:2.4cm)$) {};

\coordinate (rv1) at ($(center)+(-45:2cm)$) {};
\coordinate (lv1) at ($(center)+(45:2cm)$) {};
\coordinate (rv6) at ($(center)+(15:2cm)$) {};
\coordinate (lv6) at ($(center)+(105:2cm)$) {};
\coordinate (rv5) at ($(center)+(75:2cm)$) {};
\coordinate (lv5) at ($(center)+(165:2cm)$) {};
\coordinate (rv4) at ($(center)+(135:2cm)$) {};
\coordinate (lv4) at ($(center)+(225:2cm)$) {};
\coordinate (rv3) at ($(center)+(195:2cm)$) {};
\coordinate (lv3) at ($(center)+(285:2cm)$) {};
\coordinate (rv2) at ($(center)+(255:2cm)$) {};
\coordinate (lv2) at ($(center)+(345:2cm)$) {};

\tikzstyle{every node}=[inner sep=1pt]
\node at (Lrv1) {$v^0_1$};
\node at (Llv1) {$v^1_1$};
\node at (Lrv2) {$v^0_2$};
\node at (Llv2) {$v^1_2$};
\node at (Lrv3) {$v^0_3$};
\node at (Llv3) {$v^1_3$};
\node at (Lrv4) {$v^0_4$};
\node at (Llv4) {$v^1_4$};
\node at (Lrv5) {$v^0_5$};
\node at (Llv5) {$v^1_5$};
\node at (Lrv6) {$v^0_6$};
\node at (Llv6) {$v^1_6$};

\draw (0,0) circle (2cm);
\draw[->] (lv1)--(rv1);
\draw[->] (lv2)--(rv2);
\draw[->] (lv3)--(rv3);
\draw[->] (lv4)--(rv4);
\draw[->] (lv5)--(rv5);
\draw[->] (lv6)--(rv6);

\end{tikzpicture}
\caption{\label{fig:examp} Circular-arc graph $G$, its normalized model $\psi$ and the corresponding conformal model $\phi$.}
\end{figure}
The conformal model $\phi$ is represented by the circular word $v_2^0v_1^1v_3^0v_2^1v_4^0v_3^1v_5^0v_4^1v_6^0v_5^1v_1^0v_6^1$, that is,
$$\phi \equiv v_2^0v_1^1v_3^0v_2^1v_4^0v_3^1v_5^0v_4^1v_6^0v_5^1v_1^0v_6^1.$$ 

The elements of $V$ are called \emph{letters},
the elements of $V^{*} = \{v^{0},v^{1}: V \in V\}$ are called \emph{labeled letters}.
Given a set $A' \subset V^{*}$, by $\phi|A'$ we denote either a circular word which is the restriction of $\phi$ to the labeled letter from $A'$ or the set of all maximum contiguous subwords of the circular word $\phi$ containing all the labeled letters from $A'$.
Usually the meaning of $\phi|A'$ is clear from the context; 
otherwise, we say explicitly whether $\phi|A'$ is the circular word or the set of contiguous subwords of $\phi$.
We say $\phi|A'$ forms \emph{$k$ contiguous subwords} in $\phi$ if the set $\phi|A'$ contains exactly $k$ words.
If $k=1$, we say that $\phi|A'$ is a contiguous subword of the circular word $\phi$.
In our example, for $A' = \{v^0_1,v^1_1, v^0_6,v^1_6\}$, $\phi|A' \equiv v_6^0v_1^0v_6^1v^1_1$ if $\phi|A'$ is treated as the circular subword of $\phi$ or $\phi|A' = \{v_6^0, v^0_1v_6^1, v_1^1\}$ if $\phi|A'$ is treated as the set of all contiguous subwords containing all the labeled letters from $A'$.
In this particular case, $\phi|A'$ forms three contiguous subwords in $\phi$.

Let $A$ be a subset of $V$. 
By $A^*$ we denote the set $\{v^{0},v^{1}: v \in A\}$.
We abbreviate and we write $\phi|A$ to denote $\phi|A^{*}$.
In particular, $\phi|\{v_1,v_6\}$ means the same as $\phi|\{v^0_1,v^1_1,v^0_6,v^1_6\}$.

Let $A \subset V$ and let $A' \subset A^*$. 
If $A'$ contains exactly one labeled letter from $\{v^{0},v^{1}\}$ for every $v \in A$, 
then $A'$ is called a \emph{labeled copy} of $A$.
A word $\tau$ is a \emph{labeled permutation} of $A$ if
$\tau$ is a permutation of some labeled copy of $A$.
For example, $\{v^0_1,v^1_2,v^1_6\}$ is a labeled copy and $v^1_2v^0_1v^1_6$ is a labeled permutation of 
$\{v_1,v_2,v_6\}$.
If $A'$ is a labeled copy of $A$ or $\tau$ is a labeled permutation of $A$,
by $u^*$ we denote the unique labeled letter $u^{j} \in \{u^{0},u^{1}\}$ such that $u^j \in A'$
or $u^j \in \tau$, for $u \in A$.

Let $u'$ and $v'$ be two labeled letters in a circular word $\phi$.
We say that a labeled letter $w'$ is \emph{between $u'$ and $v'$ in $\phi$} 
if we pass $w'$ when we traverse $\phi$ from $\phi(u')$ to $\phi(v')$ in the clockwise order.
The labeled letters $v_6^0,v_5^1,v_1^0,v_6^1,v_2^0,v_1^1,v_3^0$ are between $v_4^1$ and $v_2^1$ in $\phi$ and 
the labeled letters $v_4^0,v_3^1,v_5^0$ are between $v_2^1$ and $v_4^1$ in $\phi$.

Let $\psi$ be a circular-arc model of $G$.
Let $L$ be any line in the plane and let $\psi^R$
be the reflection of $\psi$ over $L$ -- see Figure \ref{fig:reflection}.
\begin{figure}[htp!]
\centering
\begin{tikzpicture}[scale=0.80,>=latex,shorten >=-0.4pt,shorten <=-0.4pt]

\coordinate (center) at (0,0) {};

\coordinate (lv10) at ($(center)+(270:2.4cm)$) {};
\coordinate (lv11) at ($(center)+(90:2.4cm)$) {};

\coordinate (lv20) at ($(center)+(0:2.4cm)$) {};
\coordinate (lv21) at ($(center)+(180:2.4cm)$) {};

\coordinate (lv30) at ($(center)+(165:2.4cm)$) {};
\coordinate (lv31) at ($(center)+(30:2.4cm)$) {};

\coordinate (lv40) at ($(center)+(60:2.4cm)$) {};
\coordinate (lv41) at ($(center)+(-60:2.4cm)$) {};

\tikzstyle{every node}=[inner sep=1pt]
\node at (lv10) {$v^0_1$};
\node at (lv11) {$v^1_1$};
\node at (lv20) {$v^0_2$};
\node at (lv21) {$v^1_2$};
\node at (lv30) {$v^0_3$};
\node at (lv31) {$v^1_3$};
\node at (lv40) {$v^0_4$};
\node at (lv41) {$v^1_4$};

\draw (0,0) circle (2cm);

\draw[very thick] ([shift=(180:1.92cm)]0,0) arc (180:360:1.92cm);
\draw[very thick,->] ([shift=(360:1.92cm)]0,0) -- ([shift=(180:1.92cm)]0,0);

\draw[very thick] ([shift=(30:1.92cm)]0,0) arc (30:165:1.92cm);
\draw[very thick,->] ([shift=(165:1.92cm)]0,0) -- ([shift=(30:1.92cm)]0,0);

\draw[red,very thick] ([shift=(90:2cm)]0,0) arc (90:270:2cm);
\draw[red,very thick,->] ([shift=(270:2cm)]0,0) -- ([shift=(90:2cm)]0,0) ;

\draw[red,very thick] ([shift=(-60:2cm)]0,0) arc (-60:60:2cm);
\draw[red,very thick,->] ([shift=(60:2cm)]0,0) -- ([shift=(-60:2cm)]0,0);

\draw[thick, -] (3.5,-3) -- (3.5,3);

\end{tikzpicture} 
\hspace{0.5cm}
\begin{tikzpicture}[scale=0.80,>=latex,shorten >=-0.4pt,shorten <=-0.4pt]

\coordinate (center) at (0,0) {};

\coordinate (lv10) at ($(center)+(270:2.4cm)$) {};
\coordinate (lv11) at ($(center)+(90:2.4cm)$) {};

\coordinate (lv20) at ($(center)+(0:2.4cm)$) {};
\coordinate (lv21) at ($(center)+(180:2.4cm)$) {};

\coordinate (lv30) at ($(center)+(150:2.4cm)$) {};
\coordinate (lv31) at ($(center)+(15:2.4cm)$) {};

\coordinate (lv40) at ($(center)+(120:2.4cm)$) {};
\coordinate (lv41) at ($(center)+(240:2.4cm)$) {};

\tikzstyle{every node}=[inner sep=1pt]
\node at (lv10) {$v^1_1$};
\node at (lv11) {$v^0_1$};
\node at (lv20) {$v^0_2$};
\node at (lv21) {$v^1_2$};
\node at (lv30) {$v^0_3$};
\node at (lv31) {$v^1_3$};
\node at (lv40) {$v^1_4$};
\node at (lv41) {$v^0_4$};

\draw (0,0) circle (2cm);

\draw[very thick] ([shift=(180:1.92cm)]0,0) arc (180:360:1.92cm);
\draw[very thick,->] ([shift=(360:1.92cm)]0,0) -- ([shift=(180:1.92cm)]0,0);

\draw[very thick] ([shift=(15:1.92cm)]0,0) arc (15:150:1.92cm);
\draw[very thick,->] ([shift=(150:1.92cm)]0,0) -- ([shift=(15:1.92cm)]0,0);

\draw[red,very thick] ([shift=(270:2cm)]0,0) arc (270:450:2cm);
\draw[red,very thick,->]  ([shift=(450:2cm)]0,0) -- ([shift=(270:2cm)]0,0);

\draw[red,very thick] ([shift=(120:2cm)]0,0) arc (120:240:2cm);
\draw[red,very thick,->] ([shift=(240:2cm)]0,0) -- ([shift=(120:2cm)]0,0);

\draw[white, thick, -] (3.5,-3) -- (3.5,3);
\end{tikzpicture} 

\caption{\label{fig:reflection} Circular-arc graph $G$, its normalized model $\psi$, and its reflection $\psi^R$.}
\end{figure}
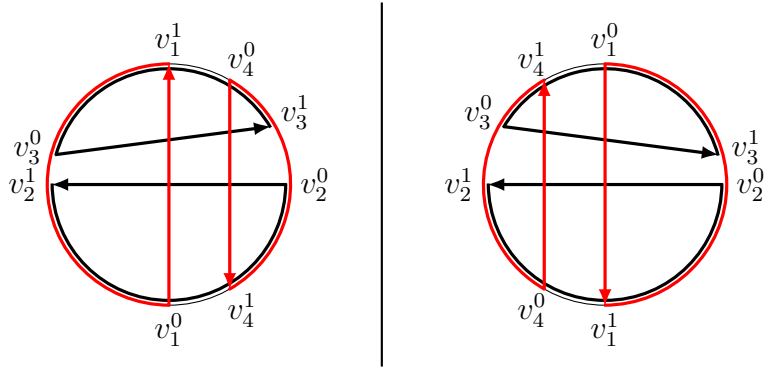
Clearly, $\psi^R$ is also a circular-arc model of $G$.
Now, suppose $\phi$ and $\phi^R$ are conformal models of $G$ associated with $\psi$ and $\psi^R$.
Note that $\phi^R$ is obtained from $\phi$ as follows:
we traverse the circular word $\phi$ in the anti-clockwise order and we replace $v^{0}$ by $v^{1}$
and $v^{1}$ by $v^{0}$, for every $v \in V$.
Indeed, in our example
$$ \phi \equiv v_2^0v_4^1v_1^0v_2^1v_3^0v_1^1v_4^0v_3^1 \text{ and } \phi^R \equiv v_3^0v_4^1v_1^0v_3^1v_2^0v_1^1v_4^0v_2^1.$$
The conformal model $\phi^R$ obtained this way is called the \emph{reflection} of $\phi$.
Note the following relation between $\phi$ and $\phi^R$:
for every $u \sim v$ the circular circular word $u^{0}v^{0}u^1v^1$ appears in $\phi$ iff 
the circular word $u^0v^1u^1v^0$ appears in $\phi^R$ --
see for the oriented chords $\phi(v_1),\phi(v_2)$ and $\phi^R(v_1),\phi^R(v_2)$ in Figure \ref{fig:reflection}.

Let $G_{ov}$ be the overlap graph of $G$ and let $\psi$ be a non-oriented chord model of $G_{ov}$.
The \emph{word representation} $\tau$ of $\psi$ is obtained similarly to the word representation of a conformal model of $G_{ov}$, except that we append $v$ to $\tau$ whenever we pass the 
end of the chord $\psi(v)$ for $v \in V$.
We write $\psi \equiv \tau$ if $\tau$ is a word representation of $\psi$.
Two chord models $\psi_1$ and $\psi_2$ of $G_{ov}$ are \emph{equivalent}, 
written $\psi_1 \equiv \psi_2$, if their word representations are equal.
The reflection $\psi^R$ of a chord model $\psi$ of $G_{ov}$ is defined analogously.

We use similar notations for circle graphs $G_{ov}$ and their chord models as for conformal models $\phi$.

Suppose $G=(V,E)$ is a graph with no twins and no universal vertices and suppose $G_{ov} = (V,{\sim})$
is the overlap graph associated with $G$.
We denote the complement of $G_{ov}$ by $(V,{\parallel})$.
If $U$ is a subset of $V$, by $G[U]$, $(U,{\sim})$, and $(U,{\parallel})$ we denote the subgraphs of $G$, $(V,{\sim})$, and $(V,{\parallel})$ induced by the set $U$, respectively.
For two sets $U_1,U_2 \subset V$, we write $U_1 \sim U_2$ ($U_1 \parallel U_2$)
if $u_1 \sim u_2$ ($u_1 \parallel u_2$, respectively) for every $u_1 \in U_1$ and $u_2 \in U_2$.

In the rest of the paper we will require an analogue of Theorem \ref{thm:normalized_conformal_correspondence}
extended on the induced subgraphs of $G$ and $G_{ov}$.
\begin{definition}
Suppose $U$ is a non-empty subset of $V$.
A circular-arc model $\psi$ of $G[U]$ is \emph{normalized} if every pair of vertices $(v,u)$ from $U$ satisfies conditions \ref{def:normalized_model}.\eqref{item:norm_di}-\eqref{item:norm_ov}.
\end{definition}
Note that the pair $(v,u)$ needs to satisfy conditions \ref{def:normalized_model}.\eqref{item:norm_di}-\eqref{item:norm_ov} with respect to $M_G$, not with respect to $M_{G[U]}$.
In particular, for any non-empty subset $U$ of $V$, if $\psi$ is a normalized model of $G$,
then $\psi$ restricted to $U$ is a normalized model of $G[U]$.

\begin{definition}
Let $U$ be a non-empty subset of $V$.
An oriented chord model $\phi$ of $(U,{\sim})$ is \emph{conformal}
if for every $v \in U$ the oriented chords $\phi(u)$ for
$u \in \leftside(v) \cap U$ lie on the left side of $\phi(v)$ 
and the oriented chords $\phi(u)$ for $u \in \rightside(v) \cap U$ 
lie on the right side of $\phi(v)$.
\end{definition}
Clearly, if $\phi$ is a conformal model of $(U,{\sim})$, then $\phi|U$ is a conformal model of $(U,{\sim})$.
\begin{theorem}
\label{thm:subgraphs_normalized_conformal_correspondence}
Let $G$ be a circular-arc graph and let $U$ be a non-empty subset of $V$.
There is a one-to-one correspondence between the normalized models of $G[U]$ 
and the conformal models of $(U,{\sim})$.
\end{theorem}

\section{Tools}
\label{sec:tools}
\subsection{The structure of all representations of a circle graph}
\label{sec:split_decomposition}
The description of the structure of all chord models of circle graphs, presented in this subsection, is 
taken from the article \cite{CFK13} by Chaplick, Fulek, and Klav\'{i}k.
The concept of split decomposition is due to Cunningham \cite{Cun82}, Theorem \ref{thm:cicle_representation_no_split} is due to
Gabor, Supowit, and Hsu \cite{GSH89}, relation $\diamond$ is due to Chaplick, Fulek, and Klav\'{i}k \cite{CFK13}, which were inspired by Naji \cite{Naji85} (see \cite{CFK13} for more details), maximal splits are due to Chaplick, Fulek, and Klav\'{i}k \cite{CFK13}.

Suppose $G_{ov} = (V,{\sim})$ is a connected circle graph.
A tuple $(A, \alpha(A), B, \alpha(B))$ is a \emph{split} in $G_{ov}$ if: 
\begin{itemize}
\item The sets $A$, $B$, $\alpha(A)$, $\alpha(B)$ form a partition of $V$,
\item We have $A \neq \emptyset$ and $B \neq \emptyset$, but possibly $\alpha(A) = \emptyset$ or $\alpha(B) = \emptyset$,
\item We have $A \sim B$.
\item We have $\alpha(A) \parallel (B \cup \alpha(B))$ and $\alpha(B) \parallel A \cup \alpha(A)$,
\end{itemize}
see Figure \ref{fig:split}.
Since $G_{ov}$ is connected, $(A,\alpha(A),B,\alpha(B))$ can be uniquely recovered from the sets $A$ and $B$.
Hence, without loosing any information, we say $(A,\alpha(A),B,\alpha(B))$ is just the split between $A$ and $B$,
and we denote $(A,\alpha(A),B,\alpha(B))$ simply by $(A,B)$.

\begin{figure}[!htp]

\begin{tikzpicture}[scale=0.4]
    \coordinate (A_center) at (-2,0) {};
    \coordinate (lA) at (-2,2.5) {};
    
    \coordinate (Ar1) at ($(A_center)+(-40:1 and 2)$) {};
    \coordinate (Ar2) at ($(A_center)+(-20:1 and 2)$) {};
    \coordinate (Ar3) at ($(A_center)+(0:1 and 2)$) {};
    \coordinate (Ar4) at ($(A_center)+(20:1 and 2)$) {};
    \coordinate (Ar5) at ($(A_center)+(40:1 and 2)$) {};
    \coordinate (Al1) at ($(A_center)+(220:1 and 2)$) {};
    \coordinate (Al2) at ($(A_center)+(200:1 and 2)$) {};
    \coordinate (Al3) at ($(A_center)+(180:1 and 2)$) {};
    \coordinate (Al4) at ($(A_center)+(160:1 and 2)$) {};
    \coordinate (Al5) at ($(A_center)+(140:1 and 2)$) {};

    \coordinate (B_center) at (2,0) {};    
    \coordinate (lB) at (2,2.5) {};

    \coordinate (Br1) at ($(B_center)+(-40:1 and 2)$) {};
    \coordinate (Br2) at ($(B_center)+(-20:1 and 2)$) {};
    \coordinate (Br3) at ($(B_center)+(0:1 and 2)$) {};
    \coordinate (Br4) at ($(B_center)+(20:1 and 2)$) {};
    \coordinate (Br5) at ($(B_center)+(40:1 and 2)$) {};
    \coordinate (Bl1) at ($(B_center)+(220:1 and 2)$) {};
    \coordinate (Bl2) at ($(B_center)+(200:1 and 2)$) {};
    \coordinate (Bl3) at ($(B_center)+(180:1 and 2)$) {};
    \coordinate (Bl4) at ($(B_center)+(160:1 and 2)$) {};
    \coordinate (Bl5) at ($(B_center)+(140:1 and 2)$) {};
    
    \coordinate (AA_center) at (-5.5,0) {};
    \coordinate (lAA) at (-5.5,2.2) {};
    \coordinate (AAr1) at ($(AA_center)+(-40:1 and 1.5)$) {};
    \coordinate (AAr2) at ($(AA_center)+(-20:1 and 1.5)$) {};
    \coordinate (AAr3) at ($(AA_center)+(0:1 and 1.5)$) {};
    \coordinate (AAr4) at ($(AA_center)+(20:1 and 1.5)$) {};
    \coordinate (AAr5) at ($(AA_center)+(40:1 and 1.5)$) {};

    \coordinate (AB_center) at (5.5,0) {};
    \coordinate (lAB) at (5.5,2.2) {};
    \coordinate (ABl1) at ($(AB_center)+(220:1 and 1.5)$) {};
    \coordinate (ABl2) at ($(AB_center)+(200:1 and 1.5)$) {};
    \coordinate (ABl3) at ($(AB_center)+(180:1 and 1.5)$) {};
    \coordinate (ABl4) at ($(AB_center)+(160:1 and 1.5)$) {};
    \coordinate (ABl5) at ($(AB_center)+(140:1 and 1.5)$) {};

    \draw[fill=gray!60] (A_center) ellipse (1 and 2);
    \draw[fill=gray!60] (B_center) ellipse (1 and 2);
    \draw[fill=gray!30] (AA_center) ellipse (1 and 1.5);
    \draw[fill=gray!30] (AB_center) ellipse (1 and 1.5);

    \tikzstyle{every node}=[inner sep=1pt]
    \node at (lA) {$A$};
    \node at (lB) {$B$};
    \node at (lAA) {$\alpha(A)$};
    \node at (lAB) {$\alpha(B)$};

    \draw[-] (Ar1)--(Bl1);
    \draw[-] (Ar1)--(Bl2);
    \draw[-] (Ar1)--(Bl3);
    \draw[-] (Ar1)--(Bl4);
    \draw[-] (Ar1)--(Bl5);

    \draw[-] (Ar2)--(Bl1);
    \draw[-] (Ar2)--(Bl2);
    \draw[-] (Ar2)--(Bl3);
    \draw[-] (Ar2)--(Bl4);
    \draw[-] (Ar2)--(Bl5);
    
    \draw[-] (Ar3)--(Bl1);
    \draw[-] (Ar3)--(Bl2);
    \draw[-] (Ar3)--(Bl3);
    \draw[-] (Ar3)--(Bl4);
    \draw[-] (Ar3)--(Bl5);
    
    \draw[-] (Ar4)--(Bl1);
    \draw[-] (Ar4)--(Bl2);
    \draw[-] (Ar4)--(Bl3);
    \draw[-] (Ar4)--(Bl4);
    \draw[-] (Ar4)--(Bl5);
    
    \draw[-] (Ar5)--(Bl1);
    \draw[-] (Ar5)--(Bl2);
    \draw[-] (Ar5)--(Bl3);
    \draw[-] (Ar5)--(Bl4);
    \draw[-] (Ar5)--(Bl5);

    \draw[-] (AAr1)--(Al1);
    \draw[-] (AAr2)--(Al2);
    \draw[-] (AAr3)--(Al3);
    \draw[-] (AAr4)--(Al4);
    \draw[-] (AAr5)--(Al5);

    \draw[-] (Br1)--(ABl1);
    \draw[-] (Br2)--(ABl2);
    \draw[-] (Br3)--(ABl3);
    \draw[-] (Br4)--(ABl4);
    \draw[-] (Br5)--(ABl5);
    
    \draw[black] (-8,-4)--(-7.5,-4);
\end{tikzpicture}
\hspace{0.5cm}
\begin{tikzpicture}[scale=0.6,>=latex,shorten >=-0.4pt,shorten <=-0.4pt]
\coordinate (center) at (0,0) {};

\coordinate (tauA) at ($(center)+(90:2.4)$) {};
\coordinate (tau'A) at ($(center)+(270:2.4)$) {};
\coordinate (tauB) at ($(center)+(0:2.4)$) {};
\coordinate (tau'B) at ($(center)+(180:2.4)$) {};

\coordinate (A) at ($(center)+(70:1.3)$) {};
\coordinate (B) at ($(center)+(200:1.3)$) {};

\coordinate (u1) at ($(center)+(120:2)$) {};
\coordinate (u2) at ($(center)+(110:2)$) {};
\coordinate (u3) at ($(center)+(93:2)$) {};
\coordinate (u4) at ($(center)+(87:2)$) {};
\coordinate (u5) at ($(center)+(70:2)$) {};
\coordinate (u6) at ($(center)+(60:2)$) {};

\coordinate (b1) at ($(center)+(240:2)$) {};
\coordinate (b2) at ($(center)+(250:2)$) {};
\coordinate (b3) at ($(center)+(267:2)$) {};
\coordinate (b4) at ($(center)+(273:2)$) {};
\coordinate (b5) at ($(center)+(290:2)$) {};
\coordinate (b6) at ($(center)+(300:2)$) {};

\coordinate (l1) at ($(center)+(160:2)$) {};
\coordinate (l2) at ($(center)+(177:2)$) {};
\coordinate (l3) at ($(center)+(183:2)$) {};
\coordinate (l4) at ($(center)+(200:2)$) {};

\coordinate (r1) at ($(center)+(3:2)$) {};
\coordinate (r2) at ($(center)+(-3:2)$) {};

\draw (center) circle (2cm);

\draw[thick,-] (u1)--(u6);
\draw[thick,-] (u2)--(u5);
\draw[thick,-] (u3)--(b3);
\draw[thick,-] (u4)--(b4);
\draw[thick,-] (b2)--(b5);

\draw[thick,-] (l2)--(r1);
\draw[thick,-] (l3)--(r2);
\draw[thick,-] (l1)--(l4);

\draw[red, very thick,<-] ([shift=(50:2cm)]0,0) arc (50:130:2cm);
\draw[blue, very thick,<-] ([shift=(230:2cm)]0,0) arc (230:310:2cm);

\draw[black, very thick,<-] ([shift=(-30:2cm)]0,0) arc (-30:30:2cm);
\draw[black, very thick,<-] ([shift=(150:2cm)]0,0) arc (150:210:2cm);

\tikzstyle{every node}=[inner sep=1pt]
\node[red] at (tauA) {$\tau_A$};
\node[blue] at (tau'A) {$\tau'_A$};;
\node[black] at (tauB) {$\tau_B$};
\node[black] at (tau'B) {$\tau'_B$};
\node at (A) {$A$};
\node at (B) {$B$};

\draw[black] (-2.5,-3.0)--(-2,-3.0);
\end{tikzpicture}
\hspace{0.25cm}
\begin{tikzpicture}[scale=0.6,>=latex,shorten >=-0.4pt,shorten <=-0.4pt]
\coordinate (center) at (0,0) {};

\coordinate (tauA) at ($(center)+(90:2.4)$) {};
\coordinate (tau'A) at ($(center)+(270:2.4)$) {};
\coordinate (tauB) at ($(center)+(0:2.4)$) {};
\coordinate (tau'B) at ($(center)+(180:2.4)$) {};

\coordinate (A) at ($(center)+(70:1.3)$) {};
\coordinate (B) at ($(center)+(200:1.3)$) {};

\coordinate (u1) at ($(center)+(120:2)$) {};
\coordinate (u2) at ($(center)+(110:2)$) {};
\coordinate (u3) at ($(center)+(93:2)$) {};
\coordinate (u4) at ($(center)+(87:2)$) {};
\coordinate (u5) at ($(center)+(70:2)$) {};
\coordinate (u6) at ($(center)+(60:2)$) {};

\coordinate (b1) at ($(center)+(240:2)$) {};
\coordinate (b2) at ($(center)+(250:2)$) {};
\coordinate (b3) at ($(center)+(267:2)$) {};
\coordinate (b4) at ($(center)+(273:2)$) {};
\coordinate (b5) at ($(center)+(290:2)$) {};
\coordinate (b6) at ($(center)+(300:2)$) {};

\coordinate (l1) at ($(center)+(160:2)$) {};
\coordinate (l2) at ($(center)+(177:2)$) {};
\coordinate (l3) at ($(center)+(183:2)$) {};
\coordinate (l4) at ($(center)+(200:2)$) {};

\coordinate (r1) at ($(center)+(3:2)$) {};
\coordinate (r2) at ($(center)+(-3:2)$) {};

\draw (center) circle (2cm);

\draw[thick,-] (u2)--(u5);
\draw[thick,-] (u3)--(b3);
\draw[thick,-] (u4)--(b4);
\draw[thick,-] (b2)--(b5);
\draw[thick,-] (b1)--(b6);

\draw[thick,-] (l2)--(r1);
\draw[thick,-] (l3)--(r2);
\draw[thick,-] (l1)--(l4);

\draw[blue, very thick,<-] ([shift=(50:2cm)]0,0) arc (50:130:2cm);
\draw[red, very thick,<-] ([shift=(230:2cm)]0,0) arc (230:310:2cm);

\draw[black, very thick,<-] ([shift=(-30:2cm)]0,0) arc (-30:30:2cm);
\draw[black, very thick,<-] ([shift=(150:2cm)]0,0) arc (150:210:2cm);

\tikzstyle{every node}=[inner sep=1pt]
\node[blue] at (tauA) {$\tau'_A$};
\node[red] at (tau'A) {$\tau_A$};;
\node[black] at (tauB) {$\tau_B$};
\node[black] at (tau'B) {$\tau'_B$};
\node at (A) {$A$};
\node at (B) {$B$};

\draw[black] (-2.5,-3.0)--(-2,-3.0);

\end{tikzpicture}
\caption{\label{fig:split} Split $(\alpha(A),A,\alpha(B),B)$ in $G_{ov}$ and two possible chord models of $G_{ov}$: $\tau_A\tau_B\tau'_A\tau'_B$ and $\tau'_A\tau_B\tau_A\tau'_B$.}
\end{figure}
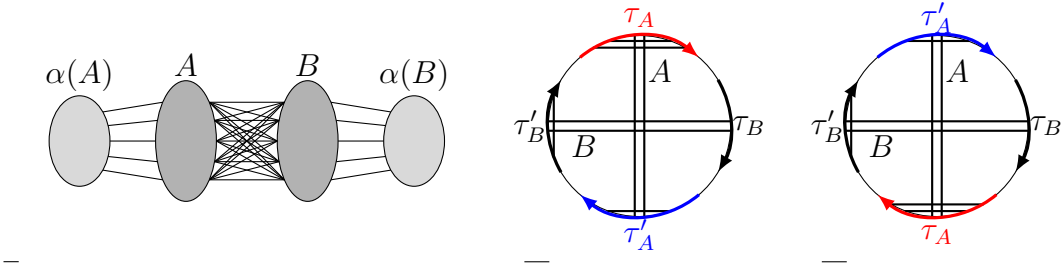

A split $(A,B)$ is \emph{non-trivial} if $|A \cup \alpha(A)| \geq 2$ and $|B \cup \alpha(B)| \geq 2$;
otherwise $(A,B)$ is \emph{trivial}.
\begin{theorem}[\cite{GSH89}]
\label{thm:cicle_representation_no_split}
If $G_{ov}$ has no non-trivial split, $G_{ov}$ has only two chord models, one being the reflection of the other.
\end{theorem}
On the other hand, if $G_{ov}$ has non-trivial splits, 
$G_{ov}$ may have many non-equivalent chord models -- see Figure \ref{fig:split}.

A split in $G_{ov}$ between $A$ and $B$ is \emph{maximal}
if there is no split in $G_{ov}$ between $A'$ and $B'$, where $A'$ and $B'$ are such that $A \subseteq A'$, $B \subseteq B'$, and $|A| < |A'|$ or $|B| < |B|'$.
Lemma~1 in \cite{CFK13} provides the following characterization of maximal splits in $G_{ov}$: 
a split between $A$ and $B$ is maximal if and only if
there exists no $C \subseteq \alpha(A)$ such that $(C,{\sim})$ is a connected component in $(\alpha(A),{\sim})$
and for every vertex $u  \in C$ either $u \sim A$ or $u \parallel A$, 
and similarly for $\alpha(B)$ and $B$.
This observation allows to present the algorithm for computing a maximal split in $G_{ov}$ (see \cite{CFK13} for more details):
\begin{itemize}
 \item start with any non-trivial split between $A$ and $B$,
 \item while there exists $C$ as described in Lemma 1 of \cite{CFK13}: 
 if $C \subseteq \alpha(A)$ set $A = A$ and $B = B \cup C'$, 
 and if $C \subseteq \alpha(B)$, set $B = B$ and $A = A \cup C'$,
 where $C'$ is the set of all vertices from $C$ adjacent to $A$ if $C \subset \alpha(A)$ or adjacent to $B$ if $C \subset \alpha(B)$, respectively,
 \item return $(A,B)$.
\end{itemize}

Suppose $(A,B)$ is a maximal split in $G_{ov}$ produced by the above algorithm.
Note that $(A,B)$ might be trivial. 
Then, Lemma~2 of~\cite{CFK13} proves the following property:
if $(A,B)$ is trivial with $|A|= \{a\}$ and $\alpha(A) = \emptyset$,
then $a$ is an \emph{articulation} vertex in $G_{ov}$, i.e. $(V \setminus \{a\}, {\sim})$ is disconnected.

\subsection{The structure of chord models of $G_{ov}$ with respect to a non-trivial maximal split $(A,B)$}
\label{subsec:structure_non_trivial_split}
Suppose $G_{ov}$ has a non-trivial maximal split $(A,B)$.
Let $C = A\cup B$. 
Following \cite{CFK13}, let $\diamond$ be the smallest equivalence relation on $C$ containing all the pairs $(u,v) \in C \times C$ such that:
\begin{itemize}
\item $u \parallel v$,
\item $u,v$ are connected by a path in $(V,{\sim})$ with 
all the inner vertices in $\alpha(A) \cup \alpha(B)$.
\end{itemize}
Suppose $C_1,\ldots,C_k$ are the equivalence classes of $\diamond$.
Note that $C_i \subseteq A$ or $C_i \subseteq B$ for every $i \in [k]$, and
hence $k \geq 2$.
Observe that: 
\begin{itemize}
\item $C_i \sim C_j$ for every $i \neq j$, $i,j \in [k]$.
\end{itemize}
Following \cite{CFK13}, one can uniquely partition the set $V \setminus C$ into 
the sets $\alpha(C_1),\ldots,\alpha(C_k)$ ($\alpha(C_i)$ might be empty) so as:
\begin{itemize}
 \item $\alpha(C_i) \parallel (\alpha(C_j) \cup C_j)$ for every $i \neq j$, $i,j \in [k]$.
\end{itemize}
See Figure \ref{fig:non_trivial_maximal_split} for an illustration.

\input splits_non_trivial_maximal_split.tex

Further, let $G_i$ by a graph obtained from $G_{ov}$ by contracting the vertices
from $V \setminus (C_i \cup \alpha(C_i))$ into a single vertex $v_i$.
Thus, $G_i$ is such that $V(G_i) = C_i \cup \alpha(C_i) \cup \{v_i\}$, 
$v_iv \in E(G_i)$ for every $v \in C_i$, $v_iv \notin E(G_i)$ for every $v \in \alpha(C_i)$, 
and $uv \in E(G_i)$ iff $u \sim v$ for every $u,v \in C_i \cup \alpha(C_i)$.
Note that every chord model of $G_i$ has the form $v_i \tau_i v_i \tau'_i$, where
every $v \in C_i$ occurs in both words $\tau$ and $\tau'$ exactly once and 
every $v \in \alpha(C_i)$ occurs twice either in $\tau_i$ or in $\tau'_i$.
The next theorem describes the relationship between the set of all chord models of $G_{ov}$ and 
the set of all chord models of $G_i$.

\begin{theorem}[Proposition~1 from \cite{CFK13}] 
\label{thm:non_trivial_split_representations}
The following statements hold:
\begin{enumerate}
\item If $v_i \tau_i v_i \tau'_i$ is a chord model of $G_i$ for $i \in [k]$,
$i_1, \ldots, i_k$ is a permutation of $[k]$, and the words $\mu_i,\mu'_i$ are such that $\{\mu_i,\mu'_i\} = \{\tau_i,\tau'_i\}$ for $i \in [k]$, then 
$$\tau \equiv \mu_{i_1}\ldots \mu_{i_k}\mu'_{i_1}\ldots \mu'_{i_k},$$
is a chord model of $G_{ov}$.
\item If $\tau$ is a chord model of $G_{ov}$, then 
$$\tau \equiv \mu_{i_1}\ldots \mu_{i_k}\mu'_{i_1}\ldots \mu'_{i_k},$$
where $i_1, \ldots, i_k$ is a permutation of $[k]$  and $v_{i_j} \mu_{i_j} v_{i_j} \mu'_{i_j}$ is a chord model of $G_{i_j}$ for $j \in [k]$.
\end{enumerate}
\end{theorem}
See Figure \ref{fig:non_trivial_maximal_split} for an illustration.

\subsection{The structure of chord models of $G_{ov}$ with respect to a trivial maximal split $(A,B)$}
\label{subsec:structure_trivial_split}
Suppose $(A,B)$ is a trivial maximal split in $G_{ov}$.
Without loss of generality we assume that $A = \{a\}$ and $\alpha(A) = \emptyset$.
We recall that $a$ is an articulation of $G_{ov}$ by Lemma~2 in~\cite{CFK13}.
Suppose that $D_1,\dots,D_k \subset V\setminus \{a\}$ are such that $(D_i,{\sim})$
is a connected component of $(V \setminus \{a\},{\sim})$ for every $i \in [k]$.
Clearly, $k \geq 2$ as $a$ is an articulation in $G_{ov}$.
Let $C_i = \{v \in D_i: v \sim a\}$ and $\alpha(C_i) = \{v \in D_i: v \parallel a\}$.
Let $G_i$ be the restriction of $G_{ov}$ to the set $\{a\} \cup C_i \cup \alpha(C_i)$, 
i.e. $G_i = (\{a\} \cup C_i \cup \alpha(C_i), {\sim})$.
Note that every chord model of $G_i$ has the form $a \tau_i a \tau'_i$, where
every $v \in C_i$ occurs in both words $\tau_i$ and $\tau'_i$ exactly once and
every $v \in \alpha(C_i)$ occurs twice in either $\tau_i$ or in $\tau'_i$.
The next theorem describes the relation between the set of all chord models of $G_{ov}$ and 
the set of all chord models of $G_i$.

\begin{theorem}[Proposition 2 in \cite{CFK13}] 
\label{thm:trivial_split_representations}
The following statements hold:
\begin{enumerate}
\item If $a \tau_i a \tau'_i$ is a chord model of $G_i$ for $i \in [k]$,
$i_1, \ldots, i_k$ is a permutation of $[k]$, and the words $\mu_i,\mu'_i$ are such that $\{\mu_i,\mu'_i\} = \{\tau_i,\tau'_i\}$ for $i \in [k]$, then 
$$\tau \equiv a\mu_{i_1}\ldots \mu_{i_k} a\mu'_{i_k}\ldots \mu'_{i_1}$$
is a chord model of $G_{ov}$.
\item If $\tau$ is a chord model of $G_{ov}$, then 
$$\tau \equiv a\mu_{i_1} \ldots \mu_{i_k}a\mu'_{i_k} \ldots \mu'_{i_1},$$
where $i_1, \ldots, i_k$ is a permutation of $[k]$ and $a \mu_{i_j} a \mu'_{i_j}$ is a chord model of $G_{i_j}$ for $j \in [k]$.
\end{enumerate}
\end{theorem}
See Figure \ref{fig:trivial_maximal_split} for an illustration.
\begin{figure}[!htp]
\begin{tikzpicture}[scale=0.35]
    \coordinate (a) at (0,-2) {};
    \coordinate (la) at (0,-3) {};
    \coordinate (C1_center) at (-6,2) {};
    \coordinate (C2_center) at (0,2) {};
    \coordinate (C3_center) at (6,2) {};
    
    \coordinate (AC1_center) at (-6,5) {};
    \coordinate (AC2_center) at (0,5) {};
    \coordinate (AC3_center) at (6,5) {};

    \coordinate (lC1) at ($(C1_center)+(180:2.8 and 1)$) {};
    \coordinate (lC2) at ($(C2_center)+(180:2.8 and 1)$) {};
    \coordinate (lC3) at ($(C3_center)+(180:2.8 and 1)$) {};

    \coordinate (lAC1) at ($(AC1_center)+(180:2.6 and 1)$) {};
    \coordinate (lAC2) at ($(AC2_center)+(180:2.6 and 1)$) {};
    \coordinate (lAC3) at ($(AC3_center)+(180:2.6 and 1)$) {};

    \coordinate (C1l1) at ($(C1_center)+(240:2 and 1)$) {};
    \coordinate (C1l2) at ($(C1_center)+(270:2 and 1)$) {};
    \coordinate (C1l3) at ($(C1_center)+(290:2 and 1)$) {};
    \coordinate (C1l4) at ($(C1_center)+(310:2 and 1)$) {};

    \coordinate (C1u1) at ($(C1_center)+(120:2 and 1)$) {};
    \coordinate (C1u2) at ($(C1_center)+(90:2 and 1)$) {};
    \coordinate (C1u3) at ($(C1_center)+(60:2 and 1)$) {};

    \coordinate (AC1l1) at ($(AC1_center)+(240:1 and 1)$) {};
    \coordinate (AC1l2) at ($(AC1_center)+(270:1 and 1)$) {};
    \coordinate (AC1l3) at ($(AC1_center)+(300:1 and 1)$) {};

    \coordinate (C2l1) at ($(C2_center)+(300:2 and 1)$) {};
    \coordinate (C2l2) at ($(C2_center)+(285:2 and 1)$) {};
    \coordinate (C2l3) at ($(C2_center)+(270:2 and 1)$) {};
    \coordinate (C2l4) at ($(C2_center)+(255:2 and 1)$) {};
    \coordinate (C2l5) at ($(C2_center)+(240:2 and 1)$) {};

    \coordinate (C2u1) at ($(C2_center)+(120:2 and 1)$) {};
    \coordinate (C2u2) at ($(C2_center)+(90:2 and 1)$) {};
    \coordinate (C2u3) at ($(C2_center)+(60:2 and 1)$) {};

    \coordinate (AC2l1) at ($(AC2_center)+(240:1 and 1)$) {};
    \coordinate (AC2l2) at ($(AC2_center)+(270:1 and 1)$) {};
    \coordinate (AC2l3) at ($(AC2_center)+(300:1 and 1)$) {};

    \coordinate (C3l1) at ($(C3_center)+(300:2 and 1)$) {};
    \coordinate (C3l2) at ($(C3_center)+(270:2 and 1)$) {};
    \coordinate (C3l3) at ($(C3_center)+(250:2 and 1)$) {};
    \coordinate (C3l4) at ($(C3_center)+(230:2 and 1)$) {};
    
    \coordinate (C3u1) at ($(C3_center)+(120:2 and 1)$) {};
    \coordinate (C3u2) at ($(C3_center)+(90:2 and 1)$) {};
    \coordinate (C3u3) at ($(C3_center)+(60:2 and 1)$) {};

    \coordinate (AC3l1) at ($(AC3_center)+(240:1 and 1)$) {};
    \coordinate (AC3l2) at ($(AC3_center)+(270:1 and 1)$) {};
    \coordinate (AC3l3) at ($(AC3_center)+(300:1 and 1)$) {};

    \draw[fill=gray!60] (C1_center) ellipse (2 and 1);
    \draw[fill=gray!60] (C2_center) ellipse (2 and 1);
    \draw[fill=gray!60] (C3_center) ellipse (2 and 1);

    \draw[fill=gray!30] (AC1_center) ellipse (1 and 1);
    \draw[fill=gray!30] (AC2_center) ellipse (1 and 1);
    \draw[fill=gray!30] (AC3_center) ellipse (1 and 1);

    \tikzstyle{every node}=[inner sep=1pt]
    \node at (lC1) {$C_1$};
    \node at (lC2) {$C_2$};
    \node at (lC3) {$C_3$};
    \node at (lAC1) {$\alpha(C_1)$};
    \node at (lAC2) {$\alpha(C_2)$};
    \node at (lAC3) {$\alpha(C_3)$};
    \node at (la) {$a$};
    
    \tikzstyle{every node}=[circle,minimum size=5pt,inner sep=0pt,draw,fill]
    \node at (a) {};

    \draw[-] (a)--(C1l1);
    \draw[-] (a)--(C1l2);
    \draw[-] (a)--(C1l3);
    \draw[-] (a)--(C1l4);
    \draw[-] (a)--(C2l1);
    \draw[-] (a)--(C2l2);
    \draw[-] (a)--(C2l3);
    \draw[-] (a)--(C2l4);
    \draw[-] (a)--(C2l5);
    \draw[-] (a)--(C3l1);
    \draw[-] (a)--(C3l2);
    \draw[-] (a)--(C3l3);
    \draw[-] (a)--(C3l4);
    
    \draw[-] (C1u1)--(AC1l1);
    \draw[-] (C1u2)--(AC1l2);
    \draw[-] (C1u3)--(AC1l3);
    
    \draw[-] (C2u1)--(AC2l1);
    \draw[-] (C2u2)--(AC2l2);
    \draw[-] (C2u3)--(AC2l3);

    \draw[-] (C3u1)--(AC3l1);
    \draw[-] (C3u2)--(AC3l2);
    \draw[-] (C3u3)--(AC3l3);
\draw[white] (-2.5,-4)--(-2,-4);
\end{tikzpicture}
\hspace{0.5cm}
\begin{tikzpicture}[scale=0.7,>=latex,shorten >=-0.4pt,shorten <=-0.4pt]
\coordinate (center) at (0,0) {};

\coordinate (tau2) at ($(center)+(45:2.4)$) {};
\coordinate (tau'2) at ($(center)+(315:2.4)$) {};

\coordinate (tau3) at ($(center)+(90:2.4)$) {};
\coordinate (tau'3) at ($(center)+(270:2.4)$) {};

\coordinate (tau4) at ($(center)+(135:2.4)$) {};
\coordinate (tau'4) at ($(center)+(225:2.4)$) {};

\coordinate (r) at ($(center)+(0:2)$) {};
\coordinate (ur) at ($(center)+(45:2)$) {};
\coordinate (u) at ($(center)+(90:2)$) {};
\coordinate (ul) at ($(center)+(135:2)$) {};
\coordinate (l) at ($(center)+(180:2)$) {};
\coordinate (bl) at ($(center)+(225:2)$) {};
\coordinate (b) at ($(center)+(270:2)$) {};
\coordinate (br) at ($(center)+(315:2)$) {};

\draw (center) circle (2cm);

\draw[very thick,-] (r)--(l);
\draw[very thick,-] (ur)--(br);
\draw[very thick,-] (u)--(b);
\draw[very thick,-] (ul)--(bl);

\draw[very thick] ([shift=(30:2cm)]0,0) arc (30:60:2cm);
\draw[very thick] ([shift=(75:2cm)]0,0) arc (75:105:2cm);
\draw[very thick] ([shift=(120:2cm)]0,0) arc (120:150:2cm);
\draw[very thick] ([shift=(210:2cm)]0,0) arc (210:240:2cm);
\draw[very thick] ([shift=(255:2cm)]0,0) arc (255:285:2cm);
\draw[very thick] ([shift=(300:2cm)]0,0) arc (300:330:2cm);

\tikzstyle{every node}=[inner sep=1pt]
\node at (tau1) {$a$};
\node at (tau'1) {$a$};
\node at (tau2) {$\tau_1$};
\node at (tau'2) {$\tau'_1$};
\node at (tau3) {$\tau'_2$};
\node at (tau'3) {$\tau_2$};
\node at (tau4) {$\tau_3$};
\node at (tau'4) {$\tau'_3$};

\draw[white] (-2.5,-3)--(-2,-3);
\end{tikzpicture}
\begin{tikzpicture}[scale=0.7,>=latex,shorten >=-0.4pt,shorten <=-0.4pt]
\coordinate (center) at (0,0) {};

\coordinate (tau2) at ($(center)+(45:2.4)$) {};
\coordinate (tau'2) at ($(center)+(315:2.4)$) {};

\coordinate (tau3) at ($(center)+(90:2.4)$) {};
\coordinate (tau'3) at ($(center)+(270:2.4)$) {};

\coordinate (tau4) at ($(center)+(135:2.4)$) {};
\coordinate (tau'4) at ($(center)+(225:2.4)$) {};

\coordinate (r) at ($(center)+(0:2)$) {};
\coordinate (ur) at ($(center)+(45:2)$) {};
\coordinate (u) at ($(center)+(90:2)$) {};
\coordinate (ul) at ($(center)+(135:2)$) {};
\coordinate (l) at ($(center)+(180:2)$) {};
\coordinate (bl) at ($(center)+(225:2)$) {};
\coordinate (b) at ($(center)+(270:2)$) {};
\coordinate (br) at ($(center)+(315:2)$) {};

\draw (center) circle (2cm);

\draw[very thick,-] (r)--(l);
\draw[very thick,-] (ur)--(br);
\draw[very thick,-] (u)--(b);
\draw[very thick,-] (ul)--(bl);

\draw[very thick] ([shift=(30:2cm)]0,0) arc (30:60:2cm);
\draw[very thick] ([shift=(75:2cm)]0,0) arc (75:105:2cm);
\draw[very thick] ([shift=(120:2cm)]0,0) arc (120:150:2cm);
\draw[very thick] ([shift=(210:2cm)]0,0) arc (210:240:2cm);
\draw[very thick] ([shift=(255:2cm)]0,0) arc (255:285:2cm);
\draw[very thick] ([shift=(300:2cm)]0,0) arc (300:330:2cm);

\tikzstyle{every node}=[inner sep=1pt]
\node at (tau1) {$a$};
\node at (tau'1) {$a$};
\node at (tau2) {$\tau'_2$};
\node at (tau'2) {$\tau_2$};
\node at (tau3) {$\tau_1$};
\node at (tau'3) {$\tau'_1$};
\node at (tau4) {$\tau_3$};
\node at (tau'4) {$\tau'_3$};

\draw[white] (-2.5,-3)--(-2,-3);
\end{tikzpicture}
\caption{\label{fig:trivial_maximal_split} Maximal trivial split. 
Given $a\tau_ia\tau'_i$ is a chord model of $G_i$ for $i \in [3]$,
two examples of chord models of $G_{ov}$ obtained from these models, namely $a\tau_3\tau'_2\tau_1a\tau'_1\tau_2\tau'_3$ and $a\tau_3\tau_1\tau'_2a\tau_2\tau'_1\tau'_3$, are shown to the right.}
\end{figure}
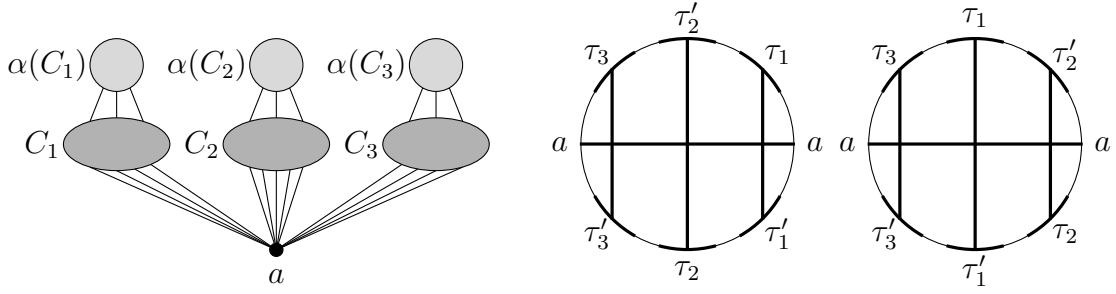
Note that the above theorem is valid for any vertex $a$ in $G_{ov}$ such that
$a$ is the articulation point in $G_{ov}$.

\subsection{Modular decomposition of $G_{ov}$}
\label{subsec:modular_decomposition}
The results presented in this subsection are due to Gallai \cite{Gal67}.

Let $G = (V,E)$ be a graph with no twins and no universal vertices and 
let $G_{ov} = (V,{\sim})$ be the overlap graph associated with $G$.

A non-empty set $M\subseteq V$ is a \emph{module} in $G_{ov}$ 
if $x\sim M$ or $x \parallel M$ for every $x\in V\smallsetminus M$.
The singleton sets and the whole $V$ are the \emph{trivial} modules of $G_{ov}$. 
The graph $(U,{\sim})$ is \emph{prime} if $(U,{\sim})$ has no modules other than the trivial ones.

A module $M$ of $G_{ov}$ is \emph{strong} if $M\subset N$, $N\subset M$, or $M\cap N=\emptyset$ for every other module $N$ in $G_{ov}$.
In particular, two strong modules are either nested or disjoint.
The \emph{modular decomposition} of $G_{ov}$, denoted by $\mathcal{M}(G_{ov})$, 
is the set containing all strong modules of $G_{ov}$.
The set $\mathcal{M}(G_{ov})$, ordered by inclusion, forms a tree in which $V$ is the root, 
maximal proper subsets in $\mathcal{M}(G_{ov})$ of a strong module $M$ are the children of $M$, 
and the leaves are the singleton modules $\{x\}$ for $x\in V$.
The children of a non-singleton module $M\in\mathcal{M}(G_{ov})$ form a partition of~$M$.

A module $M \in \mathcal{M}(G_{ov})$ is \emph{serial} if $M_1\sim M_2$ for every two children $M_1$ and $M_2$ of $M$, 
\emph{parallel} if $M_1 \parallel M_2$ for every two children $M_1$ and $M_2$ of $M$, and \emph{prime} otherwise.
Equivalently, $M \in \mathcal{M}$ is serial if $(M,{\parallel})$ is disconnected, 
parallel if $(M,{\sim})$ is disconnected, and prime if both $(M,{\sim})$ and $(M,{\parallel})$ are connected.

\subsection{Permutation subgraphs of $G_{ov}$ and the structure of its permutation models} 
\label{subsec:permutation_graphs_and_their_models}
Let $G = (V,E)$ be a circular-arc graph with no twins and no universal vertices and
let $G_{ov}$ be the overlap graph associated with $G$. 
Let $U$ be a subset of $V$.
The graph $(U,{\sim})$ is a \emph{permutation subgraph of $G_{ov}$} if there
exists a pair $(\tau^0,\tau^1)$, where $\tau^{0}$ and $\tau^{1}$ are permutations of $U$, such
that for every $x,y \in U$:
$$x \sim y \iff  
\begin{array}{l}
\text{$x$ appears before $y$ in both $\tau^{0}$ and $\tau^{1}$, or} \\
\text{$y$ appears before $x$ in both $\tau^{0}$ and $\tau^{1}$.}
\end{array}
$$
If this is the case, $(\tau^{0},\tau^{1})$ is called a \emph{permutation model} of $(U,{\sim})$.
See Figure \ref{fig:non_oriented_permutation_graph} for an example of a permutation graph and its permutation model.

\begin{figure}[htp!]
\centering
\begin{tikzpicture}[xscale=0.85,yscale=0.5,>=latex,shorten >=-0.4pt,shorten <=-0.4pt]
\coordinate (center) at (0,0) {};
\coordinate (ua) at ($(center)+(105:2cm)$) {};
\coordinate (ub) at ($(center)+(90:2cm)$) {};
\coordinate (uc) at ($(center)+(75:2cm)$) {};

\coordinate (lua) at ($(center)+(105:2.4cm)$) {};
\coordinate (lub) at ($(center)+(90:2.4cm)$) {};
\coordinate (luc) at ($(center)+(75:2.4cm)$) {};

\coordinate (bc) at ($(center)+(270:2cm)$) {};
\coordinate (bb) at ($(center)+(255:2cm)$) {};
\coordinate (ba) at ($(center)+(285:2cm)$) {};

\coordinate (lbc) at ($(center)+(270:2.4cm)$) {};
\coordinate (lbb) at ($(center)+(255:2.4cm)$) {};
\coordinate (lba) at ($(center)+(285:2.4cm)$) {};

\coordinate (tau0) at ($(center)+(60:2.4cm)$) {};
\coordinate (tau1) at ($(center)+(240:2.4cm)$) {};

\tikzstyle{every node}=[inner sep=1pt]
\begin{footnotesize}

\node at (lua) {$a$};
\node at (lba) {$a$};

\node at (lub) {$b$};
\node at (lbb) {$b$};

\node at (luc) {$c$};
\node at (lbc) {$c$};

\node at (tau0) {$\tau^0$};
\node at (tau1) {$\tau^1$};

\end{footnotesize}

\draw[very thick,<-] ([shift=(60:2cm)]0,0) arc (60:120:2cm);
\draw[very thick,<-] ([shift=(240:2cm)]0,0) arc (240:300:2cm);

\draw (ua)--(ba);
\draw (ub)--(bb);
\draw (uc)--(bc);

\end{tikzpicture}

\caption{\label{fig:non_oriented_permutation_graph} 
Permutation model $(\tau^0,\tau^1) = (abc,acb)$ of the permutation graph $(\{a,b,c\}, \{a \sim b, a \sim c)\}$.
}
\end{figure}
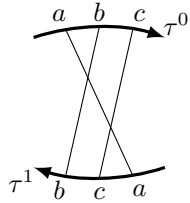

\begin{definition}
A module $U$ in $G_{ov}$ is \emph{proper} if there is a vertex $x \in V \setminus U$ such that $x \sim U$.
\end{definition}
The next claim shows that every proper module $U$ in $G_{ov}$ induces a permutation
subgraph in $G_{ov}$.
\begin{claim}
\label{claim:permutation_graphs_in_G_ov}
Suppose $U$ is a proper module in $G_{ov}$ such that $x \sim U$ for some $x \in V \setminus U$.
Then, for any chord model $\psi$ of $G_{ov}$, 
$$\psi|(U \cup \{x\}) \equiv x \tau x \tau',$$
where $(\tau,\tau')$ and $(\tau',\tau)$ are permutation models of $(M,{\sim})$.
In particular, $(U,{\sim})$ is a permutation subgraph of $G_{ov}$.
\end{claim}
\begin{proof}
Let $\psi$ be a chord model of $G_{ov}$.
Note that every chord $\psi(u)$ for $u \in U$ has 
its endpoints on different sides of the chord $\psi(x)$.
Thus, $\psi|(M \cup \{x\}) \equiv x \tau x \tau'$, where
$\tau$ and $\tau'$ are permutations of $U$.
Clearly, for every $u,v \in U$, $\psi(u)$ intersect $\psi(v)$ 
iff either $u$ appears before $v$ in both $\tau$ and $\tau'$ or
$v$ appears before $u$ in both $\tau$ and $\tau'$.
In particular, both $(\tau,\tau')$ and $(\tau',\tau)$ are permutation models of $(U,{\sim})$.
\end{proof}
Let $U$ be a proper module in $G_{ov}$.
Now, our goal is to describe all permutation models of $(U,{\sim})$.
To accomplish our task we use the modular decomposition of $(U,{\sim})$.
Note that the modular decomposition of $(U,{\sim})$ is associated with $\mathcal{M}(G_{ov})$ by 
the following equation:
$$\mathcal{M}(U,{\sim}) = \{M \in \mathcal{M}(G_{ov}): M \subseteq U\} \cup \{U\}.$$
An \emph{orientation} $(U,{\prec})$ of $(U,{\sim})$ is a binary relation on $U$ such that for every $u,v \in U$:
$$u \sim v \iff \text{either } u \prec v \text{ or } v \prec u.$$
In other words, an orientation $(U,\prec)$ arises by orienting every edge $u \sim v$ of $(U,{\sim})$ either
from $u$ to $v$ (denoted $u \prec v$) or from $v$ to $u$ (denoted $v \prec u$).
An orientation $(U,{\prec})$ of $(U,{\sim})$ is \emph{transitive} if ${\prec}$ is a transitive relation on $U$.
The connections between transitive orientations of $(U,{\sim})$ and the modular decomposition of $(U,{\sim})$
have been established by Gallai \cite{Gal67}.
\begin{theorem}[\cite{Gal67}]
\label{thm:prime_graph_has_two_transitive_orientations}
If $M_1,M_2 \in \mathcal{M}(U,{\sim})$ are such that $M_1 \sim M_2$, then every
transitive orientation $(U,{\prec})$ satisfies either $M_1 \prec M_2$ or $M_2 \prec M_1$.
\end{theorem}
Let $M$ be a strong module in $\mathcal{M}(U,{\sim})$.
The edge relation $\sim$ in $(M,{\sim})$ restricted to the edges joining the vertices from two different children of $M$ is denoted by $\sim_M$.
If $x\sim y$, then $x\sim_My$ for exactly one module $M\in\mathcal{M}(U,{\sim})$.
Hence, the set $\{{\sim_M} : M \in \mathcal{M}(U,{\sim})\}$ forms a partition of the edge set ${\sim}$ of the graph $(U,{\sim})$.
\begin{theorem}[\cite{Gal67}]
\label{thm:transitive_orientations_versus_transitive_orientations_of_strong_modules}
There is a one-to-one correspondence between the set of transitive orientations $(U,{\prec})$ of $(U,{\sim})$
and the families $$\{(M,{\prec_M}): M \in \mathcal{M}(U,{\prec}) \text{ and } \prec_M \text{ is a transitive orientation of $(M, \sim_M)$}\}$$
given by $x \prec y \iff x \prec_M y$, where $M$ is the module in $\mathcal{M}$ such that $x \sim_M y$.
\end{theorem}
The above theorem asserts that every transitive orientation of $(U,{\sim})$ restricted to the edges of the graph $(M, {\sim_M})$ induces a transitive orientation of $(M, {\sim_M})$, for every $M \in \mathcal{M}(U,{\sim})$.
On the other hand, every transitive orientation of $(U,{\sim})$ can be obtained 
by independent transitive orientation of $(M,{\sim_M})$, for $M \in \mathcal{M}(U,{\sim})$.
Gallai \cite{Gal67} characterized all possible transitive orientation of strong modules $(M,{\sim_M})$.
\begin{theorem}[\cite{Gal67}]
Let $M$ be a prime module in $\mathcal{M}(U,{\sim})$. 
Then, $(M,{\sim_M})$ has two transitive orientations, one being the reverse of the other.
\end{theorem}
A parallel module $(M,{\sim_M})$ has exactly one (empty) transitive orientation.
The transitive orientations of serial modules $(M,{\sim})$ correspond to the total orderings of its children, that
is, every transitive orientation of $(M,{\sim_M})$ is of the form $M_{i_1} \prec \ldots \prec M_{i_k}$, where
$i_1 \ldots i_k$ is a permutation of $[k]$ and $M_1, \ldots,M_k$ are the children of $M$ in $\mathcal{M}(U,{\sim})$.

Since $(U,{\sim})$ is a proper module in $G_{ov}$, then 
$(U,{\sim})$ admits a permutation model $(\tau^{0}, \tau^{1})$.
Note that $(\tau^0,\tau^1)$ yields transitive orientations of the graphs $(U,{\sim})$ and $(U,{\parallel})$
given by:
\begin{equation}
\label{eq:models_of_permutation_graphs_1}
\begin{array}{lll}
x \prec y & \iff & x \text{ occurs before } y \text{ in } \tau^0 \text{ and } x \sim y, \\
x < y & \iff &  x \text{ occurs before } y \text{ in } \tau^0 \text{ and } x \parallel y. \\
\end{array}
\end{equation}
In particular, the orientations $(U,{\prec})$ and $(U,{<})$ are consistent with the word $\tau^0$.
On the other hand, given transitive orientations ${\prec}$ and ${<}$ of $(U,{\sim})$ and $(U,{\parallel})$, respectively, one can construct a permutation model $(\tau^{0}, \tau^{1})$ of $(U,{\sim})$ such that
\begin{equation}
\label{eq:models_of_permutation_graphs_2}
\begin{array}{lll}
x \text{ occurs before } y \text{ in } \tau^{0} \iff x \prec y \text{ or } x < y,\\
x \text{ occurs before } y \text{ in } \tau^{1} \iff x \prec y \text{ or } y < x.\\
\end{array}
\end{equation}
\begin{theorem}[\cite{DM41}]
\label{thm:permutation_models_transitive_orientations}
Let $(U,{\sim})$ be a proper submodule of $G_{ov}$.
There is a one-to-one correspondence between permutation models $(\tau^{0}, \tau^{1})$ of $(U,{\sim})$
and the pairs $({<}, {\prec})$ of transitive orientations of $(U,{\parallel})$ and $(U,{\sim})$, respectively, established by equations \eqref{eq:models_of_permutation_graphs_1} and~\eqref{eq:models_of_permutation_graphs_2}.
\end{theorem}

\subsection{Modular decomposition $\mathcal{M}(G_{ov})$ and chord models of $G_{ov}$}
In this subsection we describe properties of chord models of $G_{ov}$
with respect to the modular decomposition of the graph $G_{ov}$.

The purpose of the next lemmas is to describe the restrictions of 
chord models of $G_{ov}$ to the modules $M$ from $\mathcal{M}(G_{ov})$.
For this purpose, for a given module $M \in \mathcal{M}(G_{ov})$, we define
\begin{equation*}
\begin{array}{ccl}
N[M] &=& \{x \in V \setminus M: x \sim M\}, \\
C[M] &=& \text{the connected component of $G_{ov}$ containing the module $M$.}
\end{array}
\end{equation*}
Note that $C[M]$ is not defined in the case when $M=V$ and $G_{ov}$ is disconnected.

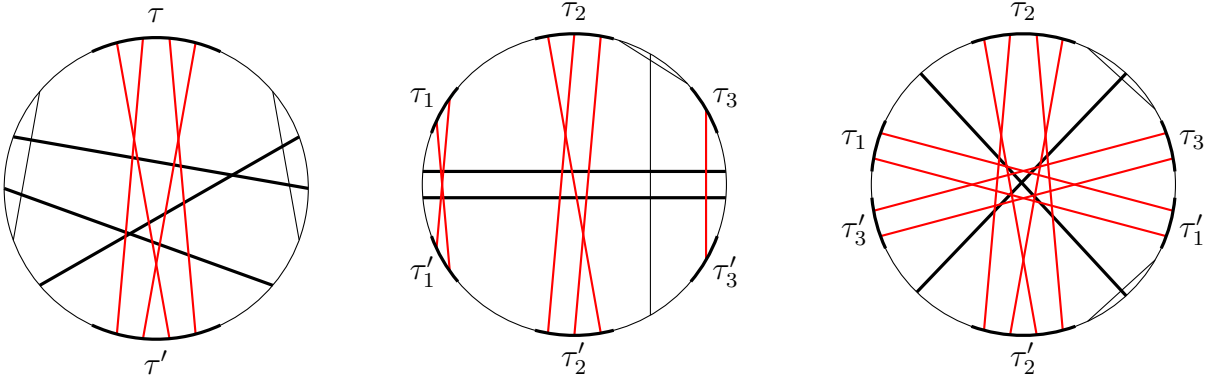
\begin{figure}[htp!]
\centering
\begin{tikzpicture}[scale=1]
\coordinate (center) at (0,0) {};
\coordinate (ltau) at ($(center)+(90:2.3cm)$) {};
\coordinate (ltau') at ($(center)+(270:2.3cm)$) {};

\coordinate (u4) at ($(center)+(75:2cm)$) {};
\coordinate (u3) at ($(center)+(85:2cm)$) {};
\coordinate (u2) at ($(center)+(95:2cm)$) {};
\coordinate (u1) at ($(center)+(105:2cm)$) {};

\coordinate (b4) at ($(center)+(255:2cm)$) {};
\coordinate (b3) at ($(center)+(265:2cm)$) {};
\coordinate (b2) at ($(center)+(275:2cm)$) {};
\coordinate (b1) at ($(center)+(285:2cm)$) {};

\coordinate (l5) at ($(center)+(140:2cm)$) {};
\coordinate (l4) at ($(center)+(160:2cm)$) {};
\coordinate (l3) at ($(center)+(180:2cm)$) {};
\coordinate (l2) at ($(center)+(200:2cm)$) {};
\coordinate (l1) at ($(center)+(220:2cm)$) {};

\coordinate (r5) at ($(center)+(-40:2cm)$) {};
\coordinate (r4) at ($(center)+(-20:2cm)$) {};
\coordinate (r3) at ($(center)+(0:2cm)$) {};
\coordinate (r2) at ($(center)+(20:2cm)$) {};
\coordinate (r1) at ($(center)+(40:2cm)$) {};

\draw ($(center)$) circle (2cm);

\draw[black] (l2)-- (l5);
\draw[black] (r1)-- (r4);

\draw[black,very thick] (l1)-- (r2);
\draw[black,very thick] (l3)-- (r5);
\draw[black,very thick] (l4)-- (r3);

\draw[thick,red] (u1)-- (b2);
\draw[thick,red] (u2)-- (b4);
\draw[thick,red] (u3)-- (b1);
\draw[thick,red] (u4)-- (b3);

\draw[very thick] ([shift=(245:2cm)]0,0) arc (245:295:2cm);
\draw[very thick] ([shift=(65:2cm)]0,0) arc (65:115:2cm);

\tikzstyle{every node}=[inner sep=1pt]
\node at (ltau) {$\tau$};
\node at (ltau') {$\tau'$};

\end{tikzpicture} 
\hspace{1cm}
\begin{tikzpicture}[scale=1]
\coordinate (center) at (0,0) {};
\coordinate (ltau1) at ($(center)+(150:2.3cm)$) {};
\coordinate (ltau'1) at ($(center)+(210:2.3cm)$) {};
\coordinate (ltau2) at ($(center)+(90:2.3cm)$) {};
\coordinate (ltau'2) at ($(center)+(270:2.3cm)$) {};
\coordinate (ltau3) at ($(center)+(30:2.3cm)$) {};
\coordinate (ltau'3) at ($(center)+(330:2.3cm)$) {};

\coordinate (tau1u1) at ($(center)+(145:2cm)$) {};
\coordinate (tau1u2) at ($(center)+(155:2cm)$) {};

\coordinate (tau1b1) at ($(center)+(205:2cm)$) {};
\coordinate (tau1b2) at ($(center)+(215:2cm)$) {};

\coordinate (tau2u1) at ($(center)+(80:2cm)$) {};
\coordinate (tau2u2) at ($(center)+(90:2cm)$) {};
\coordinate (tau2u3) at ($(center)+(100:2cm)$) {};

\coordinate (tau2b1) at ($(center)+(260:2cm)$) {};
\coordinate (tau2b2) at ($(center)+(270:2cm)$) {};
\coordinate (tau2b3) at ($(center)+(280:2cm)$) {};

\coordinate (tau3u1) at ($(center)+(30:2cm)$) {};
\coordinate (tau3b1) at ($(center)+(330:2cm)$) {};

\draw ($(center)$) circle (2cm);

\draw[black,very thick] ($(center)+(5:2cm)$)--($(center)+(175:2cm)$);
\draw[black,very thick] ($(center)+(-5:2cm)$)--($(center)+(185:2cm)$);
\draw[black] ($(center)+(60:2cm)$)--($(center)+(300:2cm)$);
\draw[black] ($(center)+(42:2cm)$)--($(center)+(73:2cm)$);

\draw[red,thick] (tau1u1)-- (tau1b1);
\draw[red,thick] (tau1u2)-- (tau1b2);

\draw[red,thick] (tau2u1)-- (tau2b2);
\draw[red,thick] (tau2u2)-- (tau2b1);
\draw[red,thick] (tau2u3)-- (tau2b3);

\draw[red,thick] (tau3u1)-- (tau3b1);

\draw[very thick] ([shift=(20:2cm)]0,0) arc (20:40:2cm);
\draw[very thick] ([shift=(75:2cm)]0,0) arc (75:105:2cm);
\draw[very thick] ([shift=(140:2cm)]0,0) arc (140:160:2cm);

\draw[very thick] ([shift=(200:2cm)]0,0) arc (200:220:2cm);
\draw[very thick] ([shift=(255:2cm)]0,0) arc (255:285:2cm);
\draw[very thick] ([shift=(320:2cm)]0,0) arc (320:340:2cm);

\tikzstyle{every node}=[inner sep=1pt]
\node at (ltau1) {$\tau_1$};
\node at (ltau'1) {$\tau'_1$};
\node at (ltau2) {$\tau_2$};
\node at (ltau'2) {$\tau'_2$};
\node at (ltau3) {$\tau_3$};
\node at (ltau'3) {$\tau'_3$};

\end{tikzpicture} 
\hspace{1cm}
\begin{tikzpicture}[scale=1]
\coordinate (center) at (0,0) {};
\coordinate (ltau1) at ($(center)+(165:2.3cm)$) {};
\coordinate (ltau'1) at ($(center)+(345:2.3cm)$) {};
\coordinate (ltau3) at ($(center)+(195:2.3cm)$) {};
\coordinate (ltau'3) at ($(center)+(15:2.3cm)$) {};
\coordinate (ltau2) at ($(center)+(90:2.3cm)$) {};
\coordinate (ltau'2) at ($(center)+(270:2.3cm)$) {};

\coordinate (u4) at ($(center)+(75:2cm)$) {};
\coordinate (u3) at ($(center)+(85:2cm)$) {};
\coordinate (u2) at ($(center)+(95:2cm)$) {};
\coordinate (u1) at ($(center)+(105:2cm)$) {};

\coordinate (b4) at ($(center)+(255:2cm)$) {};
\coordinate (b3) at ($(center)+(265:2cm)$) {};
\coordinate (b2) at ($(center)+(275:2cm)$) {};
\coordinate (b1) at ($(center)+(285:2cm)$) {};

\coordinate (tau1u1) at ($(center)+(160:2cm)$) {};
\coordinate (tau1u2) at ($(center)+(170:2cm)$) {};
\coordinate (tau1b1) at ($(center)+(340:2cm)$) {};
\coordinate (tau1b2) at ($(center)+(350:2cm)$) {};

\coordinate (tau3u1) at ($(center)+(10:2cm)$) {};
\coordinate (tau3u2) at ($(center)+(20:2cm)$) {};
\coordinate (tau3b1) at ($(center)+(190:2cm)$) {};
\coordinate (tau3b2) at ($(center)+(200:2cm)$) {};

\draw ($(center)$) circle (2cm);

\draw[very thick] ($(center)+(47.5:2cm)$)-- ($(center)+(225.5:2cm)$);
\draw ($(center)+(30:2cm)$)-- ($(center)+(65:2cm)$);

\draw[very thick] ($(center)+(132.5:2cm)$)-- ($(center)+(312.5:2cm)$);
\draw ($(center)+(295:2cm)$)-- ($(center)+(330:2cm)$);

\draw[thick,red] (tau1u1)-- (tau1b2);
\draw[thick,red] (tau1u2)-- (tau1b1);

\draw[thick,red] (tau3u1)-- (tau3b2);
\draw[thick,red] (tau3u2)-- (tau3b1);

\draw[thick,red] (u1)-- (b2);
\draw[thick,red] (u2)-- (b4);
\draw[thick,red] (u3)-- (b1);
\draw[thick,red] (u4)-- (b3);

\draw[very thick] ([shift=(5:2cm)]0,0) arc (5:25:2cm);
\draw[very thick] ([shift=(70:2cm)]0,0) arc (70:110:2cm);
\draw[very thick] ([shift=(155:2cm)]0,0) arc (155:175:2cm);

\draw[very thick] ([shift=(185:2cm)]0,0) arc (185:205:2cm);
\draw[very thick] ([shift=(250:2cm)]0,0) arc (250:290:2cm);
\draw[very thick] ([shift=(335:2cm)]0,0) arc (335:355:2cm);

\tikzstyle{every node}=[inner sep=1pt]
\node at (ltau1) {$\tau_1$};
\node at (ltau'1) {$\tau'_1$};
\node at (ltau2) {$\tau_2$};
\node at (ltau'2) {$\tau'_2$};
\node at (ltau3) {$\tau'_3$};
\node at (ltau'3) {$\tau_3$};

\end{tikzpicture} 

\caption{\label{fig:chord_models_for_proper_modules} The restriction of $\psi$ to $C[M]$. 
Chords associated with the module $M$ are red, chords associated with $N[M]$ are bolded. 
From left to right: $M$ is proper prime (Lemma \ref{lemma:circle_models_of_a_proper_prime_module}), $M$ is proper parallel (Lemma \ref{lemma:circle_models_of_a_parallel_module}), and $M$ is proper serial (Lemma \ref{lemma:circle_models_of_a_serial_module}
).}
\end{figure}

\begin{claim}
\label{claim:circle_models_of_proper_modules_N_M}
Suppose $M$ is a proper prime or a proper parallel module in $\mathcal{M}(G_{ov})$.
For any chord model $\psi$ of $G_{ov}$ we have
$$\psi|(M \cup N[M]) \equiv \pi \tau \pi' \tau',$$
where $(\tau,\tau')$ is a permutation model of $(M,{\sim})$ and $\pi, \pi'$ are permutations of $N[M]$.
In other words, $\psi|M$ forms two contiguous subwords in the circular word $\psi|(M \cup N[M])$ (see Figure \ref{fig:chord_models_for_proper_modules}).
\end{claim}
\begin{proof}
Since $M$ is proper, we can pick $x \in V \setminus M$ such that $x \sim M$.
Orient the chord $\psi(x)$ arbitrarily.
By Claim \ref{claim:permutation_graphs_in_G_ov} we have that
$$\psi|(M \cup \{x\})  \equiv x^{0} \tau x^{1} \tau',$$
where $(\tau,\tau')$ is a permutation model of $(M,{\sim})$.
We need to show that
$$\psi|(M_i \cup N[M]) \equiv \pi \tau \pi' \tau',$$
where $\pi, \pi'$ are some permutations of $N[M]$.

Fix $z \in N[M]$ such that $z \neq x$.
Suppose for a contradiction that the chord $\psi(z)$ has one of its ends between
the ends of the chords corresponding to the letters of $\tau$.
That is, suppose that $x^{0}\tau_1 z \tau_2 x^{1}$ is a subword of the circular word $\psi|(M \cup \{x,z\})$,
where $\tau_1$ and $\tau_2$ are non-empty words such that $\tau_1\tau_2 = \tau$.
Now, consider a partition of $M$ into two sets, $M_1$ and $M_2$:
$$
M_1 = \{u \in M: u \in \tau_1\}  \text{ and }  M_2 = \{u \in M: u \in \tau_2\}.
$$
Since every $\psi(u)$ for $u \in M$ must intersect $\psi(z)$, 
we have that $x^{1} \tau'_1 z \tau'_2 x^{0}$ is a subword of $\psi|(M \cup \{z,t\})$, where
$\tau'_i$ is a permutation of the set of the letters in $\tau_i$ for every $i \in [2]$.
It means, in particular, that every two chords $\psi(u_1)$ and $\psi(u_2)$ for $u_1 \in M_1$ and $u_2 \in M_2$ intersect.
So, we have $M_1 \sim M_2$, which contradicts that $M$ is a prime or a parallel module in $\mathcal{M}(G_{ov})$.
\end{proof}

\begin{lemma}
\label{lemma:circle_models_of_a_proper_prime_module}
Suppose $M$ is a proper prime module in $\mathcal{M}(G_{ov})$ and
suppose $\psi$ is a chord model of $G_{ov}$.
Then, 
$$\psi|M = \tau\tau',$$ 
where $(\tau,\tau')$ is a permutation model of $(M,{\sim})$. 
Moreover, $\tau$ and $\tau'$ are contiguous subwords in the circular word $\psi|C[M]$ (see Figure \ref{fig:chord_models_for_proper_modules}).
\end{lemma}
\begin{proof}
By Claim \ref{claim:circle_models_of_proper_modules_N_M},
$\psi|(M \cup N[M]) \equiv \pi \tau \pi' \tau'$, where $\pi, \pi'$ are permutations of $N[M]$.
Clearly, $N[M] \subset C[M]$.
To complete the proof of the lemma 
it suffices to show that either $\psi|(M \cup \{v\}) \equiv v v\tau \tau'$  or $\psi|(M \cup \{v\}) \equiv \tau v v \tau'$ for every $v \in C[M] \smallsetminus N[M]$.
Assume otherwise. 
Since $(C[M],{\sim})$ is connected and since $M$ is a module in $(C[M],{\sim})$, 
there is $u \in C[M] \smallsetminus N[M]$ 
such that
$$\psi|(M \cup \{u\}) \equiv \tau_1 u \tau_2 \tau'_2 u \tau'_1,$$
where $\tau_1, \tau_2$ and $\tau'_1, \tau'_2$ are such that $\tau_1\tau_2 = \tau$ and $\tau'_2\tau'_1=\tau'$,
both $\tau_1, \tau_2$ are non-empty or both $\tau'_1, \tau'_2$ are non-empty.
Since $u \parallel M$, we conclude that $\tau_i$ is a permutation of $\tau'_i$ for $i \in [2]$.
Hence, the sets
$$M_1 = \{w \in M: w \in \tau_1\} \text{ and } M_2 = \{w \in M: w \in \tau_2\}$$
are such that $M_1 \neq \emptyset$, $M_2 \neq \emptyset$, and $M_1 \parallel M_2$.
So, $(M, {\sim})$ is not connected, which contradicts the fact that $M$ is a prime module in $\mathcal{M}(G_{ov})$.
\end{proof}

\begin{lemma}
\label{lemma:circle_models_of_a_parallel_module}
Suppose $M$ is a proper parallel module in $\mathcal{M}(G_{ov})$ with children $M_1,\ldots,M_k$ and
suppose $\psi$ is a chord model of $G_{ov}$.
Then, 
$$\psi|M \equiv \tau_{i_1} \ldots \tau_{i_k} \tau'_{i_k} \ldots \tau'_{i_1},$$
where $(i_1,\ldots,i_k)$ is a permutation of $[k]$ and $(\tau_{i_j}, \tau'_{i_j})$
is a permutation model of $(M_{i_j},{\sim})$ for every $j \in [k]$.
Moreover, for every $j \in [k]$ the set $\psi|M_{i_j}$ consists of two contiguous subwords, $\tau_{i_j}$ and $\tau'_{i_j}$,  in the circular word $\psi|C[M]$ (see Figure \ref{fig:chord_models_for_proper_modules}).
\end{lemma}
\begin{proof}
Since $M$ is proper, we can pick $x \in V \setminus M$ such that $x \sim M$.
Since $M$ is parallel, $M_i$ is either serial or prime.
In particular, $(M_i,{\sim})$ is connected for every $i \in [k]$.
Since $x$ is an articulation point in $(M \cup \{x\}, {\sim})$, by Theorem \ref{thm:trivial_split_representations} we have that
$$\psi|(M \cup \{x\}) \equiv x \tau_{i_1} \ldots \tau_{i_k} x \tau'_{i_k} \ldots \tau'_{i_1},$$
where $(i_1,\ldots,i_k)$ is a permutation of $[k]$ and $(\tau_{i_j}, \tau'_{i_j})$
is a permutation model of $(M_{i_j},{\sim})$ for every $j \in [k]$.
Since $M$ is proper parallel, by Claim \ref{claim:circle_models_of_proper_modules_N_M} we have that 
$$\psi|(M \cup N[M]) \equiv \pi \tau_{i_1} \ldots \tau_{i_k} \pi' \tau'_{i_1} \ldots \tau'_{i_k},$$
where $\pi$ and $\pi'$ are permutations of $N[M]$.
Now, using an argument similar to those used in the previous lemma,
we prove that $\tau_{i_j}$ and $\tau'_{i_j}$ are contiguous subwords in the circular word $\psi|C[M]$.
\end{proof}

\begin{lemma}
\label{lemma:circle_models_of_a_serial_module}
Suppose $M$ is a serial module in $\mathcal{M}(G_{ov})$ with children $M_1,\ldots,M_k$
and suppose $\psi$ is a chord model of $G_{ov}$.
Then
$$\psi|M \equiv \tau_{i_1} \ldots \tau_{i_k} \tau'_{i_1} \ldots \tau'_{i_k},$$
where $(i_1,\ldots,i_k)$ is a circular permutation of $[k]$ 
and $(\tau_{i_j}, \tau'_{i_j})$ is a permutation model of $(M_{i_j},{\sim})$ for every $j \in [k]$.
Moreover, for every $j \in [k]$ the set $\psi|M_{i_j}$ consists of two contiguous subwords, 
$\tau_{i_j}$ and $\tau'_{i_j}$, in the circular word $\psi|C[M]$ (see Figure \ref{fig:chord_models_for_proper_modules}).
\end{lemma}
\begin{proof}
Since $M$ is serial, $M_i$ is either prime or parallel.
Moreover, since $x \sim M_i$ for every $x \in M \setminus M_i$, $M_i$ is proper.
From Theorem \ref{thm:trivial_split_representations}, we deduce that
$$\psi|M \equiv \tau_{i_1} \ldots \tau_{i_k} \tau'_{i_1} \ldots \tau'_{i_k},$$
where $(i_1,\ldots,i_k)$ is a circular permutation of $[k]$ 
and $(\tau_{i_j}, \tau'_{i_j})$ is a permutation model of $(M_{i_j},{\sim})$ for every $j \in [k]$.

Now, it remains to prove that $\tau_{i_j}$ and $\tau'_{i_j}$ are contiguous subwords in $\psi|C[M]$.
From Claim \ref{claim:circle_models_of_proper_modules_N_M} applied to every child of $M$, 
for every $x \in N[M]$ we must have
$$\psi|(M \cup \{x\}) \equiv \tau_{i_1} \ldots \tau_{i_j} x \tau_{i_{j+1}} \ldots \tau_{i_k}\tau'_{i_1} \ldots \tau'_{i_j} x \tau'_{i_{j+1}} \ldots \tau'_{i_k}$$
for some $j \in [k]$.
Assume that $\psi|M_{i_j}$ does not form two contiguous subwords in $\psi|C[M]$ for some $j \in [k]$.
By the above observation and by the connectivity of $(C[M],{\sim})$, there is $u \in C[M] \setminus N[M]$ such that
$u \parallel M$ and $\psi|(M_{i_j} \cup \{u\}) \equiv \tau_1 u \tau_2 \tau'_2 u \tau'_1$, where
$\tau_1, \tau_2$ and $\tau'_1, \tau'_2$ are such that $\tau_1\tau_2 = \tau_{i_j}$
and $\tau'_{2}\tau'_1 = \tau'_{i_j}$, and both $\tau_1, \tau_2$ or both $\tau'_1, \tau'_2$ are non-empty.
Since $u \parallel M_{i_j}$, $\tau_i$ is a permutation of $\tau'_i$ for $i \in [2]$.
Then, note that $\psi(u)$ must intersect every chord from $\psi(M \setminus M_{i_j})$,
which is not possible as $u \parallel M$.
\end{proof}

\section{The structure of all conformal models of $G_{ov}$}
\label{sec:conformal_models}
The goal of this section is to describe the structure of all conformal models of $G_{ov}$.
We split our work into subsections that cover the following cases:
\begin{itemize}
 \item Subsection \ref{subsec:conformal_models_of_serial_modules}: $(V,{\parallel})$ is disconnected, which corresponds to the case when $V$ is a serial module in $\mathcal{M}(G_{ov})$.
 \item Subsection \ref{subsec:conformal_models_of_improper_prime_modules}: $(V,{\sim})$ and $(V,{\parallel})$ are connected, which corresponds to the case when $V$ is an improper prime module in $\mathcal{M}(G_{ov})$.
 \item Subsection \ref{subsec:conformal_models_of_improper_parallel_modules}: $(V,{\sim})$ is disconnected, which corresponds to the case when $V$ is an improper parallel module in $\mathcal{M}(G_{ov})$.
\end{itemize}
Moreover, in the subsequent subsections we characterize all conformal models of
$(M,{\sim})$, where $M$ is a serial, an improper prime, and an improper parallel module in $\mathcal{M}(G_{ov})$, respectively.
We start with some preparatory results contained in Subsection \ref{subsec:conformal_models_of_proper_prime_parallel_modules}, where
we examine the restrictions of conformal models of $G_{ov}$ to proper prime and proper parallel modules in $\mathcal{M}(G_{ov})$.

\subsection{Proper prime and proper parallel modules of $G_{ov}$}
\label{subsec:conformal_models_of_proper_prime_parallel_modules}
The ideas presented in this subsection naturally naturally extend Spinrad's work
on co-bipartite circular-arc graphs \cite{Spin88}.
In particular, the algorithm described in Claim \ref{claim:invariants_for_proper_parallel_and_proper_prime_modules}
was used by Spinrad \cite{Spin88} in his recognition algorithm for the class of co-bipartite circular-arc graphs \cite{Spin88}.

Suppose $M$ is a proper prime or a proper parallel module in $\mathcal{M}(G_{ov})$.
Suppose $C[M]$ is the connected component containing $M$ 
and suppose $r$ is a fixed vertex in $M$, called the \emph{representant} of $M$.

Let $\phi$ be a conformal model of $G_{ov}$ and let $\phi'$
be the restriction of $\phi$ to $C[M]$.
By Lemmas \ref{lemma:circle_models_of_a_proper_prime_module} and \ref{lemma:circle_models_of_a_parallel_module}, 
$\phi'|M$ forms two contiguous subwords, $\tau^{0}_{\phi}$ and $\tau^{1}_{\phi}$, in the circular word $\phi'$. 
We indexed $\tau^0_{\phi},\tau^1_{\phi}$ such that $\tau^{j}_{\phi}$ contains the letter $r^{j}$ for $j \in \{0,1\}$. 
Suppose also that $M^{0}_{\phi}$ and $M^{1}_{\phi}$ are the sets consisting of all labeled letters in the words $\tau^{0}_{\phi}$ and $\tau^{1}_{\phi}$, respectively.
Note that $M^{0}_{\phi}$ and $M^{1}_{\phi}$ are labeled copies of $M$ and $\{M_{\phi}^{0},M_{\phi}^{1}\}$ forms a partition of $M^{*}$.
Note that the pair $(\tau^0_{\phi}, \tau^1_{\phi})$ is
an oriented permutation model of $(M,{\sim})$.
The non-oriented permutation model $(\tau^0_{\phi}, \tau^1_{\phi})$ corresponds, according to Theorem \ref{thm:permutation_models_transitive_orientations}, to the pair of transitive orientations $({<^{0}_{\phi}},{\prec^1_{\phi}})$ 
of $(M,{\parallel})$ and $(M,{\sim})$, respectively.
It turns out that the transitive orientation ${<^{0}_\phi}$ and the sets $M^{0}_{\phi}, M^{1}_{\phi}$ are independent on the choice of a conformal model $\phi$ of $G_{ov}$.
\begin{claim}
\label{claim:invariants_for_proper_parallel_and_proper_prime_modules}
There is a transitive orientation $(M,{<_M})$ of $(M,{\parallel})$ and there are labeled copies $M^{0}$ and $M^{1}$ of $M$ forming a partition of $M^{*}$, such that
$$(M^{0}_{\phi},M^{1}_{\phi},<^{0}_{\phi}) = (M^0,M^1,{<_{M}}) \quad \text{for every conformal model $\phi$ of $G_{ov}$.}$$
\end{claim}
\begin{proof}
Suppose $\phi$ is a conformal model of $G_{ov}$ and suppose $u \in M$.
We say that:
\begin{itemize}
 \item $\phi(u)$ \emph{is oriented from $M_{\phi}^{0}$ to $M^{1}_{\phi}$}
if $u^{0} \in M_{\phi}^{0}$ and $u^{1} \in M_{\phi}^{1}$,
\item  $\phi(u)$ \emph{is oriented from $M_{\phi}^{1}$ to $M^{0}_{\phi}$} if $u^{0} \in M_{\phi}^{1}$ and $u^{1} \in M_{\phi}^{0}$.
\end{itemize}
Since in any conformal model $\phi$ the chord $\phi(r)$ is oriented from 
$M^{0}_\phi$ to $M^{1}_{\phi}$, we add $r^0$ to $M^0$ and $r^1$ to $M^1$.
Next, the algorithm traverses the graph $(M,{\parallel})$ in the bfs order starting at the vertex $r$.
For every visited vertex $u$ the algorithm has already decided whether $u^0 \in M^0$ and $u^1 \in M^1$ or
whether $u^1 \in M^0$ and $u^0 \in M^1$.
Moreover, the algorithm keeps the invariant that $u^0 \in M^0$ and $u^1 \in M^1$
iff for every conformal model $\phi$ the chord
$\phi(u)$ is oriented from $M^{0}_{\phi}$ to $M^{1}_{\phi}$ and $u^1 \in M^0$ and $u^0 \in M^1$
iff for every conformal model $\phi$ the chord
$\phi(u)$ is oriented from $M^{1}_{\phi}$ to $M^{0}_{\phi}$.
When the algorithm visits a new vertex $v$, it iterates over visited vertices $u \in M$ such that $u \parallel v$ and does the following.
If $u^0 \in M^0$ and $u^1 \in M^1$ (in every conformal model $\phi$ the chord $\phi(u)$
is oriented from $M^0_\phi$ to $M^1_\phi$ -- see Figure \ref{fig:spinrad_algorithm}), the algorithm does the following:
\begin{itemize}
   \item if $v \in \leftside(u)$ and $u \in \leftside(v)$, then $u^0$ appears before $v^1$ in $\tau^0_\phi$, and $\phi(v)$ is oriented from $M^1_{\phi}$ to $M^{0}_\phi$ in every conformal model $\phi$ (see Figure \ref{fig:spinrad_algorithm}). 
   So, the algorithm orients $u \parallel v$ such that $u <_M v$, and inserts $v^0$ to $M_1$ and $v^1$ to $M^0$.
   \item if $v \in \leftside(u)$ and $u \in \rightside(v)$, then $u^0$ appears before $v^0$ in $\tau^0_\phi$, and $\phi(v)$ is oriented from $M^0_{\phi}$ to $M^{1}_\phi$ in every conformal model $\phi$ (see Figure \ref{fig:spinrad_algorithm}).
   So, the algorithm orients $u \parallel v$ such that $u <_M v$, and inserts $v^0$ to $M_0$ and $v^1$ to $M^1$.
   \item if $v \in \rightside(u)$ and $u \in \rightside(v)$, then $v^1$ appears before $u^0$ in $\tau^0_\phi$, and $\phi(v)$ is oriented from $M^1_{\phi}$ to $M^{0}_\phi$ in every conformal model $\phi$ (see Figure \ref{fig:spinrad_algorithm}).
   So the algorithm orients $u \parallel v$ such that $v <_M u$, and inserts $v^0$ to $M_1$ and $v^1$ to $M^0$.
   \item if $v \in \rightside(u)$ and $u \in \leftside(v)$, then $v^0$ appears before $u^0$ in $\tau^0_\phi$, and $\phi(v)$ is oriented from $M^0_{\phi}$ to $M^{1}_\phi$ in every conformal model $\phi$ (see Figure \ref{fig:spinrad_algorithm}).
   So, the algorithm orients $u \parallel v$ such that $v <_M u$, and inserts $v^0$ to $M_0$ and $v^1$ to $M^1$.
\end{itemize}
We proceed similarly for the case when $u^0 \in M^{1}$ and $u^{0} \in M^{1}$.

One can easily check that the triple $(M^0,M^1,{<_M})$ satisfies the thesis of the lemma.
\end{proof}

\begin{figure}[htp!]
\begin{tikzpicture}[scale=0.65,>=latex,shorten >=-0.4pt,shorten <=-0.4pt]

\coordinate (lu) at (0.6,0) {};
\coordinate (lv) at (-1.25,0) {};
\coordinate (ltau1) at (0,2.4) {};
\coordinate (ltau0) at (0,-2.4) {};

\tikzstyle{every node}=[inner sep=1pt]
\begin{scriptsize}
\node at (lu) {$\phi(u)$};
\node at (lv) {$\phi(v)$};
\node at (ltau0) {$\tau^0_{\phi}$};
\node at (ltau1) {$\tau^1_{\phi}$};
\end{scriptsize}
\draw (0,0) circle (2cm);
\draw[red,<-,thick] ([shift=(90:2cm)]0,0)--([shift=(270:2cm)]0,0);
\draw[black,->,thick] ([shift=(110:2cm)]0,0)--([shift=(250:2cm)]0,0);

\draw[very thick,->] ([shift=(135:2cm)]0,0) arc (135:45:2cm);
\draw[very thick,<-] ([shift=(235:2cm)]0,0) arc (235:315:2cm);

\end{tikzpicture}
\hspace{0.2cm}
\begin{tikzpicture}[scale=0.65,>=latex,shorten >=-0.4pt,shorten <=-0.4pt]

\coordinate (lu) at (0.6,0) {};
\coordinate (lv) at (-1.25,0) {};
\coordinate (ltau1) at (0,2.4) {};
\coordinate (ltau0) at (0,-2.4) {};

\tikzstyle{every node}=[inner sep=1pt]
\begin{scriptsize}
\node at (lu) {$\phi(u)$};
\node at (lv) {$\phi(v)$};
\node at (ltau0) {$\tau^0_{\phi}$};
\node at (ltau1) {$\tau^1_{\phi}$};
\end{scriptsize}
\draw (0,0) circle (2cm);
\draw[red,<-,thick] ([shift=(90:2cm)]0,0)--([shift=(270:2cm)]0,0);
\draw[black,<-,thick] ([shift=(110:2cm)]0,0)--([shift=(250:2cm)]0,0);

\draw[very thick,->] ([shift=(135:2cm)]0,0) arc (135:45:2cm);
\draw[very thick,<-] ([shift=(235:2cm)]0,0) arc (235:315:2cm);
\end{tikzpicture}
\hspace{0.2cm}
\begin{tikzpicture}[scale=0.65,>=latex,shorten >=-0.4pt,shorten <=-0.4pt]

\coordinate (lu) at (-0.6,0) {};
\coordinate (lv) at (1.25,0) {};
\coordinate (ltau1) at (0,2.4) {};
\coordinate (ltau0) at (0,-2.4) {};

\tikzstyle{every node}=[inner sep=1pt]
\begin{scriptsize}
\node at (lu) {$\phi(u)$};
\node at (lv) {$\phi(v)$};
\node at (ltau0) {$\tau^0_{\phi}$};
\node at (ltau1) {$\tau^1_{\phi}$};
\end{scriptsize}
\draw (0,0) circle (2cm);
\draw[red,<-,thick] ([shift=(90:2cm)]0,0)--([shift=(270:2cm)]0,0);
\draw[black,->,thick] ([shift=(70:2cm)]0,0)--([shift=(290:2cm)]0,0);

\draw[very thick,->] ([shift=(135:2cm)]0,0) arc (135:45:2cm);
\draw[very thick,<-] ([shift=(235:2cm)]0,0) arc (235:315:2cm);
\end{tikzpicture}
\hspace{0.2cm}
\begin{tikzpicture}[scale=0.65,>=latex,shorten >=-0.4pt,shorten <=-0.4pt]

\coordinate (lu) at (-0.6,0) {};
\coordinate (lv) at (1.25,0) {};
\coordinate (ltau1) at (0,2.4) {};
\coordinate (ltau0) at (0,-2.4) {};

\tikzstyle{every node}=[inner sep=1pt]
\begin{scriptsize}
\node at (lu) {$\phi(u)$};
\node at (lv) {$\phi(v)$};
\node at (ltau0) {$\tau^0_{\phi}$};
\node at (ltau1) {$\tau^1_{\phi}$};
\end{scriptsize}
\draw (0,0) circle (2cm);
\draw[red,<-,thick] ([shift=(90:2cm)]0,0)--([shift=(270:2cm)]0,0);
\draw[black,<-,thick] ([shift=(70:2cm)]0,0)--([shift=(290:2cm)]0,0);

\draw[very thick,->] ([shift=(135:2cm)]0,0) arc (135:45:2cm);
\draw[very thick,<-] ([shift=(235:2cm)]0,0) arc (235:315:2cm);
\end{tikzpicture}

\caption{\label{fig:spinrad_algorithm}}
\end{figure}
We say that $u \in M$ is \emph{oriented from $M^0$ to $M^1$ ($M^1$ to $M^0$)} if $u^0 \in M^0$ and $u^1 \in M^1$
($u^0 \in M^1$ and $u^1 \in M^0$, respectively).
So, $u$ is oriented from $M^0$ to $M^1$ iff $\phi(u)$ is oriented from $M^0_{\phi}$ to $M^1_{\phi}$ in any conformal model $\phi$.

The tuple $\mathbb{M} = (M^{0},M^{1},{<_M})$ defined by the above claim is called the \emph{metaedge} of the module $M$.
Given the metaedge $\mathbb{M}$, we define so-called admissible models for $\mathbb{M}$, which are
exactly oriented permutation models of $(M,{\sim})$ that may appear as the restrictions of conformal models of $G_{ov}$ to the sets $M^0$ and $M^1$.
\begin{definition}
\label{def:admissible_models_for_proper_submodules}
Let $M$ be a proper prime or a proper parallel module in $\mathcal{M}(G_{ov})$ and let 
$\mathbb{M} = (M^{0},M^{1},{<_M})$ be the metaedge of $M$.
A pair $(\tau^{0}, \tau^{1})$ is \emph{an admissible model for $\mathbb{M}$} if:
\begin{itemize}
 \item $\tau^0$ is a permutation of $M^{0}$,
 \item $\tau^1$ is a permutation of $M^{1}$,
 \item $(\tau^{0},\tau^{1})$ is an oriented permutation model of $(M,{\sim})$ that corresponds to the pair $({<}, {\prec})$ of transitive orientations of $(M,{\parallel})$ and $(M,{\sim})$, respectively, where ${<} = {<_{M}}$.
\end{itemize}
\end{definition}
See Figure \ref{fig:admissible_models_for_metaedges} for an illustration.

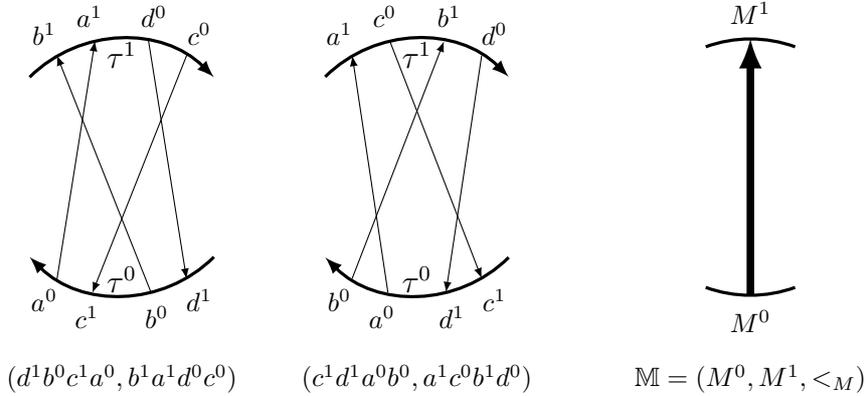
\begin{figure}[htp!]
\centering
\begin{tikzpicture}[xscale=0.85,yscale=0.85,>=latex,shorten >=-0.4pt,shorten <=-0.4pt]
\coordinate (center) at (0,0) {};
\coordinate (g1) at ($(center)+(120:2cm)$) {};
\coordinate (g2) at ($(center)+(102:2cm)$) {};
\coordinate (g3) at ($(center)+(78:2cm)$) {};
\coordinate (g4) at ($(center)+(60:2cm)$) {};
\coordinate (tau1) at ($(center)+(90:1.75cm)$) {};

\coordinate (d1) at ($(center)+(240:2cm)$) {};
\coordinate (d2) at ($(center)+(257:2cm)$) {};
\coordinate (d3) at ($(center)+(283:2cm)$) {};
\coordinate (d4) at ($(center)+(300:2cm)$) {};
\coordinate (tau0) at ($(center)+(270:1.75cm)$) {};
\coordinate (mod) at ($(center)+(270:3.3cm)$) {};

\coordinate (Lg1) at ($(center)+(120:2.4cm)$) {};
\coordinate (Lg2) at ($(center)+(102:2.4cm)$) {};
\coordinate (Lg3) at ($(center)+(78:2.4cm)$) {};
\coordinate (Lg4) at ($(center)+(60:2.4cm)$) {};

\coordinate (Ld1) at ($(center)+(240:2.4cm)$) {};
\coordinate (Ld2) at ($(center)+(257:2.4cm)$) {};
\coordinate (Ld3) at ($(center)+(283:2.4cm)$) {};
\coordinate (Ld4) at ($(center)+(300:2.4cm)$) {};

\tikzstyle{every node}=[inner sep=1pt]
\begin{footnotesize}
\node at (Lg1) {$b^1$};
\node at (Lg2) {$a^1$};
\node at (Lg3) {$d^0$};
\node at (Lg4) {$c^0$};

\node at (Ld1) {$a^0$};
\node at (Ld2) {$c^1$};
\node at (Ld3) {$b^0$};
\node at (Ld4) {$d^1$};
\node at (mod) {$(d^1b^0c^1a^0,b^1a^1d^0c^0)$};

\end{footnotesize}

\node at (tau0) {$\tau^0$};
\node at (tau1) {$\tau^1$};

\draw[very thick,->] ([shift=(135:2cm)]0,0) arc (135:45:2cm);
\draw[very thick,->] ([shift=(315:2cm)]0,0) arc (315:225:2cm);

\draw[->] (d1)--(g2);
\draw[->] (g4)--(d2);
\draw[->] (d3)--(g1);
\draw[->] (g3)--(d4);

\end{tikzpicture}
\hspace{0.5cm}
\begin{tikzpicture}[xscale=0.85,yscale=0.85,>=latex,shorten >=-0.2pt,shorten <=-0.2pt]
\coordinate (center) at (0,0) {};
\coordinate (g1) at ($(center)+(120:2cm)$) {};
\coordinate (g2) at ($(center)+(102:2cm)$) {};
\coordinate (g3) at ($(center)+(78:2cm)$) {};
\coordinate (g4) at ($(center)+(60:2cm)$) {};
\coordinate (tau1) at ($(center)+(90:1.75cm)$) {};
\coordinate (mod) at ($(center)+(270:3.3cm)$) {};

\coordinate (d1) at ($(center)+(240:2cm)$) {};
\coordinate (d2) at ($(center)+(257:2cm)$) {};
\coordinate (d3) at ($(center)+(283:2cm)$) {};
\coordinate (d4) at ($(center)+(300:2cm)$) {};
\coordinate (tau0) at ($(center)+(270:1.75cm)$) {};

\coordinate (Lg1) at ($(center)+(120:2.4cm)$) {};
\coordinate (Lg2) at ($(center)+(102:2.4cm)$) {};
\coordinate (Lg3) at ($(center)+(78:2.4cm)$) {};
\coordinate (Lg4) at ($(center)+(60:2.4cm)$) {};

\coordinate (Ld1) at ($(center)+(240:2.4cm)$) {};
\coordinate (Ld2) at ($(center)+(257:2.4cm)$) {};
\coordinate (Ld3) at ($(center)+(283:2.4cm)$) {};
\coordinate (Ld4) at ($(center)+(300:2.4cm)$) {};

\tikzstyle{every node}=[inner sep=1pt]
\begin{footnotesize}
\node at (Lg1) {$a^1$};
\node at (Lg2) {$c^0$};
\node at (Lg3) {$b^1$};
\node at (Lg4) {$d^0$};

\node at (Ld1) {$b^0$};
\node at (Ld2) {$a^0$};
\node at (Ld3) {$d^1$};
\node at (Ld4) {$c^1$};

\node at (mod) {$(c^1d^1a^0b^0,a^1c^0b^1d^0)$};
\end{footnotesize}
\node at (tau0) {$\tau^0$};
\node at (tau1) {$\tau^1$};

\draw[very thick,->] ([shift=(135:2cm)]0,0) arc (135:45:2cm);
\draw[very thick,->] ([shift=(315:2cm)]0,0) arc (315:225:2cm);

\draw[->] (d1)--(g3);
\draw[->] (d2)--(g1);
\draw[->] (g4)--(d3);
\draw[->] (g2)--(d4);
\end{tikzpicture}
\hspace{1cm}
\begin{tikzpicture}[xscale=0.85,yscale=0.85,>=latex,shorten >=-0.2pt,shorten <=-0.2pt]
\coordinate (center) at (0,0) {};
\coordinate (g1) at ($(center)+(90:2cm)$) {};
\coordinate (d1) at ($(center)+(270:2cm)$) {};

\coordinate (Lg1) at ($(center)+(90:2.4cm)$) {};
\coordinate (Ld1) at ($(center)+(270:2.4cm)$) {};
\coordinate (LLd1) at ($(center)+(270:3.3cm)$) {};

\tikzstyle{every node}=[inner sep=1pt]
\begin{footnotesize}
\node at (Lg1) {$M^1$};
\node at (Ld1) {$M^0$};
\node at (LLd1) {$\mathbb{M} = (M^0,M^1,{<_M})$};
\end{footnotesize}

\draw[very thick] ([shift=(70:2cm)]0,0) arc (70:110:2cm);
\draw[very thick] ([shift=(250:2cm)]0,0) arc (250:290:2cm);

\draw[line width=1mm,->] (d1)--(g1);
\end{tikzpicture}

\caption{\label{fig:admissible_models_for_metaedges} 
Two admissible models for the metaedge $\mathbb{M} = (M^0,M^1,{<_{M}})$, 
where $M^{0} = \{d^1,b^0,c^1,a^0\}$, $M^1 = \{b^1,a^1,d^0,c^0\}$, and ${<_M} = \{(d,a),(d,b),(c,a)\}$.
In these two models the relative position of non-intersecting chords is the same.}
\end{figure}

Given the above definition, we can summarize the results of this subsection with the following lemma.
\begin{lemma}
\label{lem:restriction_to_a_proper_module_is_admissible}
Suppose $\phi$ is a conformal model of $G_{ov}$, $M$ is a proper prime or a proper parallel module in $\mathcal{M}(G_{ov})$, and $\mathbb{M} = (M^0,M^1,{<_M})$ is the metaedge of $M$.
Then, the pair $(\phi|M^0, \phi|M^1)$ is an admissible model for $\mathbb{M}$.
\end{lemma}
In our figures, we represent the metaedge $\mathbb{M}$ by the bolded oriented chord -- see Figure \ref{fig:admissible_models_for_metaedges}.
If $(\phi|M^0,\phi|M^1)$ is admissible for $\mathbb{M}$, the endpoints of $\mathbb{M}$ indicate the 
positions of the subwords $\phi|M^0$ and $\phi|M^1$ in the circular word $\phi$.

\subsection{Conformal models of serial modules}
\label{subsec:conformal_models_of_serial_modules}
Suppose $M$ is a serial module in $G_{ov}$ with children $M_1,\ldots,M_k$.
Since $M$ is serial, every $M_i$ is a proper prime or a proper parallel module in $\mathcal{M}(G_{ov})$.
Suppose that the representant $r_i$ of $M_i$ is fixed and
$\mathbb{M}_i = (M_i^{0}$, $M_i^{1},{<_{M_i}})$ is the metaedge of $M_i$.
Having in mind Lemmas~\ref{lemma:circle_models_of_a_serial_module} and~\ref{lem:restriction_to_a_proper_module_is_admissible}, 
one can obtain a theorem describing all conformal models of $(M,{\sim})$, which
actually follows from the work done by Spinrad \cite{Spin88} and Hsu \cite{Hsu95}.
\begin{theorem}
\label{thm:description_of_all_conformal_models_of_serial_modules}
Suppose $M$ is a serial module in $G_{ov}$ with children $M_1,\ldots,M_k$.
Every conformal model $\phi$ of $(M,{\sim})$ has the form
$$\phi \equiv \mu_{i_1} \ldots \mu_{i_k} \mu'_{i_1} \ldots \mu'_{i_k},$$
where $i_1,\ldots,i_k$ is a permutation of $[k]$ and for every 
$j \in [k]$ either the pair $(\mu_{i_j}, \mu'_{i_j})$ or the pair $(\mu'_{i_j}, \mu_{i_j})$ 
is an admissible model for $\mathbb{M}_{i_j}$.

On the other hand, for every permutation $(i_1,\ldots,i_k)$ of $[k]$, every admissible model $(\tau_{i}, \tau'_{i})$ of $\mathbb{M}_{i}$, and every two words $\mu_i,\mu'_i$ such that $\{\tau_{i},\tau'_{i}\} = \{\mu_{i},\mu'_{i}\}$, a circular word
$$\phi \equiv \mu_{i_1} \ldots \mu_{i_k} \mu'_{i_1} \ldots \mu'_{i_k}$$
is a conformal model of $(M,{\sim})$.
\end{theorem}
The above theorem can be used to characterize all conformal models of $G_{ov}$ 
in the case when the graph $(V,{\parallel})$ is disconnected.
Indeed, in this case $V$ is serial in $G_{ov}$ and Theorem \ref{thm:description_of_all_conformal_models_of_serial_modules} applies.
The schematic picture of some conformal models of $(V,{\sim})$ is shown in Figure~\ref{fig:serial_case_schematic_view}. 
To get the full picture of a normalized model one needs to replace (expand) every metaedge $(M^0_i,M^1_i,{<_{M_i}})$ with an admissible model hidden behind this metaedge (two admissible models for $(M^0_i,M^1_i,{<_{M_i}})$ are shown at the bottom of the figure).

\begin{figure}[htp!]
\begin{tikzpicture}[scale=0.75,>=latex,shorten >=-0.2pt,shorten <=-0.2pt]
\coordinate (center) at (0.0,-0.5) {};

\coordinate (k11) at ($(center)+(90:2cm)$) {};
\coordinate (k10) at ($(center)+(270:2cm)$) {};
\coordinate (k21) at ($(center)+(0:2cm)$) {};
\coordinate (k20) at ($(center)+(180:2cm)$) {};
\coordinate (k31) at ($(center)+(225:2cm)$) {};
\coordinate (k30) at ($(center)+(45:2cm)$) {};
\coordinate (k41) at ($(center)+(315:2cm)$) {};
\coordinate (k40) at ($(center)+(135:2cm)$) {};

\coordinate (lk11) at ($(center)+(90:2.4cm)$) {};
\coordinate (lk10) at ($(center)+(270:2.4cm)$) {};
\coordinate (lk21) at ($(center)+(0:2.4cm)$) {};
\coordinate (lk20) at ($(center)+(180:2.4cm)$) {};
\coordinate (lk31) at ($(center)+(45:2.4cm)$) {};
\coordinate (lk30) at ($(center)+(225:2.4cm)$) {};
\coordinate (lk41) at ($(center)+(135:2.4cm)$) {};
\coordinate (lk40) at ($(center)+(315:2.4cm)$) {};

\tikzstyle{every node}=[inner sep=1pt]
\begin{footnotesize}
\node at (lk11) {$M_1^1$};
\node at (lk10) {$M_1^0$};
\node at (lk21) {$M_2^1$};
\node at (lk20) {$M_2^0$};
\node at (lk31) {$M_3^0$};
\node at (lk30) {$M_3^1$};
\node at (lk41) {$M_4^0$};
\node at (lk40) {$M_4^1$};
\end{footnotesize}

\draw (center) circle (2cm);

\draw[line width=1mm,->] (k20)--(k21);
\draw[line width=1mm,->] (k30)--(k31);
\draw[line width=1mm,->] (k40)--(k41);
\draw[line width=1mm,->] (k10)--(k11);

\draw[white] (-3.0,-3.2)--(-3.0,-3.0);
\draw[white] (3.0,2.3)--(3.0,2.1);
\end{tikzpicture}
\begin{tikzpicture}[scale=0.75,>=latex,shorten >=-0.2pt,shorten <=-0.2pt]
\coordinate (center) at (0.0,-0.5) {};

\coordinate (k11) at ($(center)+(45:2cm)$) {};
\coordinate (k10) at ($(center)+(225:2cm)$) {};
\coordinate (k21) at ($(center)+(0:2cm)$) {};
\coordinate (k20) at ($(center)+(180:2cm)$) {};
\coordinate (k31) at ($(center)+(270:2cm)$) {};
\coordinate (k30) at ($(center)+(90:2cm)$) {};
\coordinate (k41) at ($(center)+(315:2cm)$) {};
\coordinate (k40) at ($(center)+(135:2cm)$) {};

\coordinate (lk11) at ($(center)+(45:2.4cm)$) {};
\coordinate (lk10) at ($(center)+(225:2.4cm)$) {};
\coordinate (lk21) at ($(center)+(0:2.4cm)$) {};
\coordinate (lk20) at ($(center)+(180:2.4cm)$) {};
\coordinate (lk31) at ($(center)+(270:2.4cm)$) {};
\coordinate (lk30) at ($(center)+(90:2.4cm)$) {};
\coordinate (lk41) at ($(center)+(135:2.4cm)$) {};
\coordinate (lk40) at ($(center)+(315:2.4cm)$) {};

\tikzstyle{every node}=[inner sep=1pt]
\begin{footnotesize}
\node at (lk11) {$M_1^1$};
\node at (lk10) {$M_1^0$};
\node at (lk21) {$M_2^1$};
\node at (lk20) {$M_2^0$};
\node at (lk31) {$M_3^1$};
\node at (lk30) {$M_3^0$};
\node at (lk41) {$M_4^0$};
\node at (lk40) {$M_4^1$};
\end{footnotesize}

\draw (center) circle (2cm);

\draw[line width=1mm,->] (k20)--(k21);
\draw[line width=1mm,->] (k40)--(k41);
\draw[line width=1mm,->,red] (k10)--(k11);
\draw[line width=1mm,->,red] (k30)--(k31);

\draw[white] (-3.0,-3.2)--(-3.0,-3.0);
\draw[white] (3.0,2.3)--(3.0,2.1);
\end{tikzpicture}
\begin{tikzpicture}[scale=0.75,>=latex,shorten >=-0.2pt,shorten <=-0.2pt]
\coordinate (center) at (0.0,-0.5) {};

\coordinate (k11) at ($(center)+(270:2cm)$) {};
\coordinate (k10) at ($(center)+(90:2cm)$) {};
\coordinate (k21) at ($(center)+(0:2cm)$) {};
\coordinate (k20) at ($(center)+(180:2cm)$) {};
\coordinate (k31) at ($(center)+(225:2cm)$) {};
\coordinate (k30) at ($(center)+(45:2cm)$) {};
\coordinate (k41) at ($(center)+(315:2cm)$) {};
\coordinate (k40) at ($(center)+(135:2cm)$) {};

\coordinate (lk11) at ($(center)+(270:2.4cm)$) {};
\coordinate (lk10) at ($(center)+(90:2.4cm)$) {};
\coordinate (lk21) at ($(center)+(0:2.4cm)$) {};
\coordinate (lk20) at ($(center)+(180:2.4cm)$) {};
\coordinate (lk31) at ($(center)+(45:2.4cm)$) {};
\coordinate (lk30) at ($(center)+(225:2.4cm)$) {};
\coordinate (lk41) at ($(center)+(135:2.4cm)$) {};
\coordinate (lk40) at ($(center)+(315:2.4cm)$) {};

\tikzstyle{every node}=[inner sep=1pt]
\begin{footnotesize}
\node at (lk11) {$M_1^1$};
\node at (lk10) {$M_1^0$};
\node at (lk21) {$M_2^1$};
\node at (lk20) {$M_2^0$};
\node at (lk31) {$M_3^0$};
\node at (lk30) {$M_3^1$};
\node at (lk41) {$M_4^0$};
\node at (lk40) {$M_4^1$};
\end{footnotesize}

\draw (center) circle (2cm);

\draw[line width=1mm,->] (k20)--(k21);
\draw[line width=1mm,->] (k30)--(k31);
\draw[line width=1mm,->] (k40)--(k41);
\draw[line width=1mm,->,red] (k10)--(k11);

\draw[white] (-3.0,-3.2)--(-3.0,-3.0);
\draw[white] (3.0,2.3)--(3.0,2.1);
\end{tikzpicture}

\begin{tikzpicture}[xscale=0.7,yscale=0.7,>=latex,shorten >=-0.4pt,shorten <=-0.4pt]
\coordinate (center) at (0,0) {};
\coordinate (au) at ($(center)+(60:2cm)$) {};
\coordinate (du) at ($(center)+(75:2cm)$) {};
\coordinate (cu) at ($(center)+(90:2cm)$) {};
\coordinate (bu) at ($(center)+(105:2cm)$) {};
\coordinate (eu) at ($(center)+(120:2cm)$) {};

\coordinate (al) at ($(center)+(300:2cm)$) {};
\coordinate (bl) at ($(center)+(285:2cm)$) {};
\coordinate (cl) at ($(center)+(270:2cm)$) {};
\coordinate (dl) at ($(center)+(255:2cm)$) {};
\coordinate (el) at ($(center)+(240:2cm)$) {};

\coordinate (lau) at ($(center)+(60:2.4cm)$) {};
\coordinate (ldu) at ($(center)+(75:2.4cm)$) {};
\coordinate (lcu) at ($(center)+(90:2.4cm)$) {};
\coordinate (lbu) at ($(center)+(105:2.4cm)$) {};
\coordinate (leu) at ($(center)+(120:2.4cm)$) {};

\coordinate (lal) at ($(center)+(300:2.4cm)$) {};
\coordinate (lbl) at ($(center)+(285:2.4cm)$) {};
\coordinate (lcl) at ($(center)+(270:2.4cm)$) {};
\coordinate (ldl) at ($(center)+(255:2.4cm)$) {};
\coordinate (lel) at ($(center)+(240:2.4cm)$) {};

\coordinate (below) at (0,-3) {};

\tikzstyle{every node}=[inner sep=1pt]
\begin{footnotesize}

\node at (lal) {$a^0$};
\node at (lbl) {$b^0$};
\node at (lcl) {$c^1$};
\node at (ldl) {$d^0$};
\node at (lel) {$e^0$};

\node at (below) {$M^0_i$};
\end{footnotesize}

\draw[->,red] (al)--(au);
\draw[->] (bl)--(bu);
\draw[->] (cu)--(cl);
\draw[->] (dl)--(du);
\draw[->] (el)--(eu);

\draw[very thick,->] ([shift=(135:2cm)]0,0) arc (135:45:2cm);
\draw[very thick,->] ([shift=(315:2cm)]0,0) arc (315:225:2cm);

\draw[white] (-2,-3.2)--(-2,-3.0);
\draw[white] (2,2.3)--(2,2.1);
\end{tikzpicture}
\hspace{1.52cm}
\begin{tikzpicture}[xscale=0.7,yscale=0.7,>=latex,shorten >=-0.4pt,shorten <=-0.4pt]
\coordinate (center) at (0,0) {};
\coordinate (au) at ($(center)+(60:2cm)$) {};
\coordinate (du) at ($(center)+(75:2cm)$) {};
\coordinate (bu) at ($(center)+(90:2cm)$) {};
\coordinate (cu) at ($(center)+(105:2cm)$) {};
\coordinate (eu) at ($(center)+(120:2cm)$) {};

\coordinate (al) at ($(center)+(300:2cm)$) {};
\coordinate (cl) at ($(center)+(285:2cm)$) {};
\coordinate (bl) at ($(center)+(270:2cm)$) {};
\coordinate (dl) at ($(center)+(255:2cm)$) {};
\coordinate (el) at ($(center)+(240:2cm)$) {};

\coordinate (lau) at ($(center)+(60:2.4cm)$) {};
\coordinate (ldu) at ($(center)+(75:2.4cm)$) {};
\coordinate (lbu) at ($(center)+(90:2.4cm)$) {};
\coordinate (lcu) at ($(center)+(105:2.4cm)$) {};
\coordinate (leu) at ($(center)+(120:2.4cm)$) {};

\coordinate (lal) at ($(center)+(300:2.4cm)$) {};
\coordinate (lcl) at ($(center)+(285:2.4cm)$) {};
\coordinate (lbl) at ($(center)+(270:2.4cm)$) {};
\coordinate (ldl) at ($(center)+(255:2.4cm)$) {};
\coordinate (lel) at ($(center)+(240:2.4cm)$) {};

\coordinate (below) at (0,-3) {};

\tikzstyle{every node}=[inner sep=1pt]
\begin{footnotesize}

\node at (lal) {$a^0$};
\node at (lbl) {$b^0$};
\node at (lcl) {$c^1$};
\node at (ldl) {$d^0$};
\node at (lel) {$e^0$};

\node at (below) {$M^0_i$};
\end{footnotesize}

\draw[->,red] (al)--(au);
\draw[->,thick,blue] (bl)--(bu);
\draw[->,thick,blue] (cu)--(cl);
\draw[->] (dl)--(du);
\draw[->] (el)--(eu);

\draw[very thick,->] ([shift=(135:2cm)]0,0) arc (135:45:2cm);
\draw[very thick,->] ([shift=(315:2cm)]0,0) arc (315:225:2cm);

\draw[white] (-2,-3.2)--(-2,-3.0);
\draw[white] (2,2.3)--(2,2.1);
\end{tikzpicture}

\caption{\label{fig:serial_case_schematic_view}}
\end{figure}

To transform a normalized model of $G_{ov}$ into a normalized model of $G_{ov}$ 
we can permute and reorient the metaedges arbitrarily and we can replace an admissible model
hidden behind every metaedge - see Figure \ref{fig:serial_case_schematic_view}.
Theorem \ref{thm:description_of_all_conformal_models_of_serial_modules} asserts that
such operations are complete, that is, we can transform one normalized model into any other by performing the above operations.


\subsection{Conformal models of improper prime modules}
\label{subsec:conformal_models_of_improper_prime_modules}
Suppose $G$ is a circular-arc graph with no twins and no universal vertices.
Suppose $M$ is an improper prime module in $G_{ov}$.
That is, either: 
\begin{itemize}
 \item $M = V$ and $V$ is a prime module in $G_{ov}$, which means that both $(V,{\sim})$ and $(V,{\parallel})$
 are connected, or 
 \item $M \subsetneq V$, $V$ is a parallel module in $G_{ov}$, and $M$ a is prime child of $V$ in $\mathcal{M}(G_{ov})$,
 which means that $(V,{\sim})$ is disconnected, $(M,{\sim})$ is a connected component of $(V,{\sim})$ such that
 $(M,{\parallel})$ is connected.
\end{itemize}

Let $M_1,\ldots,M_k$ be the children of the module $M$.
Let $U$ be a set containing exactly one element from every $M_i$, $i \in [k]$.
Clearly, $(U,{\sim})$ contains no non-trivial modules and hence the graph $(U,{\sim})$ is prime.
The next lemma is crucial for our work.
We mention here that an analogous lemma was stated by Hsu \cite{Hsu95}. 
However, we prove it again as we work with conformal models defined in a different way.
\begin{lemma}
\label{lemma:two_models_of_a_prime_graph}
The graph $(U,{\sim})$ has exactly two conformal models, one being the reflection of the other.
\end{lemma}
\begin{proof}
Since $G$ is a circular-arc graph, $G_{ov}$ has at least one conformal model.
Since the restriction of any conformal model of $G_{ov}$ to the set $U$ is conformal, 
$(U,{\sim})$ has at least one conformal model.
Our goal is to prove that this model, up to reflection, is unique.

We prove the lemma by induction on the number of vertices in $(U,{\sim})$.
The smallest prime graph has $4$ vertices.
The only prime graph with $4$ vertices is isomorphic to the path $P_4$.
One can easily check that $P_4$ has two conformal models, 
one being the reflection of the other.
This proves the base of the induction.

Suppose $(U,{\sim})$ has at least $5$ vertices.
If $(U,{\sim})$ has no non-trivial splits, 
Theorem~\ref{thm:cicle_representation_no_split} asserts that the graph $(U,{\sim})$ has two chord models,
$\phi$ and $\phi^R$, where $\phi^R$ is the reflection of $\phi$.
Note that there is a unique orientation of the chords from $\phi$ that lead to a conformal model for $(U,{\sim})$.
Indeed, as $(U,{\sim})$ is prime, for every vertex $v \in U$ there is $u \in U$ such that $u \parallel v$.
Now, the orientation of the chord $\phi(v)$ can be decided basing on whether $u \in \leftside(v)$ or $u \in \rightside(v)$.
Similarly, there is a unique orientation of the chords in $\phi^R$ that lead to a conformal model of $(U,{\sim})$.
Clearly, the oriented conformal models $\phi$ and $\phi^R$ are the only conformal models of $(U,{\sim})$,
$\phi^R$ needs to be the reflection of $\phi$, and hence the thesis of the lemma holds.

Suppose $(U,{\sim})$ has a non-trivial split.
In this case the proof goes as follows.
We take a maximal split in $(U,{\sim})$ and then, using a structure induced by this split,
we divide $(U,{\sim})$ into so-called \emph{probes}.
A probe is a special proper induced subgraph of $(U,{\sim})$ which, as we shall prove, 
has a unique, up to reflection, conformal model.
Eventually, we show that there is a unique way to fit the models of the probes together 
to get a conformal model of $(U,{\sim})$.

\begin{definition}
A \emph{probe} in $(U,{\sim})$ is a quadruple $(y,x,X,\alpha(X))$ that satisfies the following properties:
\begin{enumerate}
 \item \label{item:probe_P_subset} $x \neq y$, $X \neq \emptyset$, $\alpha(X) \neq \emptyset$, 
 the sets $\{y,x\}$, $X$, $\alpha(X)$ are pairwise disjoint, and the set $P = \{x,y\} \cup X \cup \alpha(X)$ is a proper subset of $U$,
 \item \label{item:probe_yx_relation} $y \sim x$, $y \parallel X \cup \alpha(X)$,
 $x \sim  X$, $x \parallel \alpha(X)$, and the graph $(P,{\sim})$ is connected,
 \item \label{item:probe_neighborood_outside} for every $z \in U \setminus P$, 
 either $z \parallel (X \cup \alpha(X))$, or $z \sim X  \text{ and } z \parallel \alpha(X)$, 
 or $z \sim (X \cup \alpha(X))$.
\end{enumerate}
\end{definition}

\begin{claim}
\label{claim:probes_have_unique_model}.
Let $(y,x,X,\alpha(X))$ be a probe in $(U,{\sim})$, let $P = \{x,y\} \cup X \cup \alpha(X)$.
Then, $(P,{\sim})$ has a unique, up to reflection, conformal model.
\end{claim}
\begin{proof}
Let $Z = \{z \in X: \text{$z$ has only one neighbour in the graph $(P,{\sim})$}\}$.
Note that the only neighbor of $z \in Z$ is the vertex $x$.
Note also that $|Z| \leq 1$ as otherwise $Z$~would be a non-trivial module in $(U,{\sim})$ by property~\eqref{item:probe_neighborood_outside}.
We claim that:
\begin{itemize}
\item If $|Z| = 1$, then $\{y\} \cup Z$ is the only non-trivial module in $(P,{\sim})$.
\item If $Z = \emptyset$, then $(P,{\sim})$ has no non-trivial modules.
\end{itemize}

Suppose $M$ is a non-trivial module in $(P,{\sim})$.
We consider four cases depending on the intersection of $M$ with the set $\{y,x\}$.

Suppose $M \cap \{y,x\} = \emptyset$. 
Since $x \notin M$ and since $x \sim X$ and $x \parallel \alpha(X)$, 
we must have either $M \subseteq X$ or $M \subseteq \alpha(X)$.
Then, by property~\eqref{item:probe_neighborood_outside} of $P$, every $u \in U \setminus P$ satisfies either $u \sim M$ or $u \parallel M$.
So, $M$ is also a non-trivial module in $(U,{\sim})$, which contradicts the assumption of the lemma.

Suppose $M \cap \{y,x\}  = \{y,x\}$.
Since $X \sim x$ and $X \parallel y$, we must have $X \subseteq M$.
Since $(P,{\sim})$ is connected, we need to have $\alpha(X) \subset M$
as otherwise we would find a vertex $u \in P \setminus M$ such that $u$
is adjacent to a vertex in $M$ and non-adjacent to a vertex in $M$.
So, $M =P$, which contradicts that $M$ is a non-trivial module in $(P,{\sim})$.

Suppose $M \cap \{y,x\} = \{x\}$.
Since $y \sim x$ and $ y \parallel X \cup \alpha(X)$, 
we must have $M \cap (X \cup \alpha(X))= \emptyset$.
It follows that $M$ is trivial in $(P,{\sim})$, a contradiction.

Suppose $M \cap \{y,x\} = \{y\}$.
Note that $M \cap \alpha(X) = \emptyset$.
Otherwise, $x$ from outside $M$ is adjacent to $y$ in $M$ and
non-adjacent to a vertex in $M \cap \alpha(X)$, which can not be the case.
Let $M_X = M \cap X$.
If $M_X = \emptyset$, then $M = \{y\}$ and $M$ is trivial.
So, we must have $M_X \neq \emptyset$.
Notice that for every vertex $t \in (X \cup \alpha(X)) \setminus M_X$ we have that $t \parallel M_x$.
Otherwise, $t$ from outside $M$ would have a neighbor in $M$ and the non-neighbor $y$ in $M$.
If $|M_X| \geq 2$, then by property \eqref{item:probe_neighborood_outside} of~$P$, $M_X$ would be a non-trivial module in $(U,{\sim})$, which can not be the case.
So, $M$ might be a non-trivial module of $(P,{\sim})$ only when $|M_X| = 1$, i.e., when $M_X = \{z\}$ for some $z \in X$.
In this case, $z$ is adjacent only to the vertex $x$ in $(P,{\sim})$, which shows $Z = \{z\}$.
So, we have $M = \{y,z\}$, which completes the proof of our subclaim.

Now, we show that $(P,{\sim})$ has a unique, up to reflection, conformal model.

Suppose $Z = \emptyset$.
As we have shown, the graph $(P,{\sim})$ contains no non-trivial modules. 
Since $P$ has strictly fewer vertices than $(U,{\sim})$, 
from the inductive hypothesis we get that $(P,{\sim})$ has a unique, up to reflection, conformal model.

Suppose $Z = \{z\}$.
Then $\{y,z\}$ is the only non-trivial module in $(P,{\sim})$.
Since $(P,{\sim})$ and $(P,{\parallel})$ are connected, the graph $(P \setminus \{z\}, {\sim})$ is prime.
By the inductive hypothesis, $(P \setminus \{z\},{\sim})$ has exactly two conformal models, $\phi$ and $\phi^R$, where $\phi^R$ is the reflection of~$\phi$.
Note that the vertex $x$ is an articulation point in the graph $(P \setminus \{z\},{\sim})$.
Suppose that $(P \setminus \{z,x\},{\sim})$ has exactly $k$ connected components, 
say $D_1,\ldots,D_k$, for some $k \geq 2$.
Note that $D_i  = \{y\}$ for some $i \in [k]$.
By Theorem \ref{thm:trivial_split_representations}, $\phi \equiv x^0\tau_{i_1}\ldots \tau_{i_{k}}x^1 \tau'_{i_k}\ldots\tau'_{i_1}$, where $i_1,\ldots,i_k$ is a permutation of~$[k]$ and $x^0 \tau_{i_j} x^1 \tau'_{i_j}$ is a conformal model of $(\{x\} \cup D_{i_j},{\sim})$ for $j \in [k]$.
We show that there is a unique extension of $\phi$ by the oriented chord $\phi(z)$ 
such that the extended $\phi$ is conformal for $(P,{\sim})$.
Clearly, the extended $\phi$ must be of the form:
$$\phi \equiv x^0\tau_{i_1}\ldots \tau_{i_l} z' \tau_{i_{l+1}} \ldots \tau_{i_{k}}x^1 \tau'_{i_k} \ldots \tau'_{i_{l+1}} z'' \tau'_{i_{l}} \ldots \tau'_{i_1} \text{ for some } l \in \{0,\ldots,k\},$$
where  $z'$ and $z''$ are such that $\{z',z''\} = \{z^{0},z^{1}\}.$
For every $i \in [k]$ pick a vertex $a_i$ in the component $D_i$ such that $x \sim a_i$.
Note that $\phi(z)$ must be on the left side of $\phi(a_i)$ if $z \in \leftside(a_i)$ 
and on the right side of $\phi(a_i)$ if $z \in \rightside(a_i)$.
Hence, the place in $\phi$ (or equivalently, the index $l$) for the chord $\phi(z)$ is uniquely determined.
The orientation of $\phi(z)$ can be based on whether $y \in \leftside(z)$ or $y \in \rightside(z)$ holds.
\end{proof}

Suppose $(U,{\sim})$ has a non-trivial split.
We use the algorithm given in Section~\ref{sec:split_decomposition} to compute a maximal split $(A,B)$ in $(U,{\sim})$. 
Depending on whether $(A,B)$ is trivial or not, we assume the following notation:
\begin{itemize}
 \item if $(A,B)$ is non-trivial, we assume that $C_1,\ldots,C_k$ and $\alpha(C_1),\ldots,\alpha(C_k)$ are such as defined in Subsection~\ref{subsec:structure_non_trivial_split},
 \item if $(A,B)$ is trivial, we assume that $A=\{a\}$ and that $C_1,\ldots,C_k$ and $\alpha(C_1),\ldots,\alpha(C_k)$ are such as defined in Subsection~\ref{subsec:structure_trivial_split}.
\end{itemize}

We partition the set $[k]$ into two subsets, $I_1$ and $I_2$, such that:
\begin{itemize}
 \item $i \in I_1$ if $|C_i \cup \alpha(C_i)| = 1$,
 \item $i \in I_2$ if $|C_i \cup \alpha(C_i)| \geq 2$.
\end{itemize}
Note that $|I_1| \leq 1$ as otherwise $\bigcup_{i \in I_1} C_i$ would be a non-trivial module
in $(U,{\sim})$.
Without loss of generality we assume $C_1,\ldots,C_k$ are enumerated such that 
$I_1 = \{k\}$ if $I_1 \neq \emptyset$.
For $i \in I_2$ we have $\alpha(C_i) \neq \emptyset$ as otherwise
$C_i$ would be a non-trivial module in $(U,{\sim})$.
Moreover, since $(U,{\sim})$ is connected, 
some vertex in $C_i$ is adjacent to some vertex in $\alpha(C_i)$.
Hence, for every $i \in [k]$ we can pick two vertices $a_i,b_i \in C_i \cup \alpha(C_i)$ such that:
\begin{itemize}
\item $a_i \in C_i$, $b_i \in \alpha(C_i)$, and $a_i \sim b_i$, if $i \in I_2$,
\item $a_i = b_i$, where $\{a_i\}$ is the only vertex in $C_i$, if $i \in I_1$.
\end{itemize}
We split the proof into two cases, depending on whether or not the following condition is satisfied:
\begin{equation}
\label{eq:probe_condition}
\begin{array}{c}
\text{For every $i \in I_2$ there exist $x,y \in U \setminus (C_i \cup \alpha(C_i))$ such that} \\
\text{$(y,x,C_i,\alpha(C_i))$ is a probe in $(U,{\sim})$}.
\end{array}
\tag{*}
\end{equation}
We claim that condition \eqref{eq:probe_condition} is not satisfied only when $(A,B)$ is a trivial split,
$k=2$, and $|C_2 \cup \alpha(C_2)|=1$.
Let $i \in I_2$.
Suppose $k\geq 3$.
If $(A,B)$ is non-trivial, 
the set $C_i \cup \alpha(C_i)$ can be extended to a probe 
by the vertices $a_j,b_j$, where $j$ is any index in $I_2$ different than $i$.
If $(A,B)$ is trivial, the set $C_i \cup \alpha(C_i)$ can be extended to a probe
by the vertices $a, a_j$, where $j$ is any index in $[k]$ different than $i$.
Suppose $k=2$ and suppose $(A,B)$ is non-trivial.
Note that $|C_j \cup \alpha(C_j)|\geq 3$ for every $j \in [2]$.
Otherwise, the only vertex $a_j$ in $C_j$ is adjacent to the only vertex $b_j \in \alpha(C_j)$, 
and hence the split $(A,B)$ is not maximal.
Hence, the set $C_i \cup \alpha(C_i)$ can be extended to a probe by the vertices $a_j,b_j$, 
where $j$ is the index in $[2]$ different than $i$.
If $(A,B)$ is trivial and $|C_2 \cup \alpha(C_2)| \geq 2$, then the set $C_i \cup \alpha(C_i)$ can be extended to a probe by the vertices $a,a_j$, where $j$ is the index in $[2]$ different than $i$.
So, the only case when condition \eqref{eq:probe_condition} is not satisfied is when 
$(A,B)$ is a trivial split,
$k=2$, $|C_2| = 1$, and $|\alpha(C_2)|=0$.

Suppose \eqref{eq:probe_condition} is satisfied.
Let 
$$ R_i = \left\{
\begin{array}{lll}
 \{a_1, a_2, \ldots, a_i,b_i\}& \text{if} & (A,B) \text{ is non-trivial},\\
 \{a, a_1,b_1, \ldots, a_i,b_i\} & \text{if} & (A,B) \text{ is trivial},
\end{array}
\right.
$$
let
$$
S = \left\{
\begin{array}{lll}
 \{a_1, a_2\}& \text{if} & (A,B) \text{ is non-trivial},\\
 \{a,a_1\} & \text{if} & (A,B) \text{ is trivial},
\end{array}
\right.
$$
and let $R = R_k$.
Eventually, let
$$
\begin{array}{lll}
 \phi^{0}_S \equiv a_1^0a^0_2a_{1}^1a_2^1 \quad \text{and} \quad \phi^{1}_S \equiv a_1^0a^1_2a_{1}^1a_2^0 &\text{if} & (A,B) \text{ is non-trivial},\\
 \phi^{0}_S \equiv a^0a^0_1a^1a_1^1 \quad \text{and} \quad \phi^{1}_S \equiv a^0a^1_1a^1a^0_1 & \text{if} & (A,B) \text{ is trivial}.\\
\end{array}
$$
In any case, $\phi^{0}_S$ is the reflection of $\phi^{1}_S$ and any
conformal model $\phi$ of $(U,{\sim})$ extends either $\phi^{0}_S$ or $\phi^{1}_S$.
We claim that:
\begin{itemize}
 \item There is a unique conformal model $\phi^{j}_R$ of $(R,{\sim})$ such that $\phi_R|S = \phi^j_S$, for every $j \in \{0,1\}$.
 \item For every conformal model $\phi^j_R$ of $(R,{\sim})$ extending $\phi^j_S$ there is at most one conformal
 model $\phi^j$ of $(U,{\sim})$ such that $\phi^j|R \equiv \phi^j_R$, for every $j \in [2]$.
\end{itemize}
Then, $\phi^{0}_R$ must be the reflection of $\phi^{1}_R$, 
and $\phi^1$ must be the reflection of $\phi^0$.
This will show the lemma for the case when condition \eqref{eq:probe_condition} is satisfied.

First we prove our second claim.
Let $i \in I_2$ and let $\phi^{0}_i$ be the unique conformal model of $(P_i,{\sim})$ that extends $\phi^0_S$, 
where $P_i = \{y,x\} \cup C_i \cup \alpha(C_i)$ is a probe in $(U,\sim)$ for some $x,y \in U \setminus (C_i \cup \alpha(C_i))$.
Since the restriction of every conformal model of $(U,{\sim})$ to the set $P_i$ is conformal,
for every conformal model $\phi^{0}$ of $(U,{\sim})$ extending $\phi^0_S$
we must have $\phi^{0}|P_i \equiv \phi^0_i$ and hence $\phi^{0}|(P_i \setminus \{y\}) \equiv \phi^{0}_i|(P_i \setminus \{y\})$.
Assume that $\phi^{0}_i|(P_i \setminus y) \equiv^{R} x' \pi_i x'' \pi'_i$, where $\{x',x''\} = \{x^{0},x^{1}\}$ and $\pi_i, \pi'_i$ are chosen such that both the labeled letters $b^0_i,b^1_i$
appear in $\pi_i$.
Note that
$$\text{either } b_i^{0}a_i^{0}b_i^{1}, \text{ or } b_i^{1}a_i^{0}b_i^{0}, \text{ or } b_i^{0}a_i^{1}b_i^{1}, \text{ or } b_i^{1}a_i^{0}b_i^{0} \text{ is a subword of } \pi_i.$$
Having in mind Theorems \ref{thm:non_trivial_split_representations} and \ref{thm:trivial_split_representations}, we conclude that every conformal model $\phi^{0}$ of $(U,{\sim})$ extending $\phi^0_S$ must be of the form:
\begin{equation}
\label{eq:conformal_models_form}
\phi \equiv \left\{ 
\begin{array}{lll} 
\tau_{i_1} \ldots \tau_{i_k} \tau'_{i_1} \ldots \tau'_{i_k} & \text{if} & \text{$(A,B)$ is non-trivial,} \\ 
a^0\tau_{i_1} \ldots \tau_{i_k}a^1 \tau'_{i_k} \ldots \tau'_{i_1} & \text{if} & \text{$(A,B)$ is trivial,} 
\end{array} \right.
\tag{**}
\end{equation}
where $i_1, \ldots, i_k$ is a permutation of the set $[k]$ and
\begin{itemize}
 \item $\{\tau_{i_j} , \tau'_{i_j} \} = \{a^{0}_{i_j}, a^{1}_{i_j}\}$, for $i_j \in I_1$,
 \item $(\tau_{i_j},\tau'_{i_j}) = (\pi_{i_j},\pi'_{i_j})$ or $(\tau_{i_j},\tau'_{i_j}) = (\pi'_{i_j},\pi_{i_j})$, for $i_j \in I_2$.
\end{itemize}
It means, in particular, that for every conformal model $\phi^0_R$ of $(R,{\sim})$ extending $\phi^{0}_S$ there exists at most one conformal model $\phi^0$ of $(U,{\sim})$ that extends $\phi^{0}_R$.
We prove similarly that for every conformal model $\phi^1_R$ of $(R,{\sim})$ extending $\phi^1_S$ there
is at most one conformal model $\phi^1$ such that $\phi^1|R \equiv \phi^1_R$.

Now, we prove that there is a unique conformal model $\phi^0_R$ of $(R,{\sim})$ that extends $\phi^{0}_S$.
Suppose $(A,B)$ is non-trivial.
We claim that for every $i \in [2,k]$ there is a unique conformal model $\phi$ of $(R_i,{\sim})$ extending $\phi^{0}_S$.
To prove the claim for $i=2$ we need to show that
there is a unique extension of $\phi \equiv \phi^{0}_S$ by the chords $\phi(b_1)$ and $\phi(b_2)$.
The chord $\phi(b_1)$ must be placed such that it intersect one of the ends of $\phi(a_1)$, 
and which end is intersected can be decided based on whether $b_1 \in \leftside(a_2)$ or $b_1 \in \rightside(a_2)$.
The orientation of $\phi(b_1)$ can be decided based on whether $a_2 \in \leftside(b_1)$ or whether $a_2 \in \rightside(b_1)$. 
We show similarly that the placement and the orientation of $\phi(b_2)$ are uniquely determined.

Suppose $\phi$ is a unique conformal model of $(R_{i-1}, {\sim})$ extending $\phi^0_S$.
We show that there is a unique extension of $\phi$ by the chords $\phi(a_i)$ and $\phi(b_i)$.
Note that $(\{a_1,\ldots,a_{i-1}\}, {\sim})$ is a clique in $(R_k,{\sim})$, and
hence the chords $\{\phi(a_1),\ldots,\phi(a_{i-1})\}$ are pairwise intersecting.
There are $(i-1)$ possible placements for the non-oriented chord $\phi(a_i)$.
Every such a placement determines uniquely the partition of $\{b_1,\ldots,b_{i-1}\}$ into two sets $A$ and $B$ such that the chords from $\phi(A)$ are on one side of the non-oriented chord $\phi(a_i)$
and the chords from $\phi(B)$ are on the opposite side of $\phi(a_i)$.
Note that the partitions $\{A,B\}$ corresponding to different placements of $\phi(a_i)$ are different:
shifting the chord $\phi(a_i)$ by one chord $\phi(a_j)$ moves $b_j$ either from $A$ to $B$ or from $B$ to $A$.
Hence, to keep $\phi$ conformal, only one placement for $\phi(a_i)$ can be compatible with the partition
$$\{\{b_1,\ldots,b_{i-1}\} \cap \leftside(a_i), \{b_1,\ldots,b_{i-1}\} \cap \rightside(a_i)\}.$$
The orientation of $\phi(a_i)$ can be decided based on whether $b_1 \in \leftside(a_i)$ or whether $b_1 \in \rightside(a_i)$.
With a similar ideas to those used earlier, we show that there is a unique extension of $\phi$ by the chord $\phi(b_i)$ if we want to keep $\phi$ conformal.

The case when the split $(A,B)$ is trivial is handled with similar ideas.
This completes the proof of the lemma for the case when condition \eqref{eq:probe_condition} is satisfied.

Now consider the case when condition \eqref{eq:probe_condition} is not satisfied.
This happens when $k=2$ and $|C_2 \cup \alpha(C_2)|=1$.
That is, we have $C_2 = \{a_2\}$ and $\alpha(C_2) = \emptyset$.
If $|C_1 \cup \alpha(C_1)| =2$, then $C_1 = \{a_1\}$, $\alpha(C_1) = \{b_1\}$,
and $U = \{a_1,b_1,a,a_2\}$.
In this case $(U,{\sim})$ induces $P_4$ in $G_{ov}$, 
which has a unique, up to reflection, conformal model.
So, in the remaining of the proof we assume $|C_1 \cup \alpha(C_1)| \geq 3$. 
We consider two cases depending on whether or not $(\{a\} \cup C_1 \cup \alpha(C_1),{\sim})$ is prime.

Suppose $(\{a\} \cup C_1 \cup \alpha(C_1),{\sim})$ is prime. 
By the inductive hypothesis, $(\{a\} \cup C_1 \cup \alpha(C_1),{\sim})$ has two conformal models, 
$\phi$ and its reflection $\phi^R$.
Suppose that $\phi \equiv a^0 \pi a^{1} \pi'$ for some $\pi,\pi'$.
By Theorem~\ref{thm:trivial_split_representations}, there are two extensions of $\phi$
by the chord $\phi(a_2)$ that lead to a chord model of $(U,{\sim})$:
$\phi^1 \equiv a^0 a_2 \pi a^1 \pi' a_2$ or $\phi^2 \equiv a^0 \pi a_2 a^1 a_2 \pi'$.
Depending on whether $a_2 \in \leftside(a_1)$ or $a_2 \in \rightside(a_1)$, 
only one among them can be extended to a conformal model of $(U,{\sim})$.
The orientation of $\phi(a_2)$ can be decided based on whether $a_1 \in \leftside(a_2)$ or whether $a_1 \in \rightside(a_2)$.

Suppose $(\{a\} \cup C_1 \cup \alpha(C_1), {\sim})$ has a non-trivial module~$M$.
Observe that $a \in M$.
Otherwise, we would have either $M \subseteq C_1$ or $M \subseteq \alpha(C_1)$.
In both these cases, $M$ would be also a trivial module of $(U,{\sim})$, a contradiction.
For the remaining part of the proof, let $M_1 = M \cap C_1$ and 
$M_2 = C_1 \setminus M$.
We claim that $M_1 \neq \emptyset$ and $M_2 \neq \emptyset$.
Suppose that $M_1 = \emptyset$.
Then, $M \cap \alpha(C_1) \neq \emptyset$ as $M$ is a non-trivial module in $(\{a\} \cup C_1 \cup \alpha(C_1),{\sim})$.
Moreover, $(M \cap \alpha(C_1)) \sim C_1$ as
otherwise there would be a vertex in $C_1$ from outside $M$ adjacent to $a$ in $M$ 
and non-adjacent to a vertex in $(M \cap \alpha(C_1))$.
Furthermore, since $\alpha(C_1) \parallel a$,
we must have $(\alpha(C_1) \setminus M) \parallel (M \cap \alpha(C_1))$.
Hence, the split $(A,B)$ is not maximal as then the set $A=\{a\}$ could be extended by $\alpha(C_1) \cap M$, 
which is not the case.
This proves $M_1 \neq \emptyset$.
Now, we prove $M_2 \neq \emptyset$.
Assuming otherwise, since $(\{a\} \cup C_1) \subset M$ and $a \parallel \alpha(C_1)$, one can show $M = \{a\} \cup C_1 \cup \alpha(C_1)$ by connectivity of $(U,{\sim})$.
Hence, $M$ would be trivial in $(\{a\} \cup C_1 \cup \alpha(C_1), {\sim})$, which is not the case.
This proves $M_2 \neq \emptyset$.

We partition the vertices of $\alpha(C_1)$ into to sets: $\alpha(M_1)$ and $\alpha(M_2)$.
Let $(D,{\sim})$ be a connected component of $(\alpha(C_1),{\sim})$.
We have $D \subset \alpha(M_1)$ if there is an edge between some vertex in $D$ and some vertex in $M_1$;
otherwise we have $D \subset \alpha(M_2)$.
In particular, for every component $D \subset \alpha(M_2)$ we have $D \parallel M_1$ and 
some vertex from $D$ is adjacent to a vertex in $M_2$ as $(U,{\sim})$ is connected.
Observe that $D \subset M$ if $D \subset \alpha(M_1)$ and 
that $M_2 \sim (M_1 \cup \alpha(M_1))$ as otherwise
there would be a vertex in $M_2$ from outside $M$ adjacent to $a$ in $M$ and not adjacent to some vertex in $M$.
Summing up, we have that:
\begin{itemize}
 \item $M_1 \neq \emptyset$, $M_2 \neq \emptyset$, and ($\alpha(M_1) \neq \emptyset$ or $\alpha(M_2) \neq \emptyset$),
 \item $M_2 \sim (M_1 \cup \alpha(M_1))$,
 \item $\alpha(M_2) \parallel (M_1 \cup \alpha(M_1))$,
\end{itemize}
Note that not necessarily $(\alpha(M_1),\alpha(M_2)) = (\alpha(C_1) \cap M,\alpha(C_1) \setminus M)$.
It might happen that there is a unique singleton component $(\{d\},{\sim})$ in $(\alpha(C_1),{\sim})$ such that 
$d \in \alpha(M_2)$ and $d \in M$ 
(in this case $d \parallel M_1$ and $d \sim M_2$).
For any component $D$ different than $\{d\}$, $D \subset \alpha(M_1)$ iff $D \subset M$.
Having in mind the above properties, note that:
\begin{itemize}
 \item $(a_2,a,M_1,\alpha(M_1))$ is a probe in $(U,{\sim})$ if $\alpha(M_1) \neq \emptyset$,
 \item $(a_2,a,M_2,\alpha(M_2))$ is a probe in $(U,{\sim})$ if $\alpha(M_2) \neq \emptyset$.
\end{itemize}
If $\alpha(M_i) = \emptyset$ then $|M_i|=1$ as otherwise $M_i$ would be a non-trivial module in $(U,{\sim})$.

Let $S = \{a,a_2\}$. 
We show that there is a unique conformal model $\phi$ of $U$ that extends 
$\phi_S \equiv a^0a^1_2a^1a^0_2$.
Suppose $\alpha(M_1) \neq \emptyset$.
Let $P_i = \{a_2,a\} \cup M_i \cup \alpha(M_i)$ for $i \in [2]$.
By Claim \ref{claim:probes_have_unique_model}, $(P_1,{\sim})$ has a unique conformal model $\phi_1$ extending $\phi_S$, which is of the form either $\phi_1 \equiv a^0 a^{1}_2 \pi'_1 a^1 \pi'' a^0_2$ or $\phi_1 \equiv a^0 \pi'_1 a^1_2 a^1 a^0_2 \pi''$.
Suppose the first case, that is, 
$$\phi_1 \equiv a^0 a^{1}_2 \pi'_1 a^1 \pi'' a^0_2.$$
Since any conformal model $\phi$ of $(U,{\sim})$ extending $\phi_S$ must also extend $\phi_1$, 
$\phi$ must be of the form:
\begin{equation}
\label{eq:general_form_of_phi_U_prime}
\phi \equiv a^0a^1_2\tau'_{\phi} a^1 \tau''_{\phi} a^0_2,
\end{equation}
where for every $u \in C_1$ both words $\tau'_\phi$  and $\tau''_{\phi}$
contain exactly one labeled letter of $u$ and for every $u \in \alpha(C_1)$ both
the labeled letters of $u$ are either in $\tau'_{\phi}$ or $\tau''_{\phi}$.
Hence, $(P_2,{\sim})$ has a unique conformal model extending $\phi_S$ of the form
$$\phi_2 \equiv a^0a^1_2\pi'_2 a^1 \pi''a^0_2,$$
which follows from 
\eqref{eq:general_form_of_phi_U_prime}, from Claim \ref{claim:probes_have_unique_model} if $\alpha(M_2) \neq \emptyset$, and from $|M_2|=1$ if $\alpha(M_2) = \emptyset$.
\begin{claim}
\label{claim:properties_of_conformal_models_of_U}
Let $\phi$ be a conformal model of $(U,{\sim})$ extending $\phi_S$.
Then:
\begin{enumerate}
\item $\pi'_1$ and $\pi'_2$ are subwords of $\tau'_{\phi}$ and $|\pi'_1| + |\pi'_2| = |\tau'_{\phi}|$,
\item $\pi''_1$ and $\pi''_2$ are subwords of $\tau''_{\phi}$ and $|\pi''_1| + |\pi''_2| = |\tau''_{\phi}|$.
\item \label{item:prop_v_separates_u_and_a} for every $u \in \alpha(M_2)$ and every $v \in \alpha(M_1)$, either
 $\phi(u)$ and $\phi(v)$ are on the opposite site of $\phi(a)$, or
 there are on the same side of $\phi(a)$ and then 
 the chord $\phi(v)$ has the chords $\phi(u)$ and $\phi(a)$ on the opposite sides.
\end{enumerate}
\end{claim}
\begin{proof}
The first two statements are obvious.
Suppose $\phi(u)$ and $\phi(v)$ are on the same side of $\phi(a)$, 
but the chord $\phi(u)$ has the chords $\phi(v)$ and $\phi(a)$ on the opposite side.
Then, $\phi(u)$ must intersect some chord from $\phi(P_1)$ as $(P_1,{\sim})$ is connected.
However, this is not possible as $u \parallel P_1$.
\end{proof}
Our goal is to show that there is unique way to compose the words $\pi'_1$ and $\pi'_2$ and the words
$\pi''_1$ and $\pi''_2$ to get a conformal model of $(U,{\sim})$.

Suppose that there are two non-equivalent models $\phi_1$ and $\phi_2$ extending $\phi_S$.
We say that $x \in M_1 \cup \alpha(M_1)$ and $y \in M_2 \cup \alpha(M_2)$ are \emph{mixed in $\tau'_{\phi_1}$ and $\tau'_{\phi_2}$}
if $\tau'_{\phi_1}|\{x',y'\} \neq \tau'_{\phi_2}|\{x',y'\}$ for some $x' \in \{x^{0},x^{1}\}$ and $y' \in \{y^0,y^1\}$.
That is, $x'$ and $y'$ are mixed in $\tau'_{\phi_1}$ and $\tau'_{\phi_2}$ if they occur in $\tau'_{\phi_1}$ and $\tau'_{\phi_2}$ in different order.
We introduce the notion of being mixed in $\tau''_{\phi_1}$ and $\tau''_{\phi_2}$ similarly.
Clearly, if $\phi_1$ and $\phi_2$ are non-equivalent,
there are vertices $x \in M_1 \cup \alpha(M_1)$ and $y \in M_2 \cup \alpha(M_2)$ such that $x$ and $y$ are mixed either 
in $\tau'_{\phi_1}$ and $\tau'_{\phi_2}$ or in $\tau''_{\phi_1}$ and $\tau''_{\phi_2}$.
Suppose $x \in M_1 \cup \alpha(M_1)$ and $y \in M_2 \cup \alpha(M_2)$ are mixed in $\tau'_{\phi_1}$ and $\tau'_{\phi_2}$.
We claim that $x \in M_1$ and $y \in M_2$.
We can not have $x \in M_1$ and $y \in \alpha(M_2)$ 
as in any conformal model $\phi$ of $(U,{\sim})$ of the form \eqref{eq:general_form_of_phi_U_prime} the chord
$\phi(y)$ is always on the right side of $\phi(x)$ if $y \in \rightside(x)$
or is always on the left side of $\phi(x)$ if $y \in \leftside(x)$.
We can not have $x \in \alpha(M_1)$ and $y \in M_2$ as 
in every conformal model of $(U,{\sim})$ the chord $\phi(x)$ intersects $\phi(y)$.
Finally, by Claim \ref{claim:properties_of_conformal_models_of_U}.\eqref{item:prop_v_separates_u_and_a} we can not have $x \in \alpha(M_1)$ and $y \in \alpha(M_2)$.
It means that $x \in M_1$ and $y \in M_2$ and hence $x \sim y$.
However, it also means that $x$ and $y$ are mixed in $\tau''_{\phi_1}$ and $\tau''_{\phi_2}$.
Hence, from now we abbreviate and we say that $x$ and $y$ \emph{are mixed} if $x$ and $y$ are mixed in $\tau'_{\phi_1}$ and $\tau'_{\phi_2}$ and in $\tau''_{\phi_1}$ and $\tau''_{\phi_2}$.
Now, we claim that for every $x \in M_1$ and every $y \in M_2$:
\begin{equation}
\label{eqref:mixed_properties}
\{x,y\} \sim \alpha(M_1) \text{ and } \{x, y\} \parallel \alpha(M_2) \text{ if $x$ and $y$ are mixed.}
\end{equation}
We prove $\{x,y\} \sim \alpha(M_1)$. 
Clearly, $y \sim \alpha(M_1)$ and $x \parallel \alpha(M_2)$ by the properties of $P_1$ and $P_2$.
Suppose there is $v \in \alpha(M_2)$ such that $v \parallel x$.
If this is the case, the relative position of $\phi(x)$ and $\phi(v)$ is the same
in any conformal model $\phi$ of $(U,{\sim})$ of the form \eqref{eq:general_form_of_phi_U_prime}.
Since $\phi_i(y)$ intersects $\phi_i(v)$ for every $i \in [2]$, $x$ and $y$ can not be mixed.
The second statement of \eqref{eqref:mixed_properties} is proved similarly.
Note that $(C_1,\sim)$ is a permutation subgraph of $(U,{\sim})$ as $a \sim C_1$.
Hence, if $x \in M_1$ is mixed with $y \in M_2$ and $x \parallel z$ for some $z \in M_1$, 
then $z$ is also mixed with $y$.
Similarly, if $x \in M_1$ is mixed with $y \in M_2$ and $t \parallel y$ for some $t \in M_2$, 
then $t$ is mixed with $x$.
Now, let
$$W = \bigcup \{ \{z,t\}: \text{ $z$ and $t$ are mixed}\},$$
that is, $W$ contains all the elements in $C_1$ that are mixed with some other element in $C_1$.
Note that $W$ contains at least two elements as there are at least two elements that are mixed.  
Moreover, $W \subset C_1 \subsetneq U$.
Note that $W \sim (C_1 \setminus W)$ from the observation given above.
Now, by the properties from \eqref{eqref:mixed_properties} we conclude that $W$ is a non-trivial module in $(U,{\sim})$, 
which is a contradiction.

Consider the remaining case $\alpha(M_2) \neq \emptyset$ and $\alpha(M_1) = \emptyset$.
In this setting $|M_1| = 1$.
Suppose $M_1 = \{x\}$.
We show that there is a unique extension of $\phi_2 \equiv a^0_2a^1\pi'_2 a_2^1 \pi''_2a^0$
by a chord $\phi(x)$ to a conformal model of $(U,{\sim})$.
We introduce the mixing relation analogously to the previous case.
Using similar ideas as previously we prove that $u$ and $x$ can not be mixed if $u \in \alpha(C_2)$. 
So, $x$ can be mixed only with the elements in $M_2$.
Let
$$W = \bigcup \{ \{x,y\}: \text{ $x$ and $y$ are mixed}\}.$$
Clearly, if there are two non-equivalent conformal models $\phi_1$ and $\phi_2$ of $(U,{\sim})$, 
then $x$ is mixing with some element $y \in M_2$.
Then $\{x,y\} \subset W \subset C_1$.
Using similar ideas as previously, one can prove that $W$ is a non-trivial module in $(U,{\sim})$.
\end{proof}
Suppose $\phi^{0}_U$ and $\phi^{1}_U$ are two conformal models of $(U,{\sim})$ described by Lemma \ref{lemma:two_models_of_a_prime_graph}.
Inspired by Hsu \cite{Hsu95}, we define a \emph{consistent decomposition} of $M$,
the usefulness of which is highlighted in the upcoming Lemma \ref{lemma:consistent_decompositions}.
For every $i \in [k]$ we introduce an equivalence relation $K$ in the set $M_i$.
Depending on the type of $M_i$, the relation $K$ is defined as follows:
\begin{itemize}
\item If $M_i$ is prime, then $v K v'$ for every $v,v' \in M_i$.
\item If $M_i$ is parallel, then
$$
v K v'  \iff 
\begin{array}{c}
\vspace{1pt}
\text{either $\{v,v'\} \subseteq \leftside(u)$ or $\{v,v'\} \subseteq \rightside(u)$,} \\
\text{for every $u \in U \setminus M_i$.}
\end{array}
$$
\item If $M_i$ is serial, then 
$$
v K v'  \iff 
\begin{array}{c}
\vspace{1pt}
\{ \leftside(v) \cap (U \setminus M_i), \rightside(v) \cap (U \setminus M_i) \} = \\
\{ \leftside(v') \cap (U \setminus M_i), \rightside(v') \cap (U \setminus M_i)\}.
\end{array}
$$
\end{itemize}

Suppose $K(M_i)$ is a set of equivalence classes of $K$-relation in the module $M_i$.
The set $K(M_i)$ is called the \emph{consistent decomposition of $M_i$}
and the set $K(M) = \bigcup_{i=1}^{k} K(M_i)$ is called the \emph{consistent decomposition of $M$}.
The elements of $K(M_i)$ and $K(M)$ are called the \emph{consistent submodules} of $M_i$ and $M$, respectively.
See Figure \ref{fig:consistent_decomposition} for an illustration.

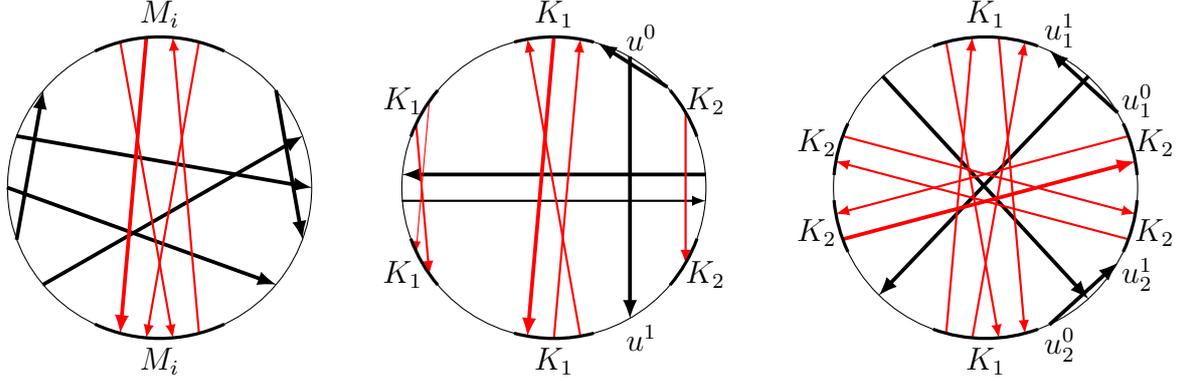
\begin{figure}[htp!]
\centering
\begin{tikzpicture}[scale=1,>=latex,shorten >=-0.2pt,shorten <=-0.2pt]
\coordinate (center) at (0,0) {};
\coordinate (ltau) at ($(center)+(90:2.3cm)$) {};
\coordinate (ltau') at ($(center)+(270:2.3cm)$) {};

\coordinate (u4) at ($(center)+(75:2cm)$) {};
\coordinate (u3) at ($(center)+(85:2cm)$) {};
\coordinate (u2) at ($(center)+(95:2cm)$) {};
\coordinate (u1) at ($(center)+(105:2cm)$) {};

\coordinate (b4) at ($(center)+(255:2cm)$) {};
\coordinate (b3) at ($(center)+(265:2cm)$) {};
\coordinate (b2) at ($(center)+(275:2cm)$) {};
\coordinate (b1) at ($(center)+(285:2cm)$) {};

\coordinate (l5) at ($(center)+(140:2cm)$) {};
\coordinate (l4) at ($(center)+(160:2cm)$) {};
\coordinate (l3) at ($(center)+(180:2cm)$) {};
\coordinate (l2) at ($(center)+(200:2cm)$) {};
\coordinate (l1) at ($(center)+(220:2cm)$) {};

\coordinate (r5) at ($(center)+(-40:2cm)$) {};
\coordinate (r4) at ($(center)+(-20:2cm)$) {};
\coordinate (r3) at ($(center)+(0:2cm)$) {};
\coordinate (r2) at ($(center)+(20:2cm)$) {};
\coordinate (r1) at ($(center)+(40:2cm)$) {};

\draw ($(center)$) circle (2cm);

\draw[line width=0.5mm,->] (l2)-- (l5);
\draw[line width=0.5mm,->] (r1)-- (r4);

\draw[line width=0.5mm,->] (l1)-- (r2);
\draw[line width=0.5mm,->] (l3)-- (r5);
\draw[line width=0.5mm,->] (l4)-- (r3);

\draw[thick,red,->] (u1)-- (b2);
\draw[line width=0.5mm,red,->] (u2)-- (b4);
\draw[thick,red,<-] (u3)-- (b1);
\draw[thick,red,->] (u4)-- (b3);

\draw[very thick] ([shift=(245:2cm)]0,0) arc (245:295:2cm);
\draw[very thick] ([shift=(65:2cm)]0,0) arc (65:115:2cm);

\tikzstyle{every node}=[inner sep=1pt]
\node at (ltau) {$M_i$};
\node at (ltau') {$M_i$};

\end{tikzpicture} 
\hspace{0.6cm}
\begin{tikzpicture}[scale=1,>=latex,shorten >=-0.2pt,shorten <=-0.2pt]
\coordinate (center) at (0,0) {};

\coordinate (u0) at ($(center)+(60:2.3cm)$) {};
\coordinate (u1) at ($(center)+(300:2.3cm)$) {};

\coordinate (ltau1) at ($(center)+(150:2.3cm)$) {};
\coordinate (ltau'1) at ($(center)+(210:2.3cm)$) {};
\coordinate (ltau2) at ($(center)+(90:2.3cm)$) {};
\coordinate (ltau'2) at ($(center)+(270:2.3cm)$) {};
\coordinate (ltau3) at ($(center)+(30:2.3cm)$) {};
\coordinate (ltau'3) at ($(center)+(330:2.3cm)$) {};

\coordinate (tau1u1) at ($(center)+(145:2cm)$) {};
\coordinate (tau1u2) at ($(center)+(155:2cm)$) {};

\coordinate (tau1b1) at ($(center)+(205:2cm)$) {};
\coordinate (tau1b2) at ($(center)+(215:2cm)$) {};

\coordinate (tau2u1) at ($(center)+(80:2cm)$) {};
\coordinate (tau2u2) at ($(center)+(90:2cm)$) {};
\coordinate (tau2u3) at ($(center)+(100:2cm)$) {};

\coordinate (tau2b1) at ($(center)+(260:2cm)$) {};
\coordinate (tau2b2) at ($(center)+(270:2cm)$) {};
\coordinate (tau2b3) at ($(center)+(280:2cm)$) {};

\coordinate (tau3u1) at ($(center)+(30:2cm)$) {};
\coordinate (tau3b1) at ($(center)+(330:2cm)$) {};

\draw ($(center)$) circle (2cm);

\draw[line width=0.5mm,->] ($(center)+(5:2cm)$)--($(center)+(175:2cm)$);
\draw[black, thick,<-] ($(center)+(-5:2cm)$)--($(center)+(185:2cm)$);
\draw[line width=0.5mm,->] ($(center)+(60:2cm)$)--($(center)+(300:2cm)$);
\draw[line width=0.5mm,->] ($(center)+(42:2cm)$)--($(center)+(73:2cm)$);

\draw[red,->] (tau1u1)-- (tau1b1);
\draw[red,thick,->] (tau1u2)-- (tau1b2);

\draw[red,thick,<-] (tau2u1)-- (tau2b2);
\draw[line width=0.5mm,red,->] (tau2u2)-- (tau2b1);
\draw[red,thick,<-] (tau2u3)-- (tau2b3);

\draw[red,thick,->] (tau3u1)-- (tau3b1);

\draw[very thick] ([shift=(20:2cm)]0,0) arc (20:40:2cm);
\draw[very thick] ([shift=(75:2cm)]0,0) arc (75:105:2cm);
\draw[very thick] ([shift=(140:2cm)]0,0) arc (140:160:2cm);

\draw[very thick] ([shift=(200:2cm)]0,0) arc (200:220:2cm);
\draw[very thick] ([shift=(255:2cm)]0,0) arc (255:285:2cm);
\draw[very thick] ([shift=(320:2cm)]0,0) arc (320:340:2cm);

\tikzstyle{every node}=[inner sep=1pt]
\node at (ltau1) {$K_1$};
\node at (ltau'1) {$K_1$};
\node at (ltau2) {$K_1$};
\node at (ltau'2) {$K_1$};
\node at (ltau3) {$K_2$};
\node at (ltau'3) {$K_2$};
\node at (u0) {$u^0$};
\node at (u1) {$u^1$};

\end{tikzpicture} 
\hspace{0.6cm}
\begin{tikzpicture}[scale=1,>=latex,shorten >=-0.2pt,shorten <=-0.2pt]
\coordinate (center) at (0,0) {};
\coordinate (u00) at ($(center)+(30:2.3cm)$) {};
\coordinate (u10) at ($(center)+(65:2.3cm)$) {};
\coordinate (u01) at ($(center)+(295:2.3cm)$) {};
\coordinate (u11) at ($(center)+(330:2.3cm)$) {};

\coordinate (ltau1) at ($(center)+(165:2.3cm)$) {};
\coordinate (ltau'1) at ($(center)+(345:2.3cm)$) {};
\coordinate (ltau3) at ($(center)+(195:2.3cm)$) {};
\coordinate (ltau'3) at ($(center)+(15:2.3cm)$) {};
\coordinate (ltau2) at ($(center)+(90:2.3cm)$) {};
\coordinate (ltau'2) at ($(center)+(270:2.3cm)$) {};

\coordinate (u4) at ($(center)+(75:2cm)$) {};
\coordinate (u3) at ($(center)+(85:2cm)$) {};
\coordinate (u2) at ($(center)+(95:2cm)$) {};
\coordinate (u1) at ($(center)+(105:2cm)$) {};

\coordinate (b4) at ($(center)+(255:2cm)$) {};
\coordinate (b3) at ($(center)+(265:2cm)$) {};
\coordinate (b2) at ($(center)+(275:2cm)$) {};
\coordinate (b1) at ($(center)+(285:2cm)$) {};

\coordinate (tau1u1) at ($(center)+(160:2cm)$) {};
\coordinate (tau1u2) at ($(center)+(170:2cm)$) {};
\coordinate (tau1b1) at ($(center)+(340:2cm)$) {};
\coordinate (tau1b2) at ($(center)+(350:2cm)$) {};

\coordinate (tau3u1) at ($(center)+(10:2cm)$) {};
\coordinate (tau3u2) at ($(center)+(20:2cm)$) {};
\coordinate (tau3b1) at ($(center)+(190:2cm)$) {};
\coordinate (tau3b2) at ($(center)+(200:2cm)$) {};

\draw ($(center)$) circle (2cm);

\draw[line width=0.5mm,->] ($(center)+(47.5:2cm)$)-- ($(center)+(225.5:2cm)$);
\draw[line width=0.5mm,->] ($(center)+(30:2cm)$)-- ($(center)+(65:2cm)$);

\draw[line width=0.5mm,->] ($(center)+(132.5:2cm)$)-- ($(center)+(312.5:2cm)$);
\draw[line width=0.5mm,->] ($(center)+(295:2cm)$)-- ($(center)+(330:2cm)$);

\draw[thick,red,->] (tau1u1)-- (tau1b2);
\draw[thick,red,<-] (tau1u2)-- (tau1b1);

\draw[line width=0.5mm,red,<-] (tau3u1)-- (tau3b2);
\draw[thick,red,->] (tau3u2)-- (tau3b1);

\draw[thick,red,->] (u1)-- (b2);
\draw[thick,red,<-] (u2)-- (b4);
\draw[thick,red,->] (u3)-- (b1);
\draw[thick,red,<-] (u4)-- (b3);

\draw[very thick] ([shift=(5:2cm)]0,0) arc (5:25:2cm);
\draw[very thick] ([shift=(70:2cm)]0,0) arc (70:110:2cm);
\draw[very thick] ([shift=(155:2cm)]0,0) arc (155:175:2cm);

\draw[very thick] ([shift=(185:2cm)]0,0) arc (185:205:2cm);
\draw[very thick] ([shift=(250:2cm)]0,0) arc (250:290:2cm);
\draw[very thick] ([shift=(335:2cm)]0,0) arc (335:355:2cm);

\tikzstyle{every node}=[inner sep=1pt]
\node at (ltau1) {$K_2$};
\node at (ltau'1) {$K_2$};
\node at (ltau2) {$K_1$};
\node at (ltau'2) {$K_1$};
\node at (ltau3) {$K_2$};
\node at (ltau'3) {$K_2$};

\node at (u00) {$u^0_1$};
\node at (u10) {$u^1_1$};
\node at (u01) {$u^0_2$};
\node at (u11) {$u^1_2$};

\end{tikzpicture} 

\caption{\label{fig:consistent_decomposition} Proper module $M_i$ and its consistent submodules.
Chords associated with the module $M_i$ are in red, chords associated with $U$ are bolded.
From left to right: proper prime $M_i$ has one consistent submodule $M_i$, 
proper parallel $M_i$ has two consistent submodules $K_1$ and $K_2$ ($u$
has the vertices from $K_1$ and $K_2$ on different sides, which proves that the elements from $K_1$ 
and the elements from $K_2$ are not in $K$-relation),
proper serial $M_i$ has two consistent submodules $K_1$ and $K_2$ (the vertices from $K_1$
have $u_1,u_2$ on the same size, the vertices from $K_2$ have $\{u_1,u_2\}$ on different sides,
which proves that the elements from $K_1$ are not in $K$-relation with the elements from $K_2$).}
\end{figure}

\begin{claim}
\label{claim:consistent_modules_structure}
Suppose $M_i$ is a child of $M$ in $\mathcal{M}(G_{ov})$.
Then:
\begin{enumerate}
 \item \label{item:consistent_modules_structure_prime} if $M_i$ is prime, then $K(M_i) = \{M_i\}$,
 \item \label{item:consistent_modules_structure_serial_parallel} if $M_i$ is serial or parallel, then every consistent submodule of $M_i$ is the union of some children of $M_i$ in $\mathcal{M}(G_{ov})$.
\end{enumerate}
In particular, every consistent submodule of $M_i$ is a submodule of $M_i$.
\end{claim}
\begin{proof}
Note that $M_i$ is proper as $M_i$ is a child of a prime module $M$.
Also, recall that $M$ is the connected component of $G_{ov}$ containing $M_i$.
Now, statement \eqref{item:consistent_modules_structure_prime} follows from Lemma \ref{lemma:circle_models_of_a_proper_prime_module} applied to the proper prime module $M_i$ and the connected component $M$ containing $M_i$.

Statement \eqref{item:consistent_modules_structure_serial_parallel} follows from Lemma \ref{lemma:circle_models_of_a_parallel_module} (Lemma \ref{lemma:circle_models_of_a_serial_module})
applied to the proper parallel (proper serial, respectively) module $M_i$ and the connected component $M$ of $G_{ov}$ containing~$M_i$.
\end{proof}

Suppose $K_1, \ldots, K_n$ is a consistent decomposition of $M$.
Recall that the set $U$ contains one element from every child of $M$.
So, $|K_i \cap U| \leq 1$ for every $i \in [n]$.
Fix a set $S = \{s_1,\ldots,s_n\} \subset M$ such that $U \subset S$ and $|K_i \cap S| = \{s_i\}$ for every $i \in [n]$. 
Call the set $S$ a \emph{skeleton} of $M$.
Thus, the skeleton of $M$ is a superset of $U$ containing exactly one element $s_i$ from every consistent submodule $K_i$ of $M$.
\begin{lemma}
\label{lemma:consistent_decompositions}
Suppose $K_1,\ldots,K_n$ is a consistent decomposition of $M$ 
and $S = \{s_1,\ldots,s_n\}$ is a skeleton of $M$.
Then:
\begin{enumerate}
\item\label{item:consistent_decompositions_conformal_models_skeleton}
 The graph $(S, {\sim})$ has exactly two conformal models, $\phi^{0}_S$ and $\phi^{1}_S$, one being the reflection of the other.
\item \label{item:consistent_decompositions_and_conformal_models}
For every conformal model $\phi$ of $(M,{\sim})$ and every $i \in [n]$, 
the words in $\phi|K_i$ form two contiguous subwords in the circular word $\phi$.
\end{enumerate}
\end{lemma}
\begin{proof}
Without loss of generality we assume that $U = \{s_1,\ldots,s_{|U|}\}$.
Let $m \in \{0,1\}$.
By Lemma \ref{lemma:two_models_of_a_prime_graph}, $(U,{\sim})$ has two conformal models
$\phi^0_U$ and its reflection $\phi^1_U$.
Our proof is based on the following claim.
\begin{equation}
\label{eq:extending_conformal_models_of_U_on_S}
\begin{array}{c}
\text{For every $j = \{|U|,\ldots,n\}$ there is a unique conformal model $\phi^{m}_{j}$}\\
\text{of $(\{s_1,\ldots,s_j\},{\sim})$ such that $\phi^{m}_{j}|U = \phi^{m}_U$.}
\end{array}
\end{equation}
Clearly, statement \eqref{item:consistent_decompositions_conformal_models_skeleton} follows
from the claim for $j = n$.

We prove our claim by induction on $j$.
For $j=|U|$ the claim is trivially satisfied.
Suppose \eqref{eq:extending_conformal_models_of_U_on_S} holds for $j = l-1$ for some $l > |U|$.
Our goal is to prove \eqref{eq:extending_conformal_models_of_U_on_S} for $j=l$.
From the inductive hypothesis, there is a unique extension $\phi^m_{l-1}$ of $\phi^{m}_U$ on the set $\{s_1,\ldots,s_{l-1}\}$.
Suppose for a contradiction that there are two non-equivalent conformal models of $(\{s_1,\ldots,s_{l}\},{\sim})$ extending $\phi^{m}_{l-1}$.
That is, suppose there are two different placements for the chord of $s_l$ in the model $\phi^m_{l-1}$
that lead to two non-equivalent conformal models of $(\{s_1,\ldots,s_l\},{\sim})$.
Equivalently, there is a circular word $\phi$ extending $\phi^{m}_{l-1}$ by the letters $x^{0},x^{1}, y^{0}, y^{1}$ such that $\phi' \equiv \phi|\{s_1,\ldots,s_{l-1},x\}$ and $\phi'' \equiv \phi|\{s_1,\ldots,s_{l-1},y\}$ are non-equivalent conformal models of $(\{s_1,\ldots,s_l\},{\sim})$
if we replace $x^0$ by $s^0_l$ and $x^1$ by $s^1_l$ in $\phi'$ and $y^0$ by $s^0_l$ and $y^1$ by $s^1_l$ in $\phi''$.
Note that for every $s \in \{s_1,\ldots,s_{l-1}\}$ the circular word $\phi$ satisfies the following properties: 
\begin{itemize}
\item if $s \in \leftside(s_l)$, then $\phi(s)$ must be on the left side of $\phi(x)$ and $\phi(y)$.
\item if $s \in \rightside(s_l)$, then $\phi(s)$ must be on the right side of $\phi(x)$ and $\phi(y)$.
\item if $s \sim s_l$, then $\phi(s)$ must intersect both $\phi(x)$ and $\phi(y)$.
\end{itemize}

We consider two cases depending on whether the chords $\phi(x)$ and $\phi(y)$ intersect in~$\phi$.

Suppose that $\phi(x)$ and $\phi(y)$ do not intersect.
Suppose $\phi(y)$ is on the right side of $\phi(x)$ and $\phi(x)$ is on the left side of $\phi(y)$.
Since $\phi'$ and $\phi''$ are non-equivalent, there is $s^{*}$ in $\{s_1,\ldots,s_{l-1}\}^{*}$
such that $\phi(s^{*})$ is on the right side of $\phi(x)$ and on the left side of $\phi(y)$.
Since the chord $\phi(s)$ can not intersect both $\phi(x)$ and $\phi(y)$, we have $s \parallel s_l$.
But then $\phi(s)$ is on the right side of $\phi(x)$ and on the left side of $\phi(y)$, 
which contradicts one of the properties of $\phi$ listed above.

Suppose that $\phi(y)$ is on the right side of $\phi(x)$ and $\phi(x)$ is on the right side of $\phi(y)$.
Let $s \in \{s_1,\ldots,s_{l-1}\}$ be such that $s \parallel s_l$. 
The chord $\phi(s)$ must lie on the right side of $\phi(x)$ and the right side of $\phi(y)$ as
any other placement of $\phi(s)$ will contradict one of the properties of $\phi$.
Furthermore, $\phi(s)$ can not have $\phi(x)$ and $\phi(y)$ on its different sides.
Hence, $\phi(s)$ has both its ends either between $\phi(x^{1})$ and $\phi(y^{0})$ or
between $\phi(y^{1})$ and $\phi(x^{0})$.
Note that $s$ and $s_l$ belong to different children of $M$.
Otherwise, supposing that $s,s_l \in M_i$ for some child module $M_i$ of $M$, 
the chord $\phi(u)$ for any $u$ such that $u \in U \setminus M_i$ and $u \sim M_i$ can not intersect $\phi(s)$, $\phi(x)$, and $\phi(y)$ at the same time. 
Thus, there is a path $P$ in the graph $(\{s_1,\ldots,s_l\},{\sim})$ joining $s$ and $s_l$ with all the inner vertices in $U$.
Then, there must be an inner vertex $s'$ in $P$ such that $\phi(s')$ has $\phi(x)$ and $\phi(y)$ on its different sides. 
This contradicts the properties of $\phi$.
The remaining cases are proved similarly.

Suppose that $\phi(x)$ and $\phi(y)$ intersect.
Without loss of generality we assume that $\phi|\{x,y\} \equiv x^{0}y^{0}x^{1}y^{1}$.
First, note that for every $s \in \{s_1,\ldots,s_{l-1}\}$ 
the chord $\phi(s)$ can not have both its ends between $\phi(x^{0})$ and $\phi(y^{0})$.
Otherwise, $\phi(s)$ would be on the left side of $\phi(x)$ and on the right side of $\phi(y)$,
which is not possible.
For the same reason, $\phi(s)$ can not have both its ends between $\phi(x^{1})$ and $\phi(y^{1})$.
Let $S'$ be the set of all $s \in \{s_1,\ldots,s_{l-1}\}$
such that $\phi(s)$ has one end between $\phi(x^{0})$ and $\phi(y^{0})$ and the other end between 
$\phi(x^{1})$ and $\phi(y^{1})$.
Clearly, $S' \neq \emptyset$ as $\phi'$ and $\phi''$ are not equivalent.
Observe that $S' \cup \{s_l\}$ is a proper module in $(\{s_1,\ldots,s_l\}, {\sim})$.
Indeed, for every $t \in \{s_1,\ldots,s_{l-1}\} \setminus S'$ 
the chord $\phi(t)$ has either both ends between $\phi(y^{0})$ and $\phi(x^{1})$ 
or between $\phi(y^{1})$ and $\phi(x^{0})$, or has one of its ends between $\phi(y^{0})$ and $\phi(x^{1})$
and the second between $\phi(y^{1})$ and $\phi(x^{0})$.
In any case, either $t \sim (S' \cup \{s_l\})$ or $t \parallel (S' \cup \{s_l\})$.
Note that $S' \cup \{s_l\}$ is properly contained in $\{s_1,\ldots,s_l\}$
as otherwise $s_l \sim S'$ and $(U,{\sim})$ would not be prime.
Since the sets $M_1,\ldots,M_k$ restricted to $\{s_1,\ldots,s_l\}$ form
a partition of $\{s_1,\ldots,s_l\}$ into $k$ maximal submodules in $(\{s_1,\ldots,s_l\},{\sim})$,
we imply that $(S' \cup \{s_l\}) \subseteq M_i$ for some $i \in [k]$.
In particular, $M_i$ must be serial as $s_l \sim S'$.
Since for every $u \in U$ such that $u \parallel M_i$ the chord
$\phi(u)$ has both its ends between $\phi(y^{0})$ and $\phi(x^{1})$ or between
$\phi(y^{1})$ and $\phi(x^{0})$, we have $s_l K s'$ for every $s' \in S'$.
However, it can not be the case as $S$ contains exactly one element from every consistent submodule of $M_i$.

To prove statement \eqref{item:consistent_decompositions_and_conformal_models}
assume that $\phi$ is a conformal model of $(M,{\sim})$.

Suppose $K_j = M_i$, where $M_i$ is a prime child of $M$.
Then, statement \eqref{item:consistent_decompositions_and_conformal_models}
follows from Lemma~\ref{lemma:circle_models_of_a_proper_prime_module}
applied to the prime module $M_i$ contained in the connected component $(M,{\sim})$.

Suppose $K_j$ is a consistent submodule of $M_i$, where $M_i$ is a serial child of $M$.
Since $M$ is prime, there is $x \in M \setminus M_i$ such that $x \sim M_i$.
Suppose that $x \in M_{i'}$ for some $i' \in [k]$ different than $i$.
By Claim \ref{claim:permutation_graphs_in_G_ov}, 
$\phi|K_j \cup \{x\} \equiv x^{0} \tau x^{1} \tau'$, where 
$(\tau, \tau^{'})$ is an oriented permutation model of $(K_j,{\sim})$.
Denote by $l^{0}$ and $l^{3}$ the first and the last letter from $K_j^{*}$, respectively,
if we traverse $\phi$ from $\phi(x^{0})$ to $\phi(x^{1})$.
Similarly, denote by $r^{0}$ and $r^3$ the first and the last letter from $K_j^{*}$, respectively,
if we traverse $\phi$ from $\phi(x^{1})$ to $\phi(x^{0})$.

We claim that there is $u \in U \setminus M_i$ such that
$\phi(u)$ has both its ends either between $\phi(l^{3})$ and $\phi(r^{0})$ or between $\phi(r^{3})$ and $\phi(l^{0})$.
Assume otherwise.
Let $T$ be the set of all $t \in M$ such that $\phi(t)$ has one end between $\phi(l^{3})$ and $\phi(r^{0})$ and the second end between $\phi(r^{3})$ and $\phi(l^{0})$.
Note that $T \neq \emptyset$ as $x \in T$.
We claim that $M_i \cup T$ is a module in $(M,{\sim})$.
Indeed, for every $v \in M \setminus (M_i \cup T)$ the chord $\phi(v)$ 
has both its ends either between $\phi(l^{0})$ and $\phi(l^3)$
or between $\phi(r^{0})$ and $\phi(r^3)$.
In particular, we have $u \parallel (M_i \cup T)$, which proves that $M_i \cup T$
is a module in $(M,{\sim})$.
Since $M$ is prime, there is $v \in M \setminus M_i$ such that $v \parallel M_i$.
In particular, $v$ is not in $T$, $M_i \cup T \subsetneq M$, and hence $(M_i \cup T)$
is a non-trivial module in $(M,{\sim})$.
However, this contradicts the fact that $M_i$ is a maximal non-trivial module in $(M,{\sim})$.
This completes the proof that there is $u \in U \setminus M_i$ such that
$\phi(u)$ has both its ends either between $\phi(l^{3})$ and $\phi(r^{0})$ or between $\phi(r^{3})$ and $\phi(l^{0})$.

Now, suppose that $\phi|K_j$ does not form two contiguous subwords in $\phi|M$.
That is, there is $y \in M \setminus K_j$ 
such that $\phi(y)$ has an end between $\phi(l^{0})$ and $\phi(l^{3})$
or between $\phi(r^{0})$ and $\phi(r^{3})$.
Assume that $\phi(y^{0})$ is between $\phi(l^{0})$ and $\phi(l^{3})$ -- the other case is handled analogously.
The end $\phi(y^{0})$ splits $K_j$ into two sets:
$$
\begin{array}{lcl}
K^{'}_j &=& \{v \in K_j: \phi(v) \text{ has an end between $\phi(l^{0})$ and $\phi(y^{0}) \}$, and} \\
K^{''}_j &=& \{v \in K_j: \phi(v) \text{ has an end between $\phi(y^{0})$ and $\phi(l^{3}) \}$.}
\end{array}
$$
By Lemma \ref{lemma:circle_models_of_a_serial_module}, $K^{'}_j$ and $K^{''}_j$ are the unions of some children of 
$M_i$.
Denote by $l^{1}$ and $l^{2}$ the last and the first labeled letter from the sets $K'_j$ and $K''_j$, respectively,
if we traverse $\phi$ from $\phi(x^{0})$ to $\phi(x^{1})$.
Similarly, denote by $r^{1}$ and $r^{2}$ the last and the first labeled letter from the sets $K'_j$ and $K''_j$, respectively, if we traverse $\phi$ from $\phi(x^{1})$ to $\phi(x^{0})$.
Assume that $y \parallel K_j$. 
Then, $y$ is not in $M_i$ as $M_i$ is serial.
Suppose that $y \in M_l$ for some $l \neq [k]$ different than $i$.
Let $P$ be a shortest path in $(M,{\sim})$ that joins $y$ and $M_i$ with all the inner vertices in $U$.
Let $v$ be a neighbor of $y$ in $P$.
Clearly, the chord $\phi(v)$ either has both ends between $\phi(l^{1})$ and $\phi(l^{2})$ or 
has one end between $\phi(l^{1})$ and $\phi(l^2)$ and the second end between $\phi(r^{1})$ and $\phi(r^2)$.
In any case, every chord from $\phi(M_l)$ has either both ends between $\phi(l^{1})$ and $\phi(l^{2})$
or both ends between $\phi(r^{1})$ and $\phi(r^{2})$.
Let $u'$ be the only vertex in $M_l \cap U$.
Now, note that $u,u' \in U\setminus M_i$ witness that $(v',v'') \notin K$ for every $v' \in K'_j$ and every $v'' \in K''_j$, which contradicts that $K_j$ is an equivalence class of $K$-relation in $M_i$.
Assume that $y \sim K_j$.
Then $\phi(y^{1})$ must be between $\phi(r^{1})$ and $\phi(r^{2})$.
If $y \in M_i$, then $y K v$ for every $v \in K_j$,
which contradicts that $K_j$ is an equivalence class of $K$-relation in $M_i$.
So, $y \notin M_i$.
Then, using an analogous argument as for the existence of $u$, 
we show that there is $u' \in U \setminus M_i$ such that $u' \parallel M_i$ and 
$\phi(u')$ has both ends either between $\phi(l^{1})$ and $\phi(l^{2})$
or between $\phi(r^{1})$ and $\phi(r^{2})$.
Consequently, $u$ and $u'$ witness that $(v',v'') \notin K$ for every $v' \in K^{'}_j$ and every $v'' \in K^{''}_j$ -- a contradiction.

Suppose $K_j$ is a consistent submodule of a parallel module $M_i$.
Since $M$ is prime, there is $x \in M \setminus M_i $ such that $x \sim M_i$.
From Claim \ref{claim:permutation_graphs_in_G_ov}, $\phi|(K_j \cup \{x\}) \equiv x^{0} \tau x^{1} \tau'$,
where $(\tau,\tau')$ is an oriented permutation model of $(K_j, \sim)$.
Let $l^{0}$ and $l^{3}$ be the first and the last letter from $K^{*}_j$
if we traverse $\phi$ from $x^{0}$ to $x^{1}$.
Similarly, let $r^{0}$ and $r^{3}$ be the first and the last letter from $K^{*}_j$
if we traverse $\phi$ from $x^{1}$ to $x^{0}$.
Suppose statement \eqref{item:consistent_decompositions_and_conformal_models} does not hold.
That is, there is $y \in (M \setminus K_j)$ such that 
$\phi(y)$ has one of its ends between $\phi(l^{0})$
and $\phi(l^3)$ or between $\phi(r^{0})$ and $\phi(r^3)$.
Suppose that $\phi(u^{0})$ lies between $\phi(l^{0})$ and $\phi(l^{3})$; the other cases are handled analogously.
Split $K_j$ into two subsets, $K'_j$ and $K''_j$, where 
$$
\begin{array}{lcl}
K'_j &=& \{v \in K_j: \phi(v) \text{ has an end between $\phi(l^{0})$ and $\phi(y^{0}) \}$, and} \\
K''_j &=& \{v \in K_j: \phi(v) \text{ has an end between $\phi(y^{0})$ and $\phi(l^{3}) \}$.}
\end{array}
$$
By Lemma \ref{lemma:circle_models_of_a_parallel_module}, $K'_j$ and $K''_j$ are the unions of some children of $M_i$.
Denote by $l^{1}$ and $l^{2}$ the last and the first labeled letter in $K'_j$ and $K''_j$, respectively,
if we traverse $\phi$ from $\phi(l^{0})$ to $\phi(l^{3})$.
Similarly, denote by $r^{1}$ and $r^{2}$ the last and the first labeled letter in $K''_j$ and $K'_j$, respectively, if we traverse $\phi$ from $\phi(r^{0})$ to $\phi(r^{3})$.
Since $\phi(y)$ can not intersect the chords from $\phi(K'_j)$ and from $\phi(K''_j)$ at the same time,
we have $y \parallel K_j$.
Suppose $\phi(y)$ has both ends between $\phi(l^{1})$ and $\phi(l^2)$.
Then we have $y \notin M_i$ as $\phi(y)$ does not intersect $\phi(x)$.
Again, using the idea of the shortest path between $y$ and $M_i$ with the inner vertices in $(U,{\sim})$, we show that there is $u \in U \setminus M_i$ such that $\phi(u)$ has one end
between $\phi(l^{1})$ and $\phi(l^{2})$ and the second one between $\phi(r^{1})$ and $\phi(r^2)$.
Then, the vertex $u$ proves that the vertices from $K'_j$ and the vertices from $K''_j$ are not in $K$-relation, which can not be the case.
So, suppose $\phi(y^{1})$ is between $\phi(r^{1})$ and $\phi(r^2)$.
Note that $y \notin M_i$. 
Otherwise, $y$ would be in $K$-relation with any vertex from $K_j$, 
which would contradict that $K_j$ is an equivalence class of $K$-relation in $M_i$.
Suppose that $y \in M_l$ for some $l \in [k]$ different than $i$.
Note that $x \notin M_l$ as $x \sim K_j$ and $y \parallel K_j$.
Hence, $x \sim M_l$.
Then, every chord from $\phi(M_l)$ has its ends on both sides of $\phi(x)$.
Moreover, every chord from $\phi(M_l)$ has one end between $\phi(l^{1})$ and $\phi(l^{2})$
and the second one between $\phi(r^{1})$ and $\phi(r^{2})$ as $M_l \parallel M_i$.
Hence, a vertex $u \in U \cap M_l$ shows that the vertices from $K'_j$
are not in $K$-relation with the vertices in $K''_j$, a contradiction.
\end{proof}
Let $\phi$ be a conformal model of $(M,{\sim})$.
By Lemma \ref{lemma:two_models_of_a_prime_graph}.\eqref{item:consistent_decompositions_conformal_models_skeleton}, \begin{equation*}
\phi|S = \phi^m_S \text{ for some } m \in [2],
\end{equation*}
where $\phi^0_S$ and its reflection $\phi^1_S$ are the only two conformal models of $(S,{\sim})$.
Let $K_i$ be a consistent submodule of $M$ for some $i \in [n]$.
Suppose $\tau^{0}_{i,\phi}$ and $\tau^1_{i,\phi}$ are two contiguous 
subwords of $\phi$ in $\phi|K_i$ enumerated such that
$\tau^{j}_{i,\phi}$ contains the letter $s^{j}_i$ for $j \in \{0,1\}$ -- see Lemma \ref{lemma:consistent_decompositions}.\eqref{item:consistent_decompositions_and_conformal_models}.
Let $K^{j}_{i,\phi}$ be the set of the letters contained in $\tau^{j}_{i,\phi}$ for $j \in \{0,1\}$.
Note that $K^{0}_{i,\phi},K^{1}_{i,\phi}$ are labeled copies of $K_i$ and $\{K^{0}_{i,\phi}, K^{1}_{i,\phi}\}$ forms a partition of $K^{*}_i$.
Clearly, $(\tau^{0}_{i,\phi}, \tau^{1}_{i,\phi})$ is an oriented permutation model of $(K_i,{\sim})$.
Let $({<^{0}_{K_i,\phi}},{\prec^{0}_{K_i,\phi}})$ be the transitive orientations of $(K_i,{\parallel})$ and $(K_i,{\sim})$, respectively, corresponding to the non-oriented permutation model $(\tau^{0}_{i,\phi}, \tau^{1}_{i,\phi})$ of $(K_i,{\sim})$.
Finally, let $\pi(\phi)$ be a circular permutation of $K^{0}_{1,\phi},K^{1}_{1,\phi}, \ldots, K^{0}_{n,\phi}, K^{1}_{n,\phi}$ that arises from $\phi$ by replacing every contiguous word $\phi|K^{j}_{i,\phi}$ by the set $K^{j}_i$.
It turns out that the sets $K^{j}_{i,\phi}$ and the transitive orientations $<^{0}_{i,\phi}$ of $(K_i,\parallel)$ do not depend on the choice of a conformal model $\phi$ of $G_{ov}$.
Moreover, $\pi(\phi)$ may take only two values depending on whether $\phi|S \equiv \phi^{0}_S$ or 
$\phi|S \equiv \phi^{1}_S$.

\begin{claim}
\label{claim:prime_module_invariants}
For every $i \in [n]$ there are labeled copies $K^{0}_i$ and $K^{1}_i$ of $K_i$ forming a partition of $K^{*}_i$ and a transitive orientation $<_{K_i}$ of $(K_i,{\parallel})$ such that
$$(K^{0}_{i, \phi}, K^{1}_{i, \phi},{<^{0}_{K_i, \phi}}) = (K^{0}_{i}, K^{1}_i,{<_{K_i}}) 
\quad
\begin{array}{c}
\text{for every conformal model $\phi$ of $(M,{\sim})$} \\
\text{and every $i \in [n]$.}
\end{array}
$$
Moreover, there are circular permutations $\pi_{0}(M), \pi_{1}(M)$ of $\{K^{0}_1,K^{1}_0, \ldots,K^{0}_n,K^{1}_n\}$, where $\pi_{0}(M)$ is the reflection of $\pi_{1}(M)$, such that 
$$\pi(\phi) = 
\left\{
\begin{array}{cll}
\pi_{0}(M) & \text{if} & \phi|S = \phi^{0}_S\\
\pi_{1}(M) & \text{if} & \phi|S = \phi^{1}_S\\
\end{array}
\right.
\quad \text{for every conformal model $\phi$ of $(M,{\sim})$.}
$$
\end{claim}
\begin{proof}
The second part of the claim follows directly from Lemma \ref{lemma:consistent_decompositions}.
In particular, note that $\pi_m(M)$ for $m \in [2]$ is obtained from $\phi^m_S$ by replacing every
labeled letter $s^0_i$ by $K^0_i$ and $s^1_i$ by $K^1_i$.
The first part of the claim is proved similarly to Claim~\ref{claim:invariants_for_proper_parallel_and_proper_prime_modules}:
in particular, for every vertex $v \in K_i$ the algorithm may decide whether $v^0 \in K^0_i$ or
$v^0 \in K^1_i$ basing on whether the vertices $s_i$ and $v$ have the vertex $s_j$ on the same side, where $s_j$
is such that $s_j \parallel K_i$.
Given the sets $K^0_i$ and $K^1_i$, the relation ${<_{K_i}}$ is computed in the same way as in Claim~\ref{claim:invariants_for_proper_parallel_and_proper_prime_modules}.
\end{proof}
The elements of the set $\{K_i^{0}, K^{1}_i,\ldots, K^{0}_n,K^{1}_n\}$ are called the \emph{slots} of $M$,
$\pi_{0}(M),\pi_{1}(M)$ are called the \emph{circular permutations of the slots} in the module $M$,
and the triple $(K^{0}_i,K^{1},{<_{K_i}})$, denoted by $\mathbb{K}_{i}$,
is called the \emph{metaedge} of $K_i$.

Now, we extend the notion of an admissible model to the metaedges associated with consistent submodules of $M$
and to the circular orientations of the slots $\pi_0(M)$ and $\pi_1(M)$.
\begin{definition}
\label{def:admissible_models_for_consistent_submodules}
Let $K_i$ be a consistent submodule of $M$ and let 
$\mathbb{K}_i = (K^{0}_i,K^{1}_i,{<_{K_i}})$ be the metaedge of $K_i$.
A pair $(\tau^{0}, \tau^{1})$ is \emph{an admissible model for $\mathbb{K}_i$} if:
\begin{itemize}
 \item $\tau^0$ is a permutation of $K^{0}_i$,
 \item $\tau^1$ is a permutation of $K^{1}_i$,
 \item $(\tau^{0},\tau^{1})$ is an oriented permutation model of $(K_i,{\sim})$ that corresponds to the pair $({<}, {\prec})$ of transitive orientations of $(K_i,{\parallel})$ and $(K_i,{\sim})$, respectively, where ${<} = {<_{K_i}}$.
\end{itemize}
\end{definition}
\begin{definition}
\label{def:admissible_models_for_prime_modules}
Let $m \in \{0,1\}$, let $M$ be an improper prime module in $\mathcal{M}(G_{ov})$, 
and let $\pi_m(M)$ be the circular permutation of the slots of $M$.
A circular word $\phi$ on the set $M^{*}$ is an \emph{admissible model for $\pi_m(M)$}
if $\phi$ arises from $\pi_{m}(M)$ by exchanging every slot $K^{j}_i$ by
a permutation $\tau^{j}_{i}$, where the words $\tau^{0}_{i}, \tau^{1}_{i}$ are such that $(\tau^{0}_{i}, \tau^{1}_{i})$ is an admissible model for $\mathbb{K}_i$.
\end{definition}
The next theorem provides a description of all conformal models of $(M,{\sim})$.
\begin{theorem}
\label{thm:description_of_all_conformal_models_of_improper_prime_modules}
Suppose $G$ is a circular-arc graph with no twins and no universal vertices.
Suppose $M$ is an improper prime module in $\mathcal{M}(G_{ov})$.
A circular word $\phi$ is a conformal model of $(M,{\sim})$ if and only if $\phi$ is an admissible model for $\pi_m(M)$ for some $m \in \{0,1\}$.
\end{theorem}
\begin{proof}
Suppose $\phi$ is conformal model of $(M,{\sim})$.
By Claim \ref{claim:prime_module_invariants} we get $\pi(\phi) = \pi_m(M)$ for some $m \in \{0,1\}$
and $(\phi|K^0_i, \phi|K^1_i)$ is an admissible model for $\mathbb{K}_i$ for every $i \in [n]$.
Thus, $\phi$ is admissible for $\pi_m(M)$.

Suppose $\phi$ is an admissible model for $\pi_m(M)$ for some $m \in \{0,1\}$.
Since $G$ is a circular-arc graph, $G_{ov}$ has a conformal model.
Hence, $(M,{\sim})$ has a conformal model, say $\phi'$.
Since the reflection of a conformal model is also conformal, we may
assume that $\pi(\phi') = \pi_m(M)$.
Now, we start with $\phi'$ and for every $i \in [n]$
we replace the words $(\phi'|K^{0}_i, \phi'|K^{1}_i)$ in $\phi'$ 
by the words $(\phi|K^{0}_i, \phi|K^{1}_i)$, respectively.
Finally, we obtain the model $\phi$.
Since ${<^{0}_{K_i,\phi}} = {<^{0}_{K_i, \phi'}} = {<_{K_i}}$, one can easily check
that after every transformation the chords $\phi(v)$ and $\phi'(v)$ have on its both sides 
the chords representing exactly the same sets of vertices.
This proves that $\phi$ is also conformal.
\end{proof}
The above theorem characterizes all conformal models of $G_{ov}$ in the case when
$(V,{\sim})$ and $(V,{\parallel})$ are connected. 
Indeed, in this case $V$ is an improper prime module in
$\mathcal{M}(G_{ov})$ and hence Theorem \ref{thm:description_of_all_conformal_models_of_improper_prime_modules} applies.
Figure \ref{fig:prime_case_schematic_view} shows a schematic picture of some conformal models of $G_{ov}$.
\begin{figure}[htp!]
\begin{tikzpicture}[scale=0.7,>=latex,shorten >=-0.2pt,shorten <=-0.2pt]
\coordinate (center) at (0,0) {};

\coordinate (k11) at ($(center)+(0:2cm)$) {};
\coordinate (k10) at ($(center)+(180:2cm)$) {};
\coordinate (k21) at ($(center)+(60:2cm)$) {};
\coordinate (k20) at ($(center)+(300:2cm)$) {};
\coordinate (k31) at ($(center)+(135:2cm)$) {};
\coordinate (k30) at ($(center)+(225:2cm)$) {};
\coordinate (k41) at ($(center)+(30:2cm)$) {};
\coordinate (k40) at ($(center)+(110:2cm)$) {};
\coordinate (k51) at ($(center)+(250:2cm)$) {};
\coordinate (k50) at ($(center)+(330:2cm)$) {};

\coordinate (lk11) at ($(center)+(0:2.4cm)$) {};
\coordinate (lk10) at ($(center)+(180:2.4cm)$) {};
\coordinate (lk21) at ($(center)+(60:2.4cm)$) {};
\coordinate (lk20) at ($(center)+(300:2.4cm)$) {};
\coordinate (lk31) at ($(center)+(135:2.4cm)$) {};
\coordinate (lk30) at ($(center)+(225:2.4cm)$) {};
\coordinate (lk41) at ($(center)+(30:2.4cm)$) {};
\coordinate (lk40) at ($(center)+(110:2.4cm)$) {};
\coordinate (lk51) at ($(center)+(250:2.4cm)$) {};
\coordinate (lk50) at ($(center)+(330:2.4cm)$) {};

\tikzstyle{every node}=[inner sep=1pt]
\begin{footnotesize}
\node at (lk11) {$K_1^1$};
\node at (lk10) {$K_1^0$};
\node at (lk21) {$K_2^1$};
\node at (lk20) {$K_2^0$};
\node at (lk31) {$K_3^1$};
\node at (lk30) {$K_3^0$};
\node at (lk41) {$K_4^1$};
\node at (lk40) {$K_4^0$};
\node at (lk51) {$K_5^1$};
\node at (lk50) {$K_5^0$};
\end{footnotesize}

\draw (0,0) circle (2cm);

\draw[line width=0.5mm,->] (k10)--(k11);
\draw[line width=0.5mm,->] (k20)--(k21);
\draw[line width=0.5mm,->] (k30)--(k31);
\draw[line width=0.5mm,->] (k40)--(k41);
\draw[line width=0.5mm,->] (k50)--(k51);

\draw (3.2,2.7)--(3.2,-2.7);

\draw[white] (-3.0,-2.7)--(-3.0,-2.5);
\draw[white] (3.0,2.7)--(3.0,2.5);
\end{tikzpicture}
\begin{tikzpicture}[scale=0.7,>=latex,shorten >=-0.2pt,shorten <=-0.2pt]
\coordinate (center) at (0,0) {};

\coordinate (k11) at ($(center)+(0:2cm)$) {};
\coordinate (k10) at ($(center)+(180:2cm)$) {};
\coordinate (k20) at ($(center)+(120:2cm)$) {};
\coordinate (k21) at ($(center)+(240:2cm)$) {};
\coordinate (k30) at ($(center)+(45:2cm)$) {};
\coordinate (k31) at ($(center)+(315:2cm)$) {};
\coordinate (k40) at ($(center)+(150:2cm)$) {};
\coordinate (k41) at ($(center)+(70:2cm)$) {};
\coordinate (k50) at ($(center)+(290:2cm)$) {};
\coordinate (k51) at ($(center)+(210:2cm)$) {};

\coordinate (lk11) at ($(center)+(0:2.4cm)$) {};
\coordinate (lk10) at ($(center)+(180:2.4cm)$) {};
\coordinate (lk20) at ($(center)+(120:2.4cm)$) {};
\coordinate (lk21) at ($(center)+(240:2.4cm)$) {};
\coordinate (lk30) at ($(center)+(45:2.4cm)$) {};
\coordinate (lk31) at ($(center)+(315:2.4cm)$) {};
\coordinate (lk40) at ($(center)+(150:2.4cm)$) {};
\coordinate (lk41) at ($(center)+(70:2.4cm)$) {};
\coordinate (lk50) at ($(center)+(290:2.4cm)$) {};
\coordinate (lk51) at ($(center)+(210:2.4cm)$) {};

\tikzstyle{every node}=[inner sep=1pt]
\begin{footnotesize}
\node at (lk11) {$K_1^1$};
\node at (lk10) {$K_1^0$};
\node at (lk20) {$K_2^0$};
\node at (lk21) {$K_2^1$};
\node at (lk30) {$K_3^0$};
\node at (lk31) {$K_3^1$};
\node at (lk40) {$K_4^0$};
\node at (lk41) {$K_4^1$};
\node at (lk50) {$K_5^0$};
\node at (lk51) {$K_5^1$};
\end{footnotesize}

\draw (0,0) circle (2cm);

\draw[line width=0.5mm,->] (k10)--(k11);
\draw[line width=0.5mm,->] (k20)--(k21);
\draw[line width=0.5mm,->] (k30)--(k31);
\draw[line width=0.5mm,->] (k40)--(k41);
\draw[line width=0.5mm,->] (k50)--(k51);

\draw[white] (-3.0,-2.7)--(-3.0,-2.5);
\draw[white] (3.0,2.7)--(3.0,2.5);
\end{tikzpicture}
\caption{\label{fig:prime_case_schematic_view}}
\end{figure}
To get the full picture of a normalized model one needs to expand every metaedge $(K^0_i,K^1_i,{<_{K_i}})$ to show an admissible model hidden behind it.

To transform a normalized model of $G_{ov}$ into a normalized model of $G_{ov}$ 
we can reflect the circular permutation of the slots of $M$ and we can replace
an admissible model hidden under every metaedge.
Theorem \ref{thm:description_of_all_conformal_models_of_improper_prime_modules} asserts that such operations are complete, that is, 
we can transform one normalized model into any other by performing such operations.

\subsection{Conformal representations of improper parallel modules}
\label{subsec:conformal_models_of_improper_parallel_modules}
The results from the first subsection of this section are taken from \cite{Hsu95} by Hsu.
The notion of an admissible model for prime children of $V$ was also partially inspired by \cite{Hsu95}.

Suppose $G$ is a circular-arc graphs with no twins and no universal vertices
such that its overlap graph $G_{ov} = (V,{\sim})$ is disconnected.
In this case, $V$ is an improper parallel module in $\mathcal{M}(G_{ov})$.
Denote the children of $V$ in $\mathcal{M}(G_{ov})$ by $\mathcal{M}(V)$.
Note that every module $M$ in $\mathcal{M}(V)$ is either improper prime or improper serial.

\subsubsection{$T_{NM}$ tree} Let $M$ be a module in $\mathcal{M}(V)$ and let $v \in V \setminus M$.
Observe that either $M \subseteq \leftside(v)$ and then we say \emph{$M$ is on the left side of $v$} 
or $M \subseteq \rightside(v)$ and then we say \emph{$M$ is on the right side of $v$}.
Let $\mathcal{M}_1$ and $\mathcal{M}_2$
be two disjoint sets of modules from $\mathcal{M}(V)$  and let $v \notin \bigcup \mathcal{M}_1 \cup \bigcup \mathcal{M}_2$.
We say that \emph{$v$ separates $\mathcal{M}_1$ and $\mathcal{M}_2$} if
either $\bigcup \mathcal{M}_1 \subset \leftside(v)$ and $\bigcup \mathcal{M}_2 \subset \rightside(v)$
or $\bigcup \mathcal{M}_2 \subset \leftside(v)$ and $\bigcup \mathcal{M}_1 \subset \rightside(v)$.
We say that \emph{$v$ separates $M_1$ and $M_2$} if $v$ separates $\{M_1\}$ and $\{M_2\}$ -- see Figure \ref{fig:T_NM_tree} for an example.
We use analogous phrases to describe the mutual position of the corresponding sets of chords 
in conformal models $\phi$ of $G_{ov}$.

Two modules $M_1, M_2 \in \mathcal{M}(V)$ are \emph{non-separated} if there is no vertex
$v \in V \setminus (M_1 \cup M_2)$ that separates $M_1$ and $M_2$.
A set $N \subseteq \mathcal{M}(V)$ is a \emph{node} in $(V,{\sim})$
if $N$ is a maximal subset of pairwise non-separated modules from $\mathcal{M}(V)$.

Let $T_{NM}$ be a bipartite graph with the set of vertices containing all the modules in $\mathcal{M}(V)$ and all the nodes in $G_{ov}$ and the set of edges joining every module $M$ and every node $N$ such that $M \in N$ 
-- see Figure \ref{fig:T_NM_tree} for an example.

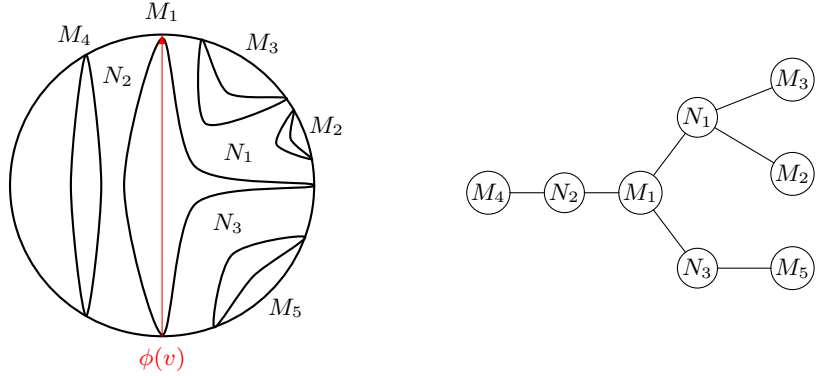
\begin{figure}[htp!]
\begin{tikzpicture}[scale=1,>=latex,shorten >=-0.4pt,shorten <=-0.4pt]
\coordinate (label) at (0,-3) {};

\coordinate (m11) at ($(center)+(90:1.95cm)$) {};
\coordinate (m112) at ($(center)+(180:0.5cm)$) {};
\coordinate (m12) at ($(center)+(270:1.98cm)$) {};
\coordinate (m13) at ($(center)+(330:0.5cm)$) {};
\coordinate (m14) at ($(center)+(0:2cm)$) {};
\coordinate (m15) at ($(center)+(30:0.5cm)$) {};
\coordinate (lm1) at ($(center)+(90:2.3cm)$) {};

\coordinate (m21) at ($(center)+(10:2cm)$) {};
\coordinate (m22) at ($(center)+(20:1.8cm)$) {};
\coordinate (m23) at ($(center)+(30:2cm)$) {};
\coordinate (m24) at ($(center)+(20:1.6cm)$) {};
\coordinate (lm2) at ($(center)+(20:2.3cm)$) {};

\coordinate (m31) at ($(center)+(35:2cm)$) {};
\coordinate (m32) at ($(center)+(55:1.5cm)$) {};
\coordinate (m33) at ($(center)+(75:2cm)$) {};
\coordinate (m34) at ($(center)+(55:1cm)$) {};
\coordinate (lm3) at ($(center)+(55:2.3cm)$) {};

\coordinate (m41) at ($(center)+(120:2cm)$) {};
\coordinate (m42) at ($(center)+(0:-1.2cm)$) {};
\coordinate (m43) at ($(center)+(240:2cm)$) {};
\coordinate (m44) at ($(center)+(0:-0.8cm)$) {};
\coordinate (lm4) at ($(center)+(120:2.3cm)$) {};

\coordinate (m51) at ($(center)+(290:2cm)$) {};
\coordinate (m52) at ($(center)+(315:1.7cm)$) {};
\coordinate (m53) at ($(center)+(340:2cm)$) {};
\coordinate (m54) at ($(center)+(315:1.3cm)$) {};
\coordinate (lm5) at ($(center)+(315:2.3cm)$) {};

\coordinate (ln1) at ($(center)+(23:1.1cm)$) {};
\coordinate (ln2) at ($(center)+(112:1.56cm)$) {};
\coordinate (ln3) at ($(center)+(330:1.0cm)$) {};

\coordinate (v) at ($(center)+(270:2.3cm)$) {};

\begin{scriptsize}
\tikzstyle{every node}=[inner sep=1pt]
\node at (lm1) {$M_1$};
\node at (lm2) {$M_2$};
\node at (lm3) {$M_3$};
\node at (lm4) {$M_4$};
\node at (lm5) {$M_5$};

\node at (ln1) {$N_1$};
\node at (ln2) {$N_2$};
\node at (ln3) {$N_3$};

\node[red] at (v) {$\phi(v)$};
\end{scriptsize}

\draw[thick] plot [smooth cycle] coordinates {(m11) (m112) (m12) (m13) (m14) (m15)};

\draw[thick] plot [smooth cycle] coordinates {(m21) (m22) (m23) (m24)};

\draw[thick] plot [smooth cycle] coordinates {(m31) (m32) (m33) (m34)};

\draw[thick] plot [smooth cycle] coordinates {(m41) (m42) (m43) (m44)};

\draw[thick] plot [smooth cycle] coordinates {(m51) (m52) (m53) (m54)};

\draw[thick] (0,0) circle (2cm);

\draw[<-, red] ($(center)+(90:2cm)$)--($(center)+(270:2cm)$);

\draw[white] (-2.5,-2.5)--(-2.5,-2.3);
\draw[white] (2.5,2.5)--(2.5,2.3);

\end{tikzpicture} 
\hspace{1cm}
\begin{tikzpicture}[scale=1,>=latex,shorten >=-0.4pt,shorten <=-0.4pt]
  \tikzstyle{every node}=[circle,minimum size=10pt,inner sep=0.5,draw];
  \begin{scriptsize}
  \node (m1) at (0.0,0) {$M_1$};
  \node (m4) at (-2,0.0) {$M_4$};
  \node (m3) at (2,1.5) {$M_3$};
  \node (m2) at (2,0.25) {$M_2$};
  \node (m5) at (2,-1) {$M_5$};
  \end{scriptsize}
  \tikzstyle{every node}=[circle,minimum size=15pt,inner sep=0.5,draw];
  \begin{scriptsize}
  \node (n1) at (0.75,1.0) {$N_1$};
  \node (n3) at (0.75,-1.0) {$N_3$};
  \node (n2) at (-1.0,0) {$N_2$};
  \end{scriptsize}
\path (m1) edge (n1); 
\path (m1) edge (n2); 
\path (m1) edge (n3); 

\path (n1) edge (m3); 
\path (n1) edge (m2); 
\path (n3) edge (m5); 
\path (m4) edge (n2); 

\draw[white] (-2.5,-2.4)--(-2.5,-1.8);
\draw[white] (2.5,2.6)--(2.5,2.3);
\end{tikzpicture}
\caption{\label{fig:T_NM_tree} The schematic picture of $T_{NM}$ tree with the nodes
$N_1,N_2,N_3$ and the modules $M_1,M_2,M_3,M_4,M_5$.
The chord $\phi(v)$ separates $\{M_4\}$ and $\{M_2,M_3,M_5\}$.
The maximal subsets consisting of non-separated modules are $\{M_1,M_2,M_3\}$, $\{M_1,M_4\}$, and $\{M_1,M_5\}$,
which correspond to the nodes $N_1$, $N_2$, and $N_3$, respectively.
}
\end{figure}

Let $N_T[M]$ denote the neighbours of a module $M$ in $T_{NM}$ and let
$N_T[N]$ denote the neighbors of a node $N$ in $T_{NM}$.
\begin{claim}
\label{claim:T_NM_tree} The following statements hold:
\begin{enumerate}
\item \label{item:T_NM_module_and_neighbors} For every module $M \in T_{NM}$ and every two nodes $N_1,N_2 \in N_T[M]$ there is a vertex $v \in M$ that separates 
the modules in $N_1 \setminus \{M\}$ and the modules in $N_2 \setminus \{M\}$.
\item \label{item:T_NM_tree} 
The bipartite graph $T_{NM}$ is a tree. 
All leaves of $T_{NM}$ are in the set $\mathcal{M}(V)$.
\end{enumerate}
\end{claim}
\begin{proof}
Let $M$ be a module in $T_{NM}$ and let $N_1, N_2$ be two different nodes adjacent to $M$ in $T_{NM}$.
Since $N_1,N_2$ are different maximal subsets of pairwise non-separated modules from $\mathcal{M}(V)$, 
there is a module $M_1 \in N_1 \setminus N_2$ and a module $M_2 \in N_2 \setminus N_1$ such that $M_1$ and $M_2$ are separated by some $v \in V \setminus (M_1 \cup M_2)$.
Suppose $M_1 \subset \leftside(v)$ and $M_2 \subset \rightside(v)$.
Suppose $v \notin M$.
Then, depending on whether $M \subset \rightside(v)$ or $M \subset \leftside(v)$, 
$v$ separates either $M$ and $M_1$ or $M$ and $M_2$.
Hence, either $N_1$ or $N_2$ is not a node of $G_{ov}$.
So, $v \in M$.
Now, the modules from $N_1 \setminus \{M\}$ are on the left side of $v$ 
and the modules from $N_2 \setminus \{M\}$ are on the right side of $v$ 
as otherwise $N_1$ or $N_2$ is not a node of $G_{ov}$.
This proves \eqref{item:T_NM_module_and_neighbors}.

Now, we show that $T_{NM}$ is a tree.
First we prove that $T_{NM}$ contains no cycles.
Suppose that $M_1N_1 \ldots M_kN_k$ is a cycle in $T_{NM}$, for some $k \geq 2$.
Since $N_1,N_k$ are neighbors of $M_1$ in $T_{NM}$, there
is $v \in M_1$ that separates the modules in $N_1 \setminus \{M_1\}$ 
and the modules in $N_k \setminus \{M_1\}$.
So, $v$ separates $M_{2}$ and $M_k$. 
In particular, $M_2$ and $M_k$ can not be in a node of $G_{ov}$,
and hence $k \geq 3$.
Since $M_2$ and $M_k$ are separated by $v$, 
there is $i \in [2,k-1]$ such that $M_i$ and $M_{i+1}$ are also separated by $v$.
So, $M_i$ and $M_{i+1}$ can not be contained in a node of $G_{ov}$,
which is not the case as $N_i$ contains both $M_i$ and $M_{i+1}$.
This shows that $T_{NM}$ is a forest.

To show that $T_{NM}$ has all leaves in the set $\mathcal{M}(V)$
it suffices to prove that every node in $G_{ov}$ contains at least two modules.
Suppose for a contradiction that $\{M\}$ is a node in $T_{NM}$, for some $M \in \mathcal{M}(V)$.
Let $\phi$ be a conformal model of $G_{ov}$.
Pick a vertex $u \in V \setminus M$ and a vertex $v \in M$ such that
$u^{*}$ and $v^{*}$ are two consecutive labeled letters in $\phi$ for some $u^{*} \in \{u^{0},u^{1}\}$
and $v^{*} \in \{v^{0},v^{1}\}$.
Then, note that $M$ and the module from $\mathcal{M}(V)$ containing $u$ are non-separated, 
which shows that $\{M\}$ can not be a node in $G_{ov}$.
It remains to show that $T_{NM}$ is connected.
Suppose $T_{NM}$ is a forest and suppose $\phi$ is any conformal model of $G_{ov}$.
Then, there are vertices $u$ and $v$ such that $u^{*}$ and $v^{*}$ are consecutive in $\phi$, 
$u^{*} \in \{u^{0},u^{1}\}$, $v^{*} \in \{v^{0},v^{1}\}$, $u \in M_u$, $v \in M_v$, and $M_u$ and $M_v$ are 
modules from $\mathcal{M}(V)$ which are
in different connected components of $T_{NM}$.
By the choice of $u$ and $v$ there is no vertex in $M \setminus (M_u \cup M_v)$
that separates $M_u$ and $M_v$.
So, $M_u$ and $M_v$ are adjacent to some node of $G_{ov}$ in $T_{NM}$,
which contradicts that $M_u$ and $M_v$ are in different connected components of $T_{NM}$.
\end{proof}

Let $M$ be a module and $N$ be a node such that $M$ and $N$ are adjacent in $T_{NM}$. 
Let $T_{NM} \setminus M$ be a forest obtained from $T_{NM}$ by deleting the module $M$ and let
$V_{T\setminus M}(N)$ be the set of all vertices contained in the modules from the connected component of $T_{NM} \setminus M$ containing~$N$.
Similarly, 
let $T_{NM} \setminus N$ be a forest obtained from $T_{NM}$ by deleting the node $N$ and let
$V_{T\setminus N}(M)$ be the set of all vertices contained in the modules from the connected component of $T_{NM} \setminus N$ containing~$M$.
In the example shown in Figure~\ref{fig:T_NM_tree}, we have
$V_{T \setminus M_1}(N_1) = M_2 \cup M_3$, $V_{T \setminus M_1}(N_2) = M_4$,
$V_{T \setminus N_3}(M_1) = M_1 \cup M_2 \cup M_3 \cup M_4$.
\begin{claim}
\label{claim:placements_of_neighboring_nodes_in_a_module}
Suppose $M$ is a module in $T_{NM}$, $N$ is a node adjacent to $M$ in $T_{NM}$ and $v$ is a vertex in $M$.
Then, either $V_{T\setminus M}(N) \subset \leftside(v)$ or $V_{T\setminus M}(N) \subset \rightside(v)$.
\end{claim}
\begin{proof}
Let $F_N$ be a connected component of $T_{NM} \setminus M$ containing $N$.
Suppose there are two modules $M_1,M_2 \in F_N$ 
such that $M_1 \subset \leftside(v)$ and $M_2 \subset \rightside(v)$.
Let $P$ be a path between $M_1$ and $M_2$ in $F_{N}$.
Clearly, there are three consecutive elements $M'_1 N' M'_2$ on the path $P$
such that $M'_1 \in \leftside(v)$ and $M'_2 \in \rightside(v)$.
Thus, $M'_1$ and $M'_2$ are separated by $v$,
which contradicts $M'_1,M'_2 \in N'$.
\end{proof}
Let $M$ be a module in $T_{NM}$, $N,N'$ be two nodes adjacent to $M$ in $T_{NM}$, and $v$ be a vertex in $M$. 
We say that:
\begin{itemize}
 \item $V_{T \setminus M}(N)$ (or shortly $N$) is \emph{on the left (right) side} of $v$ if $V_{T \setminus M}(N) \subset \leftside(v)$ ($V_{T \setminus M}(N) \subset \rightside(v)$, respectively),
 \item $v$ \emph{separates} $V_{T\setminus M}(N)$ and $V_{T\setminus M}(N')$ (or shortly $v$ \emph{separates} $N$
 and $N'$) if $v$ has $N$ and $N'$ on its different sides.
\end{itemize}
We use similar phrases to describe the mutual location of the corresponding sets of chords in conformal models of $G_{ov}$.

\begin{claim}
\label{claim:conformal_models_properties_modules_nodes_T_NM}
Let $M$ be a module in $T_{NM}$, $N$ be a node in $T_{NM}$, and $\phi$ be a conformal model of $G_{ov}$.
Then:
\begin{enumerate}
\item \label{item:conformal_models_properties_modules_T_NM} For every node $N' \in N_{T}[M]$ the set $\phi|V_{T \setminus M}(N')$ contains a single contiguous subword of $\phi$.
Moreover, for every two different nodes $N',N'' \in N_T[M]$ there is $v \in M$ 
such that $\phi(v)$ separates $\phi|V_{T \setminus M}(N')$
and $\phi|V_{T \setminus M}(N'')$.
\item \label{item:conformal_models_properties_nodes_T_NM} For every module $M' \in N_T[N]$ the set $\phi|V_{T \setminus N}(M')$ contains a single contiguous subword of $\phi$.
\end{enumerate}
\end{claim}
See Figure \ref{fig:T_NM_tree_subwords} for an illustration.
\begin{proof}
Statement \eqref{item:conformal_models_properties_modules_T_NM} follows from
Claim \ref{claim:conformal_models_properties_modules_nodes_T_NM} and Claim \ref{claim:T_NM_tree}.\eqref{item:T_NM_module_and_neighbors}

Since $N$ is a maximal subset of $\mathcal{M}(V)$ containing pairwise non-separated modules,
$\phi|M'$ is a contiguous word in the circular word $\phi|(\bigcup N)$.
Now, statement \eqref{item:conformal_models_properties_modules_T_NM} applied to the neighbors of the module $M'$ different than $N$ proves that $\phi|V_{T \setminus N}(M')$ is a contiguous subword of $\phi$.
\end{proof}

Let $\phi$ be a conformal model of $G_{ov}$,
$M$ be a module in $T_{NM}$, and $N$ be a node in $T_{NM}$.
For every $N' \in N_T[M]$ we replace the contiguous subword $\phi|V_{T \setminus M}(N')$ in $\phi$ by the letter $N'$.
We denote the circular word arisen this way by $\phi|(M \cup N_T[M])$.
Similarly, for every $M' \in N_T[N]$ we replace the contiguous subword $\phi|V_{T \setminus N}(M')$ in $\phi$ by the letter $M'$.
We denote the circular word arisen this way by $\phi|N_T[N]$.
See Figure \ref{fig:T_NM_tree_subwords} for an illustration.
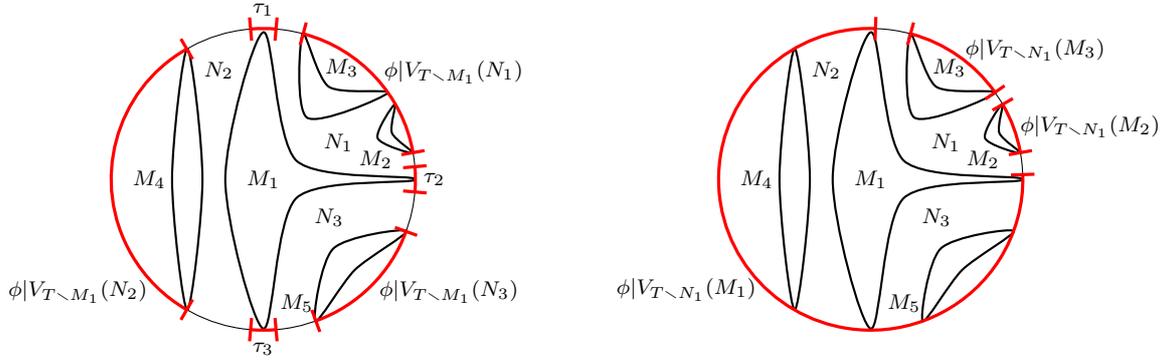
\begin{figure}[htp!]
\begin{tikzpicture}[scale=1,>=latex,shorten >=-0.4pt,shorten <=-0.4pt]
\coordinate (label) at (0,-3) {};

\coordinate (m11) at ($(center)+(90:1.95cm)$) {};
\coordinate (m112) at ($(center)+(180:0.5cm)$) {};
\coordinate (m12) at ($(center)+(270:1.98cm)$) {};
\coordinate (m13) at ($(center)+(330:0.5cm)$) {};
\coordinate (m14) at ($(center)+(0:2cm)$) {};
\coordinate (m15) at ($(center)+(30:0.5cm)$) {};
\coordinate (lm1) at ($(center)+(90:0.0cm)$) {};
\coordinate (phi_m1) at ($(center)+(225:2.07cm)$) {};

\coordinate (m21) at ($(center)+(10:2cm)$) {};
\coordinate (m22) at ($(center)+(20:1.8cm)$) {};
\coordinate (m23) at ($(center)+(30:2cm)$) {};
\coordinate (m24) at ($(center)+(20:1.6cm)$) {};
\coordinate (lm2) at ($(center)+(10:1.5cm)$) {};
\coordinate (phi_m2) at ($(center)+(20:2.05cm)$) {};

\coordinate (m31) at ($(center)+(35:2cm)$) {};
\coordinate (m32) at ($(center)+(55:1.5cm)$) {};
\coordinate (m33) at ($(center)+(75:2cm)$) {};
\coordinate (m34) at ($(center)+(55:1cm)$) {};
\coordinate (lm3) at ($(center)+(55:1.8cm)$) {};

\coordinate (m41) at ($(center)+(120:2cm)$) {};
\coordinate (m42) at ($(center)+(0:-1.2cm)$) {};
\coordinate (m43) at ($(center)+(240:2cm)$) {};
\coordinate (m44) at ($(center)+(0:-0.8cm)$) {};
\coordinate (lm4) at ($(center)+(180:1.5cm)$) {};

\coordinate (m51) at ($(center)+(290:2cm)$) {};
\coordinate (m52) at ($(center)+(315:1.7cm)$) {};
\coordinate (m53) at ($(center)+(340:2cm)$) {};
\coordinate (m54) at ($(center)+(315:1.3cm)$) {};
\coordinate (lm5) at ($(center)+(285:1.7cm)$) {};

\coordinate (ln1) at ($(center)+(26:1.1cm)$) {};
\coordinate (ln2) at ($(center)+(112:1.56cm)$) {};
\coordinate (ln3) at ($(center)+(330:1.0cm)$) {};

\coordinate (phi_n1) at ($(center)+(42.5:2.1cm)$) {};

\coordinate (phi_tau_1) at ($(center)+(90:2.25cm)$) {};

\coordinate (phi_n2) at ($(center)+(225:2.1cm)$) {};

\coordinate (phi_tau_2) at ($(center)+(270:2.25cm)$) {};

\coordinate (phi_n3) at ($(center)+(315:2.1cm)$) {};

\coordinate (phi_tau_3) at ($(center)+(0:2.25cm)$) {};

\begin{tiny}
\tikzstyle{every node}=[inner sep=1pt]
\node[anchor=west] at (phi_n1) {$\phi|V_{T\setminus M_1}(N_1)$};
\node[anchor=east] at (phi_n2) {$\phi|V_{T\setminus M_1}(N_2)$};
\node[anchor=west] at (phi_n3) {$\phi|V_{T\setminus M_1}(N_3)$};
\node at (phi_tau_1) {$\tau_1$};
\node at (phi_tau_2) {$\tau_3$};
\node at (phi_tau_3) {$\tau_2$};

\node at (lm1) {$M_1$};
\node at (lm2) {$M_2$};

\node at (lm3) {$M_3$};

\node at (lm4) {$M_4$};
\node at (lm5) {$M_5$};

\node at (ln1) {$N_1$};
\node at (ln2) {$N_2$};
\node at (ln3) {$N_3$};
\end{tiny}

\draw[thick] plot [smooth cycle] coordinates {(m11) (m112) (m12) (m13) (m14) (m15)};

\draw[thick] plot [smooth cycle] coordinates {(m21) (m22) (m23) (m24)};

\draw[thick] plot [smooth cycle] coordinates {(m31) (m32) (m33) (m34)};

\draw[thick] plot [smooth cycle] coordinates {(m41) (m42) (m43) (m44)};

\draw[thick] plot [smooth cycle] coordinates {(m51) (m52) (m53) (m54)};

\draw (0,0) circle (2cm);

\draw[very thick,red,|-|] ([shift=(10:2cm)]0,0) arc (10:75:2cm);

\draw[very thick,red,|-|] ([shift=(85:2cm)]0,0) arc (85:95:2cm);

\draw[very thick,red,|-|] ([shift=(120:2cm)]0,0) arc (120:240:2cm);

\draw[very thick,red,|-|] ([shift=(265:2cm)]0,0) arc (265:275:2cm);

\draw[very thick,red,|-|] ([shift=(290:2cm)]0,0) arc (290:340:2cm);

\draw[very thick,red,|-|] ([shift=(-5:2cm)]0,0) arc (-5:5:2cm);

\draw[white] (-3.4,-2.5)--(-3.4,-2.3);
\draw[white] (3.8,2.5)--(3.8,2.3);

\end{tikzpicture} 
\hspace{0.5cm}
\begin{tikzpicture}[scale=1,>=latex,shorten >=-0.4pt,shorten <=-0.4pt]
\coordinate (label) at (0,-3) {};

\coordinate (m11) at ($(center)+(90:1.95cm)$) {};
\coordinate (m112) at ($(center)+(180:0.5cm)$) {};
\coordinate (m12) at ($(center)+(270:1.98cm)$) {};
\coordinate (m13) at ($(center)+(330:0.5cm)$) {};
\coordinate (m14) at ($(center)+(0:2cm)$) {};
\coordinate (m15) at ($(center)+(30:0.5cm)$) {};
\coordinate (lm1) at ($(center)+(90:0.0cm)$) {};
\coordinate (phi_m1) at ($(center)+(225:2.07cm)$) {};

\coordinate (m21) at ($(center)+(10:2cm)$) {};
\coordinate (m22) at ($(center)+(20:1.8cm)$) {};
\coordinate (m23) at ($(center)+(30:2cm)$) {};
\coordinate (m24) at ($(center)+(20:1.6cm)$) {};
\coordinate (lm2) at ($(center)+(10:1.5cm)$) {};
\coordinate (phi_m2) at ($(center)+(20:2.05cm)$) {};

\coordinate (m31) at ($(center)+(35:2cm)$) {};
\coordinate (m32) at ($(center)+(55:1.5cm)$) {};
\coordinate (m33) at ($(center)+(75:2cm)$) {};
\coordinate (m34) at ($(center)+(55:1cm)$) {};
\coordinate (lm3) at ($(center)+(55:1.8cm)$) {};
\coordinate (phi_m3) at ($(center)+(55:2.1cm)$) {};

\coordinate (m41) at ($(center)+(120:2cm)$) {};
\coordinate (m42) at ($(center)+(0:-1.2cm)$) {};
\coordinate (m43) at ($(center)+(240:2cm)$) {};
\coordinate (m44) at ($(center)+(0:-0.8cm)$) {};
\coordinate (lm4) at ($(center)+(180:1.5cm)$) {};

\coordinate (m51) at ($(center)+(290:2cm)$) {};
\coordinate (m52) at ($(center)+(315:1.7cm)$) {};
\coordinate (m53) at ($(center)+(340:2cm)$) {};
\coordinate (m54) at ($(center)+(315:1.3cm)$) {};
\coordinate (lm5) at ($(center)+(285:1.7cm)$) {};

\coordinate (ln1) at ($(center)+(26:1.1cm)$) {};
\coordinate (ln2) at ($(center)+(112:1.56cm)$) {};
\coordinate (ln3) at ($(center)+(330:1.0cm)$) {};

\begin{tiny}
\tikzstyle{every node}=[inner sep=1pt]
\node at (lm1) {$M_1$};
\node[anchor=east] at (phi_m1) {$\phi|V_{T\setminus N_1}(M_1)$};
\node at (lm2) {$M_2$};
\node[anchor=west] at (phi_m2) {$\phi|V_{T\setminus N_1}(M_2)$};

\node at (lm3) {$M_3$};
\node[anchor=west] at (phi_m3) {$\phi|V_{T\setminus N_1}(M_3)$};

\node at (lm4) {$M_4$};
\node at (lm5) {$M_5$};

\node at (ln1) {$N_1$};
\node at (ln2) {$N_2$};
\node at (ln3) {$N_3$};
\end{tiny}

\draw[thick] plot [smooth cycle] coordinates {(m11) (m112) (m12) (m13) (m14) (m15)};

\draw[thick] plot [smooth cycle] coordinates {(m21) (m22) (m23) (m24)};

\draw[thick] plot [smooth cycle] coordinates {(m31) (m32) (m33) (m34)};

\draw[thick] plot [smooth cycle] coordinates {(m41) (m42) (m43) (m44)};

\draw[thick] plot [smooth cycle] coordinates {(m51) (m52) (m53) (m54)};

\draw (0,0) circle (2cm);

\draw[very thick,red,|-|] ([shift=(10:2cm)]0,0) arc (10:30:2cm);
\draw[very thick,red,|-|] ([shift=(35:2cm)]0,0) arc (35:75:2cm);
\draw[very thick,red,|-|] ([shift=(88:2cm)]0,0) arc (88:362:2cm);

\draw[white] (-3.4,-2.5)--(-3.4,-2.3);
\draw[white] (3.8,2.5)--(3.8,2.3);

\end{tikzpicture} 
\caption{\label{fig:T_NM_tree_subwords} 
A normalized model $\phi$ and contiguous subwords: $\phi|V_{T \setminus M_1}(N_1)$, $\phi|V_{T \setminus M_1}(N_2)$, and $\phi|V_{T \setminus M_1}(N_3)$ (to the left), and $\phi|V_{T \setminus N_1}(M_1)$, $\phi|V_{T \setminus N_1}(M_2)$, and $\phi|V_{T \setminus N_1}(M_3)$ (to the right). The circular word $\phi|(M_1 \cup N_T[M_1]) \equiv \tau_1N_1\tau_2N_3\tau_3N_2$ and the circular word $\phi|N_{T}[N_1] \equiv M_1M_3M_2$.}
\end{figure}
Note that $\phi|N_T[N]$ is a circular permutation of the modules from the set $N_T[N]$, for every node $N$ in $T_{NM}$.
Now, our goal is to describe the circular words $\phi|(M \cup N_T[M])$ that arise from conformal models $\phi$ for all the modules $M$ in $T_{NM}$.
To accomplish our task, we consider two cases: $M$ is prime and $M$ is serial.

\subsubsection{$M$ is prime}
Fix a prime module $M$ in $T_{NM}$.
Suppose $\phi$ is a conformal model of $G_{ov}$.
Denote by $\phi'$ the circular word $\phi|(M \cup N_T[M])$.
Clearly, $\phi|M$ is a conformal model of $(M,{\sim})$.
Let $S$ be a slot in $\pi(\phi|M)$.
Let $\tau_{\phi'}(S)$ be the smallest contiguous subword of $\phi'$ containing 
all the letters from $S$ and containing no letter from other slots of $\pi(\phi|M)$.
Clearly, $\tau_{\phi'}(S)|S$ is a permutation of $S$.
From Claim \ref{claim:conformal_models_properties_modules_nodes_T_NM}.\eqref{item:conformal_models_properties_modules_T_NM} we deduce that $\tau_{\phi'}$ has the form
$$\tau_{\phi'}(S) = \tau^{1}_{\phi'}\ N_{\phi'}^{1}\ \tau_{\phi'}^{2}\ \ldots\ \tau_{\phi'}^{l-1}\ N_{\phi'}^{(l-1)}\ \tau_{\phi'}^{l} \text{ for some $l \geq 1$},$$
where $N_{\phi'}^1, \ldots, N_{\phi'}^{(l-1)}$ is a sequence of all nodes from $N_T[M]$ occurring in $\tau_{\phi'}(S)$ (possibly empty)
and $\tau_{\phi'}^{1}, \ldots, \tau_{\phi'}^{l}$ are non-empty words satisfying $\tau_{\phi'}^{1}\cdot \ldots \cdot \tau_{\phi'}^{l} = \tau_{\phi'}(S)|S$.
Next, let
$$p_{\phi'}(S) = (S^{1}_{\phi'},N^1_{\phi'},S^{2}_{\phi'}, \ldots, S_{\phi'}^{l-1}, N_{\phi'}^{(l-1)}, S_{\phi'}^{l}),$$
where $S_{\phi'}^{i}$ is the set containing all labeled letters from the word $\tau_{\phi'}^{i}$ for $i \in [l]$.
In particular, note that $(S_{\phi'}^{1}, \ldots, S_{\phi'}^{l})$ is an ordered partition of the slot $S$ --
see Figure \ref{fig:slot_pattern} for an illustration.

\begin{figure}[!htp]

\begin{tikzpicture}[xscale=4,yscale=0.8,>=latex,shorten >=-0.4pt,shorten <=-0.4pt]
    \coordinate (center) at (-2,0) {};
    
    \coordinate (u1) at ($(center)+(110:2)$) {};
    \coordinate (lu1) at ($(center)+(110:2.1)$) {};
    \coordinate (u2) at ($(center)+(105:2)$) {};
    \coordinate (lu2) at ($(center)+(105:2.1)$) {};
    \coordinate (u3) at ($(center)+(100:2)$) {};
    \coordinate (lu3) at ($(center)+(100:2.1)$) {};
    \coordinate (u4) at ($(center)+(95:2)$) {};
    \coordinate (lu4) at ($(center)+(95:2.1)$) {};
    \coordinate (u5) at ($(center)+(90:2)$) {};
    \coordinate (lu5) at ($(center)+(90:2.1)$) {};
    \coordinate (u6) at ($(center)+(85:2)$) {};
    \coordinate (lu6) at ($(center)+(85:2.1)$) {};
    \coordinate (u7) at ($(center)+(80:2)$) {};
    \coordinate (lu7) at ($(center)+(80:2.1)$) {};
    \coordinate (u8) at ($(center)+(75:2)$) {};
    \coordinate (lu8) at ($(center)+(75:2.1)$) {};
    \coordinate (u9) at ($(center)+(70:2)$) {};
    \coordinate (lu9) at ($(center)+(70:2.1)$) {};
 
    \coordinate (b9) at ($(center)+(290:2)$) {};
    \coordinate (lb9) at ($(center)+(290:2.1)$) {};
    \coordinate (b8) at ($(center)+(285:2)$) {};
    \coordinate (lb8) at ($(center)+(285:2.1)$) {};
    \coordinate (b7) at ($(center)+(280:2)$) {};
    \coordinate (lb7) at ($(center)+(280:2.1)$) {};
    \coordinate (b6) at ($(center)+(275:2)$) {};
    \coordinate (lb6) at ($(center)+(275:2.1)$) {};
    \coordinate (b5) at ($(center)+(270:2)$) {};
    \coordinate (lb5) at ($(center)+(270:2.1)$) {};
    \coordinate (b4) at ($(center)+(265:2)$) {};
    \coordinate (lb4) at ($(center)+(265:2.1)$) {};
    \coordinate (b3) at ($(center)+(260:2)$) {};
    \coordinate (lb3) at ($(center)+(260:2.1)$) {};
    \coordinate (b2) at ($(center)+(255:2)$) {};
    \coordinate (lb2) at ($(center)+(255:2.1)$) {};
    \coordinate (b1) at ($(center)+(250:2)$) {};
    \coordinate (lb1) at ($(center)+(250:2.1)$) {};

   \coordinate (lS) at ($(center)+(118:2)$) {};
   \coordinate (lS') at ($(center)+(242:2)$) {};
    
    
    \draw ($(center) + (65:2)$) arc (65:115:2);    
    \draw ($(center) + (245:2)$) arc (245:295:2);    

    \draw[line width=1mm] ($(center) + (98:2)$) arc (98:102:2);    
    \draw[line width=1mm] ($(center) + (83:2)$) arc (83:87:2);    

    \draw[line width=1mm] ($(center) + (258:2)$) arc (258:262:2);    
    \draw[line width=1mm] ($(center) + (283:2)$) arc (283:287:2);    
    
    \tikzstyle{every node}=[inner sep=1pt]
    \node at (lS) {$S$};
    \node at (lS') {$S'$};
    \begin{footnotesize}
    \node[anchor=south] at (lu1) {$v^1_1$};
    \node[anchor=south] at (lu2) {$v^0_2$};
    \node[anchor=south] at (lu3) {$N_1$};
    \node[anchor=south] at (lu4) {$v^1_3$};
    \node[anchor=south] at (lu5) {$v^1_4$};
    \node[anchor=south] at (lu6) {$N_2$};
    \node[anchor=south] at (lu7) {$v^0_5$};
    \node[anchor=south] at (lu8) {$v^0_6$};
    \node[anchor=south] at (lu9) {$v^1_7$};

    \node[anchor=north] at (lb1) {$v^0_3$};
    \node[anchor=north] at (lb2) {$v^0_1$};
    \node[anchor=north] at (lb3) {$N_4$};
    \node[anchor=north] at (lb4) {$v^1_2$};
    \node[anchor=north] at (lb5) {$v^1_6$};
    \node[anchor=north] at (lb6) {$v^0_7$};
    \node[anchor=north] at (lb7) {$v^0_4$};
    \node[anchor=north] at (lb8) {$N_3$};
    \node[anchor=north] at (lb9) {$v^1_5$};
    \end{footnotesize}
    
    \draw[thick,<-] (u1) -- (b2);
    \draw[thick,->] (u2) -- (b4);
    \draw[thick,<-] (u4) -- (b1);
    \draw[thick,<-] (u5) -- (b7);
    \draw[thick,->] (u7) -- (b9);
    \draw[thick,->] (u8) -- (b5);
    \draw[thick,<-] (u9) -- (b6);

\end{tikzpicture}
\caption{\label{fig:slot_pattern} Pattern of the slot $S$: 
$(\{v^1_1,v^0_2\},N_1,\{v^1_3,v^1_4\},N_2,\{v^0_5,v^0_6,v^1_7\})$.
Pattern of the slot $S'$: $(\{v^1_5\},N_3,\{v^0_4,v^0_7,v^1_6,v^1_2\},N_4,\{v^0_1,v^0_3\})$.}
\end{figure}
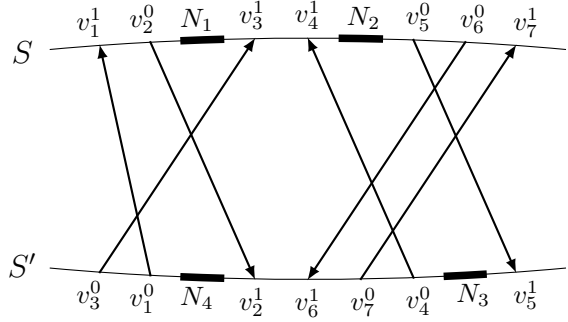

Now, let $(S,T)$ be two consecutive slots in $\pi_{\phi'|M}(M)$.
Denote by $p_{\phi'}(S,T)$ the set of all nodes from $N_T[M]$ that appear between $\tau_{\phi'}(S)$
and $\tau_{\phi'}(T)$ in the circular word $\phi'$.
Claim \ref{claim:conformal_models_properties_modules_nodes_T_NM}.\eqref{item:conformal_models_properties_modules_T_NM} proves that the set $p_{\phi'}(S,T)$ is either empty or contains exactly one node.
Similarly to the previous sections, it turns out that $p_{\phi'}(S)$ and $p_{\phi'}(S,T)$ 
do not depend on the choice of $\phi$
provided we choose $\phi$ from the set of conformal models 
admitting the same circular order of the slots.
\begin{claim}
\label{claim:ivariants_for_prime_child_of_V}
Let $m \in \{0,1\}$ and let $M$ be a prime module in $T_{NM}$.
\begin{enumerate}
\item For every slot $S$ in $\pi_m(M)$ there exists a sequence 
$$p_m(S) = (S^{1}_{m},N^1_{m},S^{2}_{m}, \ldots, S_{m}^{l-1}, N_{m}^{l-1}, S_{m}^{l})$$
where $(S^{1}_{m}, \ldots, S^{l}_m)$ is an ordered partition of $S$ and $N_m^1,\ldots,N_m^{l-1}$ are nodes from $N_T[M]$ such that $$p_m(S) = p_{\phi|(M\cup N_T[M])}(S)$$ 
for every conformal model of $G_{ov}$ such that $\pi(\phi|M) = \pi_{m}(M)$.
\item For every two consecutive slots $(S,T)$ in $\pi_{m}(M)$ there exists a set $p_m(S,T) \subset N_T[M]$ 
such that $$p_{m}(S,T) = p_{\phi|(M \cup N_T[M])}(S,T)$$ 
for every conformal model $\phi$ of $G_{ov}$ such that $\pi(\phi|M) = \pi_{m}(M)$.
\end{enumerate}
\end{claim}
\begin{proof}
The algorithm computes $p_m(S)$ and $p_m(S,T)$ as follows.
It starts with the circular orientation $\pi_m(M)$.
Then, for every node $N \in N_T[M]$
it finds a slot $S$ in $\pi_m(M)$ or two consecutive slots $(S,T)$ in $\pi_m(M)$ 
where the node $N$ must be placed, 
basing on whether $V_{T \setminus M}(N)$ is on the left or right side of $v$, for every $v \in M$.
Clearly, when the algorithm inserts $N$ to $S$,
the pattern of $S$ needs to be updated accordingly.

From the perspective of the slot $S$, the algorithm works as follows. 
Let $K$ be a consistent submodule of $M$ associated with the slot $S$.
Let $v$ be a vertex in $M \setminus K$ such that $v \sim K$; 
such a vertex exists as $M$ is prime.
Let $L$ be a consistent submodule of $M$ containing $v$.
Assume that $v$ is oriented from $L^0$ to $L^1$ and $S$ appears 
between $L^{0}$ and $L^{1}$ in $\pi_m(M)$ (for the other case the algorithm works similarly).
The algorithm starts with $p_m(S) = (S)$.
Then, for every node $N \in N_T[M]$ such that $V_{T \setminus M}(N) \subset \leftside(v)$
the algorithm computes the sets $S_1$ and $S_2$, where
$$
S_1 = 
\begin{array}{l}
\{s^{0} \in S: \text{$N$ is on the left side of $s$} \}\ \cup \\
\{s^{1} \in S: \text{$N$ is on the right side of $s$} \}
\end{array}
$$
and 
$$
S_2 =  S \setminus S_1.
$$
If $S_1 = \emptyset$ or $S_2 = \emptyset$, the algorithm does nothing.
Otherwise, the algorithm refines $p_m(S)$ by $(S_1,N,S_2)$, which means
that it finds the unique set $S' \in p_m(S)$ such that $S' \cap S_1 \neq \emptyset$ and 
$S' \cap S_2 \neq \emptyset$ and then it replaces $S'$ in $p_m(S)$ by the triple
$(S' \cap S_1, N, S' \cap S_2)$.
\end{proof}
Given a slot $S$ of $\pi_m(M)$, 
$p_m(S)$ is called the \emph{pattern} of the slot $S$ in $\pi_m(M)$,
the sequence $(N_1,\ldots,N_{l-1})$ is called the \emph{sequence of the nodes} in the slot $S$ in $\pi_m(M)$,
and $(S_1,\ldots,S_l)$ is called the \emph{ordered partition} of the slot $S$ in $\pi_m(M)$.
For every two consecutive slots $(S,T)$ in $\pi_m(M)$, 
$p_m(S,T)$ is called the \emph{set of nodes} between $S$ and $T$ in $\pi_{m}(M)$.

Suppose $K_1,\ldots,K_n$ is a consistent decomposition of $M$ and
$\mathbb{K}_i = (K^0_i,K^1_i,<_{K_i})$ is the metaedge of $K_i$ for $i \in [n]$.
With respect to the metaedge $\mathbb{K}_i$, we naturally divide the nodes from $N_T[M]$ into three categories:
\begin{itemize}
 \item \emph{$N$ is on the left side of $\mathbb{K}_i$}, written $N \in \leftside(\mathbb{K}_i)$, if $N$ is on the left side of any $v \in K_i$
 oriented from $K^0_i$ to $K^1_i$ and $N$ is on the right side of any $v \in K_i$
 oriented from $K^1_i$ to $K^0_i$,
 \item \emph{$N$ is on the right side of $\mathbb{K}_i$}, written $N \in \rightside(\mathbb{K}_i)$, if $N$ is on the right side of any $v \in K_i$ oriented from $K^0_i$ to $K^1_i$ and $N$ is on the left side of any $v \in K_i$
 oriented from $K^1_i$ to $K^0_i$ 
 \item \emph{$N$ is inside $\mathbb{K}_i$}, written $N \in \inside(\mathbb{K}_i)$, if neither of the previous conditions hold.
\end{itemize}
In particular, the set $\inside(\mathbb{K}_i)$ contains exactly the nodes that appear in the patterns
$p_0(K^0_i)$ and $p_0(K^1_i)$ as well as in the patterns $p_1(K^0_i)$ or $p_1(K^1_i)$.
Indeed, note that for every $N \in N_T[M]$, $N$ appears in $p_0(K^0_i)$ iff $N$ appears in $p_1(K^1_j)$
and $N$ appears in $p_0(K^1_j)$ iff $N$ appears in $p_1(K^0_i)$
-- see Figure~\ref{fig:slot_pattern_reflection} demonstrating the changes in the patterns after the reflection of a normalized model $\phi$.

\begin{figure}[!htp]
\begin{tikzpicture}[xscale=4,yscale=0.8,>=latex,shorten >=-0.4pt,shorten <=-0.4pt]
    \coordinate (center) at (-2,0) {};
    
    \coordinate (u1) at ($(center)+(110:2)$) {};
    \coordinate (lu1) at ($(center)+(110:2.1)$) {};
    \coordinate (u2) at ($(center)+(105:2)$) {};
    \coordinate (lu2) at ($(center)+(105:2.1)$) {};
    \coordinate (u3) at ($(center)+(100:2)$) {};
    \coordinate (lu3) at ($(center)+(100:2.1)$) {};
    \coordinate (u4) at ($(center)+(95:2)$) {};
    \coordinate (lu4) at ($(center)+(95:2.1)$) {};
    \coordinate (u5) at ($(center)+(90:2)$) {};
    \coordinate (lu5) at ($(center)+(90:2.1)$) {};
    \coordinate (u6) at ($(center)+(85:2)$) {};
    \coordinate (lu6) at ($(center)+(85:2.1)$) {};
    \coordinate (u7) at ($(center)+(80:2)$) {};
    \coordinate (lu7) at ($(center)+(80:2.1)$) {};
    \coordinate (u8) at ($(center)+(75:2)$) {};
    \coordinate (lu8) at ($(center)+(75:2.1)$) {};
    \coordinate (u9) at ($(center)+(70:2)$) {};
    \coordinate (lu9) at ($(center)+(70:2.1)$) {};
 
    \coordinate (b9) at ($(center)+(290:2)$) {};
    \coordinate (lb9) at ($(center)+(290:2.1)$) {};
    \coordinate (b8) at ($(center)+(285:2)$) {};
    \coordinate (lb8) at ($(center)+(285:2.1)$) {};
    \coordinate (b7) at ($(center)+(280:2)$) {};
    \coordinate (lb7) at ($(center)+(280:2.1)$) {};
    \coordinate (b6) at ($(center)+(275:2)$) {};
    \coordinate (lb6) at ($(center)+(275:2.1)$) {};
    \coordinate (b5) at ($(center)+(270:2)$) {};
    \coordinate (lb5) at ($(center)+(270:2.1)$) {};
    \coordinate (b4) at ($(center)+(265:2)$) {};
    \coordinate (lb4) at ($(center)+(265:2.1)$) {};
    \coordinate (b3) at ($(center)+(260:2)$) {};
    \coordinate (lb3) at ($(center)+(260:2.1)$) {};
    \coordinate (b2) at ($(center)+(255:2)$) {};
    \coordinate (lb2) at ($(center)+(255:2.1)$) {};
    \coordinate (b1) at ($(center)+(250:2)$) {};
    \coordinate (lb1) at ($(center)+(250:2.1)$) {};

   \coordinate (lS) at ($(center)+(118:2)$) {};
   \coordinate (lS') at ($(center)+(242:2)$) {};
    
    
    \draw ($(center) + (65:2)$) arc (65:115:2);    
    \draw ($(center) + (245:2)$) arc (245:295:2);    

    \draw[line width=1mm] ($(center) + (98:2)$) arc (98:102:2);    
    \draw[line width=1mm] ($(center) + (83:2)$) arc (83:87:2);    

    \draw[line width=1mm] ($(center) + (258:2)$) arc (258:262:2);    
    \draw[line width=1mm] ($(center) + (283:2)$) arc (283:287:2);    
    \draw[thick,->] ($(center) + (180:0.82)$)--($(center) + (0:0.8)$);    
    
    \tikzstyle{every node}=[inner sep=1pt]
    \node at (lS) {$K^1_i$};
    \node at (lS') {$K^0_i$};
    \begin{footnotesize}
    \node[anchor=south] at (lu1) {$v^1_1$};
    \node[anchor=south] at (lu2) {$v^0_2$};
    \node[anchor=south] at (lu3) {$N_1$};
    \node[anchor=south] at (lu4) {$v^1_3$};
    \node[anchor=south] at (lu5) {$v^1_4$};
    \node[anchor=south] at (lu6) {$N_2$};
    \node[anchor=south] at (lu7) {$v^0_5$};
    \node[anchor=south] at (lu8) {$v^0_6$};
    \node[anchor=south] at (lu9) {$v^1_7$};

    \node[anchor=north] at (lb1) {$v^0_3$};
    \node[anchor=north] at (lb2) {$v^0_1$};
    \node[anchor=north] at (lb3) {$N_4$};
    \node[anchor=north] at (lb4) {$v^1_2$};
    \node[anchor=north] at (lb5) {$v^1_6$};
    \node[anchor=north] at (lb6) {$v^0_7$};
    \node[anchor=north] at (lb7) {$v^0_4$};
    \node[anchor=north] at (lb8) {$N_3$};
    \node[anchor=north] at (lb9) {$v^1_5$};
    \end{footnotesize}
    
    \draw[thick,<-] (u1) -- (b2);
    \draw[thick,->] (u2) -- (b4);
    \draw[thick,<-] (u4) -- (b1);
    \draw[thick,<-] (u5) -- (b7);
    \draw[thick,->] (u7) -- (b9);
    \draw[thick,->] (u8) -- (b5);
    \draw[thick,<-] (u9) -- (b6);

    \draw (-1,-2.2) -- (-1,2.2);    
\end{tikzpicture}
\hspace{0.1cm}
\begin{tikzpicture}[xscale=4,yscale=0.8,>=latex,shorten >=-0.4pt,shorten <=-0.4pt]
    \coordinate (center) at (-2,0) {};
    
    \coordinate (u9) at ($(center)+(110:2)$) {};
    \coordinate (lu9) at ($(center)+(110:2.1)$) {};
    \coordinate (u8) at ($(center)+(105:2)$) {};
    \coordinate (lu8) at ($(center)+(105:2.1)$) {};
    \coordinate (u7) at ($(center)+(100:2)$) {};
    \coordinate (lu7) at ($(center)+(100:2.1)$) {};
    \coordinate (u6) at ($(center)+(95:2)$) {};
    \coordinate (lu6) at ($(center)+(95:2.1)$) {};
    \coordinate (u5) at ($(center)+(90:2)$) {};
    \coordinate (lu5) at ($(center)+(90:2.1)$) {};
    \coordinate (u4) at ($(center)+(85:2)$) {};
    \coordinate (lu4) at ($(center)+(85:2.1)$) {};
    \coordinate (u3) at ($(center)+(80:2)$) {};
    \coordinate (lu3) at ($(center)+(80:2.1)$) {};
    \coordinate (u2) at ($(center)+(75:2)$) {};
    \coordinate (lu2) at ($(center)+(75:2.1)$) {};
    \coordinate (u1) at ($(center)+(70:2)$) {};
    \coordinate (lu1) at ($(center)+(70:2.1)$) {};
 
    \coordinate (b1) at ($(center)+(290:2)$) {};
    \coordinate (lb1) at ($(center)+(290:2.1)$) {};
    \coordinate (b2) at ($(center)+(285:2)$) {};
    \coordinate (lb2) at ($(center)+(285:2.1)$) {};
    \coordinate (b3) at ($(center)+(280:2)$) {};
    \coordinate (lb3) at ($(center)+(280:2.1)$) {};
    \coordinate (b4) at ($(center)+(275:2)$) {};
    \coordinate (lb4) at ($(center)+(275:2.1)$) {};
    \coordinate (b5) at ($(center)+(270:2)$) {};
    \coordinate (lb5) at ($(center)+(270:2.1)$) {};
    \coordinate (b6) at ($(center)+(265:2)$) {};
    \coordinate (lb6) at ($(center)+(265:2.1)$) {};
    \coordinate (b7) at ($(center)+(260:2)$) {};
    \coordinate (lb7) at ($(center)+(260:2.1)$) {};
    \coordinate (b8) at ($(center)+(255:2)$) {};
    \coordinate (lb8) at ($(center)+(255:2.1)$) {};
    \coordinate (b9) at ($(center)+(250:2)$) {};
    \coordinate (lb9) at ($(center)+(250:2.1)$) {};

   \coordinate (lS) at ($(center)+(118:2)$) {};
   \coordinate (lS') at ($(center)+(242:2)$) {};
    
    
    \draw ($(center) + (65:2)$) arc (65:115:2);    
    \draw ($(center) + (245:2)$) arc (245:295:2);    

    \draw[line width=1mm] ($(center) + (93:2)$) arc (93:97:2);    
    \draw[line width=1mm] ($(center) + (78:2)$) arc (78:82:2);    

    \draw[line width=1mm] ($(center) + (253:2)$) arc (253:257:2);    
    \draw[line width=1mm] ($(center) + (278:2)$) arc (278:282:2);    

    \draw[thick,->] ($(center) + (180:0.82)$)--($(center) + (0:0.82)$);    
    
    \tikzstyle{every node}=[inner sep=1pt]
    \node at (lS) {$K^0_i$};
    \node at (lS') {$K^1_i$};
    \begin{footnotesize}
    \node[anchor=south] at (lu1) {$v^0_1$};
    \node[anchor=south] at (lu2) {$v^1_2$};
    \node[anchor=south] at (lu3) {$N_1$};
    \node[anchor=south] at (lu4) {$v^0_3$};
    \node[anchor=south] at (lu5) {$v^0_4$};
    \node[anchor=south] at (lu6) {$N_2$};
    \node[anchor=south] at (lu7) {$v^1_5$};
    \node[anchor=south] at (lu8) {$v^1_6$};
    \node[anchor=south] at (lu9) {$v^0_7$};

    \node[anchor=north] at (lb1) {$v^1_3$};
    \node[anchor=north] at (lb2) {$v^1_1$};
    \node[anchor=north] at (lb3) {$N_4$};
    \node[anchor=north] at (lb4) {$v^0_2$};
    \node[anchor=north] at (lb5) {$v^0_6$};
    \node[anchor=north] at (lb6) {$v^1_7$};
    \node[anchor=north] at (lb7) {$v^1_4$};
    \node[anchor=north] at (lb8) {$N_3$};
    \node[anchor=north] at (lb9) {$v^0_5$};
    \end{footnotesize}
    
    \draw[thick,->] (u1) -- (b2);
    \draw[thick,<-] (u2) -- (b4);
    \draw[thick,->] (u4) -- (b1);
    \draw[thick,->] (u5) -- (b7);
    \draw[thick,<-] (u7) -- (b9);
    \draw[thick,<-] (u8) -- (b5);
    \draw[thick,->] (u9) -- (b6);

\end{tikzpicture}

\caption{\label{fig:slot_pattern_reflection} 
The restriction of a conformal model $\phi$ to the set $(K_i \cup N_{T}(K_i))$ and the restriction of $\phi^R$ to the set $(K_i \cup N_{T}(K_i))$.
Assuming $\pi(\phi|M) = \pi_0(M)$ and $\pi(\phi^R|M) = \pi_1(M)$, the pair $(v^1_1v_2^0N_1v_3^1v_4^1N_2v_5^0v_6^0v_7^1, v_5^1N_3v_4^0v_7^0v_6^1v_2^1N_4v_1^0v_3^0)$ is an admissible model for $\mathbb{K}_{i,0}$, the pair $(v_7^0v_6^1v_5^1N_2v_4^0v_3^0N_1v_2^1v_1^0, v_3^1v_1^1N_4v_2^0v_6^0v_7^1v_4^1N_3v_5^0)$ is an admissible model for $\mathbb{K}_{i,1}$.}
\end{figure}
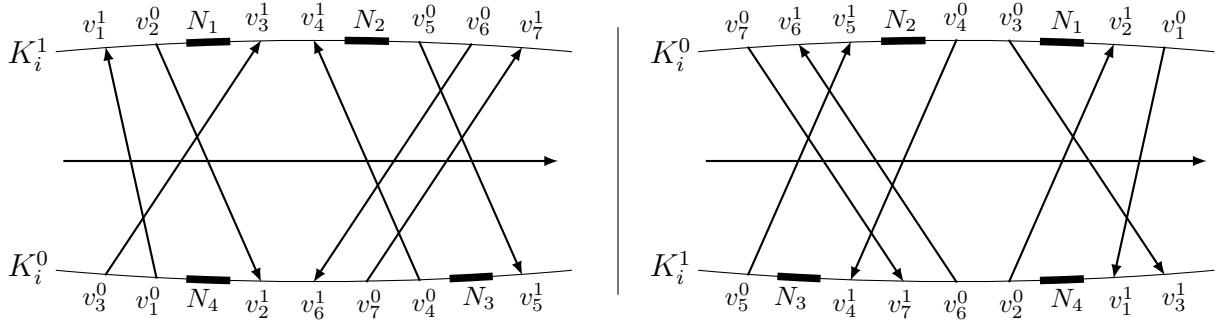

Now, we enrich the metaedge $(K^0_i,K^1_i,{<_{K_i}})$ in $\pi_m(M)$ 
by the patterns $p_m(K^{0}_i)$ and $p_m(K^1_i)$, obtaining two \emph{extended metaedges} of $K_i$,
$\mathbb{K}_{i,0}$ and $\mathbb{K}_{i,1}$, where
$$\mathbb{K}_{i,m} = (K^{0}_i, K^{1}_i,{<_{K_{i}}},p_m(K^{0}_i),p_m(K^1_i)) \text{ for every $m \in \{0,1\}$}.$$
Further, we extend the notion of an admissible model to the extended metaedges of $K_i$.
As before, the definition given below is formulated so as the restriction of $\phi|(M \cup N_T(M))$ to the slots 
$K^0_i$ and $K^1_i$ forms an admissible model for the extended $\mathbb{K}_{i,m}$ provided $\phi$ is conformal model such that $\pi_m(\phi|M) = \pi_{m}(M)$ -- see Figure~\ref{fig:slot_pattern_reflection}.

\begin{definition}
\label{def:extended_admissible_model_for_consisten_submodules}
Let $m \in \{0,1\}$, let $K_i$ be a consistent submodule of $M$, and let $\mathbb{K}_{i,m} = (K^0_i,K^1_i,{<_{K_i}},p_m(K^0_i),p_m(K^1_i))$ be the extended metaedge of $K_i$. 
A pair $(\tau^0,\tau^1)$, where $\tau^0$ and $\tau^1$ are words containing all the labeled letters from $K^0_i \cup K^1_i$ and all the nodes from $\inside(K_i)$, forms an \emph{extended admissible model for the extended metaedge $\mathbb{K}_{i,m}$} if 
\begin{itemize}
 \item $(\tau^0|K^0_i,\tau^1|K^1_i)$ is an admissible model for $(K^0_i,K^1_i,{<_{K_i}})$,
 \item the order of the nodes in $\tau^j$ equals to the order of the nodes in $p_m(K^j_i)$, 
 for every $j \in \{0,1\}$,
 \item the ordered partition of $K^j_i$ arisen from $\tau^j$ equals to the ordered partition of $K^j_i$ in $p_m(K^j_i)$, for every $j \in \{0,1\}$.
\end{itemize}
\end{definition}
\begin{claim}
Let $M$ be a prime module in $T_{NM}$.
Let $\phi$ be a conformal model of $G_{ov}$ such that $\pi(\phi|M) \equiv \pi_m(M)$ and let $\phi' \equiv \phi|(M \cup N_T[M])$.
Then:
\begin{itemize}
 \item $(\tau_{\phi'}(K^0_i), \tau_{\phi'}(K^1_i))$ is an extended admissible model for $\mathbb{K}_m$,
 \item $(\tau_{\phi'}(K^1_i), \tau_{\phi'}(K^0_i))$ is an extended admissible model for $\mathbb{K}_m$.
\end{itemize}
\end{claim}
In the next definition we use the same notation for circular words on
$M^{*} \cup N_T[M]$ as we have introduced in the beginning of this subsection for circular words $\phi' \equiv M^* \cup N_T(M)$.
\begin{definition}
Let $M$ be a prime module in $T_{NM}$ and let $m \in \{0,1\}$.
A circular word $\phi'$ on the set of letters $M^{*} \cup N_T[M]$ is
\emph{an extended admissible model for the order of the slots $\pi_m(M)$} if 
\begin{itemize}
 \item $\phi'|M$ is an admissible model of $(M,{\sim})$ such that $\pi(\phi'|M) = \pi_m(M)$,
 \item for every slot $K_i$ in $\pi_m(M)$ the pair $(\tau_{\phi'}(K^0_i), \tau_{\phi'}(K^1_i))$ is an admissible
 model for~$\mathbb{K}_{i,m}$,
 \item for every two consecutive slots $(S,T)$ in $\pi_{m}(M)$,
 $\pi_{\phi'}(S,T) = p_m(S,T)$.
\end{itemize}
A circular word $\phi'$ on the set of letters $M^{*} \cup N_T[M]$ is
\emph{an extended admissible model for $(M,{\sim})$} if $\phi'$ is extended admissible model for $\pi_m(M)$ for some $m \in \{0,1\}$.
\end{definition}
Again, the previous definition was formulated so as the following claim holds.
\begin{claim}
\label{claim:conformal_models_are_admissible}
If $\phi$ is a conformal model of $G_{ov}$ such that $\pi(\phi|M) = \pi_m(M))$ for some prime module $M$ in $T_{NM}$, 
then $\phi|(M \cup N_T[M])$ is an extended admissible model for $\pi_m(M)$.
\end{claim}
Note that the reflection of an extended admissible model for $\pi_0(M)$ is an 
extended admissible model for~$\pi_1(M)$.

\subsubsection{$M$ is serial}
Suppose that $M$ is a serial module in $T_{NM}$.
Suppose $M_1,\ldots,M_k$ are the children of $M$ in $\mathcal{M}(G_{ov})$.
Clearly, every $M_i$ is a proper prime or a proper parallel module in $\mathcal{M}(G_{ov})$.
Suppose for every $i \in [k]$ a representant of $M_i$ is fixed,
and hence the metaedge $\mathbb{M}_i = (M^0_i,M^1_i,{<_{M_i}})$ is defined.

It turns out that for serial modules in $T_{NM}$ 
we can provide a similar descriptions of the restrictions of conformal models $\phi$ to the set 
$M \cup N_T(M)$ as we have obtained for prime ones.
To have such a description,
we need to partition $M$ into \emph{consistent submodules} of $M$.
If $\inside(\mathbb{M}_i) \neq \emptyset$, 
$M_i$ is a consistent submodule of $M$.
In the set of the remaining vertices in $M$, that is, 
in the set $\bigcup \{M_i: i \in [k] \text{ and } \inside(\mathbb{M}_i) = \emptyset\}$ 
we introduce an equivalence relation $K$ defined such that
$$
\begin{array}{ccc}
u K  v\ & \iff 
\begin{array}{l}
\{\leftside(u) \cap (V \setminus M), \rightside(u) \cap (V \setminus M) \} = \\
\{\leftside(v) \cap (V \setminus M), \rightside(v) \cap (V \setminus M) \}.
\end{array}
\end{array}
$$
The equivalence classes of $K$ are the remaining consistent submodules of $M$.
Note that every consistent submodule of $M$ is the union of some children of $M$.
The set of all consistent submodules of $M$ forms a partition of $M$, called the \emph{consistent decomposition of the serial module $M$}.
Note that it might happen that $M$ has only one consistent submodule.
This take place when $\inside(\mathbb{M}_i) = \emptyset$ for every child $M_i$ of $M$
and when every conformal model $\phi$ of $G_{ov}$ is of the form
\begin{equation}
\label{eq:forms_of_phi_one_consistent_submodule_of_M}
\phi|(M \cup N_T[M]) \equiv \tau N \tau' N' \text{ or }\phi|(M \cup N_T[M]) \equiv \tau N \tau',
\end{equation}
where $(\tau,\tau')$ is the permutation model of $(M,{\sim})$, and $N,N'$ are the only nodes in $N_T[M]$ if the first case holds and $N$ is the only node in $N_T[M]$ if the latter case holds.

Suppose $K_1,\ldots,K_n$ is a consistent decomposition of $M$.
A \emph{skeleton of $M$} is a subset $\{s_1,\ldots,s_n\}$ of $M$ such that $s_i \in K_i$ for every $i \in [n]$.
The next lemma can be seen as an analogue of Lemma \ref{lemma:consistent_decompositions}.
\begin{lemma}
\label{lemma:consistent_decomposition_serial_module}
Suppose $M$ is a serial module in $\mathcal{M}(V)$, $K_1,\ldots,K_n$ is a consistent decomposition of $M$, and $S = \{s_1,\ldots,s_n\}$ is a skeleton of $M$.
Then:
\begin{enumerate}
 \item \label{item:consistent_decomposition_serial_module_two_models} There exist two conformal models $\phi^{0}_S$ and $\phi^1_S$ of $(S,{\sim})$, one being the reflection of the other, such that for every conformal model $\phi$ of $G_{ov}$ we have $\phi|S \equiv \phi^{m}_S$ for some $m \in \{0,1\}$.
 \item \label{item:consistent_decomposition_serial_module_consistency}
If $n \geq 2$, for every conformal model $\phi$ of $G_{ov}$ and every $i \in [n]$, the set $\phi|K_i$ consists of two contiguous subwords of the circular word $\phi|M$.
\end{enumerate}
\end{lemma}
\begin{proof}
Suppose $K_1,\ldots,K_n$ are enumerated such that $i < j$ if $\inside(\mathbb{K}_i) \neq \emptyset$ and $\inside(\mathbb{K}_j) = \emptyset$.
Let $S_i = \{s_1,\ldots,s_i\}$.
Note that $(S_i,{\sim})$ is a clique for every $i \in [n]$.
We claim that for every $i \in [2,n]$ there exist two models of $(S_i,{\sim})$, say $\phi^{0}_{S_i}$ 
and its reflection $\phi^{1}_{S_i}$, such that for every conformal model $\phi$ of $G_{ov}$,
either 
$$\phi|S_i \equiv \phi^{0}_{S_i} \text{ or }\phi|S_i \equiv \phi^{1}_{S_i}.$$ 
Then \eqref{item:consistent_decomposition_serial_module_two_models} follows from the claim for $i=n$.
We prove the claim by induction on $i$.
Note that $(S_2,{\sim})$ has two conformal models,
$$\phi^{0}_{S_2} \equiv s^{0}_1s^{0}_2s^{1}_1s^{1}_2 \text { and } \phi^{1}_{S_2} \equiv s^{0}_1s^{1}_2s^{1}_1s^{0}_2,$$
and $\phi^{0}_{S_2}$ is the reflection of $\phi^{1}_{S_2}$.
So, the claim holds for $i=2$.
Suppose the claim holds for $i=j-1$ for some $j \in [3,n]$.
To show the claim for $i=j$ it suffices to show that there is a unique extension $\phi^{0}_{S_j}$ of 
$\phi^{0}_{S_{j-1}}$ on the set $S_j$, such that $\phi|S_j \equiv \phi^{0}_{S_j}$ holds for every conformal model 
$\phi$ of $G_{ov}$ such that $\phi|S_{j-1} \equiv \phi^{0}_{S_{j-1}}$.
Suppose for a contradiction that there are two conformal models of $G_{ov}$, $\phi$ and $\phi'$,
such that $\phi|S_{j-1} \equiv \phi'|S_{j-1} \equiv \phi^{0}_{S_j-1}$ and $\phi|S_j \not\equiv \phi'|S_j$.
That is, the chords $\phi(s_j)$ and $\phi'(s_j)$ extend $\phi^{0}_{S_{j-1}}$ into two non-equivalent models.
\begin{figure}[!htp]
\begin{tikzpicture}[scale=0.85,>=latex,shorten >=-0.4pt,shorten <=-0.4pt]
    \coordinate (center) at (0,0) {};
    \coordinate (sl0) at ($(center)+(180:2)$) {};
    \coordinate (sl1) at ($(center)+(0:2)$) {};
    \coordinate (lsl0) at ($(center)+(180:2.1)$) {};
    \coordinate (lsl1) at ($(center)+(0:2.1)$) {};

    \coordinate (sk0) at ($(center)+(270:2)$) {};
    \coordinate (sk1) at ($(center)+(90:2)$) {};
    \coordinate (lsk0) at ($(center)+(270:2.4)$) {};
    \coordinate (lsk1) at ($(center)+(90:2.4)$) {};

    \coordinate (x0) at ($(center)+(45:2)$) {};
    \coordinate (x1) at ($(center)+(225:2)$) {};
    \coordinate (lx0) at ($(center)+(45:1.9)$) {};
    \coordinate (lx1) at ($(center)+(225:1.9)$) {};
    
    \coordinate (y0) at ($(center)+(135:2)$) {};
    \coordinate (y1) at ($(center)+(315:2)$) {};
    \coordinate (ly0) at ($(center)+(135:1.9)$) {};
    \coordinate (ly1) at ($(center)+(315:1.9)$) {};
    \tikzstyle{every node}=[inner sep=1pt]
    \begin{tiny}
    \node[anchor=east] at (lsl0) {$\phi(s^0_l)=\phi'(s^0_l)$};
    \node[anchor=west] at (lsl1) {$\phi(s^1_l) =\phi'(s^1_l)$};
    \node at (lsk0) {$\phi(s^0_k) = \phi'(s^0_k)$};
    \node at (lsk1) {$\phi(s^1_k) = \phi'(s^0_k)$};

    \node[anchor=south west] at (lx0) {$\phi(s^1_j)$};
    \node[anchor=north east] at (lx1) {$\phi(s^0_j)$};
    \node[anchor=south east] at (ly0) {$\phi'(s^0_j)$};
    \node[anchor=north west] at (ly1) {$\phi'(s^1_j)$};
    \end{tiny}
    \draw[thick,->] (sl0) -- (sl1);
    \draw[thick,->] (sk0) -- (sk1);
    \draw[thick,->] (x1) -- (x0);
    \draw[thick,->] (y0) -- (y1);

  \draw (0,0) circle (2cm);

  
\end{tikzpicture}
\begin{tikzpicture}[scale=0.85,>=latex,shorten >=-0.4pt,shorten <=-0.4pt]
    \coordinate (center) at (0,0) {};
    \coordinate (sl0) at ($(center)+(180:2)$) {};
    \coordinate (sl1) at ($(center)+(0:2)$) {};
    \coordinate (lsl0) at ($(center)+(180:2.1)$) {};
    \coordinate (lsl1) at ($(center)+(0:2.1)$) {};

    \coordinate (sk0) at ($(center)+(270:2)$) {};
    \coordinate (sk1) at ($(center)+(90:2)$) {};
    \coordinate (lsk0) at ($(center)+(270:2.4)$) {};
    \coordinate (lsk1) at ($(center)+(90:2.4)$) {};

    \coordinate (x0) at ($(center)+(30:2)$) {};
    \coordinate (x1) at ($(center)+(210:2)$) {};
    \coordinate (lx0) at ($(center)+(30:2.1)$) {};
    \coordinate (lx1) at ($(center)+(214:2.1)$) {};
    
    \coordinate (y0) at ($(center)+(60:2)$) {};
    \coordinate (y1) at ($(center)+(240:2)$) {};
    \coordinate (ly0) at ($(center)+(66:2.2)$) {};
    \coordinate (ly1) at ($(center)+(245:2.3)$) {};
    \tikzstyle{every node}=[inner sep=1pt]
    \begin{tiny}
    \node[anchor=east] at (lsl0) {$\phi(s^0_l)=\phi'(s^0_l)$};
    \node[anchor=west] at (lsl1) {$\phi(s^1_l) =\phi'(s^1_l)$};
    \node at (lsk0) {$\phi(s^0_k) = \phi'(s^0_k)$};
    \node at (lsk1) {$\phi(s^1_k) = \phi'(s^0_k)$};
    \node[anchor=west] at (lx0) {$\phi(s^0_j)$};
    \node[anchor=east] at (lx1) {$\phi(s^1_j)$};
    \node[anchor=west] at (ly0) {$\phi'(s^1_j)$};
    \node[anchor=east] at (ly1) {$\phi'(s^0_j)$};
    \end{tiny}
    \draw[thick,->] (sl0) -- (sl1);
    \draw[thick,->] (sk0) -- (sk1);
    \draw[thick,->] (x0) -- (x1);
    \draw[thick,<-] (y0) -- (y1);

  \draw (0,0) circle (2cm);

  
\end{tikzpicture}
\caption{\label{fig:clique_cases}}
\end{figure}
Hence, there are two different vertices $s_l,s_k \in S_{j-1}$ such that
$$
\phi'|\{s_l,s_k\} \equiv \phi|\{s_l,s_k\} \equiv s_k^0s_l^{0}s_k^{1}s_l^{1},
$$
but the chords $\phi(s_j)$ and $\phi'(s_j)$ 
have its endpoints in different sections $s^0_ks^0_l$, $s^0_ls^1_k$, $s^1_ks^1_l$, $s^1_ls^0_l$ of the circular word
$s_k^0s_l^{0}s_k^{1}s_l^{1}$ -- see Figure \ref{fig:clique_cases} for an illustration.

First, suppose that
$$\phi|\{s_l,s_k,s_j\} \equiv s_k^0 s_j^0 s^0_l s_k^1 s_j^1 s_l^1 \text{ and }
\phi'|\{s_l,s_k,s_j\} \equiv s_k^0 s_l^0 s_j^0 s_k^1 s_l^1 s_j^1,$$
see Figure \ref{fig:clique_cases} on the left.
Let $\phi_M = \phi|(M \cup N_T(M))$
and $\phi'_M = \phi'|(M \cup N_T(M))$.
Using similar arguments as above we may show that $\inside(\mathbb{K}_l) = \emptyset$ and $\inside(\mathbb{K}_j) = \emptyset$.
Suppose that $\inside(K_l) \neq \emptyset$.
Let $N \in N_T[M]$ be a node such that $N \in \inside(\mathbb{K}_l)$.
Consider the position of $\phi_M(N)$ and $\phi'_M(N)$ relatively 
to the chords $\phi_M(s_j)$, $\phi_M(s_k)$ in $\phi_M$ and $\phi'_M(s_k)$, $\phi'_M(s_j)$ in $\phi'_M$.
Note that in $\phi_M$ the point $\phi_M(N)$ is either on the left side of both $\phi_M(s_k)$ and $\phi_M(s_j)$ or on the right side of both $\phi_M(s_k)$ and $\phi_M(s_j)$.
However, in $\phi'_M$ the point $\phi'_M(N)$ is either on the right side of $\phi'_M(s_k)$ and the left side of $\phi'_M(s_j)$ or on the left side of $\phi'_M(s_k)$ and the right side of $\phi'_M(s_j)$.
However, this is not possible as $\phi$ and $\phi'$ are conformal models of $G_{ov}$.
So, we must have $\inside(K_l) = \emptyset$.
For the same reason we also have $\inside(\mathbb{K}_j) = \emptyset$.
To complete the proof in this case, we show that for every $N \in N_T[M]$,
$V_{T \setminus M}(N) \subseteq \leftside(s_l)$ iff $V_{T \setminus M}(N) \subseteq \leftside(s_j)$.
If this is the case, then $s_j K s_l$, which contradicts that $s_j$ and $s_l$ are from two different equivalence classes of $K$-relation.
Suppose that $V_{T \setminus M}(N) \subseteq \leftside(s_l)$ and $V_{T \setminus M}(N) \subseteq \rightside(s_j)$.
Then $\phi_M(N)$ is between $\phi_M(s^1_j)$ and $\phi_M(s^1_l)$ in $\phi_M$
and between $\phi'_M(s^0_l)$  and $\phi'_M(s^0_j)$ in $\phi'$.
Thus, $\phi_M(N)$ is on the right side of $\phi_M(s_k)$ and on the left side of $\phi'_M(s_k)$.
This can not be the case.
The second implication is proved analogously.

Next, suppose 
$$\phi|\{s_l,s_k,s_j\} \equiv s_k^0 s_j^1 s^0_l s_k^1 s_j^0 s_l^1 \text{ and }
\phi'|\{s_l,s_k,s_j\} \equiv s_k^0 s_j^0 s_l^0 s_k^1 s_j^1 s_l^1,$$
see Figure \ref{fig:clique_cases} on the right.
Let $\phi_M = \phi|(M \cup N_T(M))$
and let $\phi'_M = \phi'|(M \cup N_T(M))$.
First, note that $\inside(\mathbb{K}_l) = \inside(\mathbb{K}_k) = \emptyset$
as otherwise a node $N$ in $\inside(\mathbb{K}_l) \cup \inside(\mathbb{K}_k)$
would be such that either $\phi_M(N)$ is on the left side of 
$\phi_M(s_j)$ and $\phi'_M(N)$ is on the right side of $\phi'_M(s_j)$ or $\phi_M(N)$ is on the right side of 
$\phi_M(s_j)$ and $\phi'_M(N)$ is on the left side of $\phi'_M(s_j)$. 
This is not possible.
Now, we claim that $\leftside(\mathbb{K}_l) \cap \leftside(\mathbb{K}_s) = \emptyset$
and $\rightside(\mathbb{K}_l) \cap \rightside(\mathbb{K}_s) = \emptyset$.
This will yield $s_k K s_l$ and will lead to a contradiction.
If $N \in \leftside(\mathbb{K}_l) \cap \leftside(\mathbb{K}_s)$, then
$\phi_M(N)$ is on the right side of $\phi_M(s_j)$ and $\phi'_M(N)$ is on the left side of $\phi'_M(s_j)$, 
which is not possible.
The second case is proved analogically.

All the remaining cases corresponding to different placements of $\phi(s_j)$ and $\phi'(s_j)$ are proven in a similar way.

Let $\phi$ be any conformal model of $G_{ov}$ and let $\phi_M \equiv \phi| (M \cup N_{T}(M))$.
Statement \eqref{item:consistent_decomposition_serial_module_consistency} obviously holds if 
$K_i$ is a child of $M$ in $\mathcal{M}(G_{ov})$.
So, suppose $K_i$ is the union of at least two children of $M$ in $\mathcal{M}(G_{ov})$.
Since $n \geq 2$, there is $s_j \in S$ such that $s_j \sim K_i$.
Denote by $l^{0}$ and $l^{3}$ the first and the last labeled letter from $K_i$, respectively,
if we traverse $\phi_M$ from $s^{0}_j$ to $s^{1}_j$.
Denote by $r^{0}$ and $r^{3}$ the first and the last labeled letter from $K_i$, respectively,
if we traverse $\phi_M$ from $s^{1}_j$ to $s^{0}_j$.
To show statement \eqref{item:consistent_decomposition_serial_module_consistency} suppose for a contrary that there is $v \in M \setminus K_i$ such that $\phi(v)$ has an end, say $\phi(v')$, between $\phi(l^{0})$ and $\phi(l^3)$.
Suppose $v''$ is such that $\{v',v''\} = \{v^{0},v^{1}\}$.
Note that for every child $M'$ of $M$ such that $M' \subseteq K_i$, $\inside(M') = \emptyset$.
Thus, there are children $M_1,M_2$ of $M$ that satisfy the following properties.
If we traverse from $\phi_M(l^0)$ to $\phi_M(l^3)$ then we encounter an end of every chord from $\phi_M(M_1)$ first, then $\phi_M(v')$, and then an end of every chord from $\phi_M(M_2)$.
Since $M$ is serial,
if we traverse from $\phi_M(r^0)$ to $\phi_M(r^3)$ then we encounter an end of every chord from $\phi_M(M_2)$ first, then $\phi(v'')$, and then an end of every chord from $\phi_M(M_1)$.
Denote by $l^{1}$ and $l^{2}$ the last labeled letter from $M_1$ and the first labeled letter from $M_2$
if we traverse $\phi_M$ from $\phi_M(l^{0})$ to $\phi_M(l_2)$.
Similarly, denote by $r^{1}$ and $r^{2}$ the last labeled letter from $M_2$ and the first labeled letter from $M_1$ if we traverse $\phi_M$ from $\phi_M(r^{0})$ to $\phi_M(r^3)$.
First, note that there is a node $N' \in N_T[M]$ such that $\phi_M(N')$
is either between $\phi_M(l^1)$ and $\phi_M(l^{2})$ or between $\phi_M(r^1)$ and $\phi_M(r^2)$.
Otherwise, we have that $v K M_1$ and $v K M_2$, which contradicts $v \notin K_i$.
On the other hand, it is not possible that there is a node $N_T[M]$ between $\phi(l^1)$ and $\phi(l^2)$
and there is a node from $N_T[M]$ between $\phi(r^1)$ and $\phi(r^2)$, 
as otherwise $M_1$ is not in $K$-relation with $M_2$.
So, suppose $N' \in N_T[M]$ is between $\phi(l^1)$ and $\phi(l^2)$ and suppose there is no node between
$\phi(r^1)$ and $\phi(r^2)$.
Note that there is a node $N''$ between $\phi_M(l^{3})$ and $\phi_M(r^{0})$
or between $\phi_M(r^{3})$ and $\phi_M(l^{0})$ as otherwise 
$s_j$ would be in $K$-relation with $M_1$ and $M_2$.
In any case, $N'$ and $N''$ prove that $M_1$ and $M_2$ are not in $K$-relation, a contradiction.
\end{proof}
Because of the analogy between Lemmas~\ref{lemma:consistent_decompositions}
and Lemma \ref{lemma:consistent_decomposition_serial_module}, 
for a serial child $M$ of $V$ an analogue of Claim~\ref{claim:prime_module_invariants} is satisfied, where
for every conformal model $\phi$ the triple $(K^{0}_{i, \phi}, K^{1}_{i, \phi},{<^{0}_{K_i, \phi}})$ and
the circular order $\pi(\phi|M)$ (which may have only two elements) of the slots of $M$ are defined as for prime children of $V$.
\begin{claim}
\label{claim:serial_module_invariants}
Suppose $M$ is a serial module in $T_{NM}$, $K_1,\ldots,K_n$
is a consistent decomposition of $M$, and $S$ is a skeleton of $M$.
For every $i \in [n]$ there are labeled copies $K^{0}_i$ and $K^{1}_i$ of $K_i$ forming a partition of $K^{*}_i$ and a transitive orientation ${<_{K_i}}$ of $(K_i,{\parallel})$ such that
$$(K^{0}_{i, \phi}, K^{1}_{i, \phi},{<^{0}_{K_i, \phi}}) = (K^{0}_{i}, K^{1}_i,{<_{K_i}}) 
\quad
\begin{array}{c}
\text{for every conformal model $\phi$ of $G_{ov}$} \\
\text{and every $i \in [n]$.}
\end{array}
$$
Moreover, there are circular permutations $\pi_{0}(M),\pi_{1}(M)$ of $\{K^{0}_1,K^{1}_0, \ldots,K^{0}_n,K^{1}_n\}$ such that 
$$\pi(\phi|M) = 
\left\{
\begin{array}{cll}
\pi_{0}(M) & \text{if} & \phi|S = \phi^{0}_S\\
\pi_{1}(M) & \text{if} & \phi|S = \phi^{1}_S\\
\end{array}
\right.
\quad \text{for every conformal model $\phi$ of $G_{ov}$.}
$$
Moreover, $\pi_{0}(M)$ is the reflection of $\pi_{1}(M)$.
\end{claim}
Given the above claim, we define analogously the pattern $p_m(S)$ for every slot $S$ in $\pi_m(M)$ and the sets $p_m(S,T)$ for every two consecutive slots $(S,T)$ in $\pi_m(M)$.
With similar ideas we get an analogue of Claim \ref{claim:ivariants_for_prime_child_of_V}.
We define the extended metaedge $\mathbb{K}_{i,m}$ for every consistent submodule $K_i$ of $M$ and every $m \in \{0,1\}$, and the notions of extended admissible models for $\mathbb{K}_{i,m}$ and extended admissible models for $\pi_0(M)$, $\pi_1(M)$ and $(M,{\sim})$.

Eventually, we are ready to present operations that transform conformal models of $G_{ov}$ into 
other conformal models of $G_{ov}$.
Suppose $\phi$ is a conformal model of $G_{ov}$.
Recall that:
\begin{itemize}
\item for every node $N$ in $T_{NM}$ the circular word $\phi|N_T(N)$ is a circular permutation of the modules neighboring $N$ in $T_{NM}$,
\item for every module $M$ in $T_{NM}$ the circular word $\phi|(V \cup N_T(M))$ is an extended admissible model for $\pi_m(M)$ for some $m \in \{0,1\}$. 
Call $m$  the \emph{type} of $M$ in~$\phi$.
\end{itemize}
The first operation, which can be performed on every node of $T_{NM}$, 
allows us to change the circular permutation $\phi|N_T(N)$ around a node $N$.
Suppose that $\phi|N_T(N) \equiv M_1\ldots M_k$.
By Claim~\ref{claim:conformal_models_properties_modules_nodes_T_NM}, $\phi|V_{T \setminus N}(M_i)$ is a contiguous subword in $\phi$.
Denote this subword by $\tau_{M_i}$ and note that $\phi \equiv \tau_{M_1}\ldots \tau_{M_k}$.
Let $M_{i_1}\ldots M_{i_k}$ be a circular permutation of $N_T[N]$.
One can check that $\phi' \equiv \tau_{M_{i_1}}\ldots \tau_{M_{i_k}}$ is also a conformal model of $G_{ov}$,
which leaves the circular permutations around the remaining nodes in $T_{NM}$ and the types of all the modules in $T_{NM}$ unchanged -- see Figure \ref{fig:node_permuting} for an illustration.

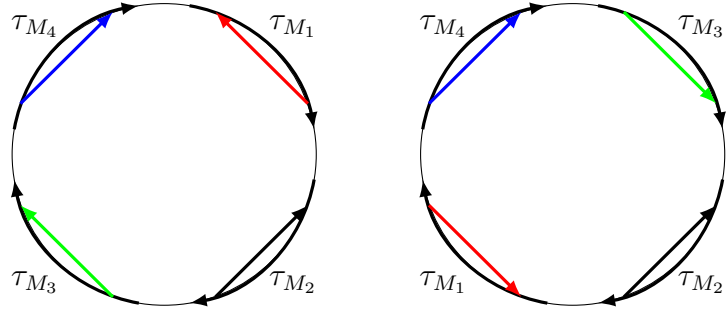
\begin{figure}[!htp]
\begin{tikzpicture}[scale=1,>=latex,shorten >=-0.4pt,shorten <=-0.4pt]
    \coordinate (center) at (0,0) {};
    \coordinate (ur1) at ($(center)+(20:2)$) {};
    \coordinate (ur2) at ($(center)+(35:2)$) {};
    \coordinate (ur3) at ($(center)+(55:2)$) {};
    \coordinate (ur4) at ($(center)+(70:2)$) {};

    \coordinate (ul1) at ($(center)+(110:2)$) {};
    \coordinate (ul2) at ($(center)+(125:2)$) {};
    \coordinate (ul3) at ($(center)+(145:2)$) {};
    \coordinate (ul4) at ($(center)+(160:2)$) {};

    \coordinate (bl1) at ($(center)+(200:2)$) {};
    \coordinate (bl2) at ($(center)+(215:2)$) {};
    \coordinate (bl3) at ($(center)+(235:2)$) {};
    \coordinate (bl4) at ($(center)+(250:2)$) {};

    \coordinate (br1) at ($(center)+(290:2)$) {};
    \coordinate (br2) at ($(center)+(305:2)$) {};
    \coordinate (br3) at ($(center)+(325:2)$) {};
    \coordinate (br4) at ($(center)+(340:2)$) {};
    
    \coordinate (lur) at ($(center)+(45:2.4)$) {};
    \coordinate (lul) at ($(center)+(135:2.4)$) {};
    \coordinate (lbl) at ($(center)+(225:2.4)$) {};
    \coordinate (lbr) at ($(center)+(315:2.4)$) {};

    \tikzstyle{every node}=[inner sep=1pt]
    \node at (lur) {$\tau_{M_1}$};
    \node at (lbr) {$\tau_{M_2}$};
    \node at (lbl) {$\tau_{M_3}$};
    \node at (lul) {$\tau_{M_4}$};
  
    \draw (0,0) circle (2cm);
    \draw[very thick,->] ($(center) + (80:2)$) arc (80:10:2);    
    \draw[very thick,->] ($(center) + (170:2)$) arc (170:100:2);    
    \draw[very thick,->] ($(center) + (260:2)$) arc (260:190:2);    
    \draw[very thick,->] ($(center) + (350:2)$) arc (350:280:2);    

    \draw[->,very thick,red] (ur1)--(ur4);    
    \draw[->,very thick,blue] (ul4)--(ul1);    
    \draw[->,very thick,green] (bl4)--(bl1);    
    \draw[->,very thick] (br1)--(br4);    
  
\end{tikzpicture}
\hspace{1cm}
\begin{tikzpicture}[scale=1,>=latex,shorten >=-0.4pt,shorten <=-0.4pt]
    \coordinate (center) at (0,0) {};
    \coordinate (ur1) at ($(center)+(20:2)$) {};
    \coordinate (ur2) at ($(center)+(35:2)$) {};
    \coordinate (ur3) at ($(center)+(55:2)$) {};
    \coordinate (ur4) at ($(center)+(70:2)$) {};

    \coordinate (ul1) at ($(center)+(110:2)$) {};
    \coordinate (ul2) at ($(center)+(125:2)$) {};
    \coordinate (ul3) at ($(center)+(145:2)$) {};
    \coordinate (ul4) at ($(center)+(160:2)$) {};

    \coordinate (bl1) at ($(center)+(200:2)$) {};
    \coordinate (bl2) at ($(center)+(215:2)$) {};
    \coordinate (bl3) at ($(center)+(235:2)$) {};
    \coordinate (bl4) at ($(center)+(250:2)$) {};

    \coordinate (br1) at ($(center)+(290:2)$) {};
    \coordinate (br2) at ($(center)+(305:2)$) {};
    \coordinate (br3) at ($(center)+(325:2)$) {};
    \coordinate (br4) at ($(center)+(340:2)$) {};
    
    \coordinate (lur) at ($(center)+(45:2.4)$) {};
    \coordinate (lul) at ($(center)+(135:2.4)$) {};
    \coordinate (lbl) at ($(center)+(225:2.4)$) {};
    \coordinate (lbr) at ($(center)+(315:2.4)$) {};

    \tikzstyle{every node}=[inner sep=1pt]
    \node at (lur) {$\tau_{M_3}$};
    \node at (lbr) {$\tau_{M_2}$};
    \node at (lbl) {$\tau_{M_1}$};
    \node at (lul) {$\tau_{M_4}$};
  
    \draw (0,0) circle (2cm);
    \draw[very thick,->] ($(center) + (80:2)$) arc (80:10:2);    
    \draw[very thick,->] ($(center) + (170:2)$) arc (170:100:2);    
    \draw[very thick,->] ($(center) + (260:2)$) arc (260:190:2);    
    \draw[very thick,->] ($(center) + (350:2)$) arc (350:280:2);    

    \draw[<-,very thick,green] (ur1)--(ur4);    
    \draw[->,blue,very thick] (ul4)--(ul1);    
    \draw[->,red,very thick] (bl1)--(bl4);    
    \draw[->,very thick] (br1)--(br4);    
  
\end{tikzpicture}
\caption{\label{fig:node_permuting} 
Permuting the modules $M_1,M_2,M_3,M_4$ around the node $N$:
$\phi\equiv \tau_{M_1}\tau_{M_2}\tau_{M_3}\tau_{M_4}$ and $\phi \equiv \tau_{M_3}\tau_{M_2}\tau_{M_1}\tau_{M_4}$ are conformal models of $G_{ov}$. Note that after permuting $\tau_1,\tau_2,\tau_3,\tau_4$ the mutual relations between the oriented chords in $\phi$ and $\phi'$ is preserved.} 
\end{figure}

The second operation, which can be performed on every module in $T_{NM}$, 
allows us to change the type of the module $M$.
Suppose $\pi(\phi|M) = \pi_m(M)$ for some $m \in \{0,1\}$.
Clearly, $\phi_M \equiv \phi|(M \cup N_T[M])$ is an extended admissible model for $\pi_n(M)$.
By Claim \ref{claim:conformal_models_properties_modules_nodes_T_NM},
 $\phi|V_{T\setminus M}(N)$ is a contiguous subword in $\phi$ for every neighbor $N$ of the module $M$.
Denote this subword by $\tau_{N}$.
Clearly, the reflection of $\phi_M$, say $\phi^R_M$, is an extended admissible model for $\phi_{1 - m}(M)$.
One can check that if we replace every node $N$ in $\phi^R(M)$ by the word $\tau_N$,
we get a normalized model $\phi'$ of $G_{ov}$ in which the type of $M$ is changed.
Moreover, the types of the remaining modules in $T_{NM}$ and the circular permutations around all the nodes in $T_{NM}$ remain unchanged
-- see Figure \ref{fig:module_reflection} for an illustration.

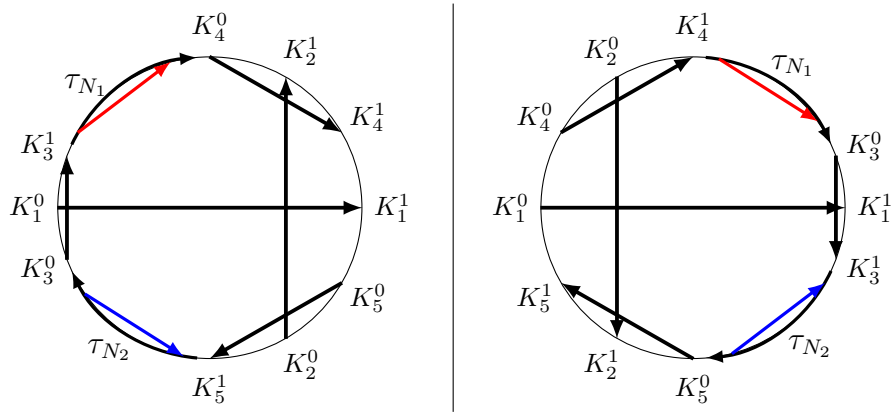
\begin{figure}[htp!]
\begin{tikzpicture}[scale=1,>=latex,shorten >=-0.2pt,shorten <=-0.2pt]
\coordinate (center) at (0,0) {};

\coordinate (k11) at ($(center)+(0:2cm)$) {};
\coordinate (k10) at ($(center)+(180:2cm)$) {};
\coordinate (k21) at ($(center)+(60:2cm)$) {};
\coordinate (k20) at ($(center)+(300:2cm)$) {};
\coordinate (k31) at ($(center)+(160:2cm)$) {};
\coordinate (k30) at ($(center)+(200:2cm)$) {};
\coordinate (k41) at ($(center)+(30:2cm)$) {};
\coordinate (k40) at ($(center)+(90:2cm)$) {};
\coordinate (k51) at ($(center)+(270:2cm)$) {};
\coordinate (k50) at ($(center)+(330:2cm)$) {};

\coordinate (lk11) at ($(center)+(0:2.4cm)$) {};
\coordinate (lk10) at ($(center)+(180:2.4cm)$) {};
\coordinate (lk21) at ($(center)+(60:2.4cm)$) {};
\coordinate (lk20) at ($(center)+(300:2.4cm)$) {};
\coordinate (lk31) at ($(center)+(160:2.4cm)$) {};
\coordinate (lk30) at ($(center)+(200:2.4cm)$) {};
\coordinate (lk41) at ($(center)+(30:2.4cm)$) {};
\coordinate (lk40) at ($(center)+(90:2.4cm)$) {};
\coordinate (lk51) at ($(center)+(270:2.4cm)$) {};
\coordinate (lk50) at ($(center)+(330:2.4cm)$) {};

\coordinate (taun1) at ($(center)+(135:2.3cm)$) {};
\coordinate (taun2) at ($(center)+(235:2.3cm)$) {};

\tikzstyle{every node}=[inner sep=1pt]
\begin{footnotesize}
\node at (lk11) {$K_1^1$};
\node at (lk10) {$K_1^0$};
\node at (lk21) {$K_2^1$};
\node at (lk20) {$K_2^0$};
\node at (lk31) {$K_3^1$};
\node at (lk30) {$K_3^0$};
\node at (lk41) {$K_4^1$};
\node at (lk40) {$K_4^0$};
\node at (lk51) {$K_5^1$};
\node at (lk50) {$K_5^0$};
\end{footnotesize}
\node at (taun1) {$\tau_{N_1}$};
\node at (taun2) {$\tau_{N_2}$};

\draw (0,0) circle (2cm);

\draw[line width=0.5mm,->] (k10)--(k11);
\draw[line width=0.5mm,->] (k20)--(k21);
\draw[line width=0.5mm,->] (k30)--(k31);
\draw[line width=0.5mm,->] (k40)--(k41);
\draw[line width=0.5mm,->] (k50)--(k51);

\draw[very thick,->] ($(center) + (155:2)$) arc (155:95:2);    
\draw[very thick,->,red] ($(center) + (150:2)$) -- ($(center) + (105:2)$);    
\draw[very thick,->] ($(center) + (265:2)$) arc (265:205:2);    
\draw[very thick,<-,blue] ($(center) + (260:2)$) -- ($(center) + (215:2)$);

\draw (3.2,2.7)--(3.2,-2.7);

\draw[white] (-3.0,-2.7)--(-3.0,-2.5);
\draw[white] (3.0,2.7)--(3.0,2.5);
\end{tikzpicture}
\begin{tikzpicture}[scale=1,>=latex,shorten >=-0.2pt,shorten <=-0.2pt]
\coordinate (center) at (0,0) {};

\coordinate (k11) at ($(center)+(0:2cm)$) {};
\coordinate (k10) at ($(center)+(180:2cm)$) {};
\coordinate (k20) at ($(center)+(120:2cm)$) {};
\coordinate (k21) at ($(center)+(240:2cm)$) {};
\coordinate (k30) at ($(center)+(20:2cm)$) {};
\coordinate (k31) at ($(center)+(340:2cm)$) {};
\coordinate (k40) at ($(center)+(150:2cm)$) {};
\coordinate (k41) at ($(center)+(90:2cm)$) {};
\coordinate (k50) at ($(center)+(270:2cm)$) {};
\coordinate (k51) at ($(center)+(210:2cm)$) {};

\coordinate (lk11) at ($(center)+(0:2.4cm)$) {};
\coordinate (lk10) at ($(center)+(180:2.4cm)$) {};
\coordinate (lk20) at ($(center)+(120:2.4cm)$) {};
\coordinate (lk21) at ($(center)+(240:2.4cm)$) {};
\coordinate (lk30) at ($(center)+(20:2.4cm)$) {};
\coordinate (lk31) at ($(center)+(340:2.4cm)$) {};
\coordinate (lk40) at ($(center)+(150:2.4cm)$) {};
\coordinate (lk41) at ($(center)+(90:2.4cm)$) {};
\coordinate (lk50) at ($(center)+(270:2.4cm)$) {};
\coordinate (lk51) at ($(center)+(210:2.4cm)$) {};

\coordinate (taun1) at ($(center)+(55:2.3cm)$) {};
\coordinate (taun2) at ($(center)+(310:2.4cm)$) {};

\tikzstyle{every node}=[inner sep=1pt]
\begin{footnotesize}
\node at (lk11) {$K_1^1$};
\node at (lk10) {$K_1^0$};
\node at (lk20) {$K_2^0$};
\node at (lk21) {$K_2^1$};
\node at (lk30) {$K_3^0$};
\node at (lk31) {$K_3^1$};
\node at (lk40) {$K_4^0$};
\node at (lk41) {$K_4^1$};
\node at (lk50) {$K_5^0$};
\node at (lk51) {$K_5^1$};
\end{footnotesize}
\node at (taun1) {$\tau_{N_1}$};
\node at (taun2) {$\tau_{N_2}$};

\draw (0,0) circle (2cm);

\draw[line width=0.5mm,->] (k10)--(k11);
\draw[line width=0.5mm,->] (k20)--(k21);
\draw[line width=0.5mm,->] (k30)--(k31);
\draw[line width=0.5mm,->] (k40)--(k41);
\draw[line width=0.5mm,->] (k50)--(k51);

\draw[very thick,->] ($(center) + (85:2)$) arc (85:25:2);    
\draw[very thick,->,red] ($(center) + (80:2)$) -- ($(center) + (35:2)$);    
\draw[very thick,->] ($(center) + (335:2)$) arc (335:275:2);    
\draw[very thick,<-,blue] ($(center) + (330:2)$) -- ($(center) + (285:2)$);    

\draw[white] (-3.0,-2.7)--(-3.0,-2.5);
\draw[white] (3.0,2.7)--(3.0,2.5);
\end{tikzpicture}
\caption{\label{fig:module_reflection} Reflecting the module $M$. Note that after the reflection,
the relative position of the oriented chords remains unchanged.}
\end{figure}

Note that we can perform the second operation so as the type of $M$ remains unchanged.
That is, we replace admissible model $\phi_M$ with another admissible model of the same type.
Such an operation exchanges some extended admissible models for some consistent submodules of $M$,
which corresponds to reordering of the oriented metaedges in some consistent submodules of $M$.

As we shall prove in Theorem \ref{thm:description_of_all_conformal_models_of_improper_parallel_modules}, 
these three operations are complete: that is, we can transform any conformal model into any other by performing a series of the above defined operations.

In the proof of Theorem \ref{thm:description_of_all_conformal_models_of_improper_parallel_modules} we use 
two obvious properties of an extended admissible model $\phi^M$ of $(M,{\sim})$:
\begin{equation}
\label{eq:conformal_models_agree_on_M}
\phi^M|M \text{ is a conformal model of } (M,{\sim}),
\tag{$C_1$}
\end{equation}
and
\begin{equation}
\label{eq:conformal_models_agree_on_N}
\begin{array}{l}
\begin{array}{lll}
  \text{$\phi^M(N)$ is on the left side of $\phi(u)$} &\iff & V_{T\setminus M}(N) \subset \leftside(u), \\
  \text{$\phi^M(N)$ is on the right side of $\phi(u)$} &\iff& V_{T \setminus M}(N) \subset \rightside(u),
\end{array}
\end{array}
\tag{$C_2$}
\end{equation}
which hold for every $N \in N_T[M]$ and every $v \in M$.

\begin{theorem}
\label{thm:description_of_all_conformal_models_of_improper_parallel_modules}
Suppose $G$ is a circular-arc graph with no twins and no universal vertices such that $G_{ov}$ is disconnected.
Then, for every module $M$ in $T_{NM}$ the circular word $\phi|(M \cup N_T[M])$ is an extended admissible model of $(M,{\sim})$ and for every node $N$ in $T_{NM}$ the circular word $\phi|N_T[N]$ is a circular permutation of $N_T[N]$.

On the other hand, given an extended admissible model $\phi^{M}$ for every module $M \in T_{NM}$
and a circular permutation $\pi(N)$ of $N_T[N]$ for every node $N \in T_{NM}$,
there is a conformal model $\phi$ of $G_{ov}$ such that:
\begin{equation}
\label{eq:conformal_model_that_agrees_on_N_and_M}
\begin{array}{rlll}
\phi|(M \cup N_T[M])& \equiv& \phi^{M} & \text{for every module $M$ in $T_{NM}$,} \\
\phi|N_T(N)&\equiv &\pi(N) & \text{for every node $N$ in $T_{NM}$.} \\
\end{array}
\tag{*}
\end{equation}
\end{theorem}
\begin{proof}
The first part of the proof follows from Claim \ref{claim:conformal_models_are_admissible}.

Choose any module $R$ in $T_{NM}$ and fix $R$ as the root of $T_{NM}$.
Let $A$ be a vertex in the rooted tree $T_{NM}$.
Denote by $V_T(A) \subset V$ the set of all vertices
contained in the modules descending $A$ in $T_{NM}$, including $A$.
By $B$ we denote the parent of $A$ if it exists or the empty word if $A = R$.
We proceed $T_{NM}$ bottom-up and for every 
vertex $A$ in the rooted tree $T_{NM}$ we construct a word $\phi'_A$ consisting of all the 
labeled letters from $V^*_T(A)$ so as the circular word $\phi_A \equiv \phi'_AB$ satisfies the following properties:
\begin{enumerate}
 \item \label{item:uv_cond_left} For every $u,v$ in $V_T(A)$, $\phi_A(u)$ is on the left side of $\phi_A(v)$ if $u \in \leftside(v)$.
 \item \label{item:uv_cond_right} For every $u,v$ in $V_T(A)$, $\phi_A(u)$ is on the right side of $\phi_A(v)$ if $u \in \rightside(v)$.
 \item \label{item:uv_overlap} For every $u,v$ in $V_T(A)$, $\phi_A(u)$ intersects $\phi_A(v)$ if $u \sim v$.
 \item \label{item:vB_cond_left} For every $v \in V_T(A)$, $\phi(B)$ is on the left side of $\phi(v)$ if $V_{T\setminus A}(B) \subset \leftside(v)$.
 \item \label{item:vB_cond_right} For every $v \in V_T(A)$, $\phi(B)$ is on the right side of $\phi(v)$ if $V_{T\setminus A}(B) \subset \rightside(v)$.
\end{enumerate}
If we construct such a word $\phi'_A$ for every vertex $A$ in $T_{NM}$, then $\phi_R$ is a conformal model of $G_{ov}$.
Moreover, from the construction of $\phi_R$ it will be clear that $\phi_R$ satisfies properties
\ref{eq:conformal_model_that_agrees_on_N_and_M}
(in particular, for every $A$ in $T_{NM}$ the word $\phi'_A$ is a contiguous subword of $\phi_R$).

Suppose $A$ is a leaf in $T_{NM}$.
So, $A$ is a module in $T_{NM}$.
Let $\phi^A$ be the admissible model of $(A,{\sim})$ given by the assumption.
Then, $\phi^A \equiv \tau^AB$, where $B$ is the parent of $A$ in $T_{NM}$.
We set $\phi'_A = \tau^A$ and we claim that $\phi'_A$ satisfies properties \eqref{item:uv_cond_left}--\eqref{item:vB_cond_right}.
By property \eqref{eq:conformal_models_agree_on_M} of $\phi^A$ we deduce that $\phi_A$ satisfies 
\eqref{item:uv_cond_left} -- \eqref{item:uv_overlap}.
By property \eqref{eq:conformal_models_agree_on_N} of $\phi^A$ we deduce that $\phi_A$ satisfies \eqref{item:vB_cond_left}--\eqref{item:vB_cond_right}.

Now, suppose that $A$ is a node.
Suppose $A_1,\ldots,A_k$ are the children of $A$ in $T_{NM}$.
Suppose $\phi'_{A_i}$ is the word that has been constructed for $A_i$.
Suppose $A_1,\ldots,A_k$ are enumerated such that $\pi(A) \equiv BA_{1} \ldots A_{k}$, where 
$\pi(A)$ is a circular permutation of $N_T[A]$ associated with the node $A$ given by the assumption.
Let $\phi'_A = \phi'_{A_{1}}\ldots\phi'_{A_k}$.
We claim that $\phi_A \equiv \phi'_AB$ satisfies conditions \eqref{item:uv_cond_left}--\eqref{item:vB_cond_right}.
Suppose $u,v \in M$. 
If $uv \in V_T(A_i)$ for some $i \in [k]$, \eqref{item:uv_cond_left}--\eqref{item:uv_overlap} are satisfied
by $\phi_A$ as they are are satisfied by $\phi_{A_i}$.
Suppose $u \in A_i$ and $v \in A_j$.
But then, $u \in V_{T \setminus A_j}(A)$.
In particular, $u \in \leftside(v)$ yields $V_{T \setminus A_j}(A) \subset \leftside(v)$ by property \eqref{eq:conformal_models_agree_on_N} of $\phi_{A_j}$.
Hence, in the word $\phi_{A_j}$ the letter $A$ is on the left side of $\phi_{A_j}(v)$.
Since $\phi_A$ arises from $\phi_{A_j}$ by substituting the letter $A$ by the word including $u^0,u^1$,
we deduce that $\phi(u)$ is on the left side of $\phi_A(v)$.
We prove similarly that $u \in \rightside(v)$ yields $\phi(u)$ is on the right side of $\phi_A(v)$.
Now, let us prove that $\phi_A$ satisfies \eqref{item:vB_cond_left}.
Let $v \in A_j$.
We have $V_{T\setminus A_j}(A) \subset \leftside(v)$ iff $V_{T\setminus A}(B) \subset \leftside(v)$ for every $v \in V_T(A)$.
Suppose $v \in A_{j}$.
Hence, $\phi_{A_j}(A)$ is on the left side of $\phi_{A_j}(v)$.
But the word $\phi_A$ arises from $\phi_{A_j}$ by substituting the letter $A$ by the word containing $B$,
so we have $\phi_A(B)$ is on the left side of $\phi_A(v)$.
Property \eqref{item:vB_cond_right} is proved similarly.

Suppose $A$ is a module in $T_{NM}$ with children $A_1,\ldots,A_k$.
Let $\phi^A$ be a an extended admissible model of $(A,{\sim})$ associated with $A$ given by the assumption.
In the word $\phi^A$ we exchange the letter $A_i$ by the word $\phi'_{A_i}$, for every $i \in [k]$.
We denote the obtained word by $\phi_A$.
We set $\phi'_A$ such that $\phi_A = \phi'_AB$.
Now, we claim that $\phi_A$ satisfies \eqref{item:uv_cond_left}--\eqref{item:vB_cond_right}.
If $u,v \in A$, properties \eqref{item:uv_cond_left}--\eqref{item:uv_overlap} follow by property \eqref{eq:conformal_models_agree_on_M} of the conformal model $\phi^A$ of $(A,{\sim})$.
If $v \in A$ and $u \in A_i$, properties \eqref{item:uv_cond_left}--\eqref{item:uv_cond_right} follow by
property \eqref{eq:conformal_models_agree_on_N} of $\phi^A$.
If $v \in A_i$, the proof of \eqref{item:uv_cond_left}--\eqref{item:uv_cond_right}
is the same as in the previous case (now $u$ might belong to $A$, but $\phi_A(u)$ is on the left side of $\phi_A(v)$ iff $\phi_{A_i}(A)$ is in the left side of $\phi_{A_i}(v)$).
It remains to prove \eqref{item:vB_cond_left}--\eqref{item:vB_cond_right}.
If $v \in A$, properties \eqref{item:vB_cond_left}--\eqref{item:vB_cond_right} follow by property \eqref{eq:conformal_models_agree_on_N} of $\phi^A$.
If $v \in A_j$ for some $j \in [k]$, properties \eqref{item:vB_cond_left}--\eqref{item:vB_cond_right} 
are proved similarly to the previous case.

Now, from the construction of $\phi_R$ we easily deduce that
$$\begin{array}{rlll}
\phi_R|(M \cup N_T[M])& \equiv& \phi^{M} & \text{for every module $M$ in $T_{NM}$,} \\
\phi_R|N_T(N)&\equiv &\pi(N) & \text{for every node $N$ in $T_{NM}$,}
\end{array}
$$
which completes the proof of the lemma.
\end{proof}


\section{Isomorphism of circular-arc graphs}
\label{sec:isomorphism_problem}

Two graphs $G$ and $H$ are \emph{isomorphic} if there 
exists a bijection $\alpha:V(G) \to V(H)$ such that 
$uv \in E(G)$ iff $\alpha(u)\alpha(v) \in E(H)$.
Our goal is to provide a polynomial-time algorithm that tests whether two circular-arc graphs $G$ and $H$
are isomorphic.

Note that every isomorphism between $G$ and $H$ maps universal vertices of $G$ into 
universal vertices in $H$.
So, if $G$ and $H$ have different number of universal vertices, then $G$ and $H$ are not isomorphic.
Otherwise, $G$ and $H$ are isomorphic iff $G'$ and $H'$ are isomorphic, where $G'$ and $H'$
are obtained from $G$ and $H$ by deleting all the universal vertices in $G$ and $H$, respectively.
So, in the rest of the paper we assume $G$ and $H$ have no universal vertices.

Suppose $G$ and $H$ are circular-arc graphs with no universal vertices.
For every vertex $v \in V(G)$ we define the set of its twins in $G$:
$$T_G(v) = \{w \in V(G): N_G[v] = N_G[w]\},$$
where $N_G[v]$ denotes the closed neighbourhood of $v$ in $G$.
Clearly, $\{T_G(v): v \in V\}$ forms a partition of the set $V(G)$.
Let $V'$ be a set containing a representant from every set $\{\{T_G(v)\}: v \in G\}$.
Let $G'$ be a graph induced by the set $V'$ in $G$ and let $m_G(v)$ denotes the size of the set $T_G(v)$, 
for every $v' \in V'$.
Note that $G'$ has no universal vertices and no twins.
We call $(G',m_G)$ a \emph{circular-arc graph with multiplicities} or simply a circular-arc graph.
A pair $(H',m_H)$ for the graph $H$ is defined analogously.
We say $(G',m_G)$ and $(H',m_H)$ are \emph{isomorphic}
if there is an isomorphism $\alpha'$ from $G'$ to $H'$ that \emph{preserves multiplicities},
that is, that satisfies $m_G(v) = m_H(\alpha'(v))$ for every for every $v \in V(G')$.
\begin{claim}
$G$ and $H$ are isomorphic if and only if $(G',m_G)$ and $(H',m_H)$ are isomorphic.
\end{claim}
\begin{proof}
Let $\alpha$ be an isomorphism between $G$ and $H$.
Note that for every $v \in V(G')$ $\alpha$ maps every twin $u$ of $v$ in $G$ into a twin $\alpha(u)$ of $\alpha(v)$ in $H$.
Hence, $\{\alpha(w): w \in T_G[v]\} = T_H(\alpha(v))$.
In other words, $\alpha$ maps the set $T_G[v]$ into the set $T_H[\alpha(v)]$.
So, $\alpha|V'$ is an isomorphism between $(G',m_G)$ and $(H',m_H)$.
The converse implication is trivial.
\end{proof}

In the rest of this section we assume that $(G,m_G)$ and $(H,m_H)$ are given on the input,
where $G$ and $H$ are circular-arc graphs with no universal vertices and no twins. 
By $G_{ov}$ and $H_{ov}$ we denote the graphs associated with $G$ and $H$, respectively. 

\begin{theorem}
\label{thm:iso_cond}
Let $(G,m_G)$ and $(H,m_H)$ be two circular-arc graphs and let 
$\alpha$ be a bijection from $V(G)$ to $V(H)$.
Then $\alpha:V(G) \to V(H)$ is an isomorphism from $(G,m_G)$ to $(H,m_H)$ iff
$\alpha$ preserves multiplicities and for every pair $(u,v) \in V(G) \times V(G)$:
\begin{enumerate}
 \item \label{item:iso_cond_1} $u \in \leftside(v)$ if and only if $\alpha(u) \in \leftside(\alpha(v))$.
 \item \label{item:iso_cond_2} $u \in \rightside(v)$ if and only if $\alpha(v) \in \rightside(\alpha(v))$.
\end{enumerate}
\end{theorem}
\begin{proof}
It is clear that every isomorphism $\alpha$ between $(G,m_G)$ and $(H,m_H)$
preserves multiplicities and satisfies \eqref{item:iso_cond_1} and \eqref{item:iso_cond_2}.

Suppose $u,v \in V(G)$.
Suppose $G_{ov}$ and $H_{ov}$ be the overlap graphs for $G$ and $H$, respectively.
Suppose that $u$ and $v$ overlap in $G_{ov}$, which is equivalent to
$u \notin (\leftside(v) \cup \rightside(v))$.
Hence, by \eqref{item:iso_cond_1} and \eqref{item:iso_cond_2}, we get $\alpha(u) \notin (\leftside(\alpha(v)) \cup \rightside(\alpha(v)))$, which yields that $\alpha(u)$ overlaps $\alpha(v)$ in $H_{ov}$.
Further, by \eqref{item:iso_cond_1} and \eqref{item:iso_cond_2} we have that:
\begin{itemize} 
\item $u \in \leftside(v)$ and $v \in \leftside(v)$ iff $\alpha(u) \in \leftside(\alpha(v))$ and $\alpha(v) \in \leftside(\alpha(v))$,
\item $u \in \leftside(v)$ and $v \in \rightside(v)$ iff $\alpha(u) \in \leftside(\alpha(v))$ and $\alpha(v) \in \rightside(\alpha(v))$,
\item $u \in \rightside(v)$ and $v \in \leftside(v)$ iff $\alpha(u) \in \rightside(\alpha(v))$ and $\alpha(v) \in \leftside(\alpha(v))$,
\item $u \in \rightside(v)$ and $v \in \rightside(v)$ iff $\alpha(u) \in \rightside(\alpha(v))$ and $\alpha(v) \in \rightside(\alpha(v))$.
\end{itemize}
Hence, $uv \in E(G)$ iff $\alpha(u)\alpha(v) \in E(H)$.
\end{proof}
If $\alpha:V(G) \to V(H)$ satisfies the conditions from the previous lemma, 
then $\alpha$ is said to be an \emph{isomorphism between $(G_{ov},m_G)$ and $(H_{ov},m_h)$}.

Suppose $\alpha$ is an isomorphism between $(G',m_G)$ and $(H',m_H)$.
Suppose $\phi$ is a conformal model of $G_{ov}$.
The \emph{image of $\phi$ by $\alpha$}, denoted $\alpha(\phi)$, is a circular word on $V^{*}(H)$ 
that arises from $\phi$ by exchanging the labeled letter $u^0$ by $\alpha^0(u)$ and the labeled letter
$u^1$ by $\alpha^1(u)$, for every $u \in V^{*}(G)$.
Since $\alpha$ is an isomorphism between $(G_{ov},m_G)$ and $(H_{ov},m_H)$,
$\alpha(\phi)$ is a conformal model of $(H_{ov},m_H)$.
We extend the notion of the \emph{image of $\tau$ by $\alpha$}, denoted $\alpha(\tau)$, on words $\tau$ consisting of letters from $V^*(G)$.

\subsection{Extended admissible models for extended metaedges}
Suppose $K$ is a consistent submodule of $G_{ov}$ for which the extended metaedge 
$$\mathbb{K}_m = (K^0,K^1,{<_K},p_m(K^0), p_m(K^1))$$ is 
defined.

The goal of this subsection is to characterize the structure of all 
extended admissible models for $\mathbb{K}_m$.
By Theorem \ref{thm:permutation_models_transitive_orientations}, 
the admissible models of $(K^0,K^1,{<_K})$ are in the correspondence with the transitive
orientations $({<}, {\prec})$ of the graphs $(K,{\parallel})$ and $(K,{\sim})$, where
${<} = {<_K}$.
So, there is a correspondence between the admissible models of $(K^0,K^1,{<_K})$
and the transitive orientations of $(K,{\sim})$.
However, not every transitive orientation $\prec$ of $(K,{\sim})$ gives back
an extended admissible model for $\mathbb{K}_m$.
Again, using Theorem \ref{thm:permutation_models_transitive_orientations}, we have
the correspondence between the extended admissible models for $\mathbb{K}_m$
and the set of transitive orientations $(K,{\prec})$ of $(K,{\sim})$
that satisfy for every $x,y \in K$ the following conditions:
\begin{equation}
\label{eq:restrictions_on_transitive_orientations}
\begin{array}{c}
\text{If $x \sim y$ and $x$ is before $y$ in $p_m(K^0)$, then $x \prec y$.} \\
\text{If $x \sim y$ and $x$ is before $y$ in $p_m(K^1)$, then $x \prec y$.}
\end{array}
\end{equation}
The abbreviation \emph{$x$ is before $y$ in $p_m(K^j)$} means that
$x^*$ and $y^*$ are in different subsets of the ordered partition associated with $p_m(K^j)$ and 
the set containing $x^*$ is before the set containing $y^*$ in $p_m(K^j)$.

Our goal is to characterize all transitive orientations $(K,{\prec})$ of $(K,{\sim})$ that satisfy 
condition~\eqref{eq:restrictions_on_transitive_orientations}. 
We call such orientations \emph{admissible} for $\mathbb{K}_m$.

Every transitive orientation $(K,{\prec})$ of $(K,{\sim})$ can be 
obtained by an independent transitive orientation of the edges of every strong module in $\mathcal{M}(K,{\sim})$ -- see Theorem~\ref{thm:transitive_orientations_versus_transitive_orientations_of_strong_modules}.
Let $A$ be a strong module in $\mathcal{M}(K,{\sim})$ with the children $A_1,\ldots,A_k$.
Recall that the graph $(A,{\sim_A})$ is obtained from $(A,{\sim})$ by restricting to the edges
between different children of~$A$.

Suppose $A$ is a prime module in $\mathcal{M}(K,{\sim})$.
By Theorem \ref{thm:prime_graph_has_two_transitive_orientations}, $(A,{\sim_A})$ has two transitive orientations, ${\prec^{0}_A}$ and ${\prec^{1}_A}$, one being the reverse of the other.
An orientation ${\prec^{i}_{A}}$ of $(A,{\sim})$ is \emph{admissible} for $\mathbb{K}_m$
if $\prec^{i}_{A}$ satisfies conditions \eqref{eq:restrictions_on_transitive_orientations} for every two vertices $x,y$ from different children of $A$.
In particular, if $\inside(A) = \emptyset$ then both ${\prec^{0}_A}$ and ${\prec^{1}_A}$ are admissible.
Otherwise, only one among ${\prec^{0}_A}$ and ${\prec^{1}_A}$
is admissible for $\mathbb{K}_m$.
Indeed, since $\inside(A) \neq \emptyset$, there are $x,y \in A$ such that $x$ is before $y$ in 
$p_m(K^j)$ for some $K^j \in \{K^0,K^1\}$.
Since $(A,{\sim_A})$ is connected, there are $x',y' \in A$ such that $x' \sim_A y'$ and $x'$ is before $y'$ in
$p_m(K^{j})$.
It means that the orientation $(A,{\prec^{i}_A})$ for which $y' \prec^i_A x'$ is not admissible for $\mathbb{K}_m$.

Now, suppose $A$ is a serial module in $\mathcal{M}(K,{\sim})$.
The transitive orientations of $(A,{\sim_A})$ are in one-to-one correspondence 
with the linear orders of $A_1,\ldots,A_k$.
That is, every transitive orientation of $(A,{\sim_A})$ is of the form $A_{i_1} \prec \ldots \prec A_{i_k}$, where $i_1, \ldots, i_k$ is a permutation of $[k]$.
Now, we define $\prec^{wo}_A$ relation on the set $\{A_1,\ldots,A_k\}$,
where 
\begin{equation}
\label{def:weak_orders}
A_i \prec^{wo}_A A_j 
\iff
\begin{array}{c}
\text{there is $x \in A_i$ and $y \in A_j$ such that $x$ appears before $y$} \\
\text{in $p_m(K^j)$ for some $K^j \in \{K^0,K^1\}$.}
\end{array}
\end{equation}
We claim that $(\{A_1,\ldots,A_k\},{\prec^{wo}_A})$ is a partial order.
Suppose $A_i,A_j,A_l$ are such that $A_i \prec^{wo}_{A} A_j \prec^{wo}_A A_l$.
It means that there are $x \in A_i$, $y,y' \in A_j$, and $z \in A_l$ such that
$x$ is before $y$ in either $p_m(K^0)$ or $p_m(K^1)$ and 
$y'$ is before $z$ in either $p_m(K^0)$ or $p_m(K^1)$.
Since $x \sim \{y,y'\} \sim z$ as $A$ is serial, 
we deduce that $x$ is before $z$ in either $p_m(K^0)$ or $p_m(K^1)$.
This proves that $(\{A_1,\ldots,A_k\},{\prec^{wo}_A})$ is a partial order.

Furthermore, we claim that $(\{A_1,\ldots,A_k\},{\prec^{wo}_A})$ is a \emph{weak order}, 
which means that $(\{A_1,\ldots,A_k\},{\prec^{wo}_A})$ can be partitioned into the set of antichains 
such that every two of them induce a complete bipartite poset in $(\{A_1,\ldots,A_k\},{\prec^{wo}_A})$.
Clearly, $(\{A_1,\ldots,A_k\},{\prec^{wo}_A})$ is a weak order iff 
it does not have three elements $A_i,A_j,A_{l}$ such that 
$A_j \prec^{wo}_A A_{l}$ and $A_i$ is incomparable 
to $A_j$ and $A_l$.
Suppose there exists $A_i,A_j,A_{l}$ with such properties.
It means that there are $x \in A_j$ and $y \in A_{l}$ such that
$x$ is before $y$ in $p_m(K^j)$ for some $K^j \in \{K^0,K^1\}$.
Let $z$ be any vertex from $A_i$.
Clearly, $z$ is either before $y$ in $p_m(K^j)$ or $z$ is after $x$ in $p_m(K^j)$.
So, $A_i$ is $\prec^{wo}_A$-comparable to $A_j$ or $A_l$.
An orientation $A_{i_1} \prec_A \ldots \prec_A A_{i_k}$ of $(A,{\sim_A})$
is \emph{admissible for $\mathbb{K}_m$} if ${\prec_A}$ extends ${\prec^{wo}_A}$.
The result of this subsection are summarized by the following claim.
\begin{claim}
\label{claim:admissible_orientations_of_M}
There is a one-to-one correspondence between the set of admissible orientations $(K,{\prec})$ for $\mathbb{K}_m$ and the families 
$$\{(A,{\prec_A}): A \in \mathcal{M}(K,{\sim}) \text{ and $\prec_A$ is an admissible orientation of $(A, {\sim_A})$ for $\mathbb{K}_m$}\}$$
given by $x \prec y \iff x \prec_A y$, where $A$ is the module in $\mathcal{M}(K,{\sim})$ such that $x \sim_A y$.
\end{claim}

\subsection{Local isomorphisms between oriented modules}
Suppose that 
$$\mathbb{M} = (M^{0},M^{1},{<_M},p(M^{0}),p(M^{1})) \text{ and } \mathbb{N} = (N^{0},N^{1},{<_N},p(N^{0}),p(N^{1}))$$ are two extended metaedges of some consistent modules
$M$ and $N$ from $\mathcal{M}(G_{ov})$ and $\mathcal{M}(H_{ov})$, respectively.
Let $u \in M$. We say that $u$ occurs in $\mathbb{M}$ at position $(i,j)$, written 
$pos(u) = (i,j)$, if $u$ is in the $i$-th subset of the ordered partition of $p(M^{0})$ 
and in the $j$-th subset of the ordered partition of $p(M^{0})$.
The position of $u \in N$ is defined analogously.

\begin{definition}
Let $A$ be a strong module in $\mathcal{M}(M,{\sim})$ and $B$ be a strong module in $\mathcal{M}(N,{\sim})$.
We say that $A$ and $B$ are \emph{locally isomorphic}
if there is a bijection $\alpha: A \to B$ that satisfies for every $u,v \in M$:
\begin{enumerate}
 \item \label{def:loc_iso_left_right} $u <_M v$ iff $\alpha(u) <_N \alpha(v)$,
 \end{enumerate}
 and for every $u \in M$:
 \begin{enumerate}[resume]
 \item \label{def:loc_iso_orientation} $u$ is oriented from $M^{0}$ to $M^{1}$ iff $\alpha(u)$ is 
oriented from $N^{0}$ to $N^{1}$,
 \item \label{def:loc_iso_pos_mul} $pos(u) = pos(\alpha(u))$ and $m_G(u) = m_H(\alpha(u))$,
 \end{enumerate}
We say that $\mathbb{M}$ and $\mathbb{N}$ are \emph{locally isomorphic} if they are locally isomorphic on $M$ and~$N$.
\end{definition}
If $\alpha$ is a local isomorphism between $\mathbb{M}$ and $\mathbb{N}$, then $|p(M^0)| = |p(N^0)|$, and $|p(M^1)| = |p(N^1)|$.
in this case we can naturally extend $\alpha$ on the set $M \cup \inside(M)$: 
$\alpha$ maps the the subsequent nodes in $p(M^0)$ into the subsequent nodes in $p(N^0)$ and the subsequent nodes in $p(M^1)$ into the subsequent nodes
in $p(N^1)$.

Let $\alpha$ be a local isomorphism between $\mathbb{M}$ and $\mathbb{N}$
and let $(\tau^{0},\tau^{1})$ be an extended admissible model of $\mathbb{N}$.
Note that the image $(\alpha(\tau^0),\alpha(\tau^1))$ of $(\tau^{0},\tau^{1})$ 
by $\alpha$ is an extended admissible model of $\mathbb{N}$.

Let $A$ be a strong module in $\mathcal{M}(A,{\sim})$.
By $height(A)$ we denote the height of $A$ in the modular decomposition tree $\mathcal{M}(M,{\sim})$, that is,
the length of the longest path from $A$ to a leaf in $\mathcal{M}(M,{\sim})$.
We define $height(B)$ for every $B \in \mathbb{M}(N,{\sim})$ similarly.
Clearly, if $height(M) \neq height(N)$, then $\mathbb{M}$ and $\mathbb{N}$ can not be locally isomorphic.

Now, we provide an algorithm that tests whether $\mathbb{M}$ and $\mathbb{N}$ are locally isomorphic.
The algorithm traverses the trees $\mathcal{M}(M,{\sim})$ and $\mathcal{M}(N,{\sim})$ in the bottom-up order and for every $A \in \mathcal{M}(M,{\sim})$ and every $B \in \mathcal{M}(N,{\sim})$ such that $height(A) = height(B)$ it checks whether $A$ and $B$ are locally isomorphic.
If this is the case, the algorithm computes a local isomorphism $\alpha'_{AB}$ between $A$ and $B$. 
On the other hand, we show that if there is a local isomorphism $\alpha$ between $\mathbb{M}$ and $\mathbb{N}$, 
then the algorithm denotes $A$ and $\alpha(A)$ as locally isomorphic, where $A$ is any strong module in $\mathcal{M}(M,{\sim})$.
Hence, the algorithm accepts $\mathbb{M}$ and $\mathbb{N}$ iff they are locally isomorphic.

We distinguish the following cases: both $A,B$ are leaves,  both $A,B$ are parallel,
both $A,B$ are prime, and both $A,B$ are serial.
In the remaining cases, $A$ and $B$ can not be locally isomorphic.

\textbf{Case 1}: $A$ is a leaf in $\mathcal{M}(M,{\sim})$ and $B$ is a leaf in $\mathcal{M}(N,{\sim})$.
Let $\{u\} = A$ and $\{v\} = B$.
The algorithm denotes $A$ and $B$ as locally isomorphic if and only if
$m_G(u) = m_H(v)$, $u^{0} \in M^{0}$ iff $v^{0}$ in $M^{0}$, and $pos(u) = pos(v)$.
If this is the case, the algorithm sets $\alpha'_{AB}(u) = v$.

Suppose $A$ and $B$ are denoted as locally isomorphic. 
One can easily check that $\alpha'_{AB}$ satisfies properties \eqref{def:loc_iso_left_right}--\eqref{def:loc_iso_pos_mul}, and hence $\alpha'_{AB}$ is a local isomorphism between $A$ and $B$.

On the other hand, if $\mathbb{M}$ and $\mathbb{N}$ are isomorphic by $\alpha$, 
$\{\alpha(u)\}$ must be a leaf in $\mathcal{M}(N,{\sim})$, and $m_G(u) = m_H(\alpha(u))$, $u^{0} \in M^{0}$ iff $\alpha^{0}(u)$ in $M^{0}$, and $pos(u) = pos(\alpha(u))$.
Hence, the algorithm denotes $A=\{u\}$ and $\alpha(A) = \{\alpha(u)\}$ as locally isomorphic.

Now, assume that $A$ and $B$ are not leaves.
Clearly, if $A$ and $B$ have different number of children or have different types, then 
$A$ and $B$ can not be locally isomorphic.
We suppose $A_1,\ldots,A_k$ are the children of $A$ in $\mathcal{M}(M,{\sim})$ enumerated such that
$\tau^{0}|A_i$ appears before $\tau^{0}|A_j$ for every $i < j$, where $(\tau^0,\tau^1)$
is a fixed extended admissible model for $\mathbb{M}$.

\textbf{Case 2}: $A$ and $B$ are parallel.
Let $B_1,\ldots,B_{k}$ be the children of $B$ enumerated according to ${<_N}$, that is,
$B_{1} <_N \ldots <_N B_{k}$.
The algorithm denotes $A$ and $B$ as locally isomorphic if and only if
for every $i \in [k]$ the sets $A_i$ and $B_{i}$ have been denoted as locally isomorphic.
If this is the case, for every $i \in [k]$ and every $u \in A_i$ 
the algorithm sets $\alpha'_{AB}(u) =v$, where $v \in B_{i}$ is such that $\alpha'_{A_iB_{i}}(u) = v$.

Suppose $A$ and $B$ are denoted as locally isomorphic.
Note that $\alpha'_{AB}$ satisfies properties \eqref{def:loc_iso_orientation}-\eqref{def:loc_iso_pos_mul} as the corresponding properties are satisfied by 
$\alpha'_{A_iB_{i}}$ for every $i \in [k]$.
Suppose $u_1, u_2$ in $A$ are such that $u_1 <_M u_2$.
If $u_1,u_2 \in A_i$ for some $i \in [k]$, then
$\alpha'_{AB}(u_1) <_N \alpha'_{AB}(u_2)$. 
Indeed, $\alpha'_{AB}(u_1) <_N \alpha'_{AB}(u_2)$ iff $\alpha'_{A_iB_i}(u_1) <_N \alpha'_{A_iB_i}(u_2)$
by definition of $\alpha'_{AB}$ and $\alpha'_{A_iB_i}(u_1) <_N \alpha'_{A_iB_i}(u_2)$ as $\alpha'_{A_i,B_i}$ satisfies \eqref{def:loc_iso_left_right}.
If $u_1 \in A_i$ and $u_{2} \in A_{j}$ for some $i < j$, then
$\alpha'_{AB}(u_1) <_N \alpha'_{AB}(u_2)$ as $\alpha'_{AB}(u_1) \in B_{i}$, $\alpha'_{AB}(u_2) \in B_{j}$, and $B_{i} <_N B_{j}$.
This proves $\alpha'_{AB}$ satisfies \eqref{def:loc_iso_left_right}.

On the other hand, suppose $\mathbb{M}$ and $\mathbb{N}$ are locally isomorphic by $\alpha$.
From inductive hypothesis, for every $i \in [k]$ the sets $A_i$ and $\alpha(A_i)$ have been marked as locally isomorphic.
Since $\alpha(A_1) <_N \ldots <_N \alpha(A_k)$ by \eqref{def:loc_iso_left_right}, 
the algorithm denotes $A$ and $\alpha(A)$ as locally isomorphic.

\textbf{Case 3}: $A$ and $B$ are prime.
For every (at most two) admissible orientations $(B,{\prec_B})$ of $(B,{\sim_B})$
the algorithm does the following.
First, it computes the order $B_1,\ldots,B_k$ of the children of $B$
in which $B_i$ is before $B_j$ iff $B_{i} <_N B_{j}$ or $B_{i} \prec_B B_j$.
The algorithm denotes $A$ and $B$ as locally isomorphic if for every $i \in [k]$
the sets $A_i$ and $B_{i}$ have been denoted as locally isomorphic.
If this is the case, for every $i \in [k]$ and every $u \in A_i$ 
the algorithm sets $\alpha'_{AB}(u) =v$, where $v \in B_{i}$ is such that $\alpha'_{A_iB_{i}}(u) = v$.

Suppose $A$ and $B$ have been marked as locally isomorphic
at the time when an admissible orientation ${\prec_B}$ of $(B,{\sim_B})$ was processed.
Note that $\alpha'_{AB}$ satisfies properties \eqref{def:loc_iso_orientation}-\eqref{def:loc_iso_pos_mul} as the corresponding properties are satisfied by 
$\alpha'_{A_iB_{i}}$ for every $i \in [k]$.
Suppose $u_1,u_2 \in M$ are such that $ u_1 < u_2$.
If $u_1,u_2 \in A_i$ for some $i \in [k]$, then 
$\alpha'_{AB}(u_1) < \alpha'_{AB}(u_2)$ as $\alpha'_{A_iB_i}(u_1) < \alpha'_{A_iB_i}(u_2)$.
Suppose $u_1 \in A_i$ and $u_j$ in $A_j$.
Then $A_i \parallel A_j$.
Note that the numbering of $B_1,\ldots,B_k$ asserts that
$A_i$ is before $A_j$ in $A_1, \ldots, A_k$ iff $B_i <_N B_j$.
This proves \eqref{def:loc_iso_left_right}.

On the other hand, suppose that $\mathbb{M}$ and $\mathbb{N}$ are locally isomorphic by $\alpha$.
Since $\alpha$ is a local isomorphism between $\mathbb{M}$ and $\mathbb{N}$, 
the image $(\alpha(\tau^{0}), \alpha(\tau^1))$ of $(\tau^0,\tau^1)$ is an extended admissible model of $\mathbb{N}$.
Then, $(\alpha(\tau^{0}), \alpha(\tau^1))$ corresponds to 
transitive orientations $({<_N},{\prec})$ of $(N,{\parallel})$ and $(N,{\sim})$, respectively, where $\prec$
is an admissible orientation for $\mathbb{N}$.
Let $\prec_{\alpha(A)}$ be the restriction of ${\prec}$ to the edges of $(\alpha(N),{\sim}_{\alpha(N)})$.
Now, when processing the admissible orientation $(A, \prec_{\alpha(A)})$ for $\mathbb{N}$, 
the algorithm orders the children of $\alpha(A)$ in the order $\alpha(A_1),\ldots,\alpha(A_k)$.
Since for every $i \in [k]$ the sets $A_i$ and $\alpha(A_i)$ have been marked as locally isomorphic,
the algorithm will denote $A$ and $\alpha(A)$ as locally isomorphic.

\textbf{Case 4:} $A$ and $B$ are serial.
Let $(A,{\prec^{wo}_A})$ and $(B,{\prec^{wo}_B})$ be the weak orders associated with the serial modules $A$
and $B$ defined by \eqref{def:weak_orders}.
For every child $A'$ of $A$ denote by $height(A')$ the height of $A'$ in $(A,{\prec^{wo}_{A}})$, that is, the size of the chain in $(A,{\prec^{wo}_A})$ which has $A'$ as its largest element.
Define $hight(B')$ for every child $B'$ of $B$ similarly.
Let $(\mathcal{A},\mathcal{B}, \mathcal{E}_{AB})$ be a bipartite graph,
where:
\begin{itemize}
 \item $\mathcal{A}$ is the set of the children $A'$ of $A$,
 \item $\mathcal{B}$ is the set of the children $B'$ of $B$,
 \item $\mathcal{E}_{AB}$ is the set of all pairs $(A',B')$ such that $A'$ and $B'$ have been denoted as locally isomorphic and $height(A') = height(B')$.
\end{itemize}
The algorithm marks $A$ and $B$ as locally isomorphic iff there is a perfect matching 
$\mathcal{M}$ in the bipartite graph $(\mathcal{A},\mathcal{B}, \mathcal{E}_{AB})$.
If such a matching exists, for every $(A', B') \in \mathcal{M}$ and every $u \in A'$ the algorithm sets $\alpha'_{AB}(u) = v$, where $v \in B'$ is such that $\alpha'_{A'B'}(u) = v$.

Suppose $A$ and $B$ have been marked as locally isomorphic.
Since for every $u_1,u_2 \in A$, $u_1 <_M u_2$ yields $u_1,u_2 \in A'$ for some child $A'$ of $A$ as $A$ is serial, we deduce that $\alpha'_{AB}$ satisfies \eqref{def:loc_iso_left_right}
as $\alpha'_{A',B'}$ satisfies \eqref{def:loc_iso_left_right} for every $(A',B') \in \mathcal{M}$.

On the other hand, suppose $\mathbb{M}$ and $\mathbb{N}$ are locally isomorphic by $\alpha$.
Since $\alpha$ preserves positions, we deduce that $height(A_i)$ in $(A,{\prec^{wo}_A})$ equals $height(\alpha(A')$ in $(\alpha(A),{\prec^{wo}_{\alpha(A)}})$.
In particular, for every $A'$ the pair $(A',\alpha(A'))$ is in the edge set of the bipartite graph
$(\mathcal{A}, B, \mathcal{E}_{AB})$.
In particular, $\{(A',\alpha(A')): A' \in \mathcal{A}\}$ establishes a perfect matching in
$(\mathcal{A}, B, \mathcal{E}_{AB})$.
Hence, the algorithm denotes $A$ and $\alpha(A)$ as locally isomorphic.

\subsection{Isomorphisms between modules with slots}
Suppose $M$ and $N$ are two modules in $G_{ov}$ and $H_{ov}$ for which 
the circular permutations of the slots $\pi_{0}(M), \pi_{1}(M)$ and $\pi_0(N),\pi_1(N)$ are defined.
Let $K_i$ be a consistent module in $M$.
For convenience, we assume that for every $m \in \{0,1\}$ 
there are two extended metaedges associated with the consistent submodule $K_i$,
$$\mathbb{K}^0_i = (K^0_i,K^1_i,{<^0_{K_i}},p_m(K^0_i),p_m(K^1_i)) \text{ and its dual }\mathbb{K}^1_i = (K^1_i,K^0_i,{<^1_{K_i}},p_m(K^1_i),p_m(K^0_i)),$$
where $<^0_{K_i} = <_{K_i}$ and $<^{1}_{i}$ is the reverse of $<^{0}_{K_i}$.
That is, $(\tau^0,\tau^1)$ is an admissible model for $\mathbb{K}^0_i$ iff
$(\tau^1,\tau^0)$ is an admissible model for $\mathbb{K}^1_i$.

Let $\pi(M)$ be an element in $\{\pi_{0}(M), \pi_{1}(M)\}$ and
$\pi(N)$ be an element in $\{\pi_{0}(M), \pi_{1}(M)\}$.
We choose a slot $K'$ in $\pi(M)$ and a slot $N'$ in $\pi(N)$ --
we say $\pi(M)$ \emph{is pinned} in $K'$ and $\pi(N)$ \emph{is pinned} in $N'$.
Let $K'_{1} \ldots K'_{2n}$ be the order of the slots in $\pi(M)$ if we traverse
$\pi(M)$ in the clockwise order starting from $K'=K'_1$.
Let $K_i$ be the consistent submodule of $M$ associated with $K'_i$, for every $i \in [2n]$.
That is, for every consistent submodule $K$ of $M$ there are different $i$ and $j$ such that $K = K_i =K_j$.
Let $\mathbb{K}'_i = (K'_i,K''_i,{<'_{K_i}},p(K'_i),p(K''_i))$ be an (extended) orientation
of $K_i$, where $p(K'_i),p(K''_i)$ are appropriate for the chosen $\pi(M)$.
We introduce similar notation for the slots of $N$.
\begin{definition}
Suppose $M$, $N$, $\pi(M)$, $\pi(N)$, $K'$, $L'$ are as given above.
We say \emph{$\pi(M)$ pinned in $K'$ is isomorphic to $\pi(N)$ pinned in $L'$} if 
there is a bijection $\alpha$ from $M \cup N_{T_G}[M]$ to $N \cup N_{T_H}[N]$ such that:
\begin{itemize}
 \item $\alpha|K_i$ establishes a local isomorphism between $\mathbb{K}'_i$ and 
 $\mathbb{L}'_i$, for every $i \in [2n]$,
 \item $\alpha|N_{T_G}(M)$ is a bijection from 
 $N_{T_G}(M)$ to $N_{T_H}(N)$ and $\alpha$ maps every node in $p(K'_i,K'_{i+1})$ into a node in $p(L'_i,L'_{i+1})$, cyclically.
\end{itemize}
\end{definition}
Note that there may exist many isomorphisms between $\pi(M)$ and $\pi(N)$ pinned in $K'$
and $L'$, respectively.
However, their restrictions to the set $N_{T_G}[M]$ are unique and are determined by the patterns of the slots in $\pi(M)$ and $\pi(N)$.

Clearly, an isomorphism between $\pi(M)$ and $\pi(N)$ can be computed by the algorithm 
presented in the previous section.

\begin{claim}
\label{claim:isomorphism_betwen_admissible_models}
Let $\alpha:M \cup N_{T_G}[M] \to N \cup N_{T_H}[N]$ be an isomorphism between $\pi(M)$ and $\pi(N)$ pinned in $K'$ and $L'$, respectively.
Then $\alpha$ satisfies the following properties for every $u,v \in M$ and every $N' \in N_{T_G}[M]$:
\begin{enumerate}
\item $u \in \leftside(v)$ iff $\alpha(u) \in \leftside(\alpha(v))$,
\item $u \in \rightside(v)$ iff $\alpha(u) \in \rightside(\alpha(v))$,
\item $N' \in \leftside(v)$ iff $\alpha(N') \in \leftside(\alpha(v))$.
\end{enumerate} 
In particular, if $\phi^M$ is an extended admissible model of $(M,{\sim})$ for $\pi(M)$, then the image $\alpha(\phi^M)$ is an extended admissible model of $(N,{\sim})$ for $\pi(N)$.
\end{claim}
\begin{proof}
If $u$ and $v$ belong to the same consistent submodule of $M$,
then $u \in \leftside(v)$ iff $\alpha(u) \in \leftside(\alpha(v))$ 
as $\alpha|K_i$ is a local isomorphism on $K_i$ and $\alpha(K_i)$.
Suppose $u$ and $v$ belong to two different submodules of $M$.
For every $u' \in \{u^{0},u^{1}\}$ let $i(u')$ be such that $u' \in K'_{i(u')}$.
Define $i(v')$ similarly.
Since $u \in \leftside(v)$, we have 
$$\pi(M)|\{K'_{i(u^{0})},K'_{i(u^{1})},K'_{i(v^{0})},K'_{i(v^{1})} \} \equiv K'_{i(v^{0})}K'_{i(u')}K'_{i(u'')}K'_{i(v^{1})}$$
for some $\{u',u''\} = \{u^{0},u^{1}\}$.
However, we have $\alpha^{0}(v) \in L_{i(v^{0})}$, $\alpha^{1}(v) \in L_{i(v^{1})}$, 
$\alpha^{0}(u) \in L_{i(u^{0})}$, $\alpha^{1}(u) \in L_{i(u^{1})}$ as local isomorphisms preserve orientations of the vertices, and hence
$$\pi(N)|\{L'_{i(u^{0})},L'_{i(u^{1})},L'_{i(v^{0})},L'_{i(v^{1})} \} \equiv L'_{i(v^{0})}L'_{i(u')}L'_{i(u'')}L'_{i(v^{1})}.$$
This proves $u \in \leftside(v)$ iff $\alpha(u) \in \leftside(\alpha(v))$.
The remaining statements of the claim are proved analogously.
\end{proof}

Now, we are ready to provide an isomorphism algorithm testing whether two circular-arc graphs
$(G,m_G)$ and $(H,m_H)$ are isomorphic. 
As in the previous sections, we consider three cases:
\begin{itemize}
 \item $V(G_{ov})$ and $V(H_{ov})$ are serial.
 \item $V(G_{ov})$ and $V({H_{ov}})$ are prime.
 \item $V(G_{ov})$ and $V(H_{ov})$ are parallel.
\end{itemize}
In the remaining cases, $(G,m_G)$ and $(H,m_H)$ can not be isomorphic.

\subsection{$V(G_{ov})$ and $V(H_{ov})$ are serial}
Suppose $(G,m_G)$ and $(H,m_H)$ are two circular-arc graphs such that both $V(G_{ov})$ and $V(H_{ov})$ are serial.
Suppose that:
\begin{itemize}
 \item $M_1,\ldots,M_k$ are the children of $V(G_{ov})$ in $\mathcal{M}(G_{ov})$,
 \item $N_1,\ldots,N_k$ are the children of $V(H_{ov})$ in $\mathcal{M}(H_{ov})$,
 \item $\mathbb{M}^0_i = (M^{0}_i,M^{1}_i,{<^{0}_{M_i}})$ and $\mathbb{M}_i^1 = (M^{1}_i,M^{0}_i,{<^{1}_{M_i}})$ are the metaedges of $M_i$, for $i \in [k]$,
 \item $\mathbb{N}^0_i = (N^{0}_i,N^{1}_i,{<^0_{N_i}})$ and $\mathbb{N}^1_i = (N^{1}_i,N^{0}_i,{<^1_{N_i}})$ are the metaedges of $N_i$, for $i \in [k]$.
\end{itemize}

The algorithm constructs a bipartite graph $G_{MN}$ between the modules $\{M_1,\ldots,M_k\}$ and $N_1,\ldots,N_k$, where there is an edge in $G_{MN}$ between $M_i$ and $N_j$ iff
there are $\mathbb{M}'_i \in \{\mathbb{M}^{0}_i,\mathbb{M}^{1}_i\}$ and $\mathbb{N}'_j \in \{\mathbb{N}^{0}_j,\mathbb{N}^{1}_j\}$ such that $\mathbb{M}'_i$ and $\mathbb{N}'_j$ are locally isomorphic.
The algorithm accepts $(G,m_G)$ and $(H,m_H)$ iff there is a perfect matching $\mathcal{M}$
in $G_{MN}$.

We claim that the algorithm accepts $(G,m_G)$ and $(H,m_H)$ iff $(G,m_G)$ and $(H,m_H)$ are isomorphic.
Suppose $\alpha$ is an isomorphism between $(G,m_G)$
and $(H, m_H)$.
Clearly, $\alpha(M_1),\ldots,\alpha(M_k)$ is a permutation of $N_1,\ldots,N_k$.
By Theorem \ref{thm:description_of_all_conformal_models_of_serial_modules}, the image $N'_i$ of $M^{0}_i$ by $\alpha$ satisfies either $N'_i = N^{0}_i$ or $N'_i = N^{1}_i$.
Now, $\mathbb{M}^{0}_i = (M^{0}_i,M^{1}_i,{<^{0}_{M_i}})$ is locally isomorphic with $\mathbb{N}'_i = (N'_i,N''_i,{<'_{N_i}})$.
So, $(M_i,\alpha(M_i))$ is an edge of $G_{MN}$ and $\{(M_i,\alpha(M_i)): i \in [k]\}$ is a perfect matching in $G_{MN}$.
So, the algorithm accepts $(G,m_G)$ and $(H,m_H)$.

Now, suppose that the algorithm accepts $(G,m_G)$ and $(H,m_H)$.
Suppose $\mathcal{M}$ is a matching between $M_1,\ldots,M_k$
and $N_1,\ldots,N_k$ in $G_{MN}$.
Without loss of generality suppose that the children of $V(H_{ov})$ are enumerated such that $(M_i,N_i) \in \mathcal{M}$ for every $i \in [k]$.
Let $\alpha_i$ be a local isomorphism between $\mathbb{M}'_i$ and 
$\mathbb{N}'_i$, where $\mathbb{M}'_i \in \{\mathbb{M}^{0}_i,\mathbb{M}^{1}_i\}$ and $\mathbb{N}'_i \in \{\mathbb{N}^{0}_i,\mathbb{N}^{1}_i\}$.
Let $\alpha:V(G_{ov}) \to V(H_{ov})$ be a mapping such that $\alpha|M_i = \alpha_i$.
Now, given $\alpha_i$ is a local isomorphism between $\mathbb{M}'_i$ and 
$\mathbb{N}'_i$, one can easily verify that $\alpha$ is an isomorphism between 
$(G,m_G)$ and $(H_{ov},m_H)$.

\subsection{$V(G_{ov})$ and $V(H_{ov})$ are prime}
Suppose $(G,m_G)$ and $(H,m_H)$ are two circular-arc graphs
such that both $V(G_{ov})$ and $V(H_{ov})$ are prime.
Suppose that $\pi_{0}(V(G_{ov}))$, $\pi_{1}(V(G_{ov}))$ and $\pi_{0}(V(H_{ov}))$, $\pi_{1}(V(H_{ov}))$
are circular permutations of the slots of $V(G_{ov})$ and $V(H_{ov})$, respectively.

The algorithm iterates over circular orders $\pi(V(G))$ in $\{\pi_{0}(V(G)), \pi_{1}(V(G))\}$
and $\pi(V(H))$ in $\{\pi_{0}(V(H)), \pi_{1}(V(H))\}$.
For every pair $(\pi(V(G)), \pi(V(H)))$ the algorithm fixes a slot $K'$ in $\pi(V(G))$.
Next, it iterates over all slots $L'$ in $\pi(V(H))$ and 
checks whether $\pi(V(G))$ and $\pi(V(H))$ pinned in $K'$ and $L'$ are isomorphic.
It accepts $(G,m_G)$ and $(H,m_h)$ iff for some choice of $\pi(V(G))$, $\pi(V(H))$, and $L'$
the algorithm finds out that $\pi(V(G))$ and $\pi(V(H))$ pinned in $K'$ and $L'$ are isomorphic.

If the algorithm accepts $(G,m_G)$ and $(H,m_H)$, then $(G,m_G)$ and $(H,m_H)$ are isomorphic,
which follows immediately from Claim \ref{claim:isomorphism_betwen_admissible_models}.
Suppose $\alpha$ is an isomorphism between $(G_{ov},m_G)$ and $(H_{ov},m_H)$.
Let $\phi$ be a conformal model of $G_{ov}$. 
The image $\alpha(\phi)$ of $\phi$ by $\alpha$ is a conformal model of $H_{ov}$.
Let $\pi(V(G)) = \pi_{\phi}(V(G))$, $\pi(V(H)) = \pi_{\alpha(\phi)}(V(H))$, 
$K'$ be a slot chosen by the algorithm 
at the time when it processes $\pi(V(G))$ and $\pi(V(H))$, and let $\alpha(K') = \{\alpha(u'): u' \in K'\}$ be the image of the slot $K'$.
Now, note that the algorithm accepts $(G,m_G)$ and $(H,m_H)$ when it processes $\pi(V(G))$, $\pi(V(H))$, and $\alpha(K')$.

\subsection{$V(G_{ov})$ and $V(H_{ov})$ are parallel}

Suppose that:
\begin{itemize}
\item $T_{G}$ and $T_H$ are $T_{NM}$ trees for $G_{ov}$ and $H_{ov}$, respectively,
\item $M_1,\ldots,M_k$ are the children of $V(G_{ov})$ and 
$K_1,\ldots,K_k$ are the children of $V(H_{ov})$.
\end{itemize}

Suppose $T_G$ is rooted in $R$, where $R$ is any leaf module in $T_G$.
For every module (node) $A$ in the rooted tree $T_G$, 
by $V_{T_G}[A]$ we denote the union of all modules descending $A$ in $T_{G}$, including $A$.
We introduce the similar notation for every module (node) $B$ in the rooted tree $T_H$.

Let $\phi$ be a conformal model of $G$.
Let $\alpha$ be an isomorphism between $G_{ov}$ and $H_{ov}$ and
let $\alpha(\phi)$ be the image of $\phi$ by $\alpha$.
Note that $\alpha$ maps every module $M$ in $T_{G}$ into a
module $\alpha(M)$ in $H_{ov}$.
In particular, the root $R$ of $T_{G}$ is mapped into a leaf module $\alpha(R)$ in $T_H$.
Let $\alpha(R)$ be the root of $T_H$.
Note that the children of every module $M$ (node $N$) are mapped into
the children of $\alpha(M)$ ($\alpha(N)$, respectively).
By Claim \ref{claim:conformal_models_properties_modules_nodes_T_NM}, 
for every node/module $A$ in $T_{G}$, 
$\phi|V_{T_G}[A]$ is a contiguous subword of $\phi$.
We denote this subword by $\phi'_A$ 
(recall the proof of Theorem \ref{thm:description_of_all_conformal_models_of_improper_parallel_modules}, 
where the word $\phi'_A$ was constructed from
extended admissible models $\phi|(M \cup N_T[M])$ for modules $M \in T_{G}$
and from circular permutations $\phi|(N \cup N_T[N])$ for nodes $N$ in $T_G$).
Similarly, the image $\alpha(\phi'_A)$ of $\phi'_A$ by $\alpha$ 
is a contiguous subword of $\alpha(\phi)$.

Now, we present the algorithm that tests whether $(G,m_G)$ and $(H,m_H)$ are isomorphic.
The algorithm picks a leaf module $R$ in $T_G$ arbitrarily and sets $R$ as the root of $T_{G}$.
Next, the algorithm iterates over all leaves $R'$ of $T_H$, sets $R'$ as the root of $T_H$, and does the following.
It traverses the trees $T_G$ and $T_H$ bottom-up
and for every two nodes $A \in T_G$ and $B \in T_H$ 
or two modules $A \in T_G$ and $B \in T_H$  
it tests whether there is a bijection $\alpha:V_{T_G}[A] \to V_{T_H}[B]$ that satisfies the following conditions:
\begin{itemize}
\item $m_G(u) = m_H(\alpha(u))$ for every $u \in V_{T_G}[A]$,
\item the image of $\phi'_A$ by $\alpha$ is a contiguous subword of some conformal model of $H_{ov}$.
\end{itemize}
We call such a mapping $\alpha$ \emph{an isomorphism between $A$ and $B$ in the rooted trees $T_G$ and~$T_H$}.

The algorithm accepts $(G,m_G)$ and $(H,m_H)$ if and only if there is a leaf module $R'$ in $T_H$ 
such that $R$ and $R'$ are isomorphic in the trees $T_G$ and $T_H$ rooted in $R$ and $R'$, respectively.

Suppose the algorithm accepts $(G,m_G)$ and $(H,m_H)$.
Note that an isomorphism $\alpha$ between $R$ and $R'$ in the rooted trees $T_G$ and $T_H$
maps bijectively the vertices from $V(G_{ov})$ into $V(H_{ov})$, preserves multiplicities,
and the image of $\phi_R$ by $\alpha$ is a conformal model of $H_{ov}$.
In particular, it means that $\alpha$ is an isomorphism between $G_{ov}$ and $H_{ov}$.

On the other hand, since $\alpha$ is an isomorphism between $G_{ov}$ and $H_{ov}$, 
$\alpha$ is an isomorphism between $R$ and $\alpha(R)$ in the trees $T_G$ and $T_H$ rooted in $R$ and $\alpha(R)$, and hence the algorithm accepts $(G,m_G)$ and $(H,m_H)$.

It remains to show how the algorithm tests whether $A$ and $B$ are isomorphic in the rooted trees $T_G$ and $T_H$.
Whenever the algorithm denotes $A$ and $B$ as isomorphic, it constructs
an isomorphism $\alpha'_{AB}$ between $A$ and $B$ witnessing its answer.
We also show that the algorithm denotes $A$ and $\alpha(A)$ as isomorphic.

Suppose that $A$ and $B$ are nodes in $T_G$ and $T_H$, respectively.
Suppose $A_1,\ldots,A_k$ are the children of $A$ in $T_G$ and $B_1,\ldots,B_l$
are the children of $B$ in $T_H$.
We define a bipartite graph $G_{AB}$ on $\{A_1,\ldots,A_k\}$ 
and $\{B_1,\ldots,B_l\}$ with the edge between $A_i$ and $B_j$ iff
$A_i$ and $B_j$ have been denoted as isomorphic.
The algorithm marks $A$ and $B$ as isomorphic iff 
there is a perfect matching $\mathcal{M}$ between $\{A_1,\ldots,A_k\}$ and 
$\{B_1,\ldots,B_l\}$ in $G_{AB}$.
If this is the case, for every $(A',B') \in \mathcal{M}$
and every $u \in A'$ the algorithm sets $\alpha'_{AB}(u) = v$, where $v \in B'$ is such that 
${\alpha'_{A'B'}(u) = v}$.

Suppose the algorithm denotes $A$ and $B$ as isomorphic.
Clearly, $\alpha'_{AB}$ preserves multiplicities.
By Theorem \ref{thm:description_of_all_conformal_models_of_improper_parallel_modules},
$\phi'_A  = \phi'_{A'_1}\ldots \phi'_{A'_k}$ for some 
permutation $A'_1,\ldots,A'_k$ of $A_1,\ldots,A_k$.
Suppose $B'_1,\ldots,B'_k$ is a permutation of the children of $B$ such that $(A'_i,B'_i) \in \mathcal{M}$.
Hence, the image $\tau'_{B'_i}$ of $\phi'_{A'_i}$ by $\alpha'_{A'_iB'_i}$ is a contiguous subword of some conformal model of $H_{ov}$, for every $i \in [k]$.
By Theorem \ref{thm:description_of_all_conformal_models_of_improper_parallel_modules}, the image $\tau'_{B'_1} \ldots \tau'_{B'_k}$ of $\phi'_A$ by $\alpha'_{AB}$ is a contiguous subword of a conformal model of $H_{ov}$.
So, $\alpha'_{AB}$ is an isomorphism between $A$ and $B$.

Suppose $G_{ov}$ and $H_{ov}$ are isomorphic.
For every child $A'$ of $A$ algorithm denoted $A'$ and $\alpha(A')$ as isomorphic.
Note that $\alpha(A_1),\ldots,\alpha(A_k)$ is the set of all the children of $B$.
Hence, $\{(A_i, \alpha(A_i)): i \in [k]\}$ is a perfect matching in $G_{AB}$.
Hence, the algorithm denotes $A$ and $\alpha(A)$ as isomorphic in $T_G$ and $T_H$.

Suppose $A$ and $B$ are leaves in $T_G$ and $T_H$, respectively.
In particular, $A$ and $B$ are modules in $T_G$ and $T_H$.
Let $p(A)$ be the parent node of $A$ in $T_G$.
Let $p(B)$ be the parent of $B$ in $T_H$.
Note that $\phi|(A \cup N_{T_G}[A])$ is an extended admissible model of $(A,{\sim})$.
Let $\pi_{m}(A)$ be the circular order of the slots in $\phi|A$, for some $m \in \{0,1\}$.
Let $S'$ be a slot in $\pi_{m}(A)$ such that
$p(A) \in p_{m}(S')$ or $p(A) \in p_{m}(S',S'')$ for some two consecutive slots $(S',S'')$ in $\pi_{m}(A)$.
Let $\pi_{0}(B)$ and $\pi_{1}(B)$ be two circular orders of the slots of $B$ in $H_{ov}$.
For every $\pi(B) \in \{\pi_{0}(B), \pi_{1}(B)\}$ the algorithm does the following.
First, it chooses a slot $T'$ in $\pi(B)$ such that
either $p(B) \in p(T')$ or 
$p(B) \in p(T',T'')$, where $T''$ is such that $(T',T'')$ are consecutive slots in $\pi(B)$ and $p(T')$ and $p(T',T'')$ are the patterns from $\pi(B)$.
Next, the algorithm checks whether $\pi_{m}(A)$ and 
$\pi(B)$ pinned in $S'$ and $T'$ are isomorphic.
Suppose $\alpha_{AB}$ is such an isomorphism.
Note that $\alpha_{AB}$ needs to map $p(A)$ into $p(B)$.
The algorithm denotes $A$ and $B$ as isomorphic, and returns $\alpha_{AB}|A$ as 
an isomorphism $\alpha'_{AB}$ between $A$ and $B$.

Suppose $A$ and $B$ are denoted as isomorphic.
Note that $A\phi'_A$ is an extended admissible model of $(A,{\sim})$.
Note that $\alpha_{AB}$ satisfies Claim \ref{claim:isomorphism_betwen_admissible_models}
and hence $\alpha_{AB}(A\phi_A)$ is an extended admissible model of $(B,{\sim})$.
In particular, by Theorem~\ref{thm:description_of_all_conformal_models_of_improper_parallel_modules}, 
the image of the word $\phi'_A$ by $\alpha'_{AB}$ can be extended to a conformal model of $H_{ov}$.
This shows that $\alpha'_{AB}$ is a local isomorphism between $A$ and $B$.

Suppose $G_{ov}$ and $H_{ov}$ are isomorphic.
Clearly, $\alpha(A)$ is a leaf in $T_H$.
Let $\pi_{m'}(\alpha(A))$ be the order of the slots of $\alpha(A)$ in the admissible model $\alpha(\phi)|\alpha(A)$.
Note that when the algorithm processes the circular order $\pi_{m'}$ of $\alpha(A)$, 
it will denote $A$ and $\alpha(A)$ as isomorphic.

Eventually, suppose $A$ and $B$ are modules in $T_G$ and $T_H$, respectively.
In this case the algorithm is similar as for the leaves, with one exception.
Suppose $\alpha_{AB}: A \cup N_{T_G}[A] \to B \cup N_{T_G}[B]$ is an isomorphism between 
$\pi_{m}(A)$ and $\pi(B)$ pinned in $S'$ and $T'$, where $m$, $\pi(B)$, $S'$, and $T'$ are defined as for leaves.
First, the algorithm checks whether $\alpha_{AB}$ maps the parent $p(M)$ of
$A$ into the parent $p(B)$ of $B$ (if the parents exist).
Moreover, for every node $N \in N_{T_G}[A]$ different than $p(A)$, the algorithm checks 
whether $N$ and $\alpha_{AB}(N)$ have been denoted as isomorphic.
If this is the case, the algorithm accepts $A$ and $B$.
The mapping $\alpha'_{AB}$ is defined as follows. 
If $u \in A$ we set $\alpha'_{AB}(u) = v$, where $v \in B$ is such that $\alpha_{AB}(u) = v$
and for every child $N$ of $A$ and every $u \in V_{T_G}[N]$ the algorithm sets $\alpha_{AB}(u) = v$, where $v \in V_{T_H}(\alpha_{AB}(N))$ is such that $\alpha'_{N\alpha_{AB}(N)}(u) = v$.

Suppose $A$ and $B$ are denoted as isomorphic.
Let $\phi^A \equiv A\phi_A|(A \cup N_{T_G}(A))$.
Note that $\phi^A$ is an extended admissible model for $A$.
By Claim \ref{claim:isomorphism_betwen_admissible_models}, 
the image of $\phi^A$ by $\alpha_{AB}$ is an extended admissible model for $B$.
Now, replace in $\alpha_{AB}(\phi^A)$ every child $N'$ of $B$ different than $p(B)$
by the word $\alpha'_{AB}(\phi'_N)$, where $N$ in $N_{T_G}[A]$ is such that $\alpha_{AB}(N) = N'$.
Note that we obtain the word $p(B)\alpha'_{AB}(\phi'_A)$.
Since $\alpha'_{AB}(\phi'_N)$ is a subword of a conformal model of $H_{ov}$,
$\alpha_{AB}(\phi^A)$ is an admissible model of $B$,
we deduce from Theorem~\ref{thm:description_of_all_conformal_models_of_improper_parallel_modules} that
$\alpha'_{AB}(\phi'_A)$ is a contiguous subword of some conformal model of $H_{ov}$.
Hence, $\alpha'_{AB}$ establishes an isomorphism between $A$ and $B$.

Suppose $G_{ov}$ and $H_{ov}$ are isomorphic.
We prove that the algorithm will denote $A$ and $\alpha(A)$ as isomorphic 
in the same way as for the leaves.

\bibliographystyle{plain}

\bibliography{lit}

\end{document}